\documentclass[a4paper]{article}
\usepackage{enumitem}
\usepackage{xcolor}
\usepackage{upgreek}

\usepackage{amsmath,amssymb,amsthm,mathtools}
\usepackage{tikz,graphicx}
\usepackage{tkz-euclide}
\usepackage{pict2e}
\usepackage{todonotes}
\usetkzobj{all}
\theoremstyle{plain}

\newtheorem*{theorem*}{Theorem}
\newtheorem{theorem}{Theorem}[section] 
\newtheorem{lemma}[theorem]{Lemma}
\newtheorem{proposition}[theorem]{Proposition}
\newtheorem{corollary}[theorem]{Corollary}

\theoremstyle{definition}
\newtheorem{definition}[theorem]{Definition}
\newtheorem{remark}[theorem]{Remark}

\usepackage{tikz}
\usepackage{pict2e}
\usetikzlibrary{arrows,calc,chains, positioning, shapes.geometric,shapes.symbols,decorations.markings,arrows.meta}
\usetikzlibrary{patterns,angles,quotes}
\usetikzlibrary{decorations.markings}
\tikzset{middlearrow/.style={
			decoration={markings,
				mark= at position 0.6 with {\arrow{#1}} ,
			},
			postaction={decorate}
		}
	}
\tikzset{->-/.style={decoration={
				markings,
				mark=at position #1 with {\arrow{latex}}},postaction={decorate}}}
	
	\tikzset{-<-/.style={decoration={
				markings,
				mark=at position #1 with {\arrowreversed{latex}}},postaction={decorate}}}

\newcommand{\ds}{\displaystyle}

\numberwithin{equation}{section}

\DeclareMathOperator{\supp}{supp}
\def\bigO{{\cal O}}

\newcommand{\lozr}{
	--++(1,1)--++(0,1)--++(-1,-1)
	--++(0,-1)
}

\newcommand{\lozd}{--++(1,1)--++(1,0)--++(-1,-1)--++(-1,0)
}

\newcommand{\lozu}{--++(1,0)--++(0,1)--++(-1,0)--++(0,-1)
}

\tikzset{
	master/.style={
		execute at end picture={
			\coordinate (lower right) at (current bounding box.south east);
			\coordinate (upper left) at (current bounding box.north west);
		}
	},
	slave/.style={
		execute at end picture={
			\pgfresetboundingbox
			\path (upper left) rectangle (lower right);
		}
	}
}

\tikzset{middlearrow/.style={
		decoration={markings,
			mark= at position 0.6 with {\arrow{#1}} ,
		},
		postaction={decorate}
	}
}

\newcommand{\re}{\text{\upshape Re\,}}
\newcommand{\im}{\text{\upshape Im\,}}

\newcommand{\C}{{\mathbb C}}

\usepackage[colorlinks=true]{hyperref}
\hypersetup{urlcolor=blue, citecolor=red, linkcolor=blue}

\DeclareMathOperator{\diag}{diag}

\def\XXint#1#2#3{{\setbox0=\hbox{$#1{#2#3}{\int}$ }
\vcenter{\hbox{$#2#3$ }}\kern-.59\wd0}}

\tikzset{
    master/.style={
        execute at end picture={
            \coordinate (lower right) at (current bounding box.south east);
            \coordinate (upper left) at (current bounding box.north west);
        }
    },
    slave/.style={
        execute at end picture={
            \pgfresetboundingbox
            \path (upper left) rectangle (lower right);
        }
    }
}
\begin{document}

\title{\vspace{-1cm}Doubly periodic lozenge tilings of a hexagon \\ and matrix valued orthogonal polynomials}
\author{Christophe Charlier\footnote{Department of Mathematics, 
				Royal Institute of Technology (KTH),
				Stockholm, Sweden.  Email: cchar@kth.se.}}

\maketitle

\begin{abstract}
We analyze a random lozenge tiling model of a large regular hexagon, whose underlying weight structure is periodic of period $2$ in both the horizontal and vertical directions. This is a determinantal point process whose correlation kernel is expressed in terms of non-Hermitian matrix valued orthogonal polynomials. This model belongs to a class of models for which the existing techniques for studying asymptotics cannot be applied. The novel part of our method consists of establishing a connection between matrix valued and scalar valued orthogonal polynomials. This allows to simplify the double contour formula for the kernel obtained by Duits and Kuijlaars by reducing the size of a Riemann-Hilbert problem. The proof relies on the fact that the matrix valued weight possesses eigenvalues that live on an underlying Riemann surface $\mathcal{M}$ of genus $0$. We consider this connection of independent interest; it is natural to expect that similar ideas can be used for other matrix valued orthogonal polynomials, as long as the corresponding Riemann surface $\mathcal{M}$ is of genus $0$.

The rest of the method consists of two parts, and mainly follows the lines of a previous work of Charlier, Duits, Kuijlaars and Lenells. First, we perform a Deift-Zhou steepest descent analysis to obtain asymptotics for the scalar valued orthogonal polynomials. The main difficulty is the study of an equilibrium problem in the complex plane. Second, the asymptotics for the orthogonal polynomials are substituted in the double contour integral and the latter is analyzed using the saddle point method. %This part requires several estimates on the phase function that appear in the integrand. 

Our main results are the limiting densities of the lozenges in the disordered flower-shaped region. However, we stress that the method allows in principle to rigorously compute other meaningful probabilistic quantities in the model. 

% Our approach can be seen as a generalization of \cite{CDKL} for a model with matrix valued orthogonal polynomials, and requires the analysis of a double contour integral involving the solution of a $4 \times 4$ Riemann-Hilbert problem.

  %We use a result of Duits and Kuijlaars to express the kernel in terms of a double integral involving a $4 \times 4$ Riemann-Hilbert problem. This Riemann-Hilbert problem is related to non-Hermitian matrix valued orthogonal polynomials. Our techniques is inspired by \cite{CDKL} and combine a Deift/Zhou steepest descent analysis of this Riemann-Hilbert problem, with saddle point methods for the analysis of the double integral. We also obtain results for the matrix valued orthogonal polynomials ???
\end{abstract}

\section{Introduction}

A lozenge tiling of a hexagon is a collection of three different types of lozenges ($\tikz[scale=.25]{ \draw (0,0)  \lozr; }$, $\tikz[scale=.25] { \draw (0,0) \lozu; }$ and $\tikz[scale=.25] { \draw (0,0) \lozd;}$) which cover this hexagon without overlaps, see Figure \ref{fig: non-intersecting paths} (left). There are finitely many such tilings; hence by assigning to each tiling $\mathcal{T}$ a non-negative weight $\mathrm{W}(\mathcal{T})$, we define a probability measure on the tilings by
\vspace{-0.15cm}\begin{align}\label{prob over tilings}
\mathbb{P}(\mathcal{T}) = \frac{\mathrm{W}(\mathcal{T})}{\sum_{\mathcal{T}'}\mathrm{W}(\mathcal{T'})},
\end{align}
\vspace{-0.05cm}where the sum is taken over all the tilings (and is assumed to be non-zero). Uniform random tilings of a hexagon (i.e. when $\mathrm{W}(\mathcal{T})=1$ for all $\mathcal{T}$) is a well-studied model. As the size of the hexagon tends to infinity (while the size of the lozenges is kept fixed), the local statistical properties of this model are described by universal processes \cite{Jptrf,BKMM,Gorin,J17}. We also refer to \cite{CKP,KO,KOS} for important early results and to \cite{BG,K} for general references on tiling models. Uniform lozenge tilings of more complicated domains (non-necessarily convex) have also been widely studied in recent years \cite{Petrov1,Petrov2,BuGo2,AvMJ,Aggarwal}.

% The most well-studied measure is the uniform measure over all the tiling. 

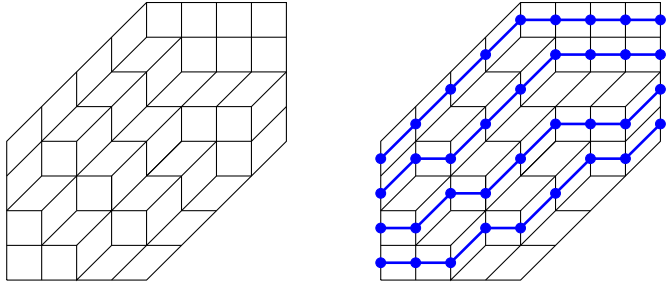
\begin{figure}
\begin{center}
\vspace{-0.5cm}
\begin{tikzpicture}[master,scale = 0.92]
\node at (0,0) {};
%\node at (7,2) {\includegraphics[width=3.5cm]{../Image/Tiling/n_4_uniform.jpg}};
\draw (0,0)--(2,0)--(4,2)--(4,4)--(2,4)--(0,2)--(0,0);
\draw (0,0.5)--(1,0.5)--(1,0)--(1.5,0.5)--(2.5,0.5);
\draw (0,1)--(0.5,1)--(1,1.5)--(1.5,1.5)--(1.5,1)--(2,1)--(3,2)--(3.5,2)--(4,2.5);
\draw (0,1)--(0.5,1.5)--(1,1.5)--(2.5,3)--(4,3);
\draw (0,1.5)--(0.5,2)--(1,2)--(2.5,3.5)--(4,3.5);
\draw (0.5,2)--(2,3.5)--(2.5,3.5);
\draw (2.5,3)--(2.5,4);
\draw (1.5,1.5)--(3,3)--(3,4);
\draw (1.5,0)--(3,1.5)--(3,2.5)--(3.5,3)--(3.5,4);
\draw (3.5,1.5)--(3.5,2.5)--(4,3);
\draw (0.5,0)--(0.5,1);
\draw (0.5,1.5)--(0.5,2.5);
\draw (1,1)--(1,2); \draw (1,2.5)--(1,3);
\draw (1.5,0.5)--(1.5,1.5);
\draw (1.5,2)--(1.5,2.5); \draw (1.5,3)--(1.5,3.5);
\draw (2,0.5)--(2,1); \draw (2,1.5)--(2,2); \draw (2,2.5)--(2,3); \draw (2,3.5)--(2,4);
\draw (2.5,1)--(2.5,1.5); \draw (2.5,2)--(2.5,2.5);
\draw (2,2.5)--(3.5,2.5);
\draw (0.5,0.5)--(1,1); \draw (1,0.5)--(1.5,1)--(1,1);
\draw (1.5,1)--(2.5,2)--(3,2);
\draw (1,2.5)--(1.5,2.5);
\draw (1.5,3)--(2,3);
\draw (1.5,2)--(2,2);
\draw (2,1.5)--(2.5,2);
\draw (2,1.5)--(2.5,1.5);
\draw (2.5,1)--(3,1); 
\draw (3,1.5)--(3.5,1.5);
\end{tikzpicture} \hspace{1cm}\begin{tikzpicture}[slave,scale = 0.92]
%\node at (0,0) {\includegraphics[width=3.5cm]{../Image/Tiling/n_4_uniform_paths.jpg}};
\node at (0,0) {};
\draw (0,0)--(2,0)--(4,2)--(4,4)--(2,4)--(0,2)--(0,0);
\draw (0,0.5)--(1,0.5)--(1,0)--(1.5,0.5)--(2.5,0.5);
\draw (0,1)--(0.5,1)--(1,1.5)--(1.5,1.5)--(1.5,1)--(2,1)--(3,2)--(3.5,2)--(4,2.5);
\draw (0,1)--(0.5,1.5)--(1,1.5)--(2.5,3)--(4,3);
\draw (0,1.5)--(0.5,2)--(1,2)--(2.5,3.5)--(4,3.5);
\draw (0.5,2)--(2,3.5)--(2.5,3.5);
\draw (2.5,3)--(2.5,4);
\draw (1.5,1.5)--(3,3)--(3,4);
\draw (1.5,0)--(3,1.5)--(3,2.5)--(3.5,3)--(3.5,4);
\draw (3.5,1.5)--(3.5,2.5)--(4,3);
\draw (0.5,0)--(0.5,1);
\draw (0.5,1.5)--(0.5,2.5);
\draw (1,1)--(1,2); \draw (1,2.5)--(1,3);
\draw (1.5,0.5)--(1.5,1.5);
\draw (1.5,2)--(1.5,2.5); \draw (1.5,3)--(1.5,3.5);
\draw (2,0.5)--(2,1); \draw (2,1.5)--(2,2); \draw (2,2.5)--(2,3); \draw (2,3.5)--(2,4);
\draw (2.5,1)--(2.5,1.5); \draw (2.5,2)--(2.5,2.5);
\draw (2,2.5)--(3.5,2.5);
\draw (0.5,0.5)--(1,1); \draw (1,0.5)--(1.5,1)--(1,1);
\draw (1.5,1)--(2.5,2)--(3,2);
\draw (1,2.5)--(1.5,2.5);
\draw (1.5,3)--(2,3);
\draw (1.5,2)--(2,2);
\draw (2,1.5)--(2.5,2);
\draw (2,1.5)--(2.5,1.5);
\draw (2.5,1)--(3,1); 
\draw (3,1.5)--(3.5,1.5);

%----- Here we draw the path and the points
\draw[blue,fill] (0,0.25) circle (0.7mm);
\draw[blue,fill] (0,0.75) circle (0.7mm);
\draw[blue,fill] (0,1.25) circle (0.7mm);
\draw[blue,fill] (0,1.75) circle (0.7mm);

\draw[blue,fill] (0.5,0.25) circle (0.7mm);
\draw[blue,fill] (0.5,0.75) circle (0.7mm);
\draw[blue,fill] (0.5,1.75) circle (0.7mm);
\draw[blue,fill] (0.5,2.25) circle (0.7mm);

\draw[blue,fill] (1,0.25) circle (0.7mm);
\draw[blue,fill] (1,1.25) circle (0.7mm);
\draw[blue,fill] (1,1.75) circle (0.7mm);
\draw[blue,fill] (1,2.75) circle (0.7mm);

\draw[blue,fill] (1.5,0.75) circle (0.7mm);
\draw[blue,fill] (1.5,1.25) circle (0.7mm);
\draw[blue,fill] (1.5,2.25) circle (0.7mm);
\draw[blue,fill] (1.5,3.25) circle (0.7mm);

\draw[blue,fill] (2,0.75) circle (0.7mm);
\draw[blue,fill] (2,1.75) circle (0.7mm);
\draw[blue,fill] (2,2.75) circle (0.7mm);
\draw[blue,fill] (2,3.75) circle (0.7mm);

\draw[blue,fill] (2.5,1.25) circle (0.7mm);
\draw[blue,fill] (2.5,2.25) circle (0.7mm);
\draw[blue,fill] (2.5,3.25) circle (0.7mm);
\draw[blue,fill] (2.5,3.75) circle (0.7mm);

\draw[blue,fill] (3,1.75) circle (0.7mm);
\draw[blue,fill] (3,2.25) circle (0.7mm);
\draw[blue,fill] (3,3.25) circle (0.7mm);
\draw[blue,fill] (3,3.75) circle (0.7mm);

\draw[blue,fill] (3.5,1.75) circle (0.7mm);
\draw[blue,fill] (3.5,2.25) circle (0.7mm);
\draw[blue,fill] (3.5,3.25) circle (0.7mm);
\draw[blue,fill] (3.5,3.75) circle (0.7mm);

\draw[blue,fill] (4,2.25) circle (0.7mm);
\draw[blue,fill] (4,2.75) circle (0.7mm);
\draw[blue,fill] (4,3.25) circle (0.7mm);
\draw[blue,fill] (4,3.75) circle (0.7mm);

\draw[blue,line width=0.35 mm] (0,0.25)--(1,0.25)--(1.5,0.75)--(2,0.75)--(3,1.75)--(3.5,1.75)--(4,2.25);
\draw[blue,line width=0.35 mm] (0,0.75)--(0.5,0.75)--(1,1.25)--(1.5,1.25)--(2.5,2.25)--(3.5,2.25)--(4,2.75);
\draw[blue,line width=0.35 mm] (0,1.25)--(0.5,1.75)--(1,1.75)--(2.5,3.25)--(4,3.25);
\draw[blue,line width=0.35 mm] (0,1.75)--(2,3.75)--(4,3.75);
\end{tikzpicture}
\end{center}
\vspace{-0.5cm}\caption{\label{fig: non-intersecting paths}A tiling of a hexagon, and the associated non-intersecting paths.}
\end{figure}

%\begin{align*}
%\mathcal{H}_{N} := \{(x,y)\in \mathbb{R}^{2}: 0 \leq x \leq N, 0 \leq y \leq N, -N \leq x-y \leq N\}, \qquad N \in \mathbb{N}_{\geq 1},
%\end{align*}

In this work, we consider the regular hexagon of (large) size $n$
\begin{align}\label{def of Hn}
\mathcal{H}_{n} := \{(x,y)\in \mathbb{R}^{2}: 0 \leq x \leq 2n, 0 \leq y \leq 2n, -n \leq x-y \leq n\}, \qquad n \in \mathbb{N}_{\geq 1},
\end{align}
but we deviate from the uniform measure and study instead measures with periodic weightings. To explain what this means, we first briefly recall a well-known one-to-one correspondence between tilings of a hexagon and certain non-intersecting paths. This bijection can be written down explicitly, but is best understood informally. The paths are obtained by drawing lines on top of two types of lozenges
\vspace{-0.3cm}
\begin{align}\label{non-intersecting paths}
\tikz[scale=.5]{ \draw (0,0)  \lozr; \draw[very thick] (0,.5)--(1,1.5); \filldraw (0,0.5) circle(3pt);  \filldraw (1,1.5) circle(3pt); }  \quad \tikz[scale=.5] { \draw (0,0) \lozu; \draw[very thick] (0,.5)--(1,.5); \filldraw (0,0.5) circle(3pt);  \filldraw (1,0.5) circle(3pt); }\quad \mbox{and} \quad   \tikz[scale=.5] { \draw (0,0) \lozd;},
\end{align}
as shown in Figure \ref{fig: non-intersecting paths} (right). The paths associated to the tilings of $\mathcal{H}_{n}$ lie on a graph $\mathcal{G}_{n}$ which depends only on the size of the hexagon, see Figure \ref{fig: lattice} (left). To each edge $\mathfrak{e}$ of $\mathcal{G}_{n}$, we assign a non-negative weight $w_{\mathfrak{e}}$. The weight of a path $\mathfrak{p}$ is then defined as $w_{\mathfrak{p}}=\prod_{\mathfrak{e} \in \mathfrak{p}} w_{\mathfrak{e}}$, and the weight of a tiling $\mathcal{T}$ as $\mathrm{W}(\mathcal{T})=\prod_{\mathfrak{p} \in \mathcal{T}} w_{\mathfrak{p}}$. Provided that at least one tiling has a positive weight, this defines a probability measure on the set of tilings by \eqref{prob over tilings}. If each edge is assigned the same weight, then we recover the uniform measure over tilings. We say that a lozenge tiling model has $p \times q$ periodic weightings if the weight structure on the edges is periodic of period $p$ in the vertical direction, and periodic of period $q$ in the horizontal direction, see Figure \ref{fig: lattice} (right) for an illustration with $p=2$ and $q=3$. Thus a $p \times q$ periodic weighting is completely determined by $2pq$ edge weights. Note that all paths share the same number of horizontal edges, and also the same number of oblique edges; hence lozenge tiling models with $1 \times 1$ periodic weightings are all equivalent to the uniform measure.

 %From now on, we discuss models with higher horizontal or vertical periodicity.

% Lozenge tiling models with $1 \times 1$ periodic weightings are a priori more general than the uniform measure (since we need to specify a weight for the horizontal edges, and another weight for the oblique ones), but it turns out this is not the case; this follows from the simple fact that each path passes through the same number of horizontal edges, and also through the same number of oblique edges.

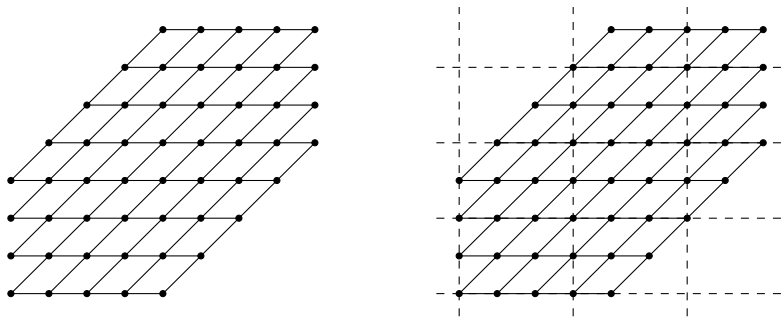
\begin{figure}[h]
\begin{center}
\vspace{0cm}
\begin{tikzpicture}[master]
\node at (0,0) {};
%\node at (1.65,1.45) {\includegraphics[width=3.5cm]{../Image/Tiling/n_4_uniform_paths.jpg}};

% The vertices
\draw[fill] (0,0) circle (0.4mm);
\draw[fill] (0.5,0) circle (0.4mm);
\draw[fill] (1,0) circle (0.4mm);
\draw[fill] (1.5,0) circle (0.4mm);
\draw[fill] (2,0) circle (0.4mm);

\draw[fill] (0,0.5) circle (0.4mm);
\draw[fill] (0.5,0.5) circle (0.4mm);
\draw[fill] (1,0.5) circle (0.4mm);
\draw[fill] (1.5,0.5) circle (0.4mm);
\draw[fill] (2,0.5) circle (0.4mm);
\draw[fill] (2.5,0.5) circle (0.4mm);

\draw[fill] (0,1) circle (0.4mm);
\draw[fill] (0.5,1) circle (0.4mm);
\draw[fill] (1,1) circle (0.4mm);
\draw[fill] (1.5,1) circle (0.4mm);
\draw[fill] (2,1) circle (0.4mm);
\draw[fill] (2.5,1) circle (0.4mm);
\draw[fill] (3,1) circle (0.4mm);

\draw[fill] (0,1.5) circle (0.4mm);
\draw[fill] (0.5,1.5) circle (0.4mm);
\draw[fill] (1,1.5) circle (0.4mm);
\draw[fill] (1.5,1.5) circle (0.4mm);
\draw[fill] (2,1.5) circle (0.4mm);
\draw[fill] (2.5,1.5) circle (0.4mm);
\draw[fill] (3,1.5) circle (0.4mm);
\draw[fill] (3.5,1.5) circle (0.4mm);

\draw[fill] (0.5,2) circle (0.4mm);
\draw[fill] (1,2) circle (0.4mm);
\draw[fill] (1.5,2) circle (0.4mm);
\draw[fill] (2,2) circle (0.4mm);
\draw[fill] (2.5,2) circle (0.4mm);
\draw[fill] (3,2) circle (0.4mm);
\draw[fill] (3.5,2) circle (0.4mm);
\draw[fill] (4,2) circle (0.4mm);

\draw[fill] (1,2.5) circle (0.4mm);
\draw[fill] (1.5,2.5) circle (0.4mm);
\draw[fill] (2,2.5) circle (0.4mm);
\draw[fill] (2.5,2.5) circle (0.4mm);
\draw[fill] (3,2.5) circle (0.4mm);
\draw[fill] (3.5,2.5) circle (0.4mm);
\draw[fill] (4,2.5) circle (0.4mm);

\draw[fill] (1.5,3) circle (0.4mm);
\draw[fill] (2,3) circle (0.4mm);
\draw[fill] (2.5,3) circle (0.4mm);
\draw[fill] (3,3) circle (0.4mm);
\draw[fill] (3.5,3) circle (0.4mm);
\draw[fill] (4,3) circle (0.4mm);

\draw[fill] (2,3.5) circle (0.4mm);
\draw[fill] (2.5,3.5) circle (0.4mm);
\draw[fill] (3,3.5) circle (0.4mm);
\draw[fill] (3.5,3.5) circle (0.4mm);
\draw[fill] (4,3.5) circle (0.4mm);

% the horizontal arrows

\draw (0,0)--(0.5,0);
\draw (0.5,0)--(1,0);
\draw (1,0)--(1.5,0);
\draw (1.5,0)--(2,0);

\draw (0,0.5)--(0.5,0.5);
\draw (0.5,0.5)--(1,0.5);
\draw (1,0.5)--(1.5,0.5);
\draw (1.5,0.5)--(2,0.5);
\draw (2,0.5)--(2.5,0.5);

\draw (0,1)--(0.5,1);
\draw (0.5,1)--(1,1);
\draw (1,1)--(1.5,1);
\draw (1.5,1)--(2,1);
\draw (2,1)--(2.5,1);
\draw (2.5,1)--(3,1);

\draw (0,1.5)--(0.5,1.5);
\draw (0.5,1.5)--(1,1.5);
\draw (1,1.5)--(1.5,1.5);
\draw (1.5,1.5)--(2,1.5);
\draw (2,1.5)--(2.5,1.5);
\draw (2.5,1.5)--(3,1.5);
\draw (3,1.5)--(3.5,1.5);

\draw (0.5,2)--(1,2);
\draw (1,2)--(1.5,2);
\draw (1.5,2)--(2,2);
\draw (2,2)--(2.5,2);
\draw (2.5,2)--(3,2);
\draw (3,2)--(3.5,2);
\draw (3.5,2)--(4,2);

\draw (1,2.5)--(1.5,2.5);
\draw (1.5,2.5)--(2,2.5);
\draw (2,2.5)--(2.5,2.5);
\draw (2.5,2.5)--(3,2.5);
\draw (3,2.5)--(3.5,2.5);
\draw (3.5,2.5)--(4,2.5);

\draw (1.5,3)--(2,3);
\draw (2,3)--(2.5,3);
\draw (2.5,3)--(3,3);
\draw (3,3)--(3.5,3);
\draw (3.5,3)--(4,3);

\draw (2,3.5)--(2.5,3.5);
\draw (2.5,3.5)--(3,3.5);
\draw (3,3.5)--(3.5,3.5);
\draw (3.5,3.5)--(4,3.5);

% the oblique arrows

\draw (0,0)--(0.5,0.5);
\draw (0.5,0)--(1,0.5);
\draw (1,0)--(1.5,0.5);
\draw (1.5,0)--(2,0.5);
\draw (2,0)--(2.5,0.5);

\draw (0,0.5)--(0.5,1);
\draw (0.5,0.5)--(1,1);
\draw (1,0.5)--(1.5,1);
\draw (1.5,0.5)--(2,1);
\draw (2,0.5)--(2.5,1);
\draw (2.5,0.5)--(3,1);

\draw (0,1)--(0.5,1.5);
\draw (0.5,1)--(1,1.5);
\draw (1,1)--(1.5,1.5);
\draw (1.5,1)--(2,1.5);
\draw (2,1)--(2.5,1.5);
\draw (2.5,1)--(3,1.5);
\draw (3,1)--(3.5,1.5);

\draw (0,1.5)--(0.5,2);
\draw (0.5,1.5)--(1,2);
\draw (1,1.5)--(1.5,2);
\draw (1.5,1.5)--(2,2);
\draw (2,1.5)--(2.5,2);
\draw (2.5,1.5)--(3,2);
\draw (3,1.5)--(3.5,2);
\draw (3.5,1.5)--(4,2);

\draw (0.5,2)--(1,2.5);
\draw (1,2)--(1.5,2.5);
\draw (1.5,2)--(2,2.5);
\draw (2,2)--(2.5,2.5);
\draw (2.5,2)--(3,2.5);
\draw (3,2)--(3.5,2.5);
\draw (3.5,2)--(4,2.5);

\draw (1,2.5)--(1.5,3);
\draw (1.5,2.5)--(2,3);
\draw (2,2.5)--(2.5,3);
\draw (2.5,2.5)--(3,3);
\draw (3,2.5)--(3.5,3);
\draw (3.5,2.5)--(4,3);

\draw (1.5,3)--(2,3.5);
\draw (2,3)--(2.5,3.5);
\draw (2.5,3)--(3,3.5);
\draw (3,3)--(3.5,3.5);
\draw (3.5,3)--(4,3.5);

\end{tikzpicture}
\hspace{1.5cm}
\begin{tikzpicture}[slave]
\node at (0,0) {};
%\node at (1.65,1.45) {\includegraphics[width=3.5cm]{../Image/Tiling/n_4_uniform_paths.jpg}};

% The vertices
\draw[fill] (0,0) circle (0.4mm);
\draw[fill] (0.5,0) circle (0.4mm);
\draw[fill] (1,0) circle (0.4mm);
\draw[fill] (1.5,0) circle (0.4mm);
\draw[fill] (2,0) circle (0.4mm);

\draw[fill] (0,0.5) circle (0.4mm);
\draw[fill] (0.5,0.5) circle (0.4mm);
\draw[fill] (1,0.5) circle (0.4mm);
\draw[fill] (1.5,0.5) circle (0.4mm);
\draw[fill] (2,0.5) circle (0.4mm);
\draw[fill] (2.5,0.5) circle (0.4mm);

\draw[fill] (0,1) circle (0.4mm);
\draw[fill] (0.5,1) circle (0.4mm);
\draw[fill] (1,1) circle (0.4mm);
\draw[fill] (1.5,1) circle (0.4mm);
\draw[fill] (2,1) circle (0.4mm);
\draw[fill] (2.5,1) circle (0.4mm);
\draw[fill] (3,1) circle (0.4mm);

\draw[fill] (0,1.5) circle (0.4mm);
\draw[fill] (0.5,1.5) circle (0.4mm);
\draw[fill] (1,1.5) circle (0.4mm);
\draw[fill] (1.5,1.5) circle (0.4mm);
\draw[fill] (2,1.5) circle (0.4mm);
\draw[fill] (2.5,1.5) circle (0.4mm);
\draw[fill] (3,1.5) circle (0.4mm);
\draw[fill] (3.5,1.5) circle (0.4mm);

\draw[fill] (0.5,2) circle (0.4mm);
\draw[fill] (1,2) circle (0.4mm);
\draw[fill] (1.5,2) circle (0.4mm);
\draw[fill] (2,2) circle (0.4mm);
\draw[fill] (2.5,2) circle (0.4mm);
\draw[fill] (3,2) circle (0.4mm);
\draw[fill] (3.5,2) circle (0.4mm);
\draw[fill] (4,2) circle (0.4mm);

\draw[fill] (1,2.5) circle (0.4mm);
\draw[fill] (1.5,2.5) circle (0.4mm);
\draw[fill] (2,2.5) circle (0.4mm);
\draw[fill] (2.5,2.5) circle (0.4mm);
\draw[fill] (3,2.5) circle (0.4mm);
\draw[fill] (3.5,2.5) circle (0.4mm);
\draw[fill] (4,2.5) circle (0.4mm);

\draw[fill] (1.5,3) circle (0.4mm);
\draw[fill] (2,3) circle (0.4mm);
\draw[fill] (2.5,3) circle (0.4mm);
\draw[fill] (3,3) circle (0.4mm);
\draw[fill] (3.5,3) circle (0.4mm);
\draw[fill] (4,3) circle (0.4mm);

\draw[fill] (2,3.5) circle (0.4mm);
\draw[fill] (2.5,3.5) circle (0.4mm);
\draw[fill] (3,3.5) circle (0.4mm);
\draw[fill] (3.5,3.5) circle (0.4mm);
\draw[fill] (4,3.5) circle (0.4mm);

% the horizontal arrows

\draw (0,0)--(0.5,0);
\draw (0.5,0)--(1,0);
\draw (1,0)--(1.5,0);
\draw (1.5,0)--(2,0);

\draw (0,0.5)--(0.5,0.5);
\draw (0.5,0.5)--(1,0.5);
\draw (1,0.5)--(1.5,0.5);
\draw (1.5,0.5)--(2,0.5);
\draw (2,0.5)--(2.5,0.5);

\draw (0,1)--(0.5,1);
\draw (0.5,1)--(1,1);
\draw (1,1)--(1.5,1);
\draw (1.5,1)--(2,1);
\draw (2,1)--(2.5,1);
\draw (2.5,1)--(3,1);

\draw (0,1.5)--(0.5,1.5);
\draw (0.5,1.5)--(1,1.5);
\draw (1,1.5)--(1.5,1.5);
\draw (1.5,1.5)--(2,1.5);
\draw (2,1.5)--(2.5,1.5);
\draw (2.5,1.5)--(3,1.5);
\draw (3,1.5)--(3.5,1.5);

\draw (0.5,2)--(1,2);
\draw (1,2)--(1.5,2);
\draw (1.5,2)--(2,2);
\draw (2,2)--(2.5,2);
\draw (2.5,2)--(3,2);
\draw (3,2)--(3.5,2);
\draw (3.5,2)--(4,2);

\draw (1,2.5)--(1.5,2.5);
\draw (1.5,2.5)--(2,2.5);
\draw (2,2.5)--(2.5,2.5);
\draw (2.5,2.5)--(3,2.5);
\draw (3,2.5)--(3.5,2.5);
\draw (3.5,2.5)--(4,2.5);

\draw (1.5,3)--(2,3);
\draw (2,3)--(2.5,3);
\draw (2.5,3)--(3,3);
\draw (3,3)--(3.5,3);
\draw (3.5,3)--(4,3);

\draw (2,3.5)--(2.5,3.5);
\draw (2.5,3.5)--(3,3.5);
\draw (3,3.5)--(3.5,3.5);
\draw (3.5,3.5)--(4,3.5);

% the oblique arrows

\draw (0,0)--(0.5,0.5);
\draw (0.5,0)--(1,0.5);
\draw (1,0)--(1.5,0.5);
\draw (1.5,0)--(2,0.5);
\draw (2,0)--(2.5,0.5);

\draw (0,0.5)--(0.5,1);
\draw (0.5,0.5)--(1,1);
\draw (1,0.5)--(1.5,1);
\draw (1.5,0.5)--(2,1);
\draw (2,0.5)--(2.5,1);
\draw (2.5,0.5)--(3,1);

\draw (0,1)--(0.5,1.5);
\draw (0.5,1)--(1,1.5);
\draw (1,1)--(1.5,1.5);
\draw (1.5,1)--(2,1.5);
\draw (2,1)--(2.5,1.5);
\draw (2.5,1)--(3,1.5);
\draw (3,1)--(3.5,1.5);

\draw (0,1.5)--(0.5,2);
\draw (0.5,1.5)--(1,2);
\draw (1,1.5)--(1.5,2);
\draw (1.5,1.5)--(2,2);
\draw (2,1.5)--(2.5,2);
\draw (2.5,1.5)--(3,2);
\draw (3,1.5)--(3.5,2);
\draw (3.5,1.5)--(4,2);

\draw (0.5,2)--(1,2.5);
\draw (1,2)--(1.5,2.5);
\draw (1.5,2)--(2,2.5);
\draw (2,2)--(2.5,2.5);
\draw (2.5,2)--(3,2.5);
\draw (3,2)--(3.5,2.5);
\draw (3.5,2)--(4,2.5);

\draw (1,2.5)--(1.5,3);
\draw (1.5,2.5)--(2,3);
\draw (2,2.5)--(2.5,3);
\draw (2.5,2.5)--(3,3);
\draw (3,2.5)--(3.5,3);
\draw (3.5,2.5)--(4,3);

\draw (1.5,3)--(2,3.5);
\draw (2,3)--(2.5,3.5);
\draw (2.5,3)--(3,3.5);
\draw (3,3)--(3.5,3.5);
\draw (3.5,3)--(4,3.5);

\draw[dashed] (-0.3,0)--(4.3,0);
\draw[dashed] (-0.3,1)--(4.3,1);
\draw[dashed] (-0.3,2)--(4.3,2);
\draw[dashed] (-0.3,3)--(4.3,3);
\draw[dashed] (0,-0.3)--(0,3.8);
\draw[dashed] (1.5,-0.3)--(1.5,3.8);
\draw[dashed] (3,-0.3)--(3,3.8);
\end{tikzpicture}
\end{center}
\vspace{-0.35cm}
\caption{\label{fig: lattice}The graph $\mathcal{G}_{4}$, and the periods of a $2 \times 3$ periodic weighting.}
\end{figure}

%Lozenge tiling models of large hexagons with periodic weightings can feature all of the three possible types of phases known in random tiling models: the solid, liquid and gas phases. A solid region (also called frozen region) is filled with one type of lozenges. In the liquid and gas phases, all three types of lozenges coexist. The difference between these two phases is reflected in the correlations between two lozenges: in the liquid region, the correlation decay is polynomial with the distance between the lozenges, while in the gas region the decay is exponential. It is known that there is no gas phase for the uniform measure (corresponding to $p=q=1$). In fact, it turns out that the smallest periods that lead to the presence of a gas phase are either $p=2,q=3$ or $p=3, q=2$.

\vspace{-0.1cm}By putting points on the paths as shown in \eqref{non-intersecting paths}, each tiling of the hexagon gives rise to a point configuration, see also Figure \ref{fig: non-intersecting paths} (right). Thus the probability measure \eqref{prob over tilings} on tilings can be viewed as a discrete point process \cite{Borodin,Soshnikov}. For lozenge tiling models with $p \times q$ periodic weightings, it follows from the Lindstr\"{o}m-Gessel-Viennot theorem \cite{GV,L} combined with the Eynard-Mehta theorem \cite{EM} that this point process is determinantal. Therefore, to understand the fine asymptotic structure (as $n \to + \infty$) it suffices to analyze the asymptotic behavior of the correlation kernel. However, until recently \cite{DK,CDKL}, the existing techniques were not appropriate for such analysis. 

\vspace{0.2cm}The main result of \cite{DK} is a double contour formula for the correlation kernels of various tiling models with periodic weightings (including lozenge tiling models of a hexagon as considered here). In this formula, the integrand is expressed in terms of the solution (denoted $Y$) to a $2p \times 2p$ Riemann-Hilbert (RH) problem. This RH problem is related to certain orthogonal polynomials (OPs), which are non-standard in two aspects:
\begin{itemize}
\item \vspace{-0.2cm}the OPs and the weight are $p \times p$ matrix valued,
\item \vspace{-0.2cm}the orthogonality conditions are non-hermitian.
\end{itemize}
The size of the RH problem, the size of the weight, and the size of the OPs all depend on $p$, but quite interestingly not on $q$.

%\vspace{0.2cm}However, until recently there were no convenient formula
%
%\vspace{0.2cm}Until recently, there were no convenient
%
%\vspace{0.2cm}In the recent paper \cite{DK}, Duits and Kuijlaars have obtained a double contour formula for the correlation kernel, which opens a path for the analysis 

%Asymptotic analysis as $n \to + \infty$ of 
\vspace{0.2cm}Lozenge tiling models of the hexagon with $p \times q$ periodic weightings are rather unexplored up to now. To the best of our knowledge, the model considered in \cite{CDKL} is the only one (other than the uniform measure) prior to the present work for which results on fine asymptotics exist. The model considered in \cite{CDKL} is $1 \times 2$ periodic and uses the formula of \cite{DK} as the starting point of the analysis. The techniques of \cite{CDKL} combine the Deift/Zhou steepest descent method \cite{DZ} of $Y$ (of size $2 \times 2$) with a non-standard saddle point analysis of the double contour integral. However, since $p=1$, the associated OPs are scalar (this fact was extensively used in the proof) and it is not clear how to generalize these techniques to the case $p \geq 2$.

\vspace{0.2cm} The aim of this paper is precisely to develop a method to handle a situation involving matrix valued orthogonal polynomials. We will implement this method on a particular lozenge tiling model with $2 \times 2$ periodic weightings, which presents one simply connected liquid region (which has the shape of a flower with $6$ petals), $6$ frozen regions, and $6$ staircase regions (also called semi-frozen regions); it will be presented in more detail in Section \ref{Section: model}. The starting point of our analysis is the double contour formula from \cite{DK} which expresses the kernel in terms of a $4 \times 4$ RH problem related to $2 \times 2$ matrix valued OPs. The method can be summarized in three steps as follows:
\begin{enumerate}
\item First, we establish a connection between matrix valued and scalar valued orthogonal polynomials. In particular, in Theorem \ref{thm: correlation kernel final scalar expression} we obtain a new expression for the kernel in terms of the solution (denoted $U$) to $2 \times 2$ RH problem related to scalar OPs. This formula allows for a simpler analysis than the original formula from \cite{DK}. 
\end{enumerate}
%\vspace{-0.2cm} 
\begin{enumerate}\setcounter{enumi}{1}
\item Second, we perform an asymptotic analysis of the RH problem for $U$ via the Deift-Zhou steepest descent method. The construction of the equilibrium measure and the associated $g$-function is the main difficulty. The remaining part of the RH analysis is rather standard.
\item Third, the asymptotics for the orthogonal polynomials are substituted in the double contour integral and the latter is analyzed using the saddle point method.
\end{enumerate}
The first step is the main novel part of the paper. The remaining two steps were first developped in \cite{CDKL} for a tiling model with $1 \times 2$ periodic weightings.

\medskip 

Our main results, which are stated in Theorem \ref{thm:main}, are the limiting densities of the different lozenges in the liquid region. However, we emphasize that the method also allows in principle to rigorously compute more sophisticated asymptotic behaviors in the model (such as the limiting process in the bulk).

%The implementation of the first step is a main contribution of the paper.  The remaining two steps are closer in spirit to \cite{CDKL} (though technically more involved).

\paragraph*{An expression for the kernel in terms of scalar OPs.} The eigenvalues and eigenvectors of the $2 \times 2$ orthogonality weight play an important role in the first step of the analysis. They are naturally defined on a $2$-sheeted Riemann surface $\mathcal{M}$, which turns out to be of genus $0$. This fact is crucial to obtain the new formula for the kernel in terms of scalar OPs. We expect that ideas similar to the ones presented here can be applied to other tiling models with periodic weightings, as long as the corresponding Riemann surface $\mathcal{M}$ is of genus $0$.

\vspace{0.2cm}Lozenge tiling models of large hexagons with periodic weightings can feature all of the three possible types of phases known in random tiling models: the solid, liquid and gas phases. A solid region (also called frozen region) is filled with one type of lozenges. In the liquid and gas phases, all three types of lozenges coexist. The difference between these two phases is reflected in the correlations between two points: in the liquid region, the correlation decay is polynomial with the distance between the points, while in the gas region the decay is exponential. It is known that there is no gas phase for the uniform measure (corresponding to $p=q=1$). In fact, it is expected that the smallest periods that lead to the presence of a gas phase are either $p=2,q=3$ or $p=3, q=2$. For models that present gas phases, we expect $\mathcal{M}$ to have genus at least $1$, and then new techniques are required. This is left for future works. 

%\subsection*{Related works}\label{Section: related works}
\paragraph*{Related works.} Random lozenge tilings of the regular hexagon is a particular example of a tiling model. We briefly review here other tiling models with periodic weightings that have been studied in the literature and for which more results are known. We also discuss the related techniques and explain why they cannot be applied in our case.

\vspace{0.2cm}The Aztec diamond is a well-studied tiling model \cite{JPS,CEP,CKP,J05}. It consists of covering the region $\{(x,y):|x|+|y| \leq n+1\}$ with $2 \times 1$ or $1 \times 2$ rectangles (called dominos), where $n  > 0$ is an integer which parametrizes the size of the covered region. Uniform domino tilings of the Aztec diamond features $4$ solid regions and one liquid region. The associated  discrete point process is determinantal, and turns out to belong to the class of Schur processes (introduced in \cite{OR1}), for which there exists a double contour integral for the kernel that is suitable for an asymptotic analysis as $n \to + \infty$. Another important Schur process is the infinite hexagon with $1 \times k$ periodic weightings. The infinite hexagon is a non-regular hexagon whose vertical side is first sent to infinity either from above or from below, see e.g. \cite[Figure 14]{BeD} for an illustration. For more examples of other interesting tiling models that fall in the Schur process class, see e.g. \cite{BG}. Uniform lozenge tilings of the finite hexagon (such as $\mathcal{H}_{n}$) do not belong to the Schur class, but have been studied using other techniques based on some connections with Hahn polynomials \cite{Jptrf}: the limiting kernel in the bulk scaling regime has been established in \cite{BKMM} using a discrete RH problem, and in \cite{Gorin} using the approach developed in \cite{BO}.
%Random lozenge tilings of the finite hexagon with periodic weightings do not belong to the Schur class (even the uniform measure). However it has been studied by several authors using some connections with Hahn polynomials \cite{Jptrf}: the limiting kernel in the bulk scaling regime has been established in \cite{BKMM} using Riemann-Hilbert (RH) techniques, and in \cite{Gorin} using the approach developed in \cite{BO}.

\vspace{0.2cm}The doubly periodic Aztec diamond exhibits all three phases. It still defines a determinantal point process, but it falls outside of the Schur process class. However, Chhita and Young found in \cite{CY2014} a formula for the correlation kernel by performing an explicit inversion of the Kasteleyn matrix. This formula was further simplified in \cite{CJ} and then used in \cite{CJ, BCJ} to obtain fine asymptotic results on the fluctuations of the liquid-gas boundary as $n \to + \infty$. This same model was analyzed soon afterward in \cite{DK} via a different (and more general) method based on matrix valued orthogonal polynomials and a related RH problem. For the doubly periodic Aztec diamond, this RH problem is surprisingly simple in the sense that it can be solved explicitly for finite $n$. The analysis of \cite{CY2014,CJ,DK} rely on the rather special integrable structure of the doubly periodic Aztec diamond. However, the approach of \cite{DK} applies to a much wider range of tiling models. Berggren and Duits \cite{BeD} have recently identified a whole class of tiling/path models for which it is possible to simplify significantly the formula of \cite{DK}. Quite remarkably, their final expression for the kernel does not involve any RH problem or OPs, which simplifies substantially the saddle point analysis. Using the results from \cite{BeD}, Berggren in \cite{Berggren} recently studied the $2 \times k$ periodic Aztec diamond, for an arbitrary $k$. The class of models for which the formula from \cite{BeD} applies roughly consists of the models with an infinite number of paths whose (possibly matrix valued) orthogonality weight has a Wiener-Hopf type factorization. This class of models contains the Schur class, but also (among others) the doubly periodic Aztec diamond and doubly periodic lozenge tilings of an infinite hexagon.

\vspace{0.2cm}However, lozenge tiling models of the finite hexagon cannot be represented as models with infinitely many paths (as opposed to the Aztec diamond and the infinite hexagon). In particular, they do not belong to the class of models studied in \cite{BeD} and thus the simplified formula from \cite{BeD} cannot be used. This fact makes lozenge tiling models of the finite hexagon harder to analyze asymptotically (see also the comment in \cite[beginning of Section 6]{BeD}). 

%Soon afterwards, Duits and Kuijlaars have obtained a formula in \cite{DK} for the correlation kernel in terms of a RH problem, as described above. This formula can be applied for a very broad class of periodic tilings model, and the authors of \cite{DK} proved that this formula is suitable for an asymptotics analysis

\paragraph{The figures.} In addition to being in bijection with non-intersecting paths, lozenge tilings of the hexagon are also in bijection with \textit{dimer coverings}, which are perfect matchings of a certain bipartite graph. We refer to \cite{J17} for more details on the correspondence with dimers (see also \cite[Figure 1]{Petrov1} for an illustration). The bijection with dimers is not used explicitly in this paper, but we do use it to generate the pictures via the shuffling algorithm \cite{ProppShuffling}.

%The rest of the paper is closer is spirit to \cite{CDKL} (though technically more involved).

\paragraph{Acknowledgment.} The first step of the method is based on an unpublished idea of Arno Kuijlaars. I am very grateful to him for allowing me and even encouraging me to use his idea. I also thank Jonatan Lenells, Arno Kuijlaars and Maurice Duits for interesting discussions, and for a careful reading of the introduction. This work is supported by the European Research Council, Grant Agreement No. 682537.

%The point is that the eigenvalues of the weight plays an important role in the analysis (see \cite{DuitsKuijlaars}), and they live on a Riemann surface whose genus may increase as either $p$ or $q$ increases. If $p = 1$, the weight is scalar and thus the genus is zero independently of $q$. It turns out that the smallest periodicity to consider in order to observe a genus $1$ is to take either 

%This means that all correlation functions can be expressed as determinants involving a function $K$ (called the kernel).

%It is heuristically believed that the so-called ``gas phase" can be observed only if the genus is one or above. 
%\newpage
\section{Model and background}\label{Section: model}
In this section, we present a lozenge tiling model with $2 \times 2$ periodic weightings. We also introduce the necessary material to invoke the double contour formula from \cite{DK} for the kernel. In particular, we present the relevant $2 \times 2$ matrix valued OPs and the associated $4 \times 4$ RH problem. %The formula from \cite{DK} will be the starting point of our analysis in Section \ref{section: first steps of the steepest descent of Y}.
\vspace{-0.2cm}
\subsubsection*{Affine transformation for certain figures of lozenge tilings}
For the presentation of the model and the results, it is convenient to define the hexagon and the lozenges as in \eqref{def of Hn}--\eqref{non-intersecting paths}. However, for the purpose of presenting certain figures of lozenge tilings, it is more pleasant to modify the hexagon and the lozenges by the following simple transformation:
\begin{align}\label{affine transformation on lozenges}
\tikz[scale=.5,baseline={([yshift=-0.5ex]current bounding box.center)}]{ \draw (0,0)--(0,1)--(1,2)--(1,1)--(0,0); \node at (1.6,0.9) {$\to$}; \draw[cm={1,0,0,1,(2.2,0)}] (0,0)--(0,1)--(0.866025,1.5)--(0.866025,0.5)--(0,0); }, \quad \tikz[scale=.5,baseline={([yshift=-0.5ex]current bounding box.center)}] { \draw (0,0)--(0,1)--(1,1)--(1,0)--(0,0); \node at (1.6,0.4) {$\to$}; \draw[cm={1,0,0,1,(2.2,0)}] (0,0)--(0,1)--(0.866025,0.5)--(0.866025,-0.5)--(0,0);} \quad \mbox{ and }  \quad   \tikz[scale=.5,baseline={([yshift=-0.5ex]current bounding box.center)}] { \draw (0,0)--(1,1)--(2,1)--(1,0)--(0,0); \node at (2.4,0.0) {$\to$}; \draw[cm={1,0,0,1,(3,0)}] (0,0)--(0.866025,0.5)--(1.73205,0)--(0.866025,-0.5)--(0,0);},
\end{align}
so that $\mathcal{H}_{n}$ is mapped by this transformation to a hexagon whose $6$ sides are of equal length. Above the definition \eqref{def of Hn} of $\mathcal{H}_{n}$, we used the standard terminology and called $\mathcal{H}_{n}$ ``the regular hexagon"; note however that $\mathcal{H}_{n}$ becomes truly regular only after applying the transformation \eqref{affine transformation on lozenges}. In the figures, we will assign the colors red, green and yellow for the three lozenges in \eqref{affine transformation on lozenges}, from left to right, respectively.
\subsection{Definition of the model}
The regular hexagon $\mathcal{H}_{n}$ has corners located at $(0,0)$, $(0,n)$, $(n,2n)$, $(2n,2n)$, $(2n,n)$ and $(n,0)$. We normalize the lozenges such that they cover each a surface of area $1$, and the vertices of the lozenges have integer coordinates. We recall that each lozenge tiling of $\mathcal{H}_{n}$ gives rise, through \eqref{non-intersecting paths}, to a system of $n$ non-intersecting paths. These paths live on the graph $\mathcal{G}_{n}$, which is illustrated in Figure \ref{fig: lattice} (left) for $n=4$. The vertices of $\mathcal{G}_{n}$ form a subset of $\mathbb{Z}\times (\mathbb{Z} + \frac{1}{2})$, and the bottom left vertex has coordinates $(0,\frac{1}{2})$. We denote the paths by
\begin{align}\label{def of the paths}
\mathfrak{p}_{j} : \{0,1,\ldots,2n\} \to \frac{1}{2}+\mathbb{Z},  \qquad j = 0,\ldots,n-1,
\end{align}
and they satisfy the initial positions $\mathfrak{p}_{j}(0) = j + \frac{1}{2}$ and ending positions $\mathfrak{p}_{j}(2n) = n + j + \frac{1}{2}$. The particular $2 \times 2$ periodic lozenge tiling model that we consider depends on a parameter $\alpha \in (0,1]$. The weightings are defined on the $2 \times 2$ bottom left block of the lattice as shown in Figure \ref{fig: 2x2 periodic weightings} (left), and is then extended by periodicity as shown in Figure \ref{fig: 2x2 periodic weightings} (right). More formally, if $\mathfrak{e} = \big( (x_{1},y_{1}+\frac{1}{2}),(x_{2},y_{2}+\frac{1}{2}) \big)$ is an edge of $\mathcal{G}_{n}$, then 
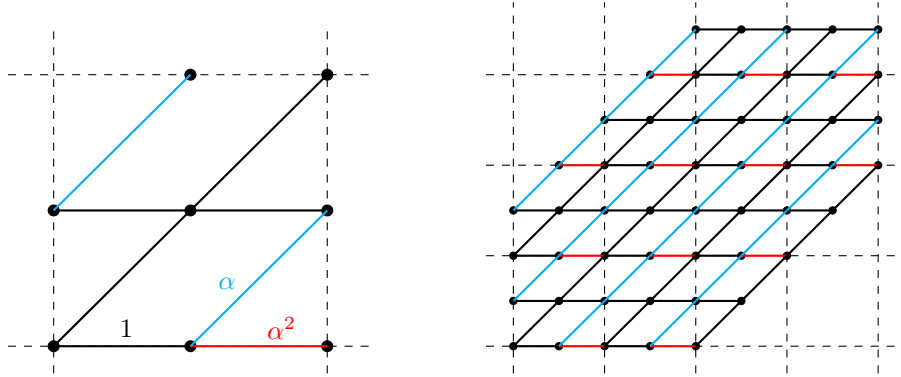
\begin{figure}
\begin{center}
\vspace{0.2cm}
\begin{tikzpicture}[master, scale = 1.2]
\node at (0,0) {};
\draw[dashed] (-0.5,0)--(3.5,0);
\draw[dashed] (-0.5,3)--(3.5,3);
\draw[dashed] (0,-0.3)--(0,3.5);
\draw[dashed] (3,-0.3)--(3,3.5);

\draw[fill] (0,0) circle (0.6mm);
\draw[fill] (1.5,0) circle (0.6mm);
\draw[fill] (3,0) circle (0.6mm);

\draw[fill] (0,1.5) circle (0.6mm);
\draw[fill] (1.5,1.5) circle (0.6mm);
\draw[fill] (3,1.5) circle (0.6mm);

\draw[fill] (1.5,3) circle (0.6mm);
\draw[fill] (3,3) circle (0.6mm);

% Horizontal arrows 
\draw[thick] (0,0)--(1.5,0);
\draw[red,thick] (1.5,0)--(3,0);

\draw[thick] (0,1.5)--(1.5,1.5);
\draw[thick] (1.5,1.5)--(3,1.5);

% Oblique arrows
\draw[thick] (0,0)--(1.5,1.5);
\draw[cyan,thick] (1.5,0)--(3,1.5);

\draw[cyan,thick] (0,1.5)--(1.5,3);
\draw[thick] (1.5,1.5)--(3,3);

\node at (0.8,0.2) {$1$};
\node at (1.9,0.7) {\color{cyan} $\alpha$};
\node at (2.5,0.2) {\color{red} $\alpha^{2}$};
\end{tikzpicture}
\hspace{1cm}
\begin{tikzpicture}[slave,scale = 1.2]
\node at (0,0) {};
%\node at (1.65,1.45) {\includegraphics[width=3.5cm]{../Image/Tiling/n_4_uniform_paths.jpg}};

\draw[dashed] (-0.3,0)--(4.3,0);
\draw[dashed] (-0.3,1)--(4.3,1);
\draw[dashed] (-0.3,2)--(4.3,2);
\draw[dashed] (-0.3,3)--(4.3,3);
\draw[dashed] (0,-0.3)--(0,3.8);
\draw[dashed] (1,-0.3)--(1,3.8);
\draw[dashed] (2,-0.3)--(2,3.8);
\draw[dashed] (3,-0.3)--(3,3.8);
\draw[dashed] (4,-0.3)--(4,3.8);

% The vertices
\draw[fill] (0,0) circle (0.4mm);
\draw[fill] (0.5,0) circle (0.4mm);
\draw[fill] (1,0) circle (0.4mm);
\draw[fill] (1.5,0) circle (0.4mm);
\draw[fill] (2,0) circle (0.4mm);

\draw[fill] (0,0.5) circle (0.4mm);
\draw[fill] (0.5,0.5) circle (0.4mm);
\draw[fill] (1,0.5) circle (0.4mm);
\draw[fill] (1.5,0.5) circle (0.4mm);
\draw[fill] (2,0.5) circle (0.4mm);
\draw[fill] (2.5,0.5) circle (0.4mm);

\draw[fill] (0,1) circle (0.4mm);
\draw[fill] (0.5,1) circle (0.4mm);
\draw[fill] (1,1) circle (0.4mm);
\draw[fill] (1.5,1) circle (0.4mm);
\draw[fill] (2,1) circle (0.4mm);
\draw[fill] (2.5,1) circle (0.4mm);
\draw[fill] (3,1) circle (0.4mm);

\draw[fill] (0,1.5) circle (0.4mm);
\draw[fill] (0.5,1.5) circle (0.4mm);
\draw[fill] (1,1.5) circle (0.4mm);
\draw[fill] (1.5,1.5) circle (0.4mm);
\draw[fill] (2,1.5) circle (0.4mm);
\draw[fill] (2.5,1.5) circle (0.4mm);
\draw[fill] (3,1.5) circle (0.4mm);
\draw[fill] (3.5,1.5) circle (0.4mm);

\draw[fill] (0.5,2) circle (0.4mm);
\draw[fill] (1,2) circle (0.4mm);
\draw[fill] (1.5,2) circle (0.4mm);
\draw[fill] (2,2) circle (0.4mm);
\draw[fill] (2.5,2) circle (0.4mm);
\draw[fill] (3,2) circle (0.4mm);
\draw[fill] (3.5,2) circle (0.4mm);
\draw[fill] (4,2) circle (0.4mm);

\draw[fill] (1,2.5) circle (0.4mm);
\draw[fill] (1.5,2.5) circle (0.4mm);
\draw[fill] (2,2.5) circle (0.4mm);
\draw[fill] (2.5,2.5) circle (0.4mm);
\draw[fill] (3,2.5) circle (0.4mm);
\draw[fill] (3.5,2.5) circle (0.4mm);
\draw[fill] (4,2.5) circle (0.4mm);

\draw[fill] (1.5,3) circle (0.4mm);
\draw[fill] (2,3) circle (0.4mm);
\draw[fill] (2.5,3) circle (0.4mm);
\draw[fill] (3,3) circle (0.4mm);
\draw[fill] (3.5,3) circle (0.4mm);
\draw[fill] (4,3) circle (0.4mm);

\draw[fill] (2,3.5) circle (0.4mm);
\draw[fill] (2.5,3.5) circle (0.4mm);
\draw[fill] (3,3.5) circle (0.4mm);
\draw[fill] (3.5,3.5) circle (0.4mm);
\draw[fill] (4,3.5) circle (0.4mm);

% the horizontal arrows

\draw[thick] (0,0)--(0.5,0);
\draw[red,thick] (0.5,0)--(1,0);
\draw[thick] (1,0)--(1.5,0);
\draw[red,thick] (1.5,0)--(2,0);

\draw[thick] (0,0.5)--(0.5,0.5);
\draw[thick] (0.5,0.5)--(1,0.5);
\draw[thick] (1,0.5)--(1.5,0.5);
\draw[thick] (1.5,0.5)--(2,0.5);
\draw[thick] (2,0.5)--(2.5,0.5);

\draw[thick] (0,1)--(0.5,1);
\draw[red,thick] (0.5,1)--(1,1);
\draw[thick] (1,1)--(1.5,1);
\draw[red,thick] (1.5,1)--(2,1);
\draw[thick] (2,1)--(2.5,1);
\draw[red,thick] (2.5,1)--(3,1);

\draw[thick] (0,1.5)--(0.5,1.5);
\draw[thick] (0.5,1.5)--(1,1.5);
\draw[thick] (1,1.5)--(1.5,1.5);
\draw[thick] (1.5,1.5)--(2,1.5);
\draw[thick] (2,1.5)--(2.5,1.5);
\draw[thick] (2.5,1.5)--(3,1.5);
\draw[thick] (3,1.5)--(3.5,1.5);

\draw[red,thick] (0.5,2)--(1,2);
\draw[thick] (1,2)--(1.5,2);
\draw[red,thick] (1.5,2)--(2,2);
\draw[thick] (2,2)--(2.5,2);
\draw[red,thick] (2.5,2)--(3,2);
\draw[thick] (3,2)--(3.5,2);
\draw[red,thick] (3.5,2)--(4,2);

\draw[thick] (1,2.5)--(1.5,2.5);
\draw[thick] (1.5,2.5)--(2,2.5);
\draw[thick] (2,2.5)--(2.5,2.5);
\draw[thick] (2.5,2.5)--(3,2.5);
\draw[thick] (3,2.5)--(3.5,2.5);
\draw[thick] (3.5,2.5)--(4,2.5);

\draw[red,thick] (1.5,3)--(2,3);
\draw[thick] (2,3)--(2.5,3);
\draw[red,thick] (2.5,3)--(3,3);
\draw[thick] (3,3)--(3.5,3);
\draw[red,thick] (3.5,3)--(4,3);

\draw[thick] (2,3.5)--(2.5,3.5);
\draw[thick] (2.5,3.5)--(3,3.5);
\draw[thick] (3,3.5)--(3.5,3.5);
\draw[thick] (3.5,3.5)--(4,3.5);

% the oblique arrows

\draw[thick] (0,0)--(0.5,0.5);
\draw[cyan,thick] (0.5,0)--(1,0.5);
\draw[thick] (1,0)--(1.5,0.5);
\draw[cyan,thick] (1.5,0)--(2,0.5);
\draw[thick] (2,0)--(2.5,0.5);

\draw[cyan,thick] (0,0.5)--(0.5,1);
\draw[thick] (0.5,0.5)--(1,1);
\draw[cyan,thick] (1,0.5)--(1.5,1);
\draw[thick] (1.5,0.5)--(2,1);
\draw[cyan,thick] (2,0.5)--(2.5,1);
\draw[thick] (2.5,0.5)--(3,1);

\draw[thick] (0,1)--(0.5,1.5);
\draw[cyan,thick] (0.5,1)--(1,1.5);
\draw[thick] (1,1)--(1.5,1.5);
\draw[cyan,thick] (1.5,1)--(2,1.5);
\draw[thick] (2,1)--(2.5,1.5);
\draw[cyan,thick] (2.5,1)--(3,1.5);
\draw[thick] (3,1)--(3.5,1.5);

\draw[cyan,thick] (0,1.5)--(0.5,2);
\draw[thick] (0.5,1.5)--(1,2);
\draw[cyan,thick] (1,1.5)--(1.5,2);
\draw[thick] (1.5,1.5)--(2,2);
\draw[cyan,thick] (2,1.5)--(2.5,2);
\draw[thick] (2.5,1.5)--(3,2);
\draw[cyan,thick] (3,1.5)--(3.5,2);
\draw[thick] (3.5,1.5)--(4,2);

\draw[cyan,thick] (0.5,2)--(1,2.5);
\draw[thick] (1,2)--(1.5,2.5);
\draw[cyan,thick] (1.5,2)--(2,2.5);
\draw[thick] (2,2)--(2.5,2.5);
\draw[cyan,thick] (2.5,2)--(3,2.5);
\draw[thick] (3,2)--(3.5,2.5);
\draw[cyan,thick] (3.5,2)--(4,2.5);

\draw[cyan,thick] (1,2.5)--(1.5,3);
\draw[thick] (1.5,2.5)--(2,3);
\draw[cyan,thick] (2,2.5)--(2.5,3);
\draw[thick] (2.5,2.5)--(3,3);
\draw[cyan,thick] (3,2.5)--(3.5,3);
\draw[thick] (3.5,2.5)--(4,3);

\draw[cyan,thick] (1.5,3)--(2,3.5);
\draw[thick] (2,3)--(2.5,3.5);
\draw[cyan,thick] (2.5,3)--(3,3.5);
\draw[thick] (3,3)--(3.5,3.5);
\draw[cyan,thick] (3.5,3)--(4,3.5);
\end{tikzpicture}
\end{center}
\caption{\label{fig: 2x2 periodic weightings} The black edges have weight $1$, the cyan edges have weight $\alpha$, and the red edges have weight $\alpha^{2}$.}
\end{figure}
\begin{align}\label{def weightings}
w_{\mathfrak{e}} =\begin{cases}
\alpha^{2}, & \mbox{if } x_{1} \mbox{ is odd, } y_{1}=y_{2}, \mbox{ and } y_{1} \mbox{ is even,} \\
\alpha, & \mbox{if } x_{1}+y_{1} \mbox{ is odd, and } y_{2} = y_{1}+1, \\
1 & \mbox{otherwise}.
\end{cases}
\end{align}
%The weightings on the edges of $\mathcal{G}_{n}$ is illustrated in Figure \ref{fig: 2x2 periodic weightings} (right). 

\vspace{0.2cm}For any values of $\alpha \in (0,1]$, the weightings \eqref{def weightings} are such that $\mathrm{W}(\mathcal{T)} > 0$ for all $\mathcal{T}$, and thus we have a well-defined probability measure via \eqref{prob over tilings}. On the other hand, if $\alpha = 0$, then several edges have weights $0$, and it is easy to see (e.g. from Figure \ref{fig: 2x2 periodic weightings} (right)) that $\mathrm{W}(\mathcal{T}) = 0$ for all $\mathcal{T}$. So in this case, \eqref{def weightings} does not induce a probability measure, and this explains why we excluded $\alpha = 0$ in the definition of the model. If $\alpha = 1$, all tilings have the same weight, and we recover the uniform distribution. Proposition \ref{prop: frozen} states that for $\alpha < 1$, there is a particular tiling $\mathcal{T}_{\max}$ that is more likely to appear than any other tiling. $\mathcal{T}_{\max}$ is illustrated in Figure \ref{fig: alpha 0 vs alpha 1} (left) for $n=60$. 
\begin{proposition}\label{prop: frozen}
Let $\alpha \in (0,1)$ and let $n \geq 1$ be an integer. There exists a unique tiling $\mathcal{T}_{\max}$ of $\mathcal{H}_{n}$ such that $\mathrm{W}(\mathcal{T}) \leq \alpha \mathrm{W}(\mathcal{T}_{\max})$ for all $\mathcal{T} \neq \mathcal{T}_{\max}$. Furthermore,
\begin{align*}
\mathrm{W}(\mathcal{T}_{\max}) = \begin{cases}
\alpha^{\frac{n^{2}}{4}}, & \mbox{if } n \mbox{ is even}, \\
\alpha^{\frac{n^{2}-1}{4}}, & \mbox{if } n \mbox{ is odd.} 
\end{cases}
\end{align*}
\end{proposition}
\begin{proof}
The proof of Proposition \ref{prop: frozen} is based from a careful inspection of $\mathcal{G}_{n}$, and we omit the details.
\end{proof}

\begin{figure}[h]
\begin{center}
\includegraphics[width=5cm]{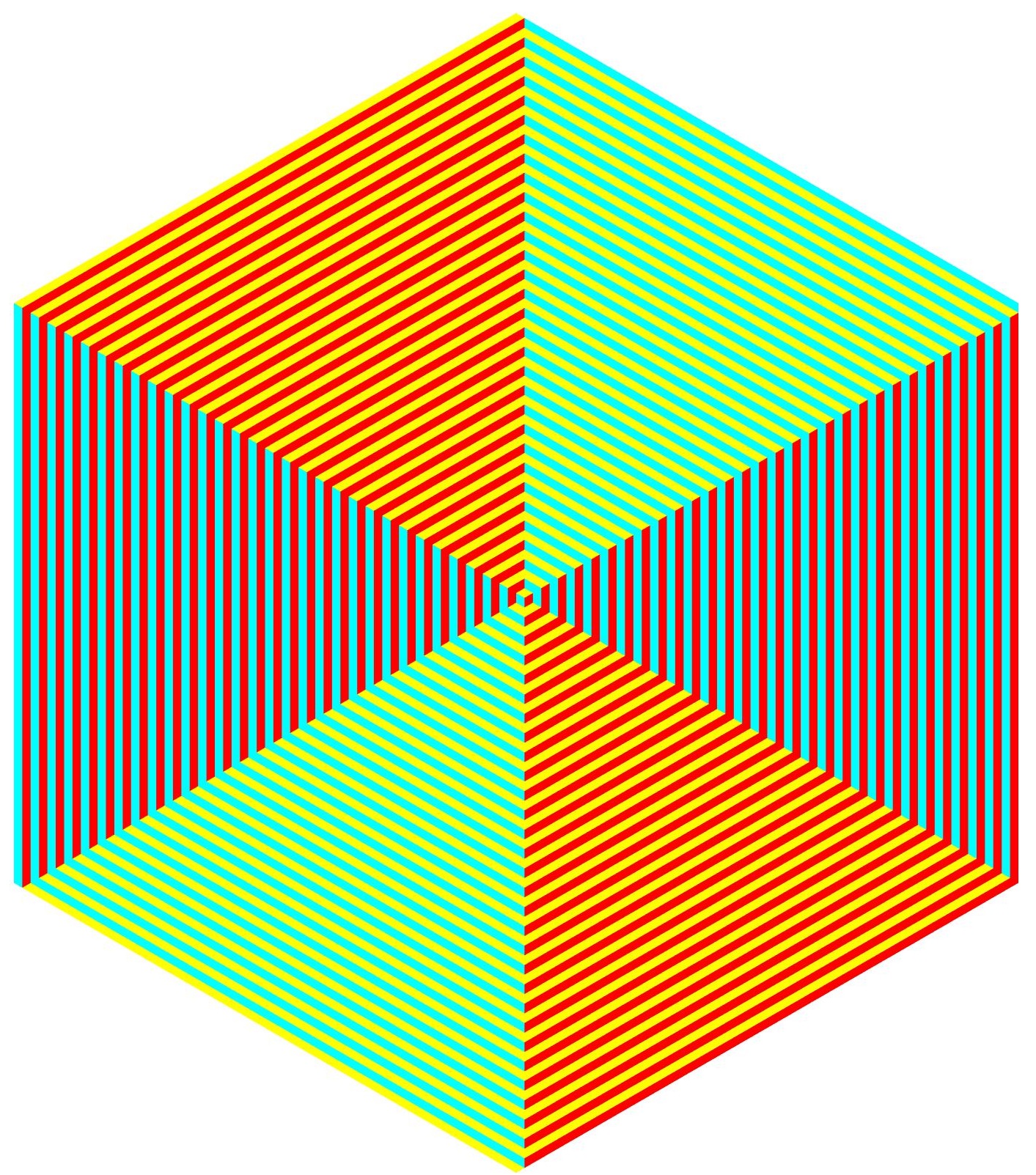}
\hspace{1cm}
\includegraphics[width=5cm]{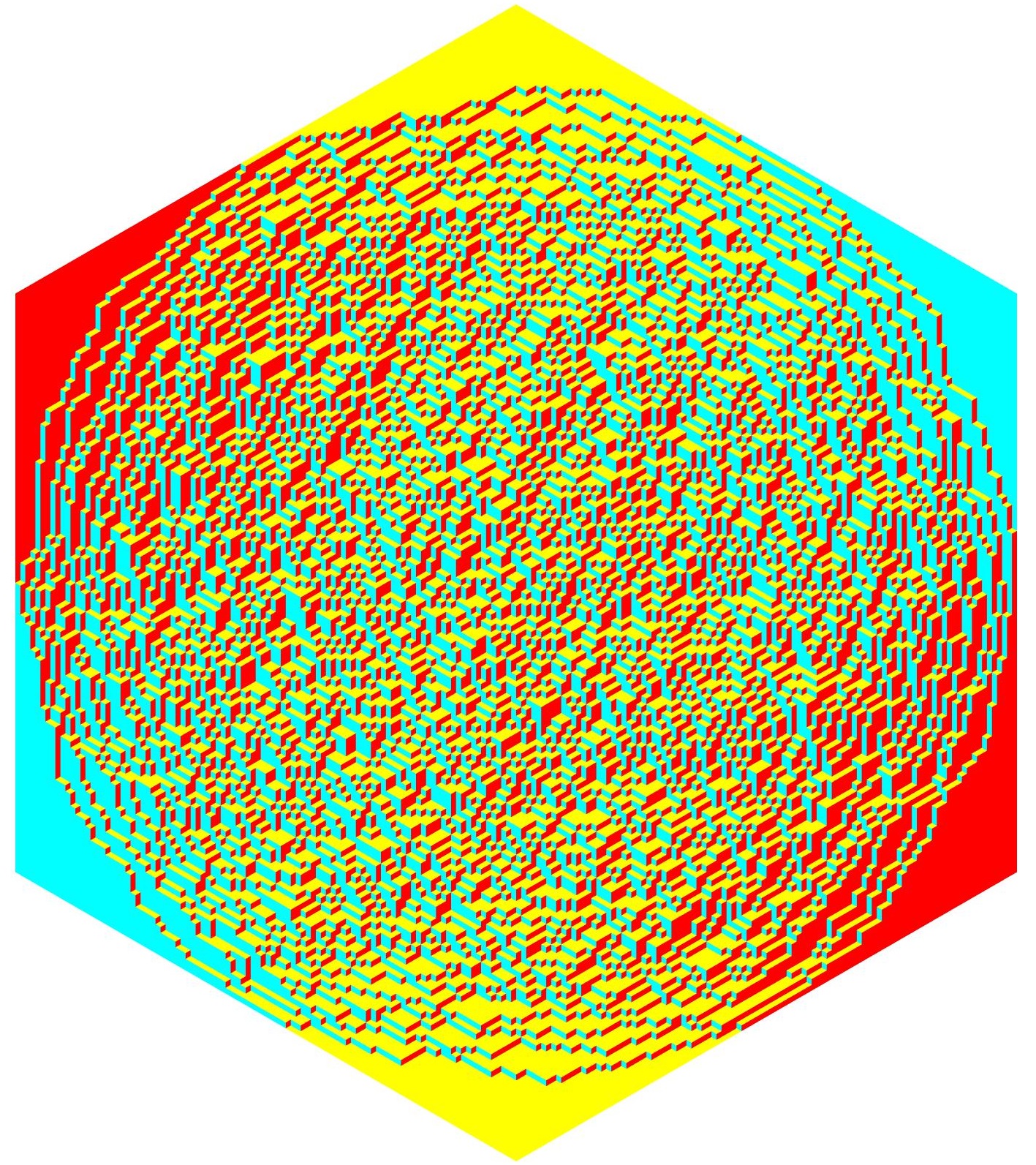}
\end{center}
\caption{\label{fig: alpha 0 vs alpha 1}Two tilings taken at random accordingly to the measure induced by \eqref{def weightings}, for $n = 60$ and $\alpha = 5\times 10^{-4}$ (left), and $n = 100$ and $\alpha = 1$ (right).}
\end{figure}

It follows from Proposition \ref{prop: frozen} that, as $\alpha \to 0$, the randomness disappears because the tiling $\mathcal{T}_{\max}$ becomes significantly more likely than any other tiling. Therefore, our model interpolates between the uniform measure over the tilings (for $\alpha = 1$) and a particular totally frozen tiling $\mathcal{T}_{\max}$ (as $\alpha \to 0$), see Figures \ref{fig: alpha 0 vs alpha 1} and \ref{fig: alpha between 0 and 1}. Intriguingly, these figures show similarities with the rectangle-triangle tiling of the hexagon obtained by Keating and Sridhar in \cite[Figure 18]{KeaSri}.
\begin{figure}
\begin{center}
\includegraphics[width=4.5cm]{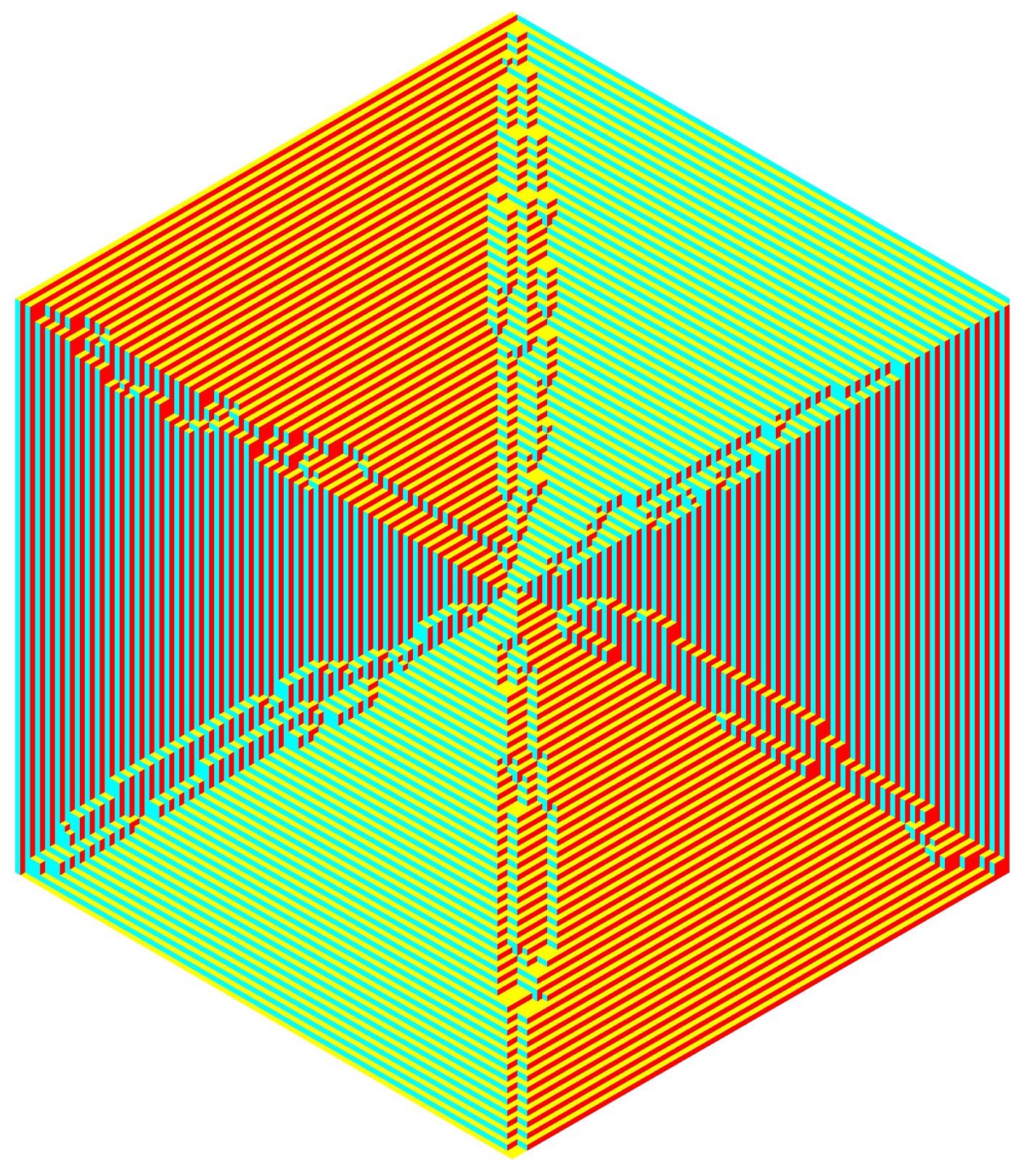}
\hspace{0.5cm}
\includegraphics[width=4.5cm]{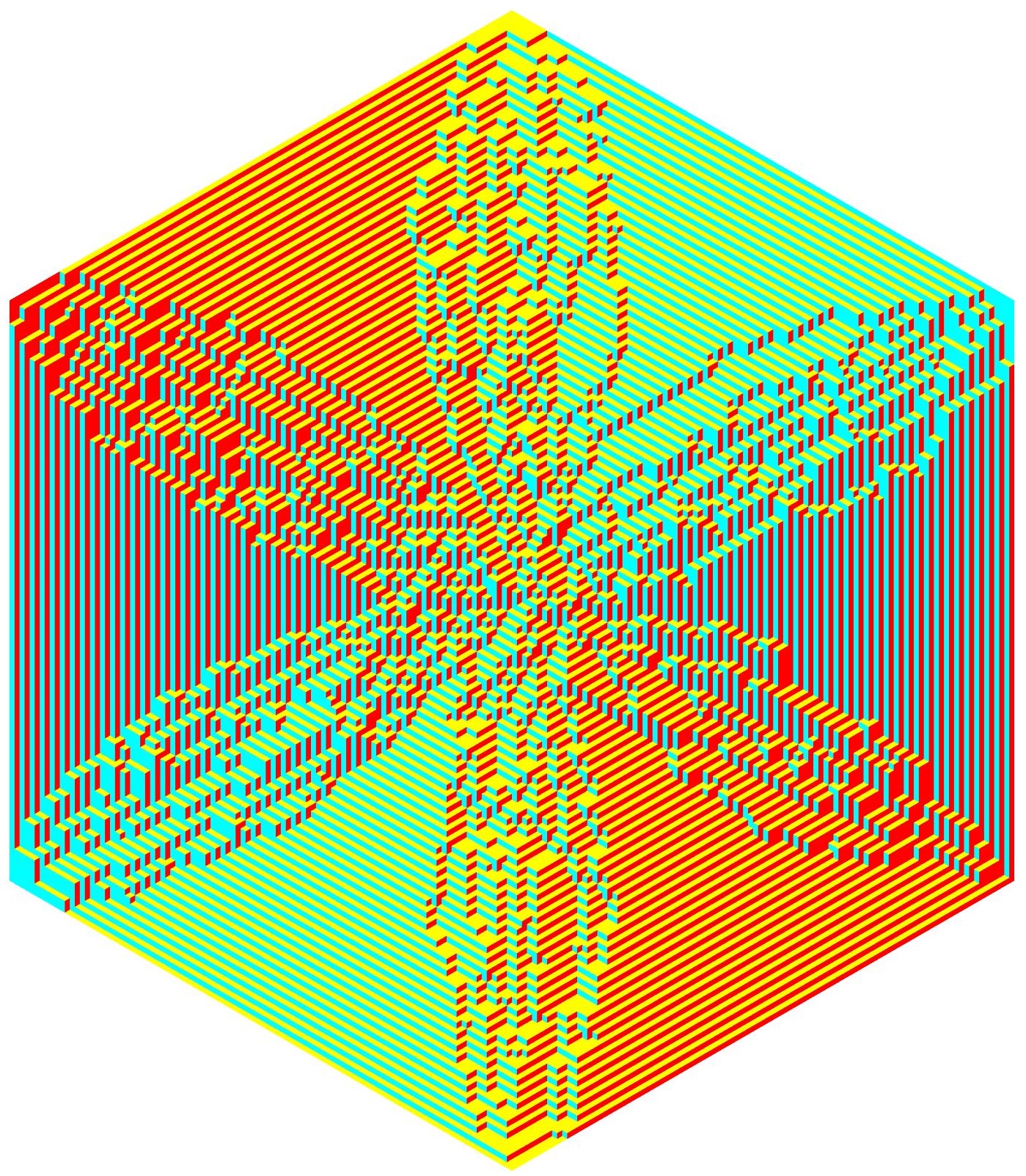}
\hspace{0.5cm}
\includegraphics[width=4.5cm]{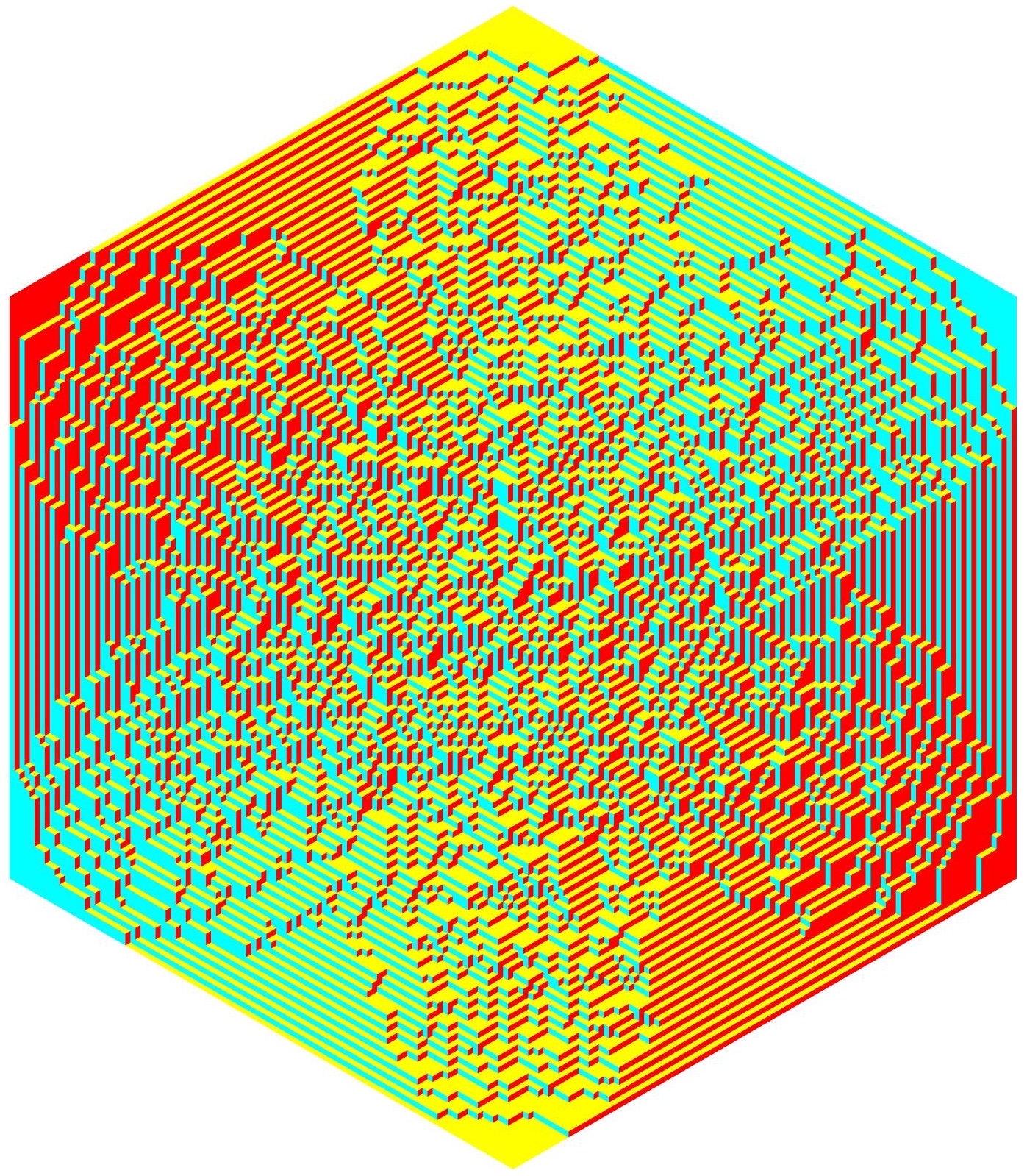}
\end{center}
\caption{\label{fig: alpha between 0 and 1}Three tilings taken at random accordingly to the measure induced by \eqref{def weightings} with $n=100$ and $\alpha = 0.01$ (left), $\alpha = 0.05$ (middle), $\alpha = 0.2$ (right).}
\end{figure}

\vspace{0.2cm}Several tiling models in the literature (e.g. those considered in \cite{BGR} and \cite{CDKL}) are defined by weightings on the lozenges, instead of weightings on the edges. To ease possible comparisons with these models, we give an alternative definition of our model. The weight $\mathrm{W}(\mathcal{T})$ of a tiling $\mathcal{T}$ can alternatively be defined as
\begin{align*}%\label{lol8}
\mathrm{W}(\mathcal{T}) = \prod_{{\tikz[scale=.18]  \draw (0,0) \lozu;} \; \in \mathcal{T}} w(\tikz[scale=.27] { \draw (0,0) \lozu;}) \prod_{\tikz[scale=.18]  {\draw (0,0) \lozr;} \; \raisebox{0.11cm}{\scriptsize $\in \hspace{-0.08cm}\mathcal{T}$}} w\Big(\raisebox{-0.15cm}{\tikz[scale=.27] \draw (0,0) \lozr; } \hspace{-0.05cm}\Big),
\end{align*}
where $w$ is the weight function over the lozenges given by
\begin{align}
& w\bigg(	\tikz[scale=.3,baseline=(current bounding box.center)]{\draw (0,0) \lozu; \draw[fill] (0,0) circle(5pt);  \node[below] (i) at (0,-.2) {\tiny{$(i,j)$}};} 	\bigg)=
\begin{cases}
\alpha^{2}, & \mbox{if }i \text{ is odd and }j \mbox{ is even,}\\
1, &  \text{otherwise}, 
\end{cases} \label{weight horizontal} 
\\
& w\bigg( \hspace{-0.15cm}\raisebox{-0.5cm}{
\tikz[scale=.27] {
\draw (0,0) \lozr; 
\draw[fill] (0,0) circle(5pt); 
\node[below] (i) at (0,-.2) {\tiny{$(i,j)$}};
}
} \hspace{-0.1cm} \bigg)=
\begin{cases}
\alpha, & \mbox{if }i+j \text{ is odd,}\\
1, &  \text{otherwise},
\end{cases}  \label{weight oblique}
\end{align}
where $\alpha \in (0,1]$. The above lozenge weightings depend only on the parity of $i$ and $j$, and thus are periodic of period $2$ in both directions. By using the correspondence \eqref{non-intersecting paths} between lozenge tilings and non-intersecting paths, it is straightforward to verify that the weightings \eqref{weight horizontal}--\eqref{weight oblique} define the same measure as the weightings \eqref{def weightings}.

%\vspace{0.2cm}In the present work, we follow \cite{CDKL} and apply the same schematic idea of combining these two techniques of asymptotic analysis. However, in our case we have $p=2$, which is an important difference with \cite{CDKL}, since we deal with $2 \times 2$ matrix valued OPs (and a $4 \times 4$ RH problem) while in \cite{CDKL} the OPs are scalar.

\subsection{Matrix valued orthogonal polynomials}\label{subsection: MVOPs}
It will be convenient for us to define $\mathcal{G}_{\infty}$ as the graph whose vertex set is $\mathbb{Z}\times (\mathbb{Z} + \frac{1}{2})$, and whose edges are of the form $\mathfrak{e}= \big( (x_{1},y_{1}+\frac{1}{2}),(x_{2},y_{2}+\frac{1}{2}) \big)$ with $x_{2} = x_{1}+1$ and $y_{2}-y_{1} \in \{0,1\}$. The weighting \eqref{def weightings} was defined on the edges of $\mathcal{G}_{n}$, but it can be straightforwardly extended to the edges of $\mathcal{G}_{\infty}$. We follow the notations of \cite[equation (4.3)]{DK}, and denote $T_{x,x+1}(y_{1},y_{2})$ for the weight associated to the edge $\mathfrak{e} = \big( (x,y_{1}+\frac{1}{2}),(x+1,y_{2}+\frac{1}{2}) \big)$ of $\mathcal{G}_{\infty}$. This weight can be obtained from \eqref{def weightings} and only depends on the parity of $x$. If $x$ is even, it is given by
\begin{align}\label{Transition x even}
T_{x,x+1}(y_{1},y_{2}) = \begin{cases}
1 & \mbox{if }y_{2}=y_{1}, \\
1 & \mbox{if }y_{2} = y_{1}+1 \mbox{ and } y_{1} \mbox{ is even}, \\
\alpha & \mbox{if } y_{2} = y_{1}+1 \mbox{ and } y_{1} \mbox{ is odd}, \\
0 & \mbox{otherwise},
\end{cases}
\end{align}
while if $x$ is odd, we have
\begin{align}\label{Transition x odd}
T_{x,x+1}(y_{1},y_{2}) = \begin{cases}
\alpha^{2} & \mbox{if }y_{2}=y_{1}, \mbox{ and } y_{1} \mbox{ is even}, \\
1 & \mbox{if }y_{2}=y_{1}, \mbox{ and } y_{1} \mbox{ is odd}, \\
\alpha & \mbox{if }y_{2} = y_{1}+1 \mbox{ and } y_{1} \mbox{ is even}, \\
1 & \mbox{if } y_{2} = y_{1}+1 \mbox{ and } y_{1} \mbox{ is odd}, \\
0 & \mbox{otherwise}.
\end{cases}
\end{align}
For each $x \in \mathbb{Z}$, $T_{x,x+1}$ is periodic of period $2$, namely $T_{x,x+1}(y_{1}+2,y_{2}+2) = T_{x,x+1}(y_{1},y_{2})$ for all $y_{1},y_{2} \in \mathbb{Z}$. The weightings $T_{x,x+1}$ can be represented as two $2 \times 2$ block Toeplitz matrices (one for $x$ even, and one for $x$ odd) that are infinite in both directions. %If $x$ is even, the diagonal block is $\begin{pmatrix}
%1 & 1 \\ 0 & 1
%\end{pmatrix}$, and the block on the first upper diagonal is $\begin{pmatrix}
%0 & 0 \\ \alpha & 0
%\end{pmatrix}$.  
These two infinite matrices can be encoded in two $2 \times 2$ matrix symbols $A_{x,x+1}(z)$, whose entries $(A_{x,x+1}(z))_{i+1,j+1}$, $0 \leq i,j \leq 1$, are given by
\begin{align*}
(A_{x,x+1}(z))_{i+1,j+1} = T_{x,x+1}(i,j) + zT_{x,x+1}(i,j+2).
\end{align*}
%Here is a more visual way of computating the entries of $A_{x,x+1}$ for $x$ even (the case of $x$ odd can also be obtained in the same way):
%\begin{center}
%\begin{tikzpicture}[master]
%\node at (0,0) {};
%\draw[dashed] (-0.5,0)--(2.5,0);
%\draw[dashed] (-0.5,2)--(2.5,2);
%\draw[dashed] (0,-0.3)--(0,2.5);
%\draw[dashed] (2,-0.3)--(2,2.5);
%
%\draw[fill] (0,0) circle (0.6mm);
%\draw[fill] (1,0) circle (0.6mm);
%\draw[fill] (2,0) circle (0.6mm);
%
%\draw[fill] (0,1) circle (0.6mm);
%\draw[fill] (1,1) circle (0.6mm);
%\draw[fill] (2,1) circle (0.6mm);
%
%\draw[fill] (1,2) circle (0.6mm);
%\draw[fill] (2,2) circle (0.6mm);
%
%\draw[fill] (2,3) circle (0.6mm);
%
%% Horizontal arrows 
%\draw[thick] (0,0)--(1,0);
%\draw[thick] (1,0)--(2,0);
%
%\draw[thick] (0,1)--(1,1);
%\draw[thick] (1,1)--(2,1);
%
%\draw[thick] (1,2)--(2,2);
%
%% Oblique arrows
%\draw[thick] (0,0)--(1,1);
%\draw[thick] (1,0)--(2,1);
%
%\draw[thick] (0,1)--(1,2);
%\draw[thick] (1,1)--(2,2);
%
%\draw[thick] (1,2)--(2,3);
%
%\node at (0.4,0.2) {$1$};
%\node at (1.2,0.4) {\color{cyan} $\alpha$};
%\node at (1.6,0.2) {\color{red} $\alpha^{2}$};
%
%\node at (1,-0.3) {$$};
%\end{tikzpicture}
%\end{center}
More explicitly, this gives
\begin{align}\label{1-transition symbol}
A_{x,x+1}(z) = \begin{cases}
\begin{pmatrix}
1 & 1 \\
\alpha z & 1
\end{pmatrix}, &  \mbox{if }x \mbox{ is even}, \\[0.4cm]
\begin{pmatrix}
\alpha^{2} & \alpha \\
z & 1
\end{pmatrix}, &  \mbox{if }x \mbox{ is odd},
\end{cases}
\end{align}
and we can retrieve the entries of $T_{x,x+1}$ from its symbol by
\begin{align*}
\begin{pmatrix}
T_{x,x+1}(2y_{1},2y_{2}) & T_{x,x+1}(2y_{1},2y_{2}+1) \\
T_{x,x+1}(2y_{1}+1,2y_{2}) & T_{x,x+1}(2y_{1}+1,2y_{2}+1)
\end{pmatrix} = \frac{1}{2\pi i}\int_{\gamma} A_{x,x+1}(z) z^{y_{1}-y_{2}} \frac{dz}{z},
\end{align*}
where $\gamma$ is any closed contour going around $0$ once in the positive direction. The symbol associated to $\mathcal{G}_{n}$ is then obtained by taking the following product (see \cite[equation (4.9)]{DK}):
\begin{align*}
A_{0,2n}(z) = \prod_{x=0}^{2n-1}A_{x,x+1}(z) = A(z)^{n},
\end{align*}
where
\begin{align}\label{def of A}
A(z) := \begin{pmatrix}
1 & 1 \\
\alpha z & 1
\end{pmatrix} \begin{pmatrix}
\alpha^{2} & \alpha \\
z & 1
\end{pmatrix} = \begin{pmatrix}
\alpha^{2}+z & 1+\alpha \\
(1+\alpha^{3})z & 1+\alpha^{2}z
\end{pmatrix}.
\end{align}
In order to limit the length and technicalities of the paper, from now we take the size of the hexagon even, i.e. $n = 2N$ where $N$ is a positive integer. This is made for convenience; the case of odd integer $n$ could also be analyzed in a similar way, but then a discussion on the partity of $n$ is needed. Since $n=2N$, following \cite[equation (4.12)]{DK}, the relevant orthogonality weight to consider is
\begin{align}
\frac{A(z)^{2N}}{z^{2N}}. \label{def of weight}
\end{align}
We consider two families $\{P_{j}\}_{j\geq 0}$ and $\{Q_{j}\}_{j\geq 0}$ of $2 \times 2$ matrix valued OPs defined by
\begin{align}
& P_{j}(z) = z^{j}I_{2}+\bigO(z^{j-1}), \qquad \mbox{as } z \to \infty \label{cond for Pn} \\
& \frac{1}{2\pi i}\int_{\gamma} P_{j}(z) \frac{A(z)^{2N}}{z^{2N}}z^{k}dz = 0_{2}, & & k = 0,\ldots,j-1, \nonumber 
\end{align}
and
\begin{align}
& \frac{1}{2\pi i}\int_{\gamma} Q_{j}(z) \frac{A(z)^{2N}}{z^{2N}}z^{j}dz = -I_{2}, \label{cond for Qn} \\
& \frac{1}{2\pi i}\int_{\gamma} Q_{j}(z) \frac{A(z)^{2N}}{z^{2N}}z^{k}dz = 0_{2}, & & k = 0,\ldots,j-1, \nonumber
\end{align}
where $0_{2}$ denotes the $2\times 2$ zero matrix, $I_{2}$ is the identity matrix, and $\gamma$ is, as before, a closed contour surrounding $0$ once in the positive direction. Since the weight \eqref{def of weight} is not hermitian, there is no guarantee that the above OPs exist for every $j$. However, it follows from \cite[Lemma 4.8 and equation (4.32)]{DK} that $P_{N}$ and $Q_{N-1}$ exist.

\subsection{The $4 \times 4$ Riemann-Hilbert problem for $Y$}
RH problems for scalar orthogonal polynomials have been introduced by Fokas, Its and Kitaev in \cite{FIK}. Here, we need the generalization of this result for matrix valued OPs, which can be found in \cite{CassaManas2012, Delvaux, GrunIglesia}. Consider the $4 \times 4$ matrix valued function $Y(z) = Y(z;\alpha,N)$ defined by
\begin{equation}\label{Y definition}
Y(z) = \begin{pmatrix}
P_{N}(z) & \displaystyle \frac{1}{2\pi i} \int_{\gamma} P_{N}(s)\frac{A^{2N}(s)}{s^{2N}}\frac{ds}{s-z} \\
Q_{N-1}(z) & \displaystyle \frac{1}{2\pi i} \int_{\gamma} Q_{N-1}(s)\frac{A^{2N}(s)}{s^{2N}}\frac{ds}{s-z}
\end{pmatrix}, \qquad z \in \mathbb{C}\setminus \gamma.
\end{equation}
The matrix $Y$ is characterized as the unique solution to the following RH problem. 
\subsubsection*{RH problem for $Y$}
\begin{itemize}
\item[(a)] $Y : \mathbb{C}\setminus \gamma \to \mathbb{C}^{4\times 4}$ is analytic.
\item[(b)] The limits of $Y(z)$ as $z$ approaches $\gamma$ from inside and outside exist, are continuous on $\gamma$ and are denoted by $Y_+$ and $Y_-$ respectively. Furthermore, they are related by
\begin{equation}\label{jump relations of Y}
Y_{+}(z) = Y_{-}(z) \begin{pmatrix}
I_{2} & \frac{A^{2N}(z)}{z^{2N}} \\ 0_{2} & I_{2}
\end{pmatrix}, \hspace{0.5cm} \mbox{ for } z \in \gamma.
\end{equation}
\item[(c)] As $z \to \infty$, we have $Y(z) = \left(I_{4} + \bigO(z^{-1})\right) \begin{pmatrix}
z^{N}I_{2} & 0_{2} \\ 0_{2} & z^{-N}I_{2}
\end{pmatrix}$.  
\end{itemize}

\subsection{Double contour formula from \cite{DK} for the kernel}
As mentioned in the introduction, the point process obtained by putting points on the paths, as shown in \eqref{non-intersecting paths}, is determinantal. We let $K$ denote the associated kernel. By definition of determinantal point processes, for integers $k \geq 1$, and $x_{1},\ldots,x_{k},y_{1},\ldots,y_{k}$ with $(x_{i},y_{i}) \neq (x_{j},y_{j})$ if $i \neq j$ we have
\begin{align}\label{def of determinantal}
\mathbb{P}\bigg[ \begin{matrix}
\mathfrak{p}_{0},\ldots,\mathfrak{p}_{2N-1} \mbox{ go through each of the points} \\
(x_{1},y_{1}+\frac{1}{2}),\ldots,(x_{k},y_{k}+\frac{1}{2})
\end{matrix} \bigg] = \det \big[ K(x_{i},y_{i},x_{j},y_{j}) \big]_{i,j=1}^{k}. 
\end{align}
The following proposition follows after specifying the general result \cite[Theorem 4.7]{DK} to our situation.\footnote{The quantities $N,M$ and $L$ in the notation of \cite{DK} are equal to $N$, $N$, and $4N$ in our notation.}
\begin{proposition}\emph{(from \cite{DK})}
Let $\alpha \in (0,1]$. For integers $x_{1},x_{2} \in \{1,2,\ldots,4N-1\}$ and $y_{1},y_{2} \in \mathbb{Z}$, we have
\begin{align}
\big[ K(x_{1},2y_{1}+j,x_{2},2y_{2}+i) \big]_{i,j=0}^{1} = - \frac{\chi_{x_{1}>x_{2}}}{2\pi i}\int_{\gamma} A_{x_{2},x_{1}}(z)z^{y_{2}-y_{1}} \frac{dz}{z} \nonumber \\
 + \frac{1}{(2\pi i)^{2}}\int_{\gamma}\int_{\gamma}\frac{A_{x_{2},4N}(w)}{w^{2N-y_{2}}}\mathcal{R}^{Y}(w,z)\frac{A_{0,x_{1}}(z)}{z^{y_{1}+1}} dzdw \label{kernel general}
\end{align}
where, $A_{a,b}$ is defined by
\begin{align*}
A_{a,b}(z) = \prod_{x=a}^{b-1}A_{x,x+1}(z), \qquad b>a,
\end{align*}
and $\mathcal{R}^{Y}$ is given by
\begin{align}\label{def of mathcal R}
\mathcal{R}^{Y}(w,z) = \frac{1}{z-w} \begin{pmatrix}
0_{2} & I_{2}
\end{pmatrix}Y^{-1}(w)Y(z)\begin{pmatrix}
I_{2} \\ 0_{2}
\end{pmatrix}.
\end{align}
\end{proposition}
As particular cases of the above, we obtain the following formulas.
\begin{corollary}
Let $\alpha \in (0,1]$. For integers $x \in \{1,\ldots,2N-1\}$ and $y \in \mathbb{Z}$, we have
\begin{align}
& \big[ K(2x,2y+j,2x,2y+i) \big]_{i,j=0}^{1} =  \frac{1}{(2\pi i)^{2}}\int_{\gamma}\int_{\gamma}\frac{A(w)^{2N-x}}{w^{2N-y}}\mathcal{R}^{Y}(w,z)\frac{A(z)^{x}}{z^{y+1}} dzdw \label{kernel diag even}
\end{align}
and
\begin{align}
& \big[ K(2x+1,2y+j,2x+1,2y+i) \big]_{i,j=0}^{1} = \label{kernel diag odd} \\
& \frac{1}{(2\pi i)^{2}}\int_{\gamma}\int_{\gamma}\begin{pmatrix}
\alpha^{2} & \alpha \\
w & 1
\end{pmatrix}\frac{A(w)^{2N-x-1}}{w^{2N-y}}\mathcal{R}^{Y}(w,z)\frac{A(z)^{x}}{z^{y+1}}\begin{pmatrix}
1 & 1 \\
\alpha z & 1
\end{pmatrix} dzdw. \nonumber
\end{align}
\end{corollary}
\begin{proof}
This simply follows from
\begin{align*}
& A_{0,2x}(z) = A(z)^{x}, & & A_{2x,4N}(w) = A(w)^{2N-x}, \\
& A_{0,2x+1}(z) = A(z)^{x}\begin{pmatrix}
1 & 1 \\
\alpha z & 1
\end{pmatrix}, & & A_{2x+1,4N}(w) = \begin{pmatrix}
\alpha^{2} & \alpha \\
w & 1
\end{pmatrix}A(w)^{2N-x-1},
\end{align*}
where we have used \eqref{1-transition symbol} and \eqref{def of A}.
\end{proof}
From \cite[Lemma 4.6]{DK}, $\mathcal{R}^{Y}(w,z)$ is the unique bivariate polynomial of degree $\leq N-1$ in both variables $w$ and $z$ which satisfies the following reproducing property
\begin{equation}\label{reproducing property Y}
\frac{1}{2\pi i} \int_{\gamma} P(w) \frac{A^{2N}(w)}{w^{2N}}\mathcal{R}^{Y}(w,z)dw = P(z),
\end{equation}
for every $2 \times 2$ matrix valued polynomial $P$ of degree $\leq N-1$. Because it satisfies \eqref{reproducing property Y}, $\mathcal{R}^{Y}(w,z)$ is called a reproducing kernel.

\section{Statement of results}\label{Section: main results}
The new double contour formula for the kernel in terms of scalar OPs is stated in Theorem \ref{thm: correlation kernel final scalar expression}. In this formula, the large $N$ behavior of the integrand is roughly $e^{N \Xi}$, for a certain phase function $\Xi$ which in our case is defined on a two-sheeted Riemann surface $\mathcal{R}_{\alpha}$. The restriction of $\Xi$ on the first and second sheet are denoted by $\Phi$ and $\Psi$, respectively. The saddle points are the solutions $\zeta \in \mathbb{C}$ for which either $\Phi'(\zeta) = 0$ or $\Psi'(\zeta) = 0$. In the liquid region, Proposition \ref{prop:saddle} states that there is a unique saddle, denoted $s$, lying in the upper half plane. This saddle plays an important role in our analysis, and some of its properties are stated in Propositions \ref{prop:hightemp} and \ref{prop: s on the Riemann surface}. The limiting densities for the lozenges in the liquid region are given explicitly in terms of $s$ in Theorem \ref{thm:main}.
\begin{remark}\label{rem:alpha is not 1}
If $\alpha = 1$, our model reduces to the uniform measure and the kernel can be expressed in terms of scalar-valued OPs. However, our approach is based on the formulas \eqref{kernel diag even}--\eqref{kernel diag odd}, and even though these formulas are still valid for $\alpha = 1$, this case requires a special attention (because of a different branch cut structure in the analysis). Since the limiting densities for the lozenges in this case are already well-known \cite{CLP}, from now we will assume that $\alpha \in (0,1)$ to avoid unnecessary discussions.
\end{remark}
\subsection{New formula for the kernel in terms of scalar OPs}\label{subsection: new formula for the kernel}
%The formulas \eqref{kernel diag even}--\eqref{kernel diag odd} for the kernel are expressed in terms of $\mathcal{R}^{Y}$, the reproducing kernel for the matrix valued OPs associated to the weight $A^{2N}(z)z^{-2N}$, see \eqref{reproducing property Y}. It requires a priori the analysis of the $4 \times 4$ RH problem for $Y$. Here we obtain another formula for the kernel in terms of the reproducing kernel $\mathcal{R}^{U}$ for certain scalar-valued OPs. This formula involves the solution $U$ to a $2 \times 2$ RH problem, and allow to simplify considerably both the Deift/Zhou steepest descent and the saddle point analysis of Sections \ref{}--\ref{}. 
We define the scalar weight $W$ by
\begin{equation}\label{def of W in intro}
W(\zeta) = \bigg( \frac{(\zeta-\alpha c)(\zeta-\alpha c^{-1})}{\zeta (\zeta-c)(\zeta-c^{-1})} \bigg)^{2N}, \qquad \mbox{ where } \qquad c = \sqrt{\frac{\alpha}{1-\alpha + \alpha^{2}}},
\end{equation}
and consider the following $2 \times 2$ RH problem.
\subsubsection*{RH problem for $U$}
\begin{itemize}
\item[(a)] $U : \mathbb{C}\setminus \gamma_{\mathbb{C}} \to \mathbb{C}^{2\times 2}$ is analytic, where $\gamma_{\mathbb{C}}$ is a closed curve surrounding $c$ and $c^{-1}$ once in the positive direction, but not surrounding $0$.
\item[(b)] The limits of $U(\zeta)$ as $\zeta$ approaches $\gamma_{\mathbb{C}}$ from inside and outside exist, are continuous on $\gamma_{\mathbb{C}}$ and are denoted by $U_+$ and $U_-$ respectively. Furthermore, they are related by
\begin{equation}\label{jump relations of U}
U_{+}(\zeta) = U_{-}(\zeta) \begin{pmatrix}
1 & W(\zeta) \\ 0 & 1
\end{pmatrix}, \hspace{0.5cm} \mbox{ for } \zeta \in \gamma_{\mathbb{C}}.
\end{equation}
\item[(c)] As $\zeta \to \infty$, we have $U(\zeta) = \left(I_{2} + \bigO(\zeta^{-1})\right) \begin{pmatrix}
\zeta^{2N} & 0 \\ 0 & \zeta^{-2N}
\end{pmatrix}$.
\end{itemize}
It is known \cite{FIK} that the solution $U$ to the above RH problem is unique (provided it exists), and can be expressed in terms of scalar-valued orthogonal polynomials as follows
\begin{align*}
U(\zeta) = \begin{pmatrix}
p_{2N}(\zeta) & \frac{1}{2\pi i} \int_{\gamma_{\mathbb{C}}} \frac{p_{2N}(\xi)W(\xi)}{\xi-\zeta}d\xi \\
q_{2N-1}(\zeta) & \frac{1}{2\pi i} \int_{\gamma_{\mathbb{C}}} \frac{q_{2N-1}(\xi)W(\xi)}{\xi-\zeta}d\xi
\end{pmatrix}, \qquad \zeta \in \mathbb{C}\setminus \gamma_{\mathbb{C}},
\end{align*}
where $p_{2N}$ and $q_{2N-1}$ are polynomials of degree $2N$ and $2N-1$ respectively, satisfying the following conditions
\begin{align}
& p_{2N}(\zeta) = \zeta^{2N}+\bigO(\zeta^{2N-1}), \qquad \mbox{as } \zeta \to \infty, \label{cond for pn} \\
& \frac{1}{2\pi i}\int_{\gamma_{\mathbb{C}}} p_{2N}(\zeta) W(\zeta)\zeta^{k}d\zeta = 0, & & k = 0,\ldots,2N-1, \nonumber 
\end{align}
and
\begin{align}
& \frac{1}{2\pi i}\int_{\gamma_{\mathbb{C}}} q_{2N-1}(\zeta) W(\zeta)\zeta^{2N-1}d\zeta = -1, \label{cond for qn} \\
& \frac{1}{2\pi i}\int_{\gamma_{\mathbb{C}}} q_{2N-1}(\zeta) W(\zeta)\zeta^{k}d\zeta = 0, & & k = 0,\ldots,2N-2. \nonumber
\end{align}
The reproducing kernel $\mathcal{R}^{U}$ is defined by
\begin{equation}\label{reproducing kernel in terms of U intro}
\mathcal{R}^{U}(\omega,\zeta) = \frac{1}{\zeta-\omega} \begin{pmatrix}
0 & 1
\end{pmatrix} U^{-1}(\omega)U(\zeta) \begin{pmatrix}
1 \\ 0
\end{pmatrix}.
\end{equation}
Now, we state our first main result.
\begin{theorem}\label{thm: correlation kernel final scalar expression}

For $x \in \{1,\ldots,2N-1\}$, $y \in \mathbb{Z}$ and $\epsilon_{x} \in \{0,1\}$, we have
\begin{align}
& \big[ K(2x+\epsilon_{x},2y+j,2x+\epsilon_{x},2y+i) \big]_{i,j=0}^{1} = \frac{1}{(2\pi i)^{2}}\int_{\gamma_{\mathbb{C}}}\int_{\gamma_{\mathbb{C}}} H_{K}(\omega,\zeta;\epsilon_{x}) \label{new formula for the kernel in the thm} \\ 
&  W(\omega) \mathcal{R}^{U}(\omega,\zeta) \frac{\omega^{N+x-y}}{\zeta^{N+x-y}} \frac{(\omega-c)^{y}(\omega-c^{-1})^{y}}{(\zeta-c)^{y}(\zeta-c^{-1})^{y}} \frac{(\zeta - \alpha c)^{x} (\zeta - \alpha c^{-1})^{x}}{(\omega - \alpha c)^{x} (\omega - \alpha c^{-1})^{x}}d\zeta d\omega, \nonumber
\end{align}
where $\gamma_{\mathbb{C}}$ is a closed curve surrounding $c$ and $c^{-1}$ once in the positive direction that does not go around $0$, and where $H_{K}(\omega,\zeta;0)$ and $H_{K}(\omega,\zeta;1)$ are given by 
\begin{align}
& H_{K}(\omega,\zeta;0) = \begin{pmatrix}
\frac{1}{\zeta - c} & \frac{c(1-\alpha)}{\alpha (\zeta -c)(\zeta - c^{-1})} \\
\frac{\alpha}{(1-\alpha)c^{2} \omega}\frac{\omega-c}{\zeta -c} & \frac{\omega - c}{c \omega (\zeta -c)(\zeta -c^{-1})}
\end{pmatrix}, \label{HK epsx 0} \\
& H_{K}(\omega,\zeta;1) = \begin{pmatrix}
\frac{c (\zeta - \alpha c)}{\zeta(\zeta -c)(\omega - \alpha c)} & \frac{(1-\alpha)c(\zeta - \alpha c)}{(\zeta -c)(\zeta - c^{-1})(\omega - \alpha c)} \\
\frac{(\zeta - \alpha c)(\omega -c)}{(1-\alpha)\zeta(\zeta -c)(\omega-\alpha c)} & \frac{(\zeta - \alpha c)(\omega -c)}{(\zeta -c)(\zeta -c^{-1})(\omega - \alpha c)}
\end{pmatrix}. \label{HK epsx 1}
\end{align}
\end{theorem}
\begin{remark}
Theorem \ref{thm: correlation kernel final scalar expression} is proved in Section \ref{section: reducing the size}. It is based on an unpublished idea of A. Kuijlaars that matrix valued orthogonal polynomials in a genus zero situation can be reduced to scalar orthogonality. In our case, the scalar orthogonality appears in \eqref{cond for pn}--\eqref{cond for qn} and a main part of the proof of Theorem \ref{thm: correlation kernel final scalar expression} consists of relating the matrix valued reproducing kernel $\mathcal{R}^Y$ from \eqref{def of mathcal R} to the scalar reproducing kernel $\mathcal{R}^U$ from \eqref{reproducing kernel in terms of U intro}.
\end{remark}

\subsection{The rational function $\mathcal{Q}$} \label{sec:zeta}

%A main step of our analysis will be to find an expression for the kernel in terms of scalar OPs. We will denote $W(\zeta)$ for the associated scalar weight, which depends on $N$. 
The function $\mathcal{Q}$ is a meromorphic function that appears in the equilibrium problem associated to the varying weight $W$. Its explicit expression is obtained after solving a non-linear system of 5 equations with 5 unknowns. Here, we just state the formula for $\mathcal{Q}$, and refer to Section \ref{section: g-function} for a more constructive approach. We define $\mathcal{Q}$ as follows
\begin{align}\label{def of Q in statement of results}
\mathcal{Q}(\zeta) = \frac{(\zeta-r_{1})^{2}(\zeta-r_{2})^{2}(\zeta-r_{3})^{2}(\zeta-r_{+})(\zeta-r_{-})}{4\zeta^{2} (\zeta-\alpha c)^{2} (\zeta-\alpha c^{-1})^{2} (\zeta-c)^{2} (\zeta-c^{-1})^{2}},
\end{align}
where $c$ is given by \eqref{def of W in intro}, $r_{1}$, $r_{2}$ and $r_{3}$ are given by
\begin{align}\label{def of r1 r2 r3}
& r_{1} = - \sqrt{\alpha}, & & r_{2} = \sqrt{\alpha} \frac{\alpha c + \sqrt{\alpha}}{c + \sqrt{\alpha}}, & & r_{3} = \sqrt{\alpha} \frac{c + \sqrt{\alpha}}{\alpha c + \sqrt{\alpha}},
\end{align}
and $r_{+}$ and $r_{-}$ are given by
\begin{align}
& r_{+} = c \, \bigg( \frac{1+\alpha}{2} + i\sqrt{3} \frac{1-\alpha}{2} \bigg), & & r_{-} = c \, \bigg( \frac{1+\alpha}{2} - i\sqrt{3} \frac{1-\alpha}{2} \bigg). \label{def of r+ r-}
\end{align}
The zero $r_{+}$ of $\mathcal{Q}$ lies in the upper half plane, $r_{-} = \overline{r_{+}}$, and the other zeros and poles of $\mathcal{Q}$ are real. Furthermore, for all $\alpha \in (0,1)$, they are ordered as follows:
\begin{align}\label{ordering of the zeros and the poles}
r_{1} < 0 < \alpha c < r_{2} < \alpha c^{-1} < c < r_{3} < c^{-1}.
\end{align}
\subsection{Lozenge probabilities}
The densities for the three types of lozenges at a point $(x,y)$, $x,y \in \{0,1,\ldots,4N\}$, are denoted by
\begin{align}\label{def of mathcal P1 P2 P3}
\mathcal{P}_{1}(x,y) = \mathbb P\Bigg(\tikz[scale=.3,baseline=(current bounding box.center)] {\draw (0,-1) \lozr; \filldraw (0,-1) circle(5pt); \draw (0,-1) node[below] {$(x,y)$}} \Bigg), \qquad \mathcal{P}_{2}(x,y) = \mathbb P\Bigg(\tikz[scale=.3,baseline=(current bounding box.center)] {\draw (0,0) \lozu; \filldraw (0,0) circle(5pt); \draw (0,0) node[below] {$(x,y)$} }\Bigg), \qquad \mathcal{P}_{3}(x,y) = \mathbb P\Bigg(\tikz[scale=.3,baseline=(current bounding box.center)] {\draw (0,0) \lozd; \filldraw (1,0) circle(5pt); \draw (1,0) node[below] {$(x,y)$}} \Bigg),
\end{align}
and satisfy $\sum_{j=1}^{3}\mathcal{P}_{j}(x,y) = 1$. Because our model is $2\times 2$ periodic, $\mathcal{P}_{1}(x,y)$, $\mathcal{P}_{2}(x,y)$ and $\mathcal{P}_{3}(x,y)$ depend crucially on the parity of $x$ and $y$, and it is convenient to consider the following matrices
\begin{align}\label{def of P1 P2 P3}
P_{j}(x,y) = \begin{pmatrix}
\mathcal{P}_{j}(2x,2y+1) & \mathcal{P}_{j}(2x+1,2y+1) \\
\mathcal{P}_{j}(2x,2y) & \mathcal{P}_{j}(2x+1,2y)
\end{pmatrix}, \qquad j=1,2,3,
\end{align}
where $x,y \in \{0,1,...,2N-1\}$. Let $\{(x_{N},y_{N})\}_{N\geq 1}$ be a sequence satisfying
\begin{align}\label{good sequence}
& \begin{cases}
\frac{x_{N}}{N} = 1+\xi + o(1),  \\
\frac{y_{N}}{N} = 1+\eta + o(1),
\end{cases} \qquad  \mbox{as } N \to + \infty,
\end{align}
where the point $(\xi,\eta)$ lies in the hexagon
\begin{equation} \label{eq:Hhexagon}
\mathcal H= \left\{(\xi,\eta) \mid  -1\leq \xi \leq 1, \ -1 \leq  \eta \leq 1,\ -1\leq  \eta-\xi \leq 1 \right\}.
\end{equation}
In Theorem \ref{thm:main}, we give explicit expressions for
\begin{align}\label{limits intro in intro}
\lim_{N \to + \infty} P_{j}(x_{N},y_{N}), \qquad j=1,2,3,
\end{align}
in case $(\xi,\eta)$ belongs to the liquid region $\mathcal{L}_{\alpha} \subset \mathcal{H}$.

%The phase functions $\Phi(\zeta)=\Phi(\zeta;\xi,\eta,\alpha)$ and $\Psi(\zeta)=\Psi(\zeta;\xi,\eta,\alpha)$ appear in our saddle point analysis and their exact expressions will be given in Section \ref{section: phase functions}. The saddles are the solution to $\Phi'(\zeta)\Psi'(\zeta)=0$. As mentioned, the limits \eqref{limits intro in intro} will be expressed in terms of particular saddle, denoted $s = s(\xi,\eta;\alpha)$. Let us now described $\mathcal{L}_{\alpha}$ and $s=s(\xi,\eta;\alpha)$ in more detail.

%Proposition \ref{prop:saddle} states that equation \eqref{eq:saddlepointeq} admits at most one solution $\zeta$ in the upper half plane $\mathbb{C}^{+} = \{\zeta \in \mathbb{C}: \im \zeta >0\}$. 
%The liquid region $\mathcal{L}_{\alpha}$ is defined as the subset of $\mathcal{H}$ for which there exists a solution $\zeta = s(\xi,\eta;\alpha)$ of \eqref{eq:saddlepointeq} in the upper half plane $\mathbb{C}^{+}=\{\zeta \in \mathbb{C}: \im \zeta >0\}$. 

%We will obtain explicit expressions for \eqref{limits intro in intro}. \eqref{new formula for the kernel in the thm} asymptotically by means of a saddle point analysis. The relevant phase functions $\Phi(\zeta)=\Phi(\zeta;\xi,\eta,\alpha)$ and $\Psi(\zeta)=\Psi(\zeta;\xi,\eta,\alpha)$ will be given in Section \ref{section: phase functions}. Here we simply state the saddle point equations $\Phi'(\zeta)\Psi'(\zeta)=0$, and refer to Section \ref{section: saddle point analysis} for more motivation.
\subsection{Saddle points and the liquid region}\label{subsection: saddle points and liquid region}
%Here we simply state the saddle point equation. We refer to 

%The phase functions $\Phi(\zeta)=\Phi(\zeta;\xi,\eta,\alpha)$ and $\Psi(\zeta)=\Psi(\zeta;\xi,\eta,\alpha)$ will appear in the saddle point analysis of \eqref{new formula for the kernel in the thm}, and their precise definitions will be given in Section \ref{section: phase functions}. As always with the saddle point method, the main contribution will come from a local analysis around the saddle, and thus they play an important role. Here, we simply state the saddle point equation $\Phi'(\zeta)\Psi'(\zeta)=0$:
For each $(\xi,\eta)\in \mathcal{H}$, there are in total $8$ saddles for the double contour integral \eqref{new formula for the kernel in the thm}, which are the solutions to  the algebraic equation 
\begin{align}\label{eq:saddlepointeq} 
\left[\frac{\xi-\eta}{2}\frac{1}{\zeta} - \frac{\xi}{2}\left( \frac{1}{\zeta-\alpha c} + \frac{1}{\zeta - \alpha c^{-1}} \right) + \frac{\eta}{2}\left( \frac{1}{\zeta - c} + \frac{1}{\zeta-c^{-1}} \right)\right]^{2} = \mathcal{Q}(\zeta),
\end{align}
where $\mathcal{Q}(\zeta)$ is given by \eqref{def of Q in statement of results}. 
Following the previous works \cite{BF,Duits1,Ok2,Petrov1, CDKL}, we define the liquid region as the subset of $\mathcal{H}$ for which there exists a saddle lying in the upper half-plane $\mathbb{C}^{+} = \{\zeta \in \mathbb{C}: \im \zeta >0\}$. Proposition \ref{prop:saddle} states that there is a unique such saddle (whenever it exists), which is denoted by $s(\xi,\eta;\alpha)$. This saddle plays a particular role in the analysis of Section \ref{section: saddle point analysis} and appears in the final formulas for the limiting densities \eqref{limits intro in intro}.
\begin{proposition}\label{prop:saddle}
Let $(\xi, \eta) \in \mathcal{H}^{\mathrm{o}}$ (the interior set of $\mathcal{H}$).
Then there exists at most one solution $\zeta=s(\xi, \eta; \alpha)$ to \eqref{eq:saddlepointeq} in $\mathbb C^+=\{\zeta \in \mathbb C \mid \im \zeta>0\}$.
\end{proposition}
%\begin{proof}
%See Subsection \ref{subsection: proof of prop 3.2}.
%\end{proof}
\begin{definition}
We define the liquid region $\mathcal L_\alpha \subset \mathcal H$ by
\begin{equation} \label{eq:liquidLalpha}
\mathcal L_\alpha= \left\{ (\xi,\eta) \in \mathcal{H}^{\mathrm{o}} \mid 
\mbox{there exists a solution } \zeta = s(\xi, \eta; \alpha) \in \mathbb{C}^+ \mbox{ to \eqref{eq:saddlepointeq}} \right\}
\end{equation}
and we define the map $s: \mathcal L_\alpha \to\mathbb{C}^+$ by $(\xi,\eta) \mapsto s(\xi,\eta; \alpha)$. 
\end{definition}
\begin{figure}
\begin{center}
\begin{tikzpicture}[master]
\node at (0,0) {\includegraphics[width=4.05cm]{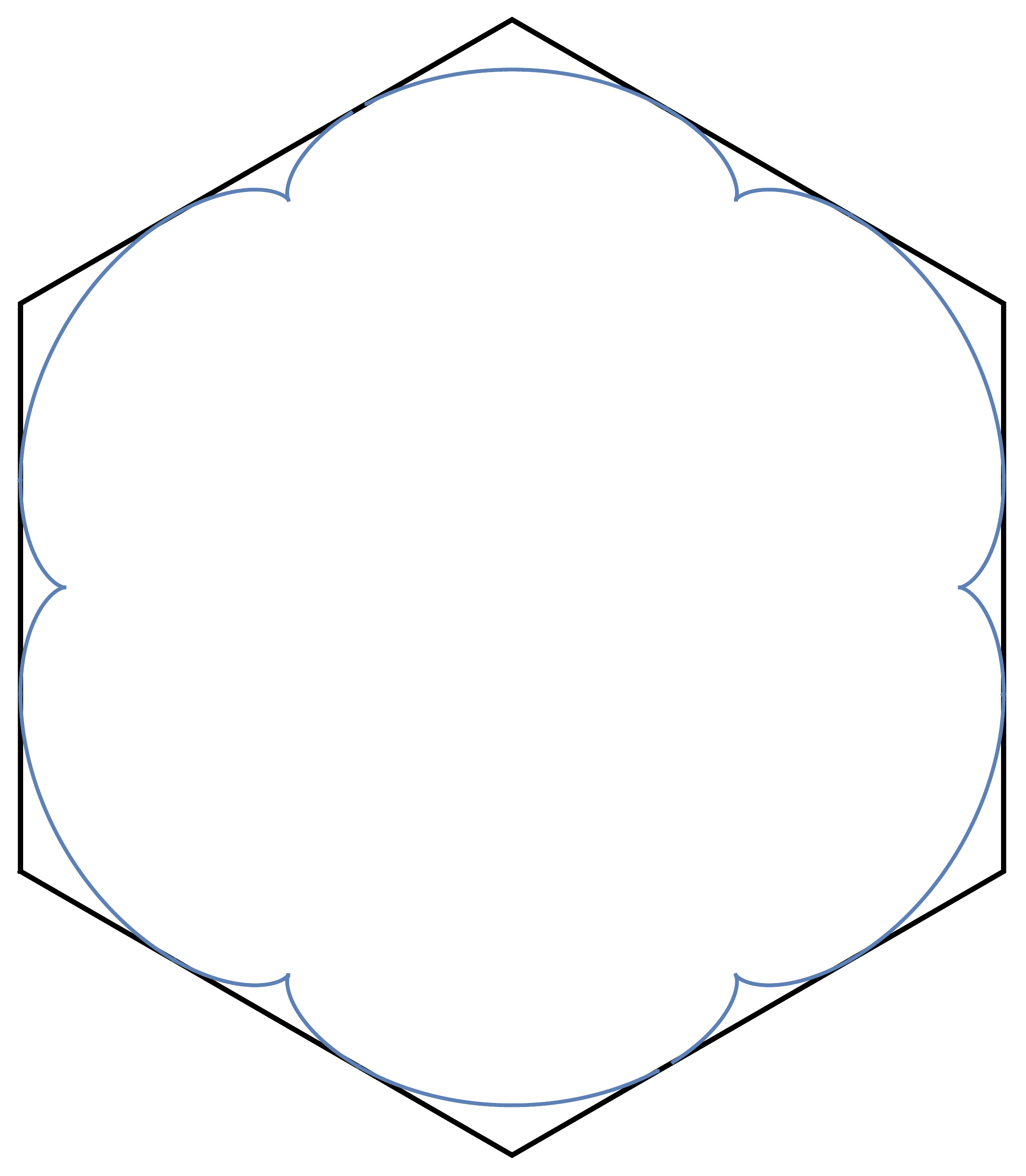}};
\draw [line width=0.65 mm,red] (0,-2.05)--(0,2.04);
\draw [line width=0.65 mm,red,dashed] (-1.78,0)--(1.77,0);
\draw [line width=0.65 mm,blue] (30:2.05)--(-150:2.05);
\draw [line width=0.65 mm,blue,dashed] (120:1.78)--(-60:1.77);
\draw [line width=0.65 mm,green] (150:2.05)--(-30:2.05);
\draw [line width=0.65 mm,green,dashed] (60:1.78)--(-120:1.77);
\end{tikzpicture}
\hspace{1cm}
\begin{tikzpicture}[slave,scale = 0.5]
%\node at (0,0) {\includegraphics[width = 12cm]{../Image/Trajectories/traj_alpha_04.jpg}};
\node at (0,0) {};

% circle gamma1
%\draw[red,line width=0.65 mm] ([shift=(-156:3.6cm)]2.1,0) arc (-156:156:3.6cm);
%\draw[red,dashed,line width=0.65 mm] ([shift=(155:3.6cm)]2.1,0) arc (155:205:3.6cm);
\draw[red,line width=0.65 mm] ([shift=(0:3.6cm)]2.1,0) arc (0:156:3.6cm);
\draw[red,dashed,line width=0.65 mm] ([shift=(155:3.6cm)]2.1,0) arc (155:180:3.6cm);

% circle gamma_0
%\draw[blue,line width=0.65 mm] ([shift=(-35:2.4cm)]-3.13,0) arc (-35:35:2.4cm);
%\draw[blue,dashed,line width=0.65 mm] ([shift=(35:2.4cm)]-3.13,0) arc (35:325:2.4cm);
\draw[blue,line width=0.65 mm] ([shift=(0:2.4cm)]-3.13,0) arc (0:35:2.4cm);
\draw[blue,dashed,line width=0.65 mm] ([shift=(35:2.4cm)]-3.13,0) arc (35:180:2.4cm);

% circle gamma alpha
%\draw[green,line width=0.65 mm] ([shift=(95:1.45cm)]-1.05,0) arc (95:265:1.45cm);
%\draw[green,dashed,line width=0.65 mm] ([shift=(-95:1.45cm)]-1.05,0) arc (-95:95:1.45cm);
\draw[green,line width=0.65 mm] ([shift=(95:1.45cm)]-1.05,0) arc (95:180:1.45cm);
\draw[green,dashed,line width=0.65 mm] ([shift=(0:1.45cm)]-1.05,0) arc (0:95:1.45cm);

\node at (2.1,0) {\color{black} \small $\bullet$};
\node at (-0.38,0) {\color{black} \small $\bullet$};
\node at (-1.05,0) {\color{black} \small $\bullet$};
\node at (-2.035,0) {\color{black} \small $\bullet$};
\node at (-3.13,0) {\color{black} \small $\bullet$};

%\node at (3,3) {\color{red} $\Sigma_{1}$};
\end{tikzpicture}
\end{center}
\caption{\label{fig: mapping s} On the left, we draw for $\alpha = 0.4$ the parts of the lines $\xi = 0$ (red), $\eta = \frac{\xi}{2}$ (dashed red), $\eta = \xi$ (blue), $\eta = - \xi$ (dashed blue), $\eta = 0$ (green) and $\eta = 2 \xi$ (dashed green) that are in the liquid region. On the right, we draw the corresponding location of $s(\xi,\eta;\alpha)$ in the upper half plane. The black dots are, from left to right, $0$, $\alpha c$, $\alpha c^{-1}$, $c$ and $c^{-1}$.}
\end{figure}
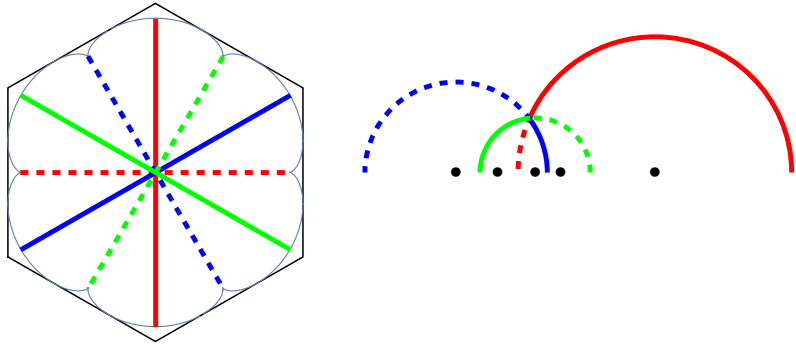
It is clear from \eqref{def of Q in statement of results} and \eqref{eq:saddlepointeq} that $(0,0) \in \mathcal{L}_{\alpha}$ and $s(0,0;\alpha) = r_{+}$ for all $\alpha \in (0,1)$. We now describe some properties of $(\xi,\eta) \mapsto s(\xi,\eta;\alpha)$. Consider the following three circles:
\begin{align*}
& \gamma_{0} = \{\zeta \in \mathbb{C}:|\zeta|= R_{0}\}, \qquad \gamma_{\alpha} = \{\zeta \in \mathbb{C}:|\zeta - \alpha c^{-1}|= R_{\alpha}\}, \qquad \gamma_{1} = \{\zeta \in \mathbb{C}:|\zeta - c^{-1}|= R_{1}\}
\end{align*}
where $R_{0} = \sqrt{\alpha}$,  $R_{\alpha} = (1-\alpha)\sqrt{\alpha}$ and $R_{1} = \frac{1-\alpha}{\sqrt{\alpha}}$ (see also Figure \ref{fig: crit traj alpha 04}). %For convenience, we orient $\gamma_{0}$, $\gamma_{\alpha}$ and $\gamma_{1}$ in the positive direction. 
It is a direct computation to verify that
\begin{align}\label{r2 r+ r- on the circle}
r_{+},r_{-},r_{2} \in \gamma_{1}, \qquad r_{+},r_{-},r_{3} \in \gamma_{\alpha} \qquad \mbox{ and } \qquad r_{+},r_{-},r_{1} \in \gamma_{0}.
\end{align}
In particular, we can write
\begin{align*}
r_{\pm} = c^{-1} + R_{1}e^{\pm i \theta_{1}} = \alpha c^{-1} + R_{\alpha} e^{\pm i \theta_{\alpha}} = R_{0} e^{\pm i \theta_{0}},
\end{align*}
for certain angles $\theta_{1} \in (\frac{2\pi}{3},\pi)$, $\theta_{\alpha} \in (\frac{\pi}{3},\frac{2\pi}{3})$ and $\theta_{0} \in (0,\frac{\pi}{3})$. We also define
\begin{align}
& \Sigma_{1} = \{\zeta \in \mathbb{C} : |\zeta - c^{-1}| = R_{1}, \; \arg z \in (-\theta_{1},\theta_{1}) \} \subset \gamma_{1}, \label{def of Sigma 1} \\
& \Sigma_{\alpha} = \{\zeta \in \mathbb{C} : |\zeta - \alpha c^{-1}| = R_{\alpha}, \; \arg z \in (-\pi,-\theta_{\alpha})\cup (\theta_{\alpha},\pi] \} \subset \gamma_{\alpha}, \label{def of Sigma alpha} \\
& \Sigma_{0} = \{\zeta \in \mathbb{C} : |\zeta| = R_{0}, \; \arg z \in (-\theta_{0},\theta_{0}) \} \subset \gamma_{0}. \label{def of Sigma 0}
\end{align}
The following proposition is illustrated in Figure \ref{fig: mapping s}.
\begin{remark}
For a given set $A$, the notation $\overline{A}$ stands for the closure of $A$.
\end{remark}
\begin{proposition}\label{prop:hightemp}
The map $(\xi,\eta) \mapsto s(\xi,\eta;\alpha)$ satisfies $s(-\xi,-\eta;\alpha) = s(\xi,\eta;\alpha)$, and 
\begin{enumerate}[label=(\alph*)]
\item\label{item a in prop mapping s} it maps $\{\xi = 0\}\cap \mathcal{L}_{\alpha}$ onto $\overline{\Sigma_{1}} \cap \mathbb{C}^{+}$,
\item\label{item b in prop mapping s} it maps $\{\eta = \frac{\xi}{2}\}\cap \mathcal{L}_{\alpha}$ onto $(\gamma_{1}\setminus \Sigma_{1}) \cap \mathbb{C}^{+}$,
\item\label{item c in prop mapping s} it maps $\{\eta = \xi\}\cap \mathcal{L}_{\alpha}$ onto $\overline{\Sigma_{0}} \cap \mathbb{C}^{+}$,
\item\label{item d in prop mapping s} it maps $\{\eta = -\xi\}\cap \mathcal{L}_{\alpha}$ onto $(\gamma_{0}\setminus \Sigma_{0}) \cap \mathbb{C}^{+}$,
\item\label{item e in prop mapping s} it maps $\{\eta = 0\}\cap \mathcal{L}_{\alpha}$ onto $\overline{\Sigma_{\alpha}} \cap \mathbb{C}^{+}$,
\item\label{item f in prop mapping s} it maps $\{\eta = 2\xi \}\cap \mathcal{L}_{\alpha}$ onto $(\gamma_{\alpha}\setminus \Sigma_{\alpha}) \cap \mathbb{C}^{+}$.
\end{enumerate}
\end{proposition}
By definition, the saddles lie in the complex plane. We show here that they can be naturally projected on a Riemann surface. Define $\mathcal{Q}(\zeta)^{1/2}$ with a branch cut joining $r_{-}$ to $r_{+}$ along $\Sigma_{1}$, such that $\mathcal{Q}(\zeta)^{1/2}\sim \frac{1}{2\zeta}$ as $\zeta \to \infty$, and denote the associated Riemann surface by $\mathcal{R}_{\alpha}$:
\begin{align*}
\mathcal{R}_{\alpha} := \{(\zeta,w) \in \mathbb{C}^{2} : w^{2}=\mathcal{Q}(\zeta) \}.
\end{align*}
This is a two-sheeted covering of the $\zeta$-plane glued along $\Sigma_{1}$, and the sheets are ordered such that $w = \mathcal{Q}(\zeta)^{1/2}$ on the first sheet and $w = -\mathcal{Q}(\zeta)^{1/2}$ on the second sheet. For each solution $\zeta$ to \eqref{eq:saddlepointeq}, there exists a $w$ satisfying $w^{2} = \mathcal{Q}(\zeta)$, and such that
\begin{align}\label{eq:saddlepointeq sqrt} 
\frac{\xi-\eta}{2}\frac{1}{\zeta} - \frac{\xi}{2}\left( \frac{1}{\zeta-\alpha c} + \frac{1}{\zeta - \alpha c^{-1}} \right) + \frac{\eta}{2}\left( \frac{1}{\zeta - c} + \frac{1}{\zeta-c^{-1}} \right) = w.
\end{align}
\begin{definition}
The map $(\xi,\eta)\mapsto w(\xi,\eta;\alpha)$ is defined by $w(\xi,\eta;\alpha)^{2} = \mathcal{Q}(s(\xi,\eta;\alpha))$, such that \eqref{eq:saddlepointeq sqrt} holds with $\zeta = s(\xi,\eta;\alpha)$ and $w = w(\xi,\eta;\alpha)$.
\end{definition}
\begin{proposition}\label{prop: s on the Riemann surface}
The map $(\xi,\eta) \mapsto \big( s(\xi,\eta;\alpha),w(\xi,\eta;\alpha) \big)$ is a diffeomorphism from $\mathcal{L}_{\alpha}$ to 
\begin{align}\label{calRplusdef}
\mathcal{R}_{\alpha}^+ := \{(\zeta,w) \in \mathcal R_{\alpha} \mid \im \zeta > 0 \}.
\end{align}
It maps the left half $\mathcal L_\alpha^{l}=\left\{ (\xi,\eta) \in \mathcal L_\alpha  \mid  \xi < 0 \right\}$ to the upper half-plane of the first sheet of $\mathcal{R}_{\alpha}$, and it maps $\mathcal L_\alpha^{r}=\left\{ (\xi,\eta) \in \mathcal L_\alpha  \mid  \xi > 0 \right\}$ to the upper half-plane of the second sheet. Moreover, its inverse $(s,w)\mapsto (\xi,\eta)=(\xi(s,w;\alpha),\eta(s,w;\alpha))$ is explicitly given by
\begin{align}\label{inverse of the diffeomorphism}
\begin{pmatrix}
\xi \\ \eta
\end{pmatrix} = \begin{pmatrix}
\re \left( \frac{-(s - \alpha)(s +\alpha)(s -c)(s - \frac{1}{c})}{(s - \alpha c)(s - \frac{\alpha}{c})(s - 1)(s + 1)} \right) & 1 \\[0.2cm]
\im \left( \frac{-(s - \alpha)(s +\alpha)(s -c)(s - \frac{1}{c})}{(s - \alpha c)(s - \frac{\alpha}{c})(s - 1)(s + 1)} \right) & 0
\end{pmatrix}^{-1} \begin{pmatrix}
\re \left( \frac{2s (s -c)(s - \frac{1}{c})}{(s - 1)(s + 1)}w \right) \\[0.2cm]
\im \left( \frac{2s (s -c)(s - \frac{1}{c})}{(s - 1)(s + 1)}w  \right)
\end{pmatrix}.
\end{align}
\end{proposition}

\paragraph*{Description of the liquid region.}
After clearing the denominator in \eqref{eq:saddlepointeq}, we get
\begin{multline} \label{eq:saddleequationhigh}
(\zeta -r_{1})^{2}(\zeta -r_{2})^{2}(\zeta -r_{3})^{2}(\zeta -r_{+})(\zeta -r_{-}) =  \\ 
\Big[ (\zeta-1)(\zeta +1)(\zeta-\alpha c)(\zeta-\tfrac{\alpha}{c})\eta - (\zeta-\alpha)(\zeta+\alpha)(\zeta-c)(\zeta-\tfrac{1}{c})\xi \Big]^{2}. 
\end{multline}
Since \eqref{eq:saddleequationhigh} is invariant under the map $(\xi,\eta) \mapsto (-\xi,-\eta)$, we conclude that $\mathcal{L}_{\alpha}$ is symmetric with respect to the origin. Also, this equation has real coefficients, so $s(\xi,\eta;\alpha)$ and $\overline{s(\xi,\eta;\alpha)}$ are both solutions whenever $(\xi,\eta) \in \mathcal{L}_{\alpha}$. At the boundary $\partial \mathcal L_\alpha$ of the liquid region, $s(\xi,\eta;\alpha)$ and $\overline{s(\xi,\eta;\alpha)}$ coalesce in the real line, so $\partial \mathcal L_\alpha$ is part of the zero set of the discriminant of \eqref{eq:saddleequationhigh} (whose expression is too long to be written down). The curve $\partial \mathcal{L}_{\alpha}$ is tangent to the hexagon at $12$ points and possesses $6$ cusp points. The tangent points can be obtained by letting $s \to s_{\star}\in\{ 0,\alpha c, \alpha c^{-1}, c,c^{-1},\infty\}$ in \eqref{inverse of the diffeomorphism}, and the cusp points by letting $s \to s_{\star}  \in \{r_{1},r_{2},r_{3}\}$ in \eqref{inverse of the diffeomorphism} (see also Figure \ref{fig: mapping s}). Figure \ref{fig: flower} illustrates $\partial \mathcal{L}_{\alpha}$ for different values of $\alpha$ (and has been generated using \eqref{inverse of the diffeomorphism}). Denote $\mathcal{F}_{\alpha,j}$, $j=1,\ldots,6$ for the regions shown in Figure \ref{fig: frozen regions} (left). They are disjoint from each other and from $\mathcal{L}_{\alpha}$, and are symmetric under $(\xi,\eta)\mapsto (-\xi,-\eta)$. As we will see, these regions are frozen (or semi-frozen). 
\begin{figure}
\begin{center}
\includegraphics[width=3.5cm]{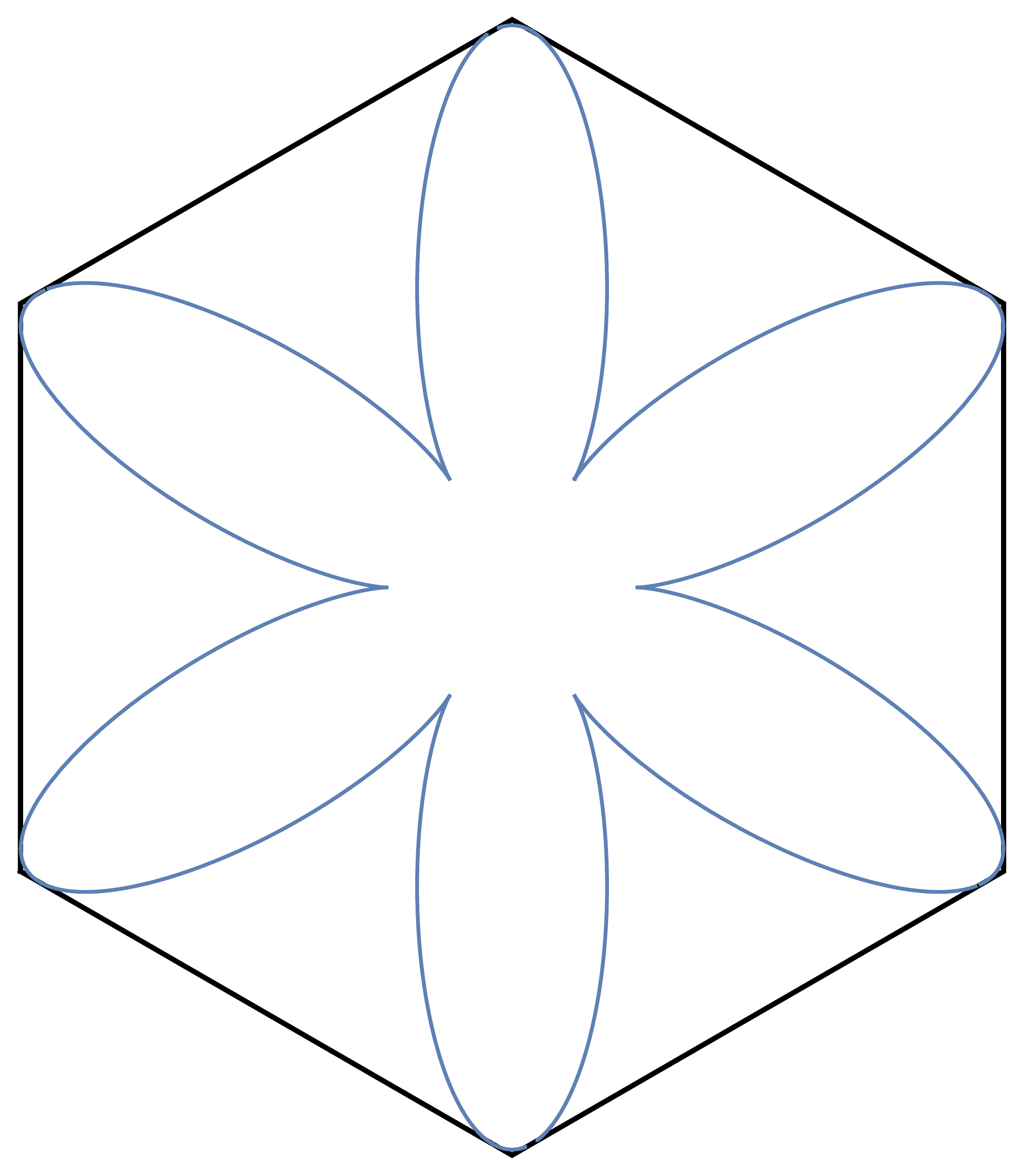}
\includegraphics[width=3.5cm]{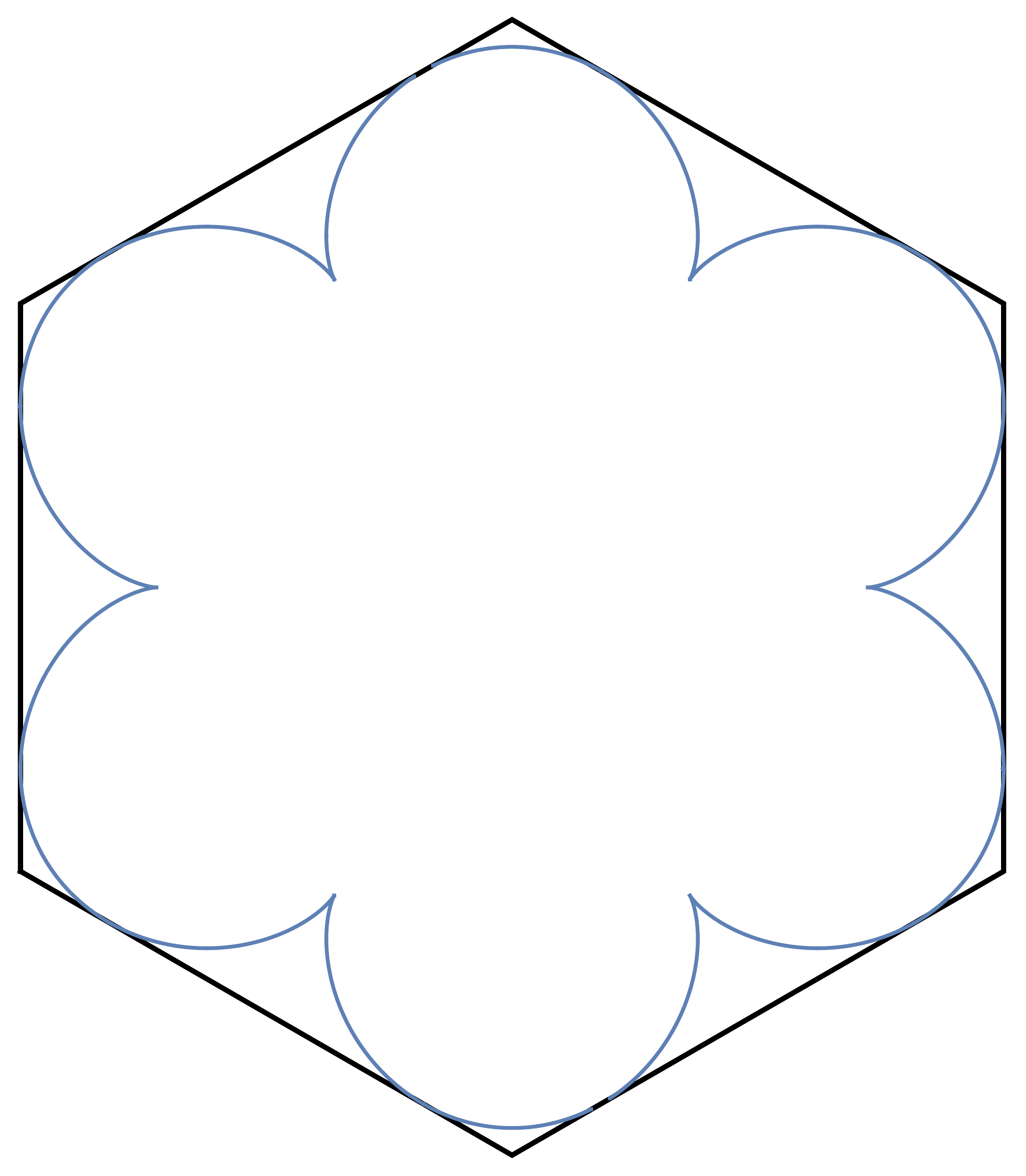}
\includegraphics[width=3.5cm]{alpha_04.pdf}
\end{center}
\caption{\label{fig: flower}The curve $\partial \mathcal{L}_{\alpha}$ with $\alpha=0.04$, $\alpha=0.2$ and $\alpha=0.4$ (from left to right).}
\end{figure}
From Propositions \ref{prop:hightemp} and \ref{prop: s on the Riemann surface}, we already infer the following:
\begin{subequations}\label{limits s sbar as xi eta}
\begin{align}
& s,\overline{s} \to s_{\star} \in (0,\alpha c), & & \mbox{as } (\xi,\eta) \to (\xi^{\star},\eta^{\star}) \in \partial \mathcal{L}_{\alpha} \cap \partial \mathcal{F}_{1,\alpha}, \\[0.1cm]
& s,\overline{s} \to s_{\star} \in (\alpha c^{-1},c), & & \mbox{as } (\xi,\eta) \to (\xi^{\star},\eta^{\star}) \in \partial \mathcal{L}_{\alpha} \cap \partial \mathcal{F}_{2,\alpha}, \\[0.1cm]
& s,\overline{s} \to s_{\star} \in (c^{-1},+\infty), & & \mbox{as } (\xi,\eta) \to (\xi^{\star},\eta^{\star}) \in \partial \mathcal{L}_{\alpha} \cap \partial \mathcal{F}_{3,\alpha}, \\[0.1cm]
& s,\overline{s} \to s_{\star} \in (\alpha c,\alpha c^{-1}), & & \mbox{as } (\xi,\eta) \to (\xi^{\star},\eta^{\star}) \in \partial \mathcal{L}_{\alpha} \cap \partial \mathcal{F}_{4,\alpha}, \\[0.1cm]
& s,\overline{s} \to s_{\star} \in (-\infty,0), & & \mbox{as } (\xi,\eta) \to (\xi^{\star},\eta^{\star}) \in \partial \mathcal{L}_{\alpha} \cap \partial \mathcal{F}_{5,\alpha}, \\[0.1cm]
& s,\overline{s} \to s_{\star} \in (c,c^{-1}), & & \mbox{as } (\xi,\eta) \to (\xi^{\star},\eta^{\star}) \in \partial \mathcal{L}_{\alpha} \cap \partial \mathcal{F}_{6,\alpha}.
\end{align} 
\end{subequations}

\newpage 

\subsection{Limiting densities in the liquid region}	
Theorem \ref{thm:main} states that the limits \eqref{limits intro in intro} are expressed in terms of the angles shown in Figure \ref{fig: the four triangles}. %We obtain the following results.

\begin{theorem} \label{thm:main}
Let $\{(x_{N},y_{N}\}_{N \geq 1}$ be a sequence satisfying \eqref{good sequence} with $(\xi,\eta) \in \mathcal{L}_{\alpha}$. We obtain the following limits:
\begin{align}
& \lim_{N \to \infty} P_{1}(x_{N},y_{N}) = \frac{1}{\pi}\begin{pmatrix}
\arg(s-\alpha c) - \arg(s) & \arg(s-\alpha c^{-1}) \\
\arg (s-\alpha c) & \arg(s-\alpha c^{-1}) - \arg s
\end{pmatrix}, \label{P1 limit main result} \\
& \lim_{N \to \infty} P_{2}(x_{N},y_{N}) = \frac{1}{\pi}\begin{pmatrix}
\arg(s-c^{-1}) - \arg(s-\alpha c) & \arg(s-c^{-1}) - \arg(s-\alpha c^{-1}) \\
\arg(s-c)-\arg(s-\alpha c) & \arg(s-c) - \arg(s-\alpha c^{-1})
\end{pmatrix}, \label{P2 limit main result} \\
& \lim_{N \to \infty} P_{3}(x_{N},y_{N}) = \frac{1}{\pi} \begin{pmatrix}
\pi-\arg(s-c^{-1})+\arg(s) & \pi-\arg(s-c^{-1}) \\
\pi-\arg(s-c) & \pi-\arg(s-c)+\arg(s)
\end{pmatrix}. \label{P3 limit main result}
\end{align}
These limits can equivalently be stated as follows:
\begin{align*}
& \lim_{N \to \infty} P_{1}(x_{N},y_{N}) = \frac{1}{\pi}\begin{pmatrix}
\phi_{1,11} & \phi_{1,12} \\
\phi_{1,21} & \phi_{1,22}
\end{pmatrix}, \\
& \lim_{N \to \infty} P_{2}(x_{N},y_{N}) = \frac{1}{\pi}\begin{pmatrix}
\phi_{2,11} & \phi_{2,12} \\
\phi_{2,21} & \phi_{2,22}
\end{pmatrix}, \\
& \lim_{N \to \infty} P_{3}(x_{N},y_{N}) = \frac{1}{\pi} \begin{pmatrix}
\phi_{3,11}^{(l)}+\phi_{3,11}^{(r)} & \phi_{3,12} \\
\phi_{3,21} & \phi_{3,22}^{(l)}+\phi_{3,22}^{(r)}
\end{pmatrix},
\end{align*}
where $\phi_{k,ij}$, $1\leq i,j,k \leq 2$, and $\phi_{3,11}^{(l)}$, $\phi_{3,11}^{(r)}$, $\phi_{3,12}$, $\phi_{3,21}$, $\phi_{3,22}^{(l)}$ and $\phi_{3,22}^{(r)}$ are the angles represented in Figure \ref{fig: the four triangles}.
\end{theorem}
%\begin{proof}
%See Section \ref{}.
%\end{proof}
\begin{figure}[t]
\begin{center}
\begin{tikzpicture}[master,scale = 0.9]
\node at (0,0) {};
\node at (2.1,0) {\color{black} \large $\bullet$};
\node at (-0.38,0) {\color{black} \large $\bullet$};
\node at (-1.05,0) {\color{black} \large $\bullet$};
\node at (-2.035,0) {\color{black} \large $\bullet$};
\node at (-3.13,0) {\color{black} \large $\bullet$};
\node at (-1.8,2) {\color{black} \large $\bullet$};

\coordinate (Start) at (-3.13,0);
\coordinate (End) at (2.1,0);
\coordinate (Middle) at (-2.035,0);
\coordinate (S) at (-1.8,2);

\draw (S)--(Start)--(End)--(S)--(Middle);

\tkzMarkAngle[fill= gray,size=0.65cm,opacity=.4](End,Start,S);

\tkzMarkAngle[fill= gray,size=0.65cm,opacity=.4](Middle,S,End);

\tkzMarkAngle[fill= gray,size=0.85cm,opacity=.4](Start,S,Middle);

\tkzMarkAngle[fill= gray,size=0.85cm,opacity=.4](S,End,Start);

\node at (-2.45,1.6) {\tiny $\phi_{1,11}$};
\node at (-1.4,1.2) {\tiny $\phi_{2,11}$};
\node at (-3.3,0.4) {\tiny $\phi_{3,11}^{(l)}$};
\node at (0.8,0.3) {\tiny $\phi_{3,11}^{(r)}$};
\end{tikzpicture}\hspace{1cm}
\begin{tikzpicture}[slave,scale = 0.9]
\node at (0,0) {};
\node at (2.1,0) {\color{black} \large $\bullet$};
\node at (-0.38,0) {\color{black} \large $\bullet$};
\node at (-1.05,0) {\color{black} \large $\bullet$};
\node at (-2.035,0) {\color{black} \large $\bullet$};
\node at (-3.13,0) {\color{black} \large $\bullet$};
\node at (-1.8,2) {\color{black} \large $\bullet$};

\coordinate (Start) at (-1.05,0);
\coordinate (End) at (2.1,0);
\coordinate (S) at (-1.8,2);

\draw (S)--(Start)--(End)--(S);

\tkzMarkAngle[fill= gray,size=0.45cm,opacity=.4](End,Start,S);

\tkzMarkAngle[fill= gray,size=0.65cm,opacity=.4](Start,S,End);

\tkzMarkAngle[fill= gray,size=0.85cm,opacity=.4](S,End,Start);

\node at (-0.7,0.5) {\tiny $\phi_{1,12}$};
\node at (-1.2,1.95) {\tiny $\phi_{2,12}$};
\node at (1,0.3) {\tiny $\phi_{3,12}$};
\end{tikzpicture}

\vspace{0.4cm}

\begin{tikzpicture}[slave,scale = 0.9]
\node at (0,0) {};
\node at (2.1,0) {\color{black} \large $\bullet$};
\node at (-0.38,0) {\color{black} \large $\bullet$};
\node at (-1.05,0) {\color{black} \large $\bullet$};
\node at (-2.035,0) {\color{black} \large $\bullet$};
\node at (-3.13,0) {\color{black} \large $\bullet$};
\node at (-1.8,2) {\color{black} \large $\bullet$};

\coordinate (Start) at (-2.035,0);
\coordinate (End) at (-0.38,0);
\coordinate (S) at (-1.8,2);

\draw (S)--(Start)--(End)--(S);

\tkzMarkAngle[fill= gray,size=0.5cm,opacity=.4](End,Start,S);

\tkzMarkAngle[fill= gray,size=0.65cm,opacity=.4](Start,S,End);

\tkzMarkAngle[fill= gray,size=0.5cm,opacity=.4](S,End,Start);

\node at (-1.6,0.55) {\tiny $\phi_{1,21}$};
\node at (-1.2,1.95) {\tiny $\phi_{2,21}$};
\node at (-0.1,0.3) {\tiny $\phi_{3,21}$};
\end{tikzpicture}\hspace{1cm}\begin{tikzpicture}[slave,scale = 0.9]
\node at (0,0) {};
\node at (2.1,0) {\color{black} \large $\bullet$};
\node at (-0.38,0) {\color{black} \large $\bullet$};
\node at (-1.05,0) {\color{black} \large $\bullet$};
\node at (-2.035,0) {\color{black} \large $\bullet$};
\node at (-3.13,0) {\color{black} \large $\bullet$};
\node at (-1.8,2) {\color{black} \large $\bullet$};

\coordinate (Start) at (-3.13,0);
\coordinate (End) at (-0.38,0);
\coordinate (Middle) at (-1.05,0);
\coordinate (S) at (-1.8,2);

\draw (S)--(Start)--(End)--(S)--(Middle);

\tkzMarkAngle[fill= gray,size=0.65cm,opacity=.4](End,Start,S);

\tkzMarkAngle[fill= gray,size=1cm,opacity=.4](Middle,S,End);

%\tkzMarkAngle[fill= gray,size=0.85cm,opacity=.4](Middle,S,End);

\tkzMarkAngle[fill= gray,size=0.65cm,opacity=.4](Start,S,Middle);

\tkzMarkAngle[fill= gray,size=0.5cm,opacity=.4](S,End,Start);

\draw[dashed] (-4,-0.2)--(-4,5.5);
\draw[dashed] (-10.5,2.5)--(2.5,2.5);

\node at (-1.9,1.2) {\tiny $\phi_{1,22}$};
\node at (-1.2,1.7) {\tiny $\phi_{2,22}$};
\node at (-2.2,0.4) {\tiny $\phi_{3,22}^{(l)}$};
\node at (-0.2,0.4) {\tiny $\phi_{3,22}^{(r)}$};

%\node at (2.1,-0.28) {$c^{-1}$};
%\node at (-0.38,-0.35) {$c$};
%\node at (-1.05,-0.28) {$\alpha c^{-1}$};
%\node at (-2.035,-0.35) {$\alpha c$};
%\node at (-3.13,-0.35) {$0$};
%\node at (-2.1,2.1) {$s$};
\end{tikzpicture}
\end{center}
\caption{\label{fig: the four triangles}In each of the four quadrants, the $5$ collinear dots represent, from left to right, the points $0$, $\alpha c$, $\alpha c^{-1}$, $c$ and $c^{-1}$. The other dot represents $s(\xi,\eta;\alpha)$. The figures are made for $\alpha = 0.4$, $\xi = -0.325$ and $\eta = 0.256$.}
\end{figure}
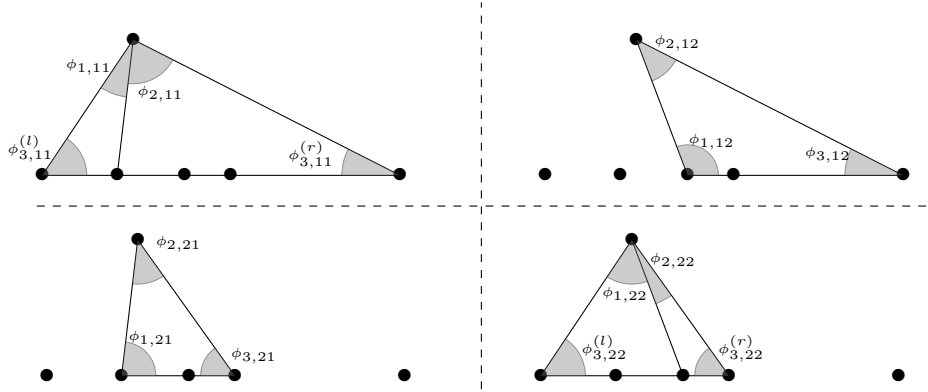
%Theorem \ref{thm:main} follows from Proposition \ref{prop:doubleintegrallimit} below, and the proof of
%this proposition will be given in Section
%\ref{sec:contour analysis}.
By combining \eqref{limits s sbar as xi eta} with Theorem \ref{thm:main}, we obtain the following immediate corollary.
\begin{corollary}\label{coro: frozen}
Let $\{(x_{N},y_{N}\}_{N \geq 1}$ be a sequence satisfying \eqref{good sequence} with $(\xi,\eta) \in \mathcal{L}_{\alpha}$. We have
\begin{align*}
& \lim_{N \to \infty} P_{j}(x_{N},y_{N}) \to \left\{ \begin{pmatrix}
1 & 1 \\ 1 & 1
\end{pmatrix}, \begin{pmatrix}
0 & 0 \\ 0 & 0
\end{pmatrix},\begin{pmatrix}
0 & 0 \\ 0 & 0
\end{pmatrix} \right\} & & \mbox{as } (\xi,\eta) \to (\xi^{\star},\eta^{\star}) \in \partial \mathcal{L}_{\alpha} \cap \partial \mathcal{F}_{1,\alpha}, \\
& \lim_{N \to \infty} P_{j}(x_{N},y_{N}) \to \left\{ \begin{pmatrix}
0 & 0 \\ 0 & 0
\end{pmatrix},\begin{pmatrix}
1 & 1 \\ 1 & 1
\end{pmatrix},\begin{pmatrix}
0 & 0 \\ 0 & 0
\end{pmatrix} \right\} & & \mbox{as } (\xi,\eta) \to (\xi^{\star},\eta^{\star}) \in \partial \mathcal{L}_{\alpha} \cap \partial \mathcal{F}_{2,\alpha}, \\
& \lim_{N \to \infty} P_{j}(x_{N},y_{N}) \to \left\{ \begin{pmatrix}
0 & 0 \\ 0 & 0
\end{pmatrix},\begin{pmatrix}
0 & 0 \\ 0 & 0
\end{pmatrix},\begin{pmatrix}
1 & 1 \\ 1 & 1
\end{pmatrix} \right\} & & \mbox{as } (\xi,\eta) \to (\xi^{\star},\eta^{\star}) \in \partial \mathcal{L}_{\alpha} \cap \partial \mathcal{F}_{3,\alpha}, \\
& \lim_{N \to \infty} P_{j}(x_{N},y_{N}) \to \left\{ \begin{pmatrix}
0 & 1 \\ 0 & 1
\end{pmatrix},\begin{pmatrix}
1 & 0 \\ 1 & 0
\end{pmatrix},\begin{pmatrix}
0 & 0 \\ 0 & 0
\end{pmatrix} \right\} & & \mbox{as } (\xi,\eta) \to (\xi^{\star},\eta^{\star}) \in \partial \mathcal{L}_{\alpha} \cap \partial \mathcal{F}_{4,\alpha}, \\
& \lim_{N \to \infty} P_{j}(x_{N},y_{N}) \to \left\{ \begin{pmatrix}
0 & 1 \\ 1 & 0
\end{pmatrix},\begin{pmatrix}
0 & 0 \\ 0 & 0
\end{pmatrix},\begin{pmatrix}
1 & 0 \\ 0 & 1
\end{pmatrix} \right\} & & \mbox{as } (\xi,\eta) \to (\xi^{\star},\eta^{\star}) \in \partial \mathcal{L}_{\alpha} \cap \partial \mathcal{F}_{5,\alpha}, \\
& \lim_{N \to \infty} P_{j}(x_{N},y_{N}) \to \left\{ \begin{pmatrix}
0 & 0 \\ 0 & 0
\end{pmatrix},\begin{pmatrix}
1 & 1 \\ 0 & 0
\end{pmatrix},\begin{pmatrix}
0 & 0 \\ 1 & 1
\end{pmatrix} \right\} & & \mbox{as } (\xi,\eta) \to (\xi^{\star},\eta^{\star}) \in \partial \mathcal{L}_{\alpha} \cap \partial \mathcal{F}_{6,\alpha},
\end{align*}
where the three matrices inside each brackets correspond, from left to right, to $j=1,2,3$.
\end{corollary}
\begin{figure}
\begin{center}
\hspace{-1cm}
\begin{tikzpicture}[master]
\node at (0,0) {\includegraphics[width=7cm]{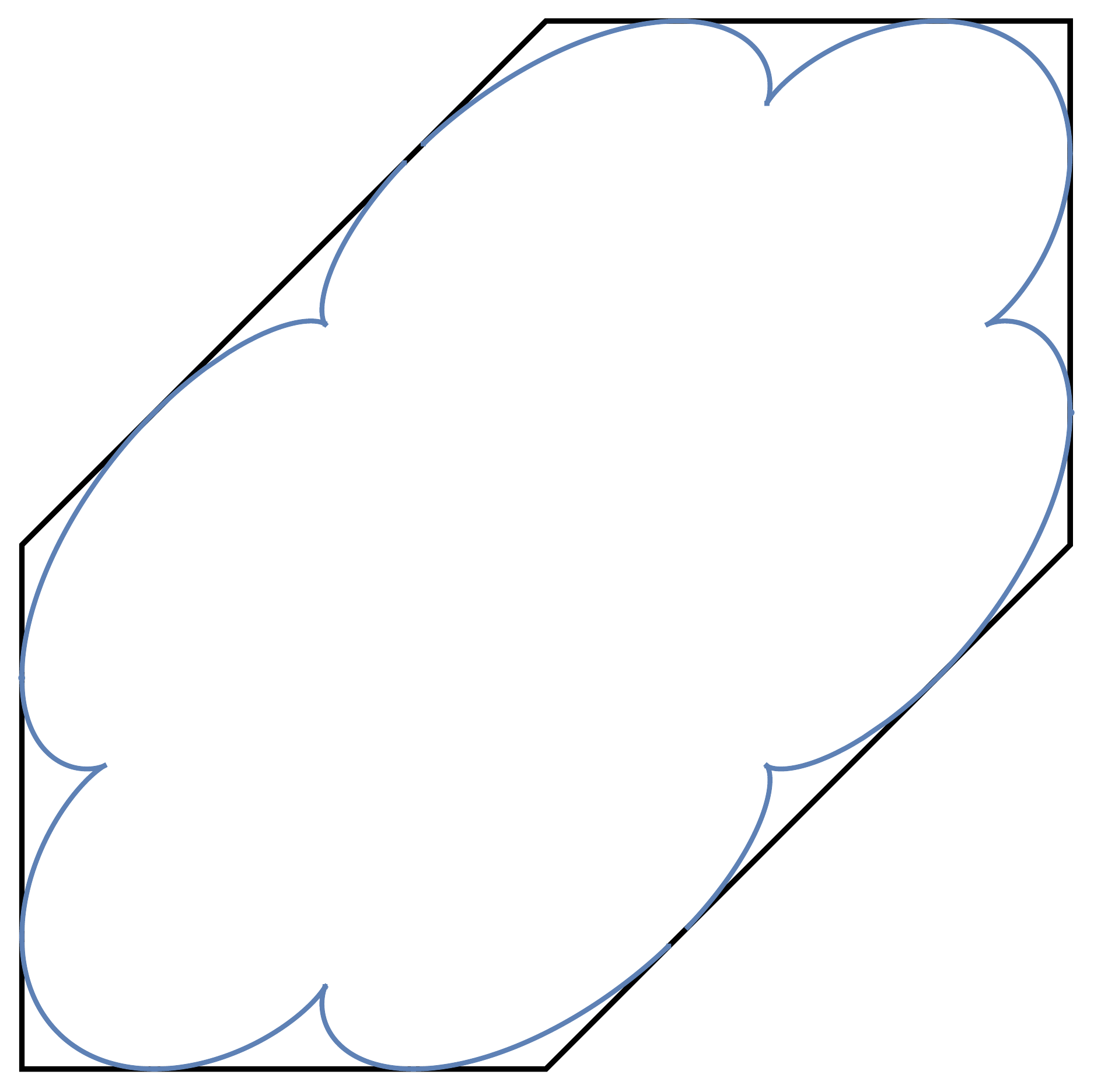}};
\node at (-3.3,0.5) {$\mathcal{F}_{1,\alpha}$};
\node at (3.3,-0.5) {$\mathcal{F}_{1,\alpha}$};
\node at (-2.95,-3.65) {$\mathcal{F}_{2,\alpha}$};
\node at (2.95,3.65) {$\mathcal{F}_{2,\alpha}$};
\node at (0,3.65) {$\mathcal{F}_{3,\alpha}$};
\node at (0,-3.65) {$\mathcal{F}_{3,\alpha}$};
\node at (-3.7,-1.6) {$\mathcal{F}_{4,\alpha}$};
\node at (3.7,1.6) {$\mathcal{F}_{4,\alpha}$};
\node at (-2,1.8) {$\mathcal{F}_{5,\alpha}$};
\node at (2,-1.8) {$\mathcal{F}_{5,\alpha}$};
\node at (-1.6,-3.65) {$\mathcal{F}_{6,\alpha}$};
\node at (1.6,3.65) {$\mathcal{F}_{6,\alpha}$};
\end{tikzpicture}
\hspace{-1cm}
\begin{tikzpicture}[slave]
\node at (0,0) {\includegraphics[width=7cm]{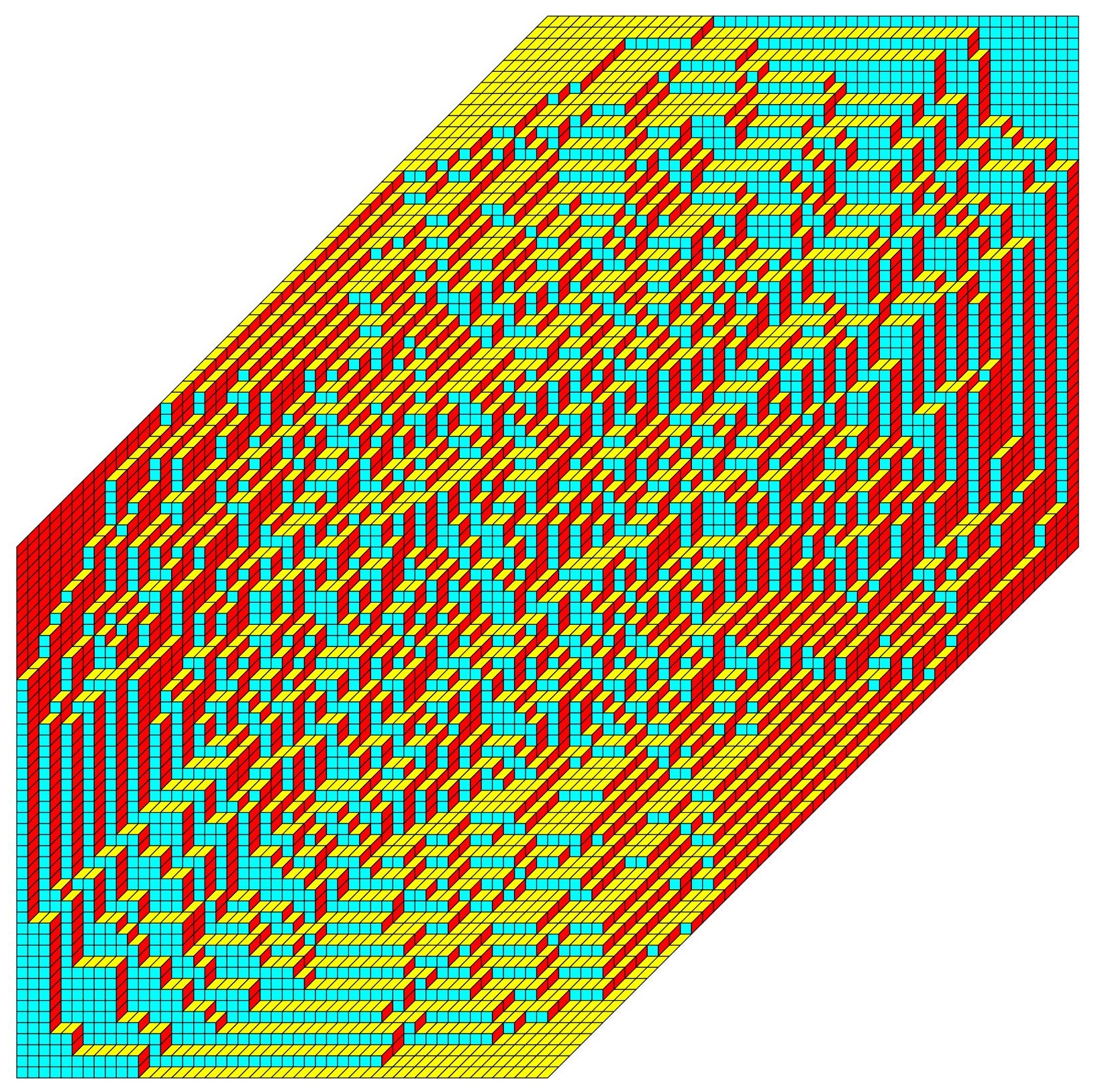}};
\node at (0,-1.2) {.};
\end{tikzpicture}
\end{center}
\caption{\label{fig: frozen regions}The six frozen regions for $\alpha = 0.3$, and a tiling of size $n=48$.}
\end{figure}
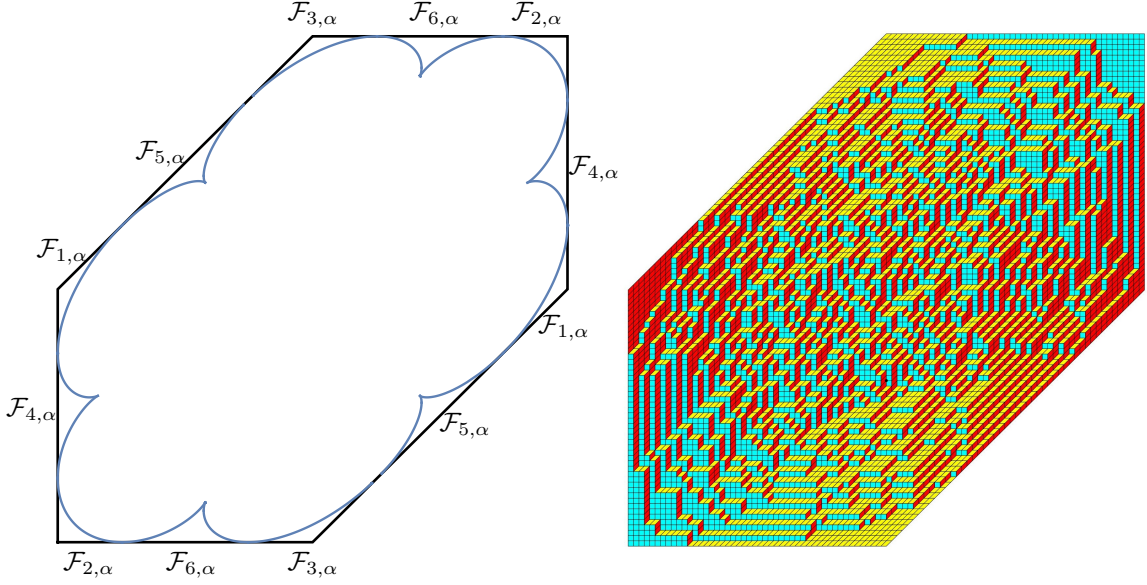

From Figure \ref{fig: frozen regions} (right), it transpires that the regions $\mathcal{F}_{j,\alpha}$, $j=1,2,3$ are frozen, and that $\mathcal{F}_{j,\alpha}$, $j=4,5,6$ are semi-frozen. More precisely, let $(x,y) \in \{0,\ldots,2N-1\}$ be such that $(\xi,\eta)=(\frac{x}{N}-1,\frac{y}{N}-1) \in \mathcal{F}_{j,\alpha}$, $j \in \{1,\ldots,6\}$. In Figure \ref{fig: frozen regions} (right), we observe that
\begin{align*}
\tikz[scale=.3,baseline=(current bounding box.center)] {\draw (0,0) \lozr; \draw (1,0) \lozr; \draw (0,1) \lozr; \draw (1,1) \lozr; \filldraw (0,0) circle(5pt); \draw (0,0) node[below] {$(2x,2y)$}; } \; , \quad \tikz[scale=.3,baseline=(current bounding box.center)] {\draw (0,0) \lozu; \draw (1,0) \lozu; \draw (0,1) \lozu; \draw (1,1) \lozu; \filldraw (0,0) circle(5pt); \draw (0,0) node[below] {$(2x,2y)$}; } \; , \quad \tikz[scale=.3,baseline=(current bounding box.center)] {\draw (0,0) \lozd; \draw (1,0) \lozd; \draw (0,1) \lozd; \draw (1,1) \lozd; \filldraw (1,0) circle(5pt); \draw (1,0) node[below] {$(2x,2y)$};} \; , \quad \tikz[scale=.3,baseline=(current bounding box.center)] {\draw (0,0) \lozu; \draw (1,0) \lozr; \draw (0,1) \lozu; \draw (1,1) \lozr; \filldraw (0,0) circle(5pt); \draw (0,0) node[below] {$(2x,2y)$};} \; , \quad \tikz[scale=.3,baseline=(current bounding box.center)] {\draw (0,0) \lozr; \draw (0,0) \lozd; \draw (-1,1) \lozd; \draw (1,1) \lozr; \filldraw (0,0) circle(5pt); \draw (0,0) node[below] {$(2x,2y)$};} \; , \quad \tikz[scale=.3,baseline=(current bounding box.center)] {\draw (-1,0) \lozd; \draw (0,0) \lozd; \draw (0,1) \lozu; \draw (1,1) \lozu; \filldraw (0,0) circle(5pt); \draw (0,0) node[below] {$(2x,2y)$};},
\end{align*}
depending on whether $(\xi,\eta) \in \mathcal{F}_{j,\alpha}$, $j=1,...,6$, respectively. Corollary \ref{coro: frozen} describes the situation at the boundary of the liquid region, and is consistent with these observations.

%We do not cover $\alpha = 1$. For $\alpha = 1$, the boundary of the liquid region is  $4 \xi^2 - 4 \xi \eta + 4 \eta^2 = 3$ and this is the equation for the ellipse that is tangent to all six sides of the hexagon (and with our choice of the repair, it looks like a circle). The semi-frozen regions disappears for $\alpha =1$.  The analysis of the later section are valid for all $\alpha \in (0,1)$. 

%\newpage
\subsection{Outline of the rest of the paper}
The proofs of Propositions \ref{prop:saddle}, \ref{prop:hightemp} and \ref{prop: s on the Riemann surface} are rather direct and are presented in Section \ref{section: easy proofs}. In Section \ref{section: first steps of the steepest descent of Y} we follow an idea of \cite{DK} and perform an eigendecomposition of the matrix valued weight. The eigenvalues and eigenvectors are naturally related to a $2$-sheeted Riemann surface $\mathcal{M}$. The proof of Theorem \ref{thm: correlation kernel final scalar expression} is given in Section \ref{section: reducing the size}, and relies on the fact that $\mathcal{M}$ is of genus $0$. The proof of Theorem \ref{thm:main} is done via a saddle point analysis in Section \ref{section: saddle point analysis}, after considerable preparations have been carried out in Sections \ref{section: lozenge probabilities}--\ref{section: phase functions}: 
\begin{itemize}
\vspace{-0.2cm}
\item[\textbullet] In Section \ref{section: lozenge probabilities}, we use Theorem \ref{thm: correlation kernel final scalar expression} to find double contour formulas for the lozenges in terms of scalar OPs. We also use the symmetry in our model to conclude that it is sufficient to prove Theorem \ref{thm:main} for the lower left quadrant of the liquid region. 
\vspace{-0.2cm}\item[\textbullet] In Section \ref{section: steepest descent for $U$}, we will perform a Deift/Zhou \cite{DZ} steepest descent analysis on the RH problem for $U$. This analysis goes via a series of transformations $U \mapsto T \mapsto S \mapsto R$. The first transformation $U \mapsto T$ uses a so-called $g$-function which is obtained in Section \ref{section: g-function}.
\vspace{-0.2cm}\item[\textbullet] The functions $\Phi$ and $\Psi$ denote the restrictions of the phase function $\Xi$ to the first and second sheet of $\mathcal{R}_{\alpha}$ and play a central role in the large $N$ analysis of the kernel. In Section \ref{section: phase functions}, we study the level set $\mathcal{N}_{\Phi} = \{ \zeta \in \mathbb{C}: \re \Phi(\zeta) = \re \Phi(s) \}$, which is of crucial importance to find the contour deformations that we need to consider for the saddle point analysis.
\end{itemize}
%The content of Section \ref{section: reducing the size} is the most novel part of this work. The rest of the paper is closer in spirit to \cite{CDKL} (though . 
As mentioned in Remark \ref{rem:alpha is not 1}, we will always assume that $\alpha \in (0,1)$, even though it will not be written explicitly.

\section{Proofs of Propositions \ref{prop:saddle}, \ref{prop:hightemp} \ref{item a in prop mapping s} and \ref{prop: s on the Riemann surface}}\label{section: easy proofs}

\subsection{Proof of Proposition \ref{prop:saddle}}\label{subsection: proof of prop 3.2}
By \eqref{eq:saddleequationhigh}, the saddles are the zeros of the polynomial $M$ given by
\begin{multline*}
M(\zeta) = (\zeta -r_{1})^{2}(\zeta -r_{2})^{2}(\zeta -r_{3})^{2}(\zeta -r_{+})(\zeta -r_{-}) - \\
\Big[ (\zeta-1)(\zeta +1)(\zeta-\alpha c)(\zeta-\tfrac{\alpha}{c})\eta - (\zeta-\alpha)(\zeta+\alpha)(\zeta-c)(\zeta-\tfrac{1}{c})\xi \Big]^{2}.
\end{multline*}
Since the coefficients of $M$ are real, Proposition \ref{prop:saddle} follows if $M$ has at least $6$ zeros on the real line. This can be proved by a direct inspection of the values of $M(\zeta)$ at $\zeta = -\infty,r_{1},0,\alpha c, r_{2}, \frac{\alpha}{c}, c, r_{3}, c^{-1},+\infty$:
\begin{align*}
& M(r_{1}) = - \frac{\alpha}{c^{2}}(1-\alpha)^{2}(c+\sqrt{\alpha})^{2} (\alpha c + \sqrt{\alpha})^{2} (\eta + \xi)^{2}, \qquad M(0) = \alpha^{4} (1-(\eta-\xi)^{2}), \\
& M(\alpha c) = (1-\alpha)^{8} c^{8}(1-\xi^{2}), \hspace{1.1cm} M(r_{2}) = - \frac{\alpha (1-\alpha)^{10}c^{8}}{(c+\sqrt{\alpha})^{8}}\big( (1+\alpha^{2})c+\sqrt{\alpha}(1+\alpha) \big)^{2}(\xi-2\eta)^{2}, \\
& M(\alpha c^{-1}) = \alpha^{4}(1-\alpha)^{8}(1-\xi^{2}), \qquad M(c) = (1-\alpha)^{8} c^{8} (1-\eta^{2}), \\
& M(r_{3}) = -\frac{\alpha (1-\alpha)^{10}c^{8}}{(\alpha c + \sqrt{\alpha})^{8}}\big( (1+\alpha^{2})c + \sqrt{\alpha}(1+\alpha) \big)^{2}(\eta - 2 \xi)^{2}, \qquad M(c^{-1}) = \frac{(1-\alpha)^{8}}{\alpha^{4}}(1-\eta^{2}).
\end{align*}
Since $(\xi,\eta) \in \mathcal{H}^{\mathrm{o}}$, where
\begin{align*}
\mathcal{H}^{\mathrm{o}}= \left\{(\xi,\eta) \mid  -1 < \xi < 1, \ -1 <  \eta < 1,\ -1 < \eta-\xi < 1 \right\},
\end{align*}
the leading coefficients of $M$ is $1-(\xi-\eta)^{2}>0$. We conclude the following:
\begin{enumerate}
\item if $\eta \neq -\xi$, $M$ has at least one simple root on $(-\infty,r_{1})$ and at least one simple root on $(r_{1},0)$, 
\item if $\eta \neq \frac{\xi}{2}$, $M$ has at least one simple root on $(\alpha c,r_{2})$ and at least one simple root on $(r_{2},\frac{\alpha}{c})$,
\item if $\eta \neq 2\xi$, $M$ has at least one simple root on $(c,r_{3})$ and at least one simple root on $(r_{3},\frac{1}{c})$. 
\end{enumerate}
Finally, other computations show that $M'(r_{1}) = 0$ if $\eta = -\xi$, that $M'(r_{2}) = 0$ if $\eta = \frac{\xi}{2}$ and that $M'(r_{3}) = 0$ if $\eta = 2 \xi$. So $M$ has at least $6$ real zeros (counting multiplicities) for each $(\xi,\eta) \in \mathcal{H}^{\mathrm{o}}$.

\subsection{Proof of Propositions \ref{prop:hightemp} and \ref{prop: s on the Riemann surface}}
We start with the proof of Proposition \ref{prop: s on the Riemann surface}. By rearranging the terms in \eqref{eq:saddlepointeq sqrt}, we see that the saddles are the solutions to
\begin{align*}
\left[ \frac{1}{2\zeta} - \frac{1}{2}\left( \frac{1}{\zeta - \alpha c} + \frac{1}{\zeta - \alpha c^{-1}} \right) \right]\xi + \left[ -\frac{1}{2\zeta} + \frac{1}{2}\left( \frac{1}{\zeta-c} + \frac{1}{\zeta-c^{-1}} \right) \right]\eta = w,
\end{align*}
where $w$ satisfies $w^{2} = \mathcal{Q}(\zeta)$. This can be rewritten as
\begin{align}\label{lol23}
\frac{-(\zeta - \alpha)(\zeta +\alpha)(\zeta -c)(\zeta - \frac{1}{c})}{(\zeta - \alpha c)(\zeta - \frac{\alpha}{c})(\zeta - 1)(\zeta + 1)}\xi + \eta = \frac{2\zeta (\zeta -c)(\zeta - \frac{1}{c})}{(\zeta - 1)(\zeta + 1)}w.
\end{align}
Taking the real and imaginary parts of \eqref{lol23}, and recalling that $\xi,\eta \in \mathbb{R}$, we get
\begin{align}\label{lol24}
\begin{pmatrix}
\re \left( \frac{-(\zeta - \alpha)(\zeta +\alpha)(\zeta -c)(\zeta - \frac{1}{c})}{(\zeta - \alpha c)(\zeta - \frac{\alpha}{c})(\zeta - 1)(\zeta + 1)} \right) & 1 \\
\im \left( \frac{-(\zeta - \alpha)(\zeta +\alpha)(\zeta -c)(\zeta - \frac{1}{c})}{(\zeta - \alpha c)(\zeta - \frac{\alpha}{c})(\zeta - 1)(\zeta + 1)} \right) & 0
\end{pmatrix}\begin{pmatrix}
\xi \\ \eta
\end{pmatrix} = \begin{pmatrix}
\re \left( \frac{2\zeta (\zeta -c)(\zeta - \frac{1}{c})}{(\zeta - 1)(\zeta + 1)}w \right) \\
\im \left( \frac{2\zeta (\zeta -c)(\zeta - \frac{1}{c})}{(\zeta - 1)(\zeta + 1)}w \right)
\end{pmatrix}.
\end{align}
Since
\begin{align*}
\frac{-(\zeta - \alpha)(\zeta +\alpha)(\zeta -c)(\zeta - \frac{1}{c})}{(\zeta - \alpha c)(\zeta - \frac{\alpha}{c})(\zeta - 1)(\zeta + 1)} = -1 + \frac{a_{1}}{\zeta-1} + \frac{a_{2}}{\zeta + 1} + \frac{a_{3}}{\zeta - \alpha c} + \frac{a_{4}}{\zeta - \frac{\alpha}{c}},
\end{align*}
with $a_{1},a_{2},a_{3},a_{4}>0$, we have
\begin{align}\label{im is neg in diffeo}
\im \frac{-(\zeta - \alpha)(\zeta +\alpha)(\zeta -c)(\zeta - \frac{1}{c})}{(\zeta - \alpha c)(\zeta - \frac{\alpha}{c})(\zeta - 1)(\zeta + 1)} < 0, \quad \mbox{ for } \im \zeta > 0.
\end{align}
Thus, the $2\times 2$ matrix at the left-hand-side of \eqref{lol24} is invertible, and we get \eqref{inverse of the diffeomorphism}. 
This shows that $(\xi,\eta) \mapsto \big(s(\xi,\eta;\alpha),w(\xi,\eta;\alpha)\big)$ is a bijection from $\mathcal{L}_{\alpha}$ to $\mathcal R_{\alpha}^+$. This mapping is clearly differentiable, and therefore it is a diffeomorphism. Replacing $(s,w) \mapsto (s,-w)$ in the right-hand-side of \eqref{inverse of the diffeomorphism}, we see that the left-hand-side becomes $(\xi,\eta) \mapsto (-\xi,-\eta)$. This implies the symmetry $s(\xi,\eta;\alpha) = s(-\xi,-\eta;\alpha)$. It remains to prove that $(\xi,\eta) \in \mathcal{L}_{\alpha}^{l}$ is mapped to a point $\big(s(\xi,\eta;\alpha),w(\xi,\eta;\alpha)\big)$ lying in the upper half plane of the first sheet. The proof of this claim is splitted in the next two lemmas. 

%\newpage

\begin{lemma}\label{lemma: im is 0 for diffeo}
We have
$\im \left(\frac{2\zeta(\zeta-c)(\zeta-c^{-1})}{(\zeta-1)(\zeta+1)}\mathcal{Q}(\zeta)^{1/2}\right) = 0$ if and only if $\zeta \in \mathbb{R}\cup \overline{\Sigma_{1}}$. 
\end{lemma}
\begin{proof}
Consider the function $f$ defined by 
\begin{align*}
f(\zeta) := \frac{(\zeta-r_{1})^{2}(\zeta-r_{2})^{2}(\zeta-r_{3})^{2}(\zeta-r_{+})(\zeta-r_{-})}{(\zeta-1)^{2}(\zeta+1)^{2}(\zeta-\alpha c)^{2}(\zeta-\alpha c^{-1})^{2}}.
\end{align*}
By the fundamental theorem, for each $x \in [0,+\infty)$, there are $8$ solutions $\zeta \in \mathbb{C}$ to $f(\zeta) = x$. The claim follows if we show that all these solutions lie on $\mathbb{R}\cup \Sigma_{1}$. First, note that the function $f$ is positive on the real line, has poles at $-1,\alpha c,\alpha c^{-1},1$, and zeros at $r_{1},r_{2},r_{3}$. Since $-1<r_{1}<\alpha c < r_{2}<\alpha c^{-1} < r_{3}<1$, the equation $f(\zeta) = x$ has at least $6$ real solutions (counting multiplicities) for each $x \in [0,+\infty)$. Furthermore, $f(\zeta) \to 1$ as $\zeta \to \pm \infty$, $f$ has a local minimum at $c^{-1} + R_{1}$, and $f(c^{-1}+R_{1}e^{it})<1$. Therefore, $f(\zeta) = x$ has $8$ solutions on $\mathbb{R}$ for each $x \in [f(c^{-1}+R_{1}),+\infty)$. It remains to show that there are two solutions on $\overline{\Sigma_{1}}$ whenever $x \in [0,f(c^{-1}+R_{1})]$. Writing $\zeta = c^{-1} + R_{1} e^{it} \in \gamma_{1}$, $t \in [-\pi,\pi]$, some computations show that
\begin{align*}
f(c^{-1} + R_{1}e^{it}) = \frac{2(\cos t-\cos \theta_{1} )\left( \cos t + \frac{\alpha^{2}+(2-\alpha)\sqrt{1-\alpha + \alpha^{2}}}{2(1-\alpha)} \right)^{2}\cos^{2}(\frac{t}{2})}{\left( \cos t + \frac{\sqrt{1-\alpha + \alpha^{2}}}{1-\alpha} \right)^{2} \left( \cos t + \frac{2-\alpha + \alpha^{2}}{2\sqrt{1-\alpha + \alpha^{2}}} \right)^{2}}.
\end{align*}
So $t \mapsto f(c^{-1} + R_{1}e^{it})$ is even, positive and decreases from $f(c^{-1}+R_{1})$ to $0$ as $t$ increases from $0$ to $\theta_{1}$, which finishes the proof.
\end{proof}
\begin{lemma}\label{lemma: proof of part (a)}
Let $(s,w) \in \mathcal{R}_{\alpha}^{+}$ such that $w = \mathcal{Q}(s)^{1/2}$ (i.e. $(s,w)$ is in the first sheet). Then, $\xi = \xi(s,w;\alpha) < 0$.
\end{lemma}
\begin{proof}
Using \eqref{inverse of the diffeomorphism} together with \eqref{im is neg in diffeo}, we infer that $\xi$ has the same sign as
\begin{equation} \label{eq:signofxi}
- \im \left(\frac{2s(s-c)(s-c^{-1})}{(s-1)(s+1)}w\right).
\end{equation}
By Lemma \ref{lemma: im is 0 for diffeo}, \eqref{eq:signofxi} is $0$ if and only if $s \in \overline{\Sigma_{1}}$, which implies that the sign of \eqref{eq:signofxi} is constant for $s \in \mathbb{C}^{+}\setminus \overline{\Sigma_{1}}$. From the expansion
\begin{align*}
-\frac{2s(s-c)(s-c^{-1})}{(s-1)(s+1)}w = 1 + \frac{a}{s} + \bigO(s^{-2}), \qquad \mbox{as } s \to \infty,
\end{align*}
where $a>0$, we conclude that \eqref{eq:signofxi} is negative for all $s$ sufficiently large and lying in $\mathbb{C}^{+}$, and the claim follows.
\end{proof}
%The expression \eqref{inverse of the diffeomorphism} for the mapping $(s,w)\mapsto (\xi,\eta)$ was particularly convenient to prove \ref{item a in prop mapping s}, because it suffices to analyze \eqref{eq:signofxi}.
This finishes the proof of Proposition \ref{prop: s on the Riemann surface} and Proposition \ref{prop:hightemp} \ref{item a in prop mapping s}. In principle, it is also possible to use \eqref{inverse of the diffeomorphism} to prove parts \ref{item b in prop mapping s}--\ref{item f in prop mapping s} of Proposition \ref{prop:hightemp}, but it leads to more involve analysis. However, by rearranging the terms in \eqref{lol23}, we can find other expressions than \eqref{inverse of the diffeomorphism} for the mapping $(s,w) \mapsto (\xi,\eta)$ that lead to simpler proofs of  \ref{item b in prop mapping s}--\ref{item f in prop mapping s}. We only sketch the proof of \ref{item e in prop mapping s}. First, we rewrite \eqref{lol23} as
\begin{align*}
\xi + \frac{-(\zeta-1)(\zeta+1)(\zeta-\alpha c)(\zeta-\alpha c^{-1})}{(\zeta-\alpha)(\zeta+\alpha)(\zeta-c)(\zeta-c^{-1})}\eta = \frac{-2\zeta(\zeta-\alpha c)(\zeta-\alpha c^{-1})}{(\zeta-\alpha)(\zeta+\alpha)}w,
\end{align*}
which implies
\begin{align*}
\begin{pmatrix}
1 & \re \left( \frac{-(\zeta-1)(\zeta+1)(\zeta-\alpha c)(\zeta-\alpha c^{-1})}{(\zeta-\alpha)(\zeta+\alpha)(\zeta-c)(\zeta-c^{-1})} \right) \\
0 & \im \left( \frac{-(\zeta-1)(\zeta+1)(\zeta-\alpha c)(\zeta-\alpha c^{-1})}{(\zeta-\alpha)(\zeta+\alpha)(\zeta-c)(\zeta-c^{-1})} \right)
\end{pmatrix} \begin{pmatrix}
\xi \\ \eta
\end{pmatrix} = \begin{pmatrix}
\re \left( \frac{-2\zeta(\zeta-\alpha c)(\zeta-\alpha c^{-1})}{(\zeta-\alpha)(\zeta+\alpha)}w \right) \\
\im \left( \frac{-2\zeta(\zeta-\alpha c)(\zeta-\alpha c^{-1})}{(\zeta-\alpha)(\zeta+\alpha)}w \right)
\end{pmatrix}.
\end{align*}
Next, we verify that
\begin{align*}
\im \left( \frac{-(\zeta-1)(\zeta+1)(\zeta-\alpha c)(\zeta-\alpha c^{-1})}{(\zeta-\alpha)(\zeta+\alpha)(\zeta-c)(\zeta-c^{-1})} \right)>0, \qquad \mbox{for } \im \zeta > 0,
\end{align*}
which implies that $\eta = \eta(\zeta,w;\alpha)$ has the same sign as 
\begin{align*}
\im \left( \frac{-2\zeta(\zeta-\alpha c)(\zeta-\alpha c^{-1})}{(\zeta-\alpha)(\zeta+\alpha)}w \right).
\end{align*}
Finally, in a similar way as in Lemma \ref{lemma: im is 0 for diffeo}, we show that this quantity is $0$ if and only if $\zeta \in \mathbb{R}\cup \Sigma_{\alpha}$, which proves part \ref{item e in prop mapping s}. We omit the proofs of parts \ref{item b in prop mapping s}, \ref{item c in prop mapping s}, \ref{item d in prop mapping s} and \ref{item f in prop mapping s}.

\section{Analysis of the RH problem for $Y$}\label{section: first steps of the steepest descent of Y}
In order to describe the behavior of $Y$ as $N \to + \infty$, one needs to control the $2 \times 2$ upper right block of the jumps, which is $A(z)^{2N}z^{-2N}$. To do this, we follow an idea of Duits and Kuijlaars \cite{DK} and proceed with the eigendecomposition of $A$. Then, we use this factorization to perform a first transformation $Y \mapsto X$ on the RH problem. 
\subsection{Eigendecomposition of $A$}
The matrix $A(z)$ defined in \eqref{def of A} has the following eigenvalues	
\begin{equation} \label{eq:lambda12} 
\lambda_{1,2}(z) = \frac{1+\alpha^{2}}{2}(1+z) \pm \frac{1-\alpha^{2}}{2}\sqrt{(z-z_{+})(z-z_{-})}, \qquad z \in \mathbb{C}\setminus [z_{-},z_{+}],
\end{equation}
where the $+$ and $-$ signs read for $\lambda_{1}$ and $\lambda_{2}$, respectively, and $z_{+}$ and $z_{-}$ are given by 
\begin{align*}
z_\pm = \frac{-(1+\alpha^{2}) \pm 2 \sqrt{\alpha(1-\alpha + \alpha^{2})}}{(1-\alpha)^{2}},
\end{align*}
and satisfy $z_{-}<-1<z_{+}<0$ and $z_{+}z_{-} = 1$. We define the square root $\sqrt{(z-z_{+})(z-z_{-})}$ such that it is analytic in $\mathbb{C}\setminus [z_{-},z_{+}]$, with an asymptotic behavior at $\infty$ given by
\begin{align*}
\sqrt{(z-z_{+})(z-z_{-})} = z + \bigO(1), \qquad \mbox{as } z \to \infty.
\end{align*}
The eigenvectors of $A$ are in the columns of the following matrix:
\begin{align} 
& E(z) = \frac{1}{1+\alpha} \begin{pmatrix} 1+\alpha & 1+\alpha \\ 
\lambda_1(z) - (\alpha^{2}+z) & \lambda_2(z) - (\alpha^{2}+z) 
\end{pmatrix} \label{def of E} \\
& =  \begin{pmatrix} 
1 & 1 \\ 
\frac{1-\alpha}{2}\big( 1-z+\sqrt{(z-z_{+})(z-z_{-})} \big)
 & \frac{1-\alpha}{2}\big( 1-z-\sqrt{(z-z_{+})(z-z_{-})} \big) \nonumber
\end{pmatrix},
\end{align}
and we have the factorization
\begin{equation}\label{eigenvalue eigenvector decomposition of A}
A(z) = E(z) \Lambda(z) E(z)^{-1},
\end{equation}
where $\Lambda(z) = \diag(\lambda_{1}(z),\lambda_{2}(z))$ is the matrix of eigenvalues. The matrix $E(z)$ is analytic for $z \in \mathbb C \setminus [z_{-},z_{+}]$,
and satisfies
\begin{align}
& E_{+}(z) = E_{-}(z)\sigma_{1}, & & z \in (z_{-},z_{+}), \label{jumps for E} \\
& E(z) = \begin{pmatrix} 
1 & 1 \\
\frac{1-\alpha+\alpha^{2}}{1-\alpha}+\bigO(z^{-1}) & -(1-\alpha) z + \bigO(1)
\end{pmatrix} & & \mbox{ as } z \to \infty, \label{asymp for E}
\end{align}
where $\sigma_1 = \begin{pmatrix} 0 & 1 \\ 1 & 0 \end{pmatrix}$.

%\begin{remark}
%If $\alpha = 1$, then $\lambda_{1,2}(z) = 1+z\pm 2 \sqrt{z}$ and 
%\begin{align*}
%E(z) = \begin{pmatrix}
%1 & 1 \\
%\sqrt{z} & -\sqrt{z}
%\end{pmatrix}.
%\end{align*}
%\end{remark}

\subsection{First transformation $Y \mapsto X$}
The first transformation of the RH problem diagonalizes the $2 \times 2$ upper right block of the jumps, and is defined by
\begin{equation}\label{def X}
X(z) = Y(z) \begin{pmatrix} E(z) & 0_2 \\ 0_2 & E(z) \end{pmatrix}.
\end{equation}
\begin{remark}
By standard arguments \cite{Deift}, we have $\det Y \equiv 1$. Note however that the $Y \mapsto X$ transformation does not preserve the unit determinant. Indeed, since $\det E(z) = -(1-\alpha)\sqrt{(z-z_{+})(z-z_{-})}$, we have $\det X(z) = (1-\alpha)^{2}(z-z_{+})(z-z_{-})$.
\end{remark}
Using the jumps for $E$ given by \eqref{jumps for E}, we verify that $X$ satisfies the following RH problem.
\subsubsection*{RH problem for $X$}
\begin{itemize}
\item[(a)] $X : \mathbb{C}\setminus (\gamma\cup [z_{-},z_{+}]) \to \mathbb{C}^{4\times 4}$ is analytic, where we recall that $\gamma$ is a closed contour surrounding $0$ once in the positive direction. 
\item[(b)] The jumps for $X$ are given by
\begin{align}
& X_{+}(z) = X_{-}(z) \begin{pmatrix}
I_{2} & \frac{\Lambda^{2N}(z)}{z^{2N}} \\ 0_{2} & I_{2}
\end{pmatrix}, & & \mbox{ for } z \in \gamma \setminus \mathcal{Z}, \\
& X_{+}(z) = X_{-}(z) \begin{pmatrix}
\sigma_{1} & 0_{2} \\ 0_{2} & \sigma_{1}
\end{pmatrix}, & & \mbox{ for } z \in (z_{-},z_{+})\setminus \mathcal{Z},
\end{align}
where $\mathcal{Z} := \gamma \cap [z_{-},z_{+}]$. Depending on $\gamma$, $\mathcal{Z}$ can be the empty set, a finite set, or an infinite set. If $\mathcal{Z}$ contains of one or several intervals, then on these intervals the jumps are 
\begin{align*}
& X_{+}(z) = X_{-}(z) \begin{pmatrix}
\sigma_{1} & 0_{2} \\ 0_{2} & \sigma_{1}
\end{pmatrix}\begin{pmatrix}
I_{2} & \frac{\Lambda_{+}^{2N}(z)}{z^{2N}} \\ 0_{2} & I_{2}
\end{pmatrix} = X_{-}(z) \begin{pmatrix}
I_{2} & \frac{\Lambda_{-}^{2N}(z)}{z^{2N}} \\ 0_{2} & I_{2}
\end{pmatrix}\begin{pmatrix}
\sigma_{1} & 0_{2} \\ 0_{2} & \sigma_{1}
\end{pmatrix}.
\end{align*}
\item[(c)] As $z \to \infty$, we have $X(z) = \left(I_{4} + \bigO(z^{-1})\right) \begin{pmatrix}
z^{N}E(z) & 0_{2} \\ 0_{2} & z^{-N}E(z)
\end{pmatrix}$. \\
As $z \to z_{-}$ or as $z \to z_{+}$, $X(z) = \bigO(1)\begin{pmatrix} E(z) & 0_2 \\ 0_2 & E(z) \end{pmatrix}$.
\end{itemize}

\section{Proof of Theorem \ref{thm: correlation kernel final scalar expression}}\label{section: reducing the size}
%In this section, we prove formula \eqref{new formula for the kernel in the thm} which expresses the kernel in terms of the solution $U$ to a $2 \times 2$ RH problem. This formula allows for a significantly simpler saddle point analysis than the formulas \eqref{kernel diag even}--\eqref{kernel diag odd}. There is a natural Riemann surface $\mathcal{M}$ associated to the eigenvalues of $A$ ($\mathcal{M}$ is defined below). Our proof of Theorem \ref{thm: correlation kernel final scalar expression} relies crucially on the fact that $\mathcal{M}$ is of genus $0$. 
%To avoid abundance of notations, we focus on the kernel along the diagonal with even horizontal coordinates \eqref{kernel diag even}. Minor adaptations of this section allow to treat similarly the kernel along the diagonal with odd horizontal coordinates \eqref{kernel diag odd}. 
First, we use the factorization of $A$ obtained in \eqref{eigenvalue eigenvector decomposition of A} together with the transformation $Y \mapsto X$ given by \eqref{def X}, to rewrite the formulas \eqref{kernel diag even}--\eqref{kernel diag odd} as follows
\begin{align}
& \big[ K(2x+\epsilon_{x},2y+j,2x+\epsilon_{x},2y+i) \big]_{i,j=0}^{1} =  \frac{1}{(2\pi i)^{2}}\int_{\gamma}\int_{\gamma} \begin{pmatrix}
\alpha^{2} & \alpha \\ w & 1
\end{pmatrix}^{\epsilon_{x}} \label{kernel X} \\ & \hspace{2cm} \times E(w)\frac{\Lambda(w)^{2N-x-\epsilon_{x}}}{w^{2N-y}}\mathcal{R}^{X}(w,z)\frac{\Lambda(z)^{x}}{z^{y+1}}E(z)^{-1} \begin{pmatrix}
1 & 1 \\ \alpha z & 1 
\end{pmatrix}^{\epsilon_{x}} dzdw, \nonumber
\end{align}
where $\mathcal{R}^{X}$ is given by
\begin{equation}\label{def of Rcal X}
\mathcal{R}^{X}(w,z) = E^{-1}(w)\mathcal{R}^{Y}(w,z)E(z) = \frac{1}{z-w} \begin{pmatrix}
0_{2} & I_{2}
\end{pmatrix}X^{-1}(w)X(z) \begin{pmatrix}
I_{2} \\ 0_{2}
\end{pmatrix}.
\end{equation}
The property \eqref{reproducing property Y} of $\mathcal{R}^{Y}$ implies the following reproducing property for $\mathcal{R}^{X}$:
\begin{equation}\label{reproducing property X}
\frac{1}{2\pi i} \int_{\gamma} P(w)E(w) \frac{\Lambda(w)^{2N}}{w^{2N}}\mathcal{R}^{X}(w,z)dw = P(z)E(z),
\end{equation}
for every $2 \times 2$ matrix valued polynomial $P$ of degree $\leq N-1$.

\vspace{0.2cm}Now, we introduce the Riemann surface $\mathcal{M}$ associated to the eigenvalues and eigenvectors of $A$. This Riemann surface is of genus $0$ and therefore there is a one-to-one map between it and the Riemann sphere (called the $\zeta$-plane).
\subsection{The Riemann surface $\mathcal{M}$ and the $\zeta$-plane}\label{subsection: 5.1}
The Riemann surface $\mathcal{M}$ is defined by
\begin{equation}
\mathcal{M} = \{(z,y)\in \mathbb{C}\times \mathbb{C}: y^2 = (z-z_{+})(z-z_{-}) \},
\end{equation}
and has genus zero. We represent it as a two-sheeted covering of the $z$-plane glued along $[z_{-},z_{+}]$. On the first sheet we require $y = z + \bigO(1)$ as $z \to \infty$, and on the second sheet we require $y = -z+ \bigO(1)$ as $z \to \infty$. To shorten the notations, a point $(z,y)$ lying on the Riemann surface will simply be denoted by $z$ when there is no confusion, that is, we will omit the $y$-coordinate. If we want to emphasize that the point $(z,y)$ is on the $j$-th sheet, $j \in \{1,2\}$, then we will use the notation $z^{(j)}$. With this convention, the two points at infinity are denoted by $\infty^{(1)}$ and $\infty^{(2)}$. The function $y$ satisfies 
\begin{align}
& y(\tfrac{1}{\alpha^{(2)}}) = - \frac{1+\alpha^{2}}{\alpha(1-\alpha)}, & & y(\alpha^{(2)}) = - \frac{1+\alpha^{2}}{1-\alpha}, \label{y computation} \\
& y(0^{(2)}) = -1, & & y(0^{(1)}) = 1. \nonumber
\end{align}
The functions $\lambda_{1}(z)$ and $\lambda_{2}(z)$ are defined on the $z$-plane (see \eqref{eq:lambda12}), and together they define a function $\lambda$ on $\mathcal{M}$ as follows:
\begin{equation}
\lambda\big((z,y)\big) = \left\{ \begin{array}{l l}
\lambda_{1}(z), & \mbox{if } (z,y) \mbox{ is on the first sheet}, \\
\lambda_{2}(z), & \mbox{if } (z,y) \mbox{ is on the second sheet}.
\end{array} \right.
\end{equation}
This is a meromorphic function on $\mathcal{M}$ with two simple poles at $\infty^{(1)}$ and $\infty^{(2)}$ and no other poles. Using \eqref{y computation}, we verify that $\lambda$ has two simple zeros at $\alpha^{(2)}$ and $\frac{1}{\alpha^{(2)}}$, and since $\mathcal{M}$ has genus $0$, $\lambda$ has no other zeros. From \eqref{jumps for E}, the matrix $E$ can also be extended to the full Riemann surface as follows
\begin{align*}
E\big((z,y)\big) & = \begin{pmatrix}
1 & 1 \\
\frac{1-\alpha}{2}(1-z+y) & \frac{1-\alpha}{2}(1-z-y)
\end{pmatrix} \\
& = \begin{cases}
E(z), & \mbox{if }(z,y) \mbox{ is on the first sheet}, \\
E(z)\sigma_{1}, & \mbox{if } (z,y) \mbox{ is on the second sheet}.
\end{cases}
\end{align*}
The function $\zeta = \zeta(z)$ defined by
\begin{equation}\label{def of zeta map}
\zeta = \,\frac{2z -(z_{+}+z_{-}) + 2 y}{z_+-z_{-}},
\end{equation}
is a conformal and bijective map from $\mathcal{M}$ to the Riemann sphere. The first sheet of $\mathcal{M}$ is mapped by \eqref{def of zeta map} to the subset $\{\zeta \in \mathbb{C}\cup\{\infty\}:|\zeta| > 1\}$ of the $\zeta$-plane, and
the second sheet is mapped to $\{\zeta \in \mathbb{C}\cup\{\infty\}: |\zeta| < 1\}$. The inverse function $z=z(\zeta)$ is given by
\begin{equation}\label{def zeta}
z = \frac{z_{+}+z_{-}}{2}+\frac{z_{+}-z_{-}}{4}\left( \zeta + \zeta^{-1} \right),
\end{equation}
where $z$ is on the first sheet if $|\zeta|>1$ and on the second sheet if $|\zeta|<1$. By definition, the above function $z(\zeta)$ vanishes at $\zeta(0^{(1)})$ and $\zeta(0^{(2)})$. Since it has simple poles at $\zeta = 0$ and $\zeta = \infty$, and since $z(\zeta) = \frac{z_{+}-z_{-}}{4}\zeta + \bigO(1)$ as $\zeta \to \infty$, \eqref{def zeta} can be rewritten as
\begin{equation}\label{z factorized in terms of zeta}
z = \frac{z_{+}-z_{-}}{4 \zeta}(\zeta-\zeta(0^{(1)}))(\zeta-\zeta(0^{(2)})).
\end{equation}
The functions $z(\zeta)$ and $\zeta(z)$ satisfy
\begin{align*}
& z(1) = z_{+}, & & z(-1) = z_{-}, & & z(\infty) = \infty^{(1)}, & & z(0) = \infty^{(2)}, \\
& \zeta(z_{+}) = 1, & & \zeta(z_{-}) = -1, & & \zeta(\infty^{(1)}) = \infty, & & \zeta(\infty^{(2)}) = 0.
\end{align*}
Also, we note that as $z \in \mathcal{M}$, $\im z = 0$, $z \notin (z_{-},z_{+})$, follows the straight line segments $[\infty^{(1)},z_{-}]$, $[z_{-},\infty^{(2)}]$, $[\infty^{(2)},z_{+}]$, $[z_{+},\infty^{(1)}]$, the function $\zeta(z)$ increases from $-\infty$ to $+\infty$. In particular, we have
\begin{align*}
\zeta(z_{-}) < \zeta(\infty^{(2)}) < \zeta(\tfrac{1}{\alpha^{(2)}}) < \zeta(\alpha^{(2)}) < \zeta(0^{(2)}) < \zeta(z_{+}) < \zeta(0^{(1)}).
\end{align*}
The following identities will be useful later, and can be verified by direct computations:
\begin{align} 
& y = \frac{z_{+}-z_{-}}{4}\left( \zeta-\zeta^{-1} \right),  \qquad \frac{dz}{y} = \frac{d\zeta}{\zeta}, \label{dz y in terms of zeta} \\
& \lambda =  \frac{z_{+}-z_{-}}{4 \zeta} (\zeta - \zeta(\tfrac{1}{\alpha^{(2)}}))(\zeta-\zeta(\alpha^{(2)})), \label{lambda in terms of zeta} \\
& \frac{d\lambda}{dz}  =  \frac{\zeta^{2}- \alpha^{2}}{\zeta^2-1}, \label{lambda prime in terms of zeta} \\
& \frac{dz}{d\zeta} = \frac{z_{+}-z_{-}}{4\zeta}\left( \zeta-\zeta^{-1} \right). \label{dz_sur_dzeta in terms of zeta}
\end{align}
We define $c$ by
\begin{align*}
c = \frac{z_{+}-z_{-}}{-(z_{+}+z_{-})+2\sqrt{z_{+}z_{-}}} = \sqrt{\frac{\alpha}{1-\alpha +\alpha^{2}}}  < 1.
\end{align*}
From straightforward calculations using \eqref{def of zeta map}, we have
\begin{align*}
& \zeta(\tfrac{1}{\alpha^{(2)}}) =  \alpha c, & & \zeta(\alpha^{(2)}) = \alpha c^{-1}, \\
& \zeta(0^{(2)}) = c, & & \zeta(0^{(1)}) = c^{-1},
\end{align*}
and
\begin{align}\label{lambda - alpha2 - z in terms of zeta}
\lambda(z) - \alpha^{2} - z = \frac{1+\alpha^{3}}{1-\alpha}\frac{\zeta -c}{\zeta}.
\end{align}
%and
%\begin{align*}
%\lambda = \frac{z_{+}-z_{-}}{4c\zeta} \left( c^{2}\zeta^{2} - \frac{\alpha(1+\alpha)(1+\beta)}{\sqrt{(\alpha+\beta)(1+\alpha\beta)}}c\zeta + \alpha^{2} \right)
%\end{align*}

\subsection{The reproducing kernel $\mathcal{R}^{\mathcal{M}}$}\label{subsection: 5.2}
For $w^{(j)}$ on the $j$-th sheet of $\mathcal{M}$ and $z^{(k)}$ on the $k$-th sheet, we define
$\mathcal{R}^{\mathcal{M}}(w^{(j)},z^{(k)})$ by
\begin{equation}\label{def of reproducing kernel M}
\mathcal{R}^{\mathcal{M}}(w^{(j)},z^{(k)}) = y(w^{(j)}) e_{j}^{T} \mathcal{R}^{X}(w,z) e_{k},
\end{equation}
where $e_{1} = \begin{pmatrix}
1 \\ 0
\end{pmatrix}$ and $ e_{2} = \begin{pmatrix}
0 \\ 1
\end{pmatrix} $. Note that $\mathcal{R}^{\mathcal{M}} : \mathcal{M}_{*} \times \mathcal{M}_{*} \to \mathbb{C}$ is scalar valued, with $\mathcal{M}_{*} =  \mathcal{M}\setminus\{\infty^{(1)},\infty^{(2)}\}$. It is convenient for us to consider formal sums of points on $\mathcal{M}$, which are called \textit{divisors} in the literature. More precisely, a divisor $D$ can be written in the form
\begin{align*}
D = \sum_{j=1}^{k}n_{j}z_{j}, \qquad k \geq 1, \quad  n_{j} \in \mathbb{Z}, \quad z_{j} \in \mathcal{M},
\end{align*}
and we say that $D \geq 0$ if $n_{1},\ldots,n_{k} \geq 0$. The divisor of a non-zero meromorphic function $f$ on $\mathcal{M}$ is defined by
\begin{align*}
\mbox{div}(f) := n_{1}z_{1}+...+n_{k_{1}}z_{k_{1}} - n_{k_{1}+1}z_{k_{1}+1} - ... - n_{k_{2}} z_{k_{2}},
\end{align*}
where $z_{1},\ldots,z_{k_{1}}$ are the zeros of $f$ of multiplicities $n_{1},\ldots,n_{k_{1}}$, respectively, and $z_{k_{1}+1},\ldots,z_{k_{2}}$ are the poles of $f$ of order $n_{k_{1}+1},\ldots,n_{k_{2}}$, respectively. Given a divisor $D$, we define $L(-D)$ as the vector space of meromorphic functions on $\mathcal{M}$ given by
\begin{align*}
L(-D) = \{ f: \mbox{div}(f) \geq -D \mbox{ or } f \equiv 0 \}.
\end{align*}
The following divisors will play an important role:
\begin{align*}
& D_{N} = (N-1) \cdot \infty^{(1)} + N \cdot \infty^{(2)}, \\
& D_{N}^{*} = N \cdot \infty^{(1)} + (N-1) \cdot \infty^{(2)}.
\end{align*}
Thus $L_{N} := L(-D_{N})$ is the vector space of meromorphic functions on $\mathcal{M}$, with poles at $\infty^{(1)}$ and $\infty^{(2)}$ only, such that the pole at $\infty^{(1)}$ is of order at most $N-1$, and the pole at $\infty^{(2)}$ is of order at most $N$. Similarly we define $L_{N}^{*} = L(-D_{N}^{*})$. From the Riemann-Roch theorem, we have
\begin{align*}
\mbox{dim } L_{N} = \mbox{dim } L_{N}^{*} = 2N,
\end{align*}
since there is no holomorphic differential (other than the zero differential) on a Riemann surface of genus $0$.

\begin{lemma}\label{lemma: reproducing kernel Riemann surface}
We have
\begin{itemize}
\item[(a)] $z \mapsto \mathcal{R}^{\mathcal{M}}(w,z) \in L_{N}$ for every $w \in \mathcal{M}_{*}$, 
\item[(b)] $w \mapsto \mathcal{R}^{\mathcal{M}}(w,z) \in L_{N}^{*}$ for every $z \in \mathcal{M}_{*}$, 
\item[(c)] $\mathcal{R}^{\mathcal{M}}$ is a reproducing kernel for $L_{N}$ in the sense that
\begin{equation}\label{reproducing property M}
\frac{1}{2\pi i} \int_{\gamma_{\mathcal{M}}} f(w) \frac{\lambda^{2N}(w)}{w^{2N}} \mathcal{R}^{\mathcal{M}}(w,z) \frac{dw}{y(w)} = f(z)
\end{equation}
for every $f \in L_{N}$, where $\gamma_{\mathcal{M}}$ is a closed contour surrounding once $0^{(1)}$ and $0^{(2)}$ on the Riemann surface $\mathcal{M}$ in the positive direction (in particular $\gamma_{\mathcal{M}}$ visits both sheets).
\end{itemize}
\end{lemma}
\begin{proof}
Using the definitions of $\mathcal{R}^{\mathcal{M}}$ and $\mathcal{R}^{X}$ given by \eqref{def of reproducing kernel M} and \eqref{def of Rcal X}, we can rewrite $\mathcal{R}^{\mathcal{M}}$ as
\begin{equation}\label{kernel Riemann f and g}
\mathcal{R}^{\mathcal{M}}(w,z) = \frac{\sum_{j=1}^{4}g_{j}(w)f_{j}(z)}{z-w}, \qquad w,z \in \mathcal{M}_{*},
\end{equation}
where
\begin{equation*}
f_{j}(z) = \left\{ \begin{array}{l l}
X_{j1}(z), & \mbox{if } z = z^{(1)}, \\
X_{j2}(z), & \mbox{if } z = z^{(2)},
\end{array} \right. 
\end{equation*}
and
\begin{align}\label{def of gj in reducing}
g_{j}(w) = y(w) \left\{ \begin{array}{l l}
(X^{-1})_{3j}(w), & \mbox{if } w = w^{(1)}, \\
(X^{-1})_{4j}(w), & \mbox{if } w = w^{(2)}.
\end{array} \right.
\end{align}
From properties (a) and (b) of the RH problem for $X$, the functions $f_{j}$ are analytic in $\mathcal{M}_{*}$. By combining the large $z$ asymptotics of $E(z)$ (given by \eqref{asymp for E}) with property (c) of the RH problem for $X$, we obtain
\begin{align*}
X(z) \begin{pmatrix}
I_{2} \\ 0_{2}
\end{pmatrix} = \begin{pmatrix}
\bigO(z^{N}) & \bigO(z^{N}) \\
\bigO(z^{N}) & \bigO(z^{N+1}) \\
\bigO(z^{N-1}) & \bigO(z^{N}) \\
\bigO(z^{N-1}) & \bigO(z^{N}) 
\end{pmatrix}, \quad \mbox{as } z \to \infty,
\end{align*}
from which we conclude that the functions $f_{j}$'s have poles of order at most $N$ at $\infty^{(1)}$ and at most $N+1$ at $\infty^{(2)}$. Therefore, we have shown that
\begin{align*}
f_{j} \in L(-(D_{N}+\infty^{(1)}+\infty^{(2)})), \qquad j = 1,2,3,4.
\end{align*}
The numerator in \eqref{kernel Riemann f and g} is, for each fixed $w \in \mathcal{M}_{*}$ a linear combination of the functions $f_{j}$, so belong to $L(-(D_{N}+\infty^{(1)}+\infty^{(2)}))$ as a function of $z$. By definitions of $\mathcal{R}^{\mathcal{M}}$ and $\mathcal{R}^{X}$, the numerator vanishes for $z = w^{(1)}$ and for $z = w^{(2)}$. Thus the division by $z-w$ in \eqref{kernel Riemann f and g} does not introduce any poles, but it reduces the order of the poles at $\infty^{(1)}$ and $\infty^{(2)}$ by one, and therefore $z \mapsto \mathcal{R}^{\mathcal{M}}(w,z) \in L_{N}$ as claimed in part (a). Now we turn to the proof of part (b). First, we note that
\begin{align*}
E(w)^{-1} = \frac{-\frac{1}{1-\alpha}}{\sqrt{(w-z_{+})(w-z_{-})}} \begin{pmatrix}
\frac{1-\alpha}{2}\big( 1-w-\sqrt{(w-z_{+})(w-z_{-})} \big) & -1 \\
- \frac{1-\alpha}{2}\big( 1-w+\sqrt{(w-z_{+})(w-z_{-})} \big) & 1
\end{pmatrix}.
\end{align*}
Therefore, since $\det Y \equiv 1$, by using condition (c) of the RH problem for $X$, we have
\begin{align*}
X^{-1}(w) = \begin{pmatrix}
E^{-1}(w) & 0_{2} \\
0_{2} & E^{-1}(w)
\end{pmatrix} \begin{pmatrix}
\bigO(1) & \bigO(1) \\
\bigO(1) & \bigO(1)
\end{pmatrix} \qquad \mbox{as } w \to z_{\star} \in \{z_{+},z_{-}\},
\end{align*}
and we conclude from \eqref{def of gj in reducing} that the functions $g_{j}$ are also analytic in $\mathcal{M}_{*}$. On the other hand, by using the asymptotics $Y(w)=I_{4}+\bigO(w^{-1})$ as $w \to \infty$ together with the fact that $\det Y \equiv 1$, we can obtain asymptotics for $X^{-1}(w)$ as $w \to \infty$ using \eqref{def X}. After some simple computations, we get
\begin{equation*}
y(w) \begin{pmatrix}
0_{2} & I_{2}
\end{pmatrix}X^{-1}(w) =  \begin{pmatrix}
\bigO(w^{N}) & \bigO(w^{N}) & \bigO(w^{N+1}) & \bigO(w^{N})  \\
\bigO(w^{N-1}) & \bigO(w^{N-1}) & \bigO(w^{N}) & \bigO(w^{N}) 
\end{pmatrix},
\end{equation*}
from which it follows that
\begin{align*}
g_{j} \in L(-(D_{N}^{*}+\infty^{(1)}+\infty^{(2)})), \qquad j = 1,2,3,4.
\end{align*}
We conclude the proof of part (b) as in part (a), by noting that $\mathcal{R}^{\mathcal{M}}(w,z)$ in \eqref{kernel Riemann f and g} has no pole at $z=w$ (on any sheet). Finally, let us take $P(w) = p(w) e_{1}^{T} = p(w) \begin{pmatrix}
1 & 0
\end{pmatrix}$ in \eqref{reproducing property X}, with $p$ a scalar polynomial satisfying $\mbox{deg }p \leq N-1$. Since $e_{1}^{T} E(w) = \begin{pmatrix}
1 & 1
\end{pmatrix} = e_{1}^{T} + e_{2}^{T}$, it gives
\begin{align*}
p(z)\begin{pmatrix}
1 & 1
\end{pmatrix} = \frac{1}{2\pi i} \int_{\gamma}p(w)(e_{1}^{T}+e_{2}^{T}) \frac{\Lambda(w)^{2N}}{w^{2N}}\mathcal{R}^{X}(w,z)dw.
\end{align*}
By multiplying the above from the right by $e_{k}$, we obtain
\begin{align*}
p(z) = & \; \frac{1}{2\pi i} \int_{\gamma}p(w) \frac{\lambda_{1}(w)^{2N}}{w^{2N}}e_{1}^{T}\mathcal{R}^{X}(w,z)e_{k}dw  + \frac{1}{2\pi i} \int_{\gamma}p(w) \frac{\lambda_{2}(w)^{2N}}{w^{2N}}e_{2}^{T}\mathcal{R}^{X}(w,z)e_{k}dw.
\end{align*}
We denote $\gamma^{(1)}$ and $\gamma^{(2)}$ for the projections of $\gamma$ on the first and second sheet of $\mathcal{M}$, respectively. Using \eqref{def of reproducing kernel M}, the above can be rewritten as
\begin{align*}
p(z) = & \; \frac{1}{2\pi i} \int_{\gamma^{(1)}}p(w) \frac{\lambda(w)^{2N}}{w^{2N}}\mathcal{R}^{\mathcal{M}}(w,z^{(k)})\frac{dw}{y(w)} + \frac{1}{2\pi i} \int_{\gamma^{(2)}}p(w) \frac{\lambda(w)^{2N}}{w^{2N}}\mathcal{R}^{\mathcal{M}}(w,z^{(k)})\frac{dw}{y(w)},
\end{align*}
for every $z \in \mathbb{C}$ and for any $k \in \{1,2\}$. The two integrals combine to one integral over a contour $\gamma_{\mathcal{M}}$ surrounding both $0^{(1)}$ and $0^{(2)}$ on $\mathcal{M}$ in the positive direction, and thus
\begin{align}\label{lol8}
p(z) = \frac{1}{2 \pi i} \int_{\gamma_{\mathcal{M}}}p(w) \frac{\lambda(w)^{2N}}{w^{2N}} \mathcal{R}^{\mathcal{M}}(w,z) \frac{dw}{y(w)}, \qquad \mbox{deg } p \leq N-1,
\end{align}
for every $z \in \mathcal{M}_{*}$. Let us now take $P(w) = p(w)e_{2}^{T} = p(w) \begin{pmatrix}
0 & 1
\end{pmatrix}$ in \eqref{reproducing property X}, and note that
\begin{equation*}
e_{2}^{T}E(w) = \frac{1}{1+\alpha} \begin{pmatrix}
\lambda_{1}(w)-\alpha^{2}-w & \lambda_{2}(w)-\alpha^{2}-w
\end{pmatrix}.
\end{equation*}
The two above entries together define the meromorphic function $w \in \mathcal{M} \mapsto \frac{1}{1+\alpha}(\lambda(w)-\alpha^{2} - w)$ on $\mathcal{M}$. By proceeding in a similar way as for \eqref{lol8}, we obtain this time
\begin{align*}
p(z)(\lambda(z)-(\alpha^{2}+z)) = \frac{1}{2 \pi i} \int_{\gamma_{\mathcal{M}}}p(w)(\lambda(w)-(\alpha^{2}+w)) \frac{\lambda(w)^{2N}}{w^{2N}} \mathcal{R}^{\mathcal{M}}(w,z) \frac{dw}{y(w)},
\end{align*}
for all scalar polynomials $p$ with $\mbox{deg }p \leq N-1$ and for all $z \in \mathcal{M}_{*}$. Therefore, for any function $f$ in the form
\begin{align*}
f(z) = p_{1}(z) + p_{2}(z)(\lambda(z)-\alpha^{2}-z)
\end{align*}
with $p_{1}$, $p_{2}$ two polynomials of degree $\leq N-1$, we have proved that
\begin{align*}
f(z) = \frac{1}{2 \pi i} \int_{\gamma_{\mathcal{M}}}f(w) \frac{\lambda(w)^{2N}}{w^{2N}} \mathcal{R}^{\mathcal{M}}(w,z) \frac{dw}{y(w)}.
\end{align*}
Let $L := \{f:f(z) = p_{1}(z) + p_{2}(z)(\lambda(z)-\alpha^{2}-z)\mbox{ with }p_{1},p_{2} \mbox{ two polynomials of degree } \leq N-1\}$. Since $z \mapsto \lambda-\alpha^{2}-z$ has a simple pole at $\infty^{(2)}$ (and no other poles), we conclude that $L \subseteq L_{N}$. Note also that $\mbox{dim } L = \mbox{dim } L_{N} = 2N$, and thus we have $L=L_{N}$. This finishes the proof.
\end{proof}

\subsection{The reproducing kernel $\mathcal{R}^{U}$}\label{subsection: 5.3}
To ease the notations, we define $z = z(\zeta)$ and $w = w(\omega)$ by
\begin{align}
& z = \frac{z_{+}+z_{-}}{2}+\frac{z_{+}-z_{-}}{4}\left( \zeta + \zeta^{-1} \right), \qquad \; \; \zeta \in \mathbb{C}\cup  \{\infty\}, \; z \in \mathcal{M}, \label{z in terms of zeta lol} \\
& w = \frac{z_{+}+z_{-}}{2}+\frac{z_{+}-z_{-}}{4}\left( \omega + \omega^{-1} \right), \qquad \omega \in \mathbb{C}\cup  \{\infty\}, \; w \in \mathcal{M}, \label{w in terms of omega lol}
\end{align}
with the same convention as in \eqref{def zeta}, that is, $z$ (resp. $w$) is on the first sheet if $|\zeta|>1$ (resp. $|\omega|>1$), and on the second sheet if $|\zeta|<1$ (resp. $|\omega|<1$). We define $\mathcal{R}^{U}$ in terms of $\mathcal{R}^{\mathcal{M}}$ as follows
\begin{equation}\label{def of Reproducing kernel C}
\mathcal{R}^{U}(\omega,\zeta) = \omega^{N-1}\zeta^{N}\mathcal{R}^{\mathcal{M}}(w(\omega),z(\zeta)).
\end{equation}
%Let $W$ be the scalar weight defined in \eqref{def of W in intro}. Proposition \ref{prop: reprod Rcal U} states that $\mathcal{R}^{U}$ is the reproducing kernel with weight $W$ for the scalar polynomials of degree $\leq 2N-1$.
\begin{proposition}\label{prop: reprod Rcal U}
Let $W$ and $c$ be defined as in \eqref{def of W in intro}. $\mathcal{R}^{U}$ is a bivariate polynomial of degree $\leq 2N-1$ in both $\omega$ and $\zeta$. It satisfies
\begin{equation}\label{reproducing property C}
\frac{1}{2\pi i}\int_{\gamma_{\mathbb{C}}} p(\omega) W(\omega) \mathcal{R}^{U}(\omega,\zeta)d\omega = p(\zeta)
\end{equation}
for every scalar polynomial $p$ of degree $\leq 2N-1$, where $\gamma_{\mathbb{C}}$ is a closed curve in the complex plane going around $c$ and $c^{-1}$ once in the positive direction, but not going around $0$.
\end{proposition}
\begin{proof}
From part (a) of Lemma \ref{lemma: reproducing kernel Riemann surface}, for each $w \in \mathcal{M}_{*}$, the function $z \mapsto \mathcal{R}^{\mathcal{M}}(w,z)$ is meromorphic on $\mathcal{M}$, with a pole of order at most $N-1$ at $\infty^{(1)}$ and a pole of order at most $N$ at $\infty^{(2)}$. Since $z(0) = \infty^{(2)}$ and $z(\infty) = \infty^{(1)}$, we conclude that for each $\omega \in \mathbb{C}$, the function $\zeta \mapsto  \mathcal{R}^{\mathcal{M}}(w(\omega),z(\zeta))$ is meromorphic on $\mathbb{C}\cup \{\infty\}$, with a pole of order at most $N-1$ at $\infty$ and another pole of order at most $N$ at $0$. Therefore, for each $\omega \in \mathbb{C}$, the function $\zeta \mapsto \mathcal{R}^{U}(\omega,\zeta)$ is a polynomial of degree at most $2N-1$. From part (b) of Lemma \ref{lemma: reproducing kernel Riemann surface}, we conclude similarly that for each $\zeta \in \mathbb{C}$, the function $\omega \mapsto \mathcal{R}^{U}(\omega,\zeta)$ is a polynomial of degree at most $2N-1$. So we have proved that $\mathcal{R}^{U}$ is a bivariate polynomial of degree $\leq 2N-1$ in both $\omega$ and $\zeta$. 

Now, we turn to the proof of \eqref{reproducing property C}. It can be directly verified from \eqref{w in terms of omega lol} (see also \eqref{z factorized in terms of zeta}) that $\omega(0^{(1)}) = c^{-1}$, $\omega(0^{(2)})=c$, $(\partial_{\omega}w)(c^{-1})>0$ and $(\partial_{\omega}w)(c)<0$. In particular, the map $w \mapsto \omega(w)$ is conformal in small neighborhoods of $0^{(1)}$ and $0^{(2)}$. Since conformal maps preserve orientation, the curve $\gamma_{\mathcal{M}}$ which surrounds both $0^{(1)}$ and $0^{(2)}$ once in the positive direction, is mapped by $w \mapsto \omega(w)$ onto a curve $\gamma_{\mathbb{C}}$ on the complex plane, which surrounds $c$ and $c^{-1}$ once in the positive direction. Furthermore, since $\omega(\infty^{(2)}) = 0$, the curve $\gamma_{\mathbb{C}}$ does not surround $0$. By changing variables $(w,z) \mapsto (\omega,\zeta)$ in \eqref{reproducing property M}, and by using \eqref{z factorized in terms of zeta}, \eqref{dz y in terms of zeta} and \eqref{lambda in terms of zeta}, we obtain 
\begin{equation*}
\begin{array}{r c l}
f(z(\zeta)) & = & \ds \frac{1}{2\pi i} \int_{\gamma_{\mathbb{C}}}f(w(\omega)) \frac{\lambda^{2N}(w(\omega))}{w(\omega)^{2N}}\mathcal{R}^{\mathcal{M}}(w(\omega),z(\zeta)) \frac{dw(\omega)}{y(w(\omega))} \\
& = & \ds \frac{1}{2\pi i} \int_{\gamma_{\mathbb{C}}}f(w(\omega)) \bigg( \frac{(\omega-\alpha c)(\omega-\alpha c^{-1})}{(\omega-c)(\omega-c^{-1})} \bigg)^{2N} \mathcal{R}^{\mathcal{M}}(w(\omega),z(\zeta))\frac{d\omega}{\omega},
\end{array}
\end{equation*}
for every $f \in L_{N}$. Since $f \in L_{N}$, the function $\zeta \mapsto f(z(\zeta))$ is meromorphic on the Riemann sphere, with a pole of degree at most $N$ at $\zeta = 0$ and a pole of degree at most $N-1$ at $\zeta = \infty$. In other words, $\zeta \mapsto \zeta^{N}f(z(\zeta)) =: p(\zeta)$ is a polynomial of degree at most $2N-1$. By multiplying the above equality by $\zeta^{N}$, we thus have
\begin{equation*}
p(\zeta) = \frac{1}{2\pi i}\int_{\gamma_{\mathbb{C}}} \frac{p(\omega)}{\omega^{N}} \bigg( \frac{(\omega-\alpha c)(\omega-\alpha c^{-1})}{(\omega-c)(\omega-c^{-1})} \bigg)^{2N} \mathcal{R}^{\mathcal{M}}(w(\omega),z(\zeta))\zeta^{N}\frac{d\omega}{\omega}.
\end{equation*}
We obtain the claim after substituting \eqref{def of Reproducing kernel C} in the above expression.
\end{proof}
Now, we prove formula \eqref{reproducing kernel in terms of U intro}, which expresses $\mathcal{R}^{U}$ in terms of the solution $U$ to the $2 \times 2$ RH problem presented in Section \ref{subsection: new formula for the kernel}.
\begin{proposition}
The reproducing kernel $\mathcal{R}^{U}$ defined by \eqref{def of Reproducing kernel C} can be rewritten in terms of $U$ as follows
\begin{equation}\label{reproducing kernel in terms of U}
\mathcal{R}^{U}(\omega,\zeta) = \frac{1}{\zeta-\omega} \begin{pmatrix}
0 & 1
\end{pmatrix} U^{-1}(\omega)U(\zeta) \begin{pmatrix}
1 \\ 0
\end{pmatrix}.
\end{equation}
\end{proposition}
\begin{proof}
By \cite[Lemma 4.6 (c)]{DK}, there is a unique bivariate polynomial $\mathcal{R}^{U}$ of degree $\leq 2N-1$ in both $\omega$ and $\zeta$ which satisfies \eqref{reproducing property C}. Therefore, it suffices to first replace $\mathcal{R}^{U}$ in the left-hand-side of \eqref{reproducing property C} by 
\begin{align*}
\frac{1}{\zeta-\omega} \begin{pmatrix}
0 & 1
\end{pmatrix} U^{-1}(\omega)U(\zeta) \begin{pmatrix}
1 \\ 0
\end{pmatrix},
\end{align*}
and then to verify that \eqref{reproducing property C} still holds with this replacement. The rest of the proof goes exactly as in the proof of \cite[Proposition 4.9]{DK}, so we omit it.
\end{proof}
\subsection{Proof of formula \eqref{new formula for the kernel in the thm}}
Now, using the results of Sections \ref{subsection: 5.1}, \ref{subsection: 5.2} and \ref{subsection: 5.3}, we give a proof for formula \eqref{new formula for the kernel in the thm}. From \eqref{kernel diag even}--\eqref{kernel diag odd}, for $x \in \{1,\ldots,2N-1\}$, $y \in \mathbb{Z}$ and $\epsilon_{x} \in \{0,1\}$, we have
\begin{align}
& \big[ K(2x+\epsilon_{x},2y+j,2x+\epsilon_{x},2y+i) \big]_{i,j=0}^{1} \nonumber \\ 
& = \frac{1}{(2\pi i)^{2}}\int_{\gamma}\int_{\gamma}\begin{pmatrix}
\alpha^{2} & \alpha \\
w & 1
\end{pmatrix}^{\epsilon_{x}}\frac{A(w)^{2N-x-\epsilon_{x}}}{w^{2N-y}}\mathcal{R}^{Y}(w,z)\frac{A(z)^{x}}{z^{y+1}}\begin{pmatrix}
1 & 1 \\
\alpha z & 1
\end{pmatrix}^{\epsilon_{x}} dzdw, \label{K in terms of Rcal Y unified}
\end{align}
where $\gamma$ is a closed contour surrounding $0$ once in the positive direction. The proof consists of using the successive transformations $\mathcal{R}^{Y} \mapsto \mathcal{R}^{X} \mapsto \mathcal{R}^{\mathcal{M}} \mapsto \mathcal{R}^{U}$. We first use the eigendecomposition \ref{eigenvalue eigenvector decomposition of A} of $A$ and the $\mathcal{R}^{Y} \mapsto \mathcal{R}^{X}$ transformation given in \eqref{def of Rcal X} to rewrite \eqref{K in terms of Rcal Y unified} as
\begin{align}
& \big[ K(2x+\epsilon_{x},2y+j,2x+\epsilon_{x},2y+i) \big]_{i,j=0}^{1} \nonumber \\ 
& = \frac{1}{(2\pi i)^{2}}\int_{\gamma}\int_{\gamma} \begin{pmatrix}
\alpha^{2} & \alpha \\
w & 1
\end{pmatrix}^{\epsilon_{x}} E(w)\frac{\Lambda(w)^{2N-x-\epsilon_{x}}}{w^{2N-y}} \mathcal{R}^{X}(w,z) \frac{\Lambda(z)^{x}}{z^{y+1}}E(z)^{-1}\begin{pmatrix}
1 & 1 \\
\alpha z & 1
\end{pmatrix}^{\epsilon_{x}} dzdw. \label{K in terms of Rcal X unified}
\end{align}
Using \eqref{def of E}, we can write $E(w)$ and $E(z)^{-1}$ as
\begin{align}
E(w) & \;  = \begin{pmatrix}
1 & 1 \\
\frac{\lambda(w^{(1)})-\alpha^{2}-w}{1+\alpha} & \frac{\lambda(w^{(2)})-\alpha^{2}-w}{1+\alpha}
\end{pmatrix}, \qquad w \in \mathbb{C}, \label{lol3} \\
E(z)^{-1} & \; = \frac{1}{1-\alpha}\begin{pmatrix}
\frac{(1+\alpha^{3})z}{y(z^{(1)})(\lambda(z^{(1)})-\alpha^{2}-z)} & \frac{1}{y(z^{(1)})} \\
\frac{(1+\alpha^{3})z}{y(z^{(2)})(\lambda(z^{(2)})-\alpha^{2}-z)} & \frac{1}{y(z^{(2)})}
\end{pmatrix}, \qquad z \in \mathbb{C}, \label{lol4}
\end{align}
where we have also used the relation
\begin{equation*}
(\lambda_{1}-\alpha^{2}-z)(\lambda_{2}-\alpha^{2}-z) = -(1+\alpha)(1+\alpha^{3})z
\end{equation*}
to obtain \eqref{lol4}. The identities \eqref{lol3} and \eqref{lol4} allow to rewrite the integrand of \eqref{K in terms of Rcal X unified} by noting that
\begin{align*}
&E(w)\frac{\Lambda(w)^{2N-x-\epsilon_{x}}}{w^{2N-y}} \mathcal{R}^{X}(w,z) \frac{\Lambda(z)^{x}}{z^{y+1}}E(z)^{-1} = \sum_{j,k=1}^{2}\begin{pmatrix}
1 \\
\frac{\lambda(w^{(j)})-\alpha^{2}-w}{1+\alpha}
\end{pmatrix}\lambda(w^{(j)})^{2N-x-\epsilon_{x}} \\
& \times e_{j}^{T} \frac{y(w^{(j)})\mathcal{R}^{X}(w,z)}{w^{2N-y}z^{y+1}}e_{k} \lambda(z^{(k)})^{x} \begin{pmatrix}
\frac{(1+\alpha^{3})z}{(1-\alpha)(\lambda(z^{(k)})-\alpha^{2}-z)} & \frac{1}{1-\alpha}
\end{pmatrix}\frac{1}{y(w^{(j)})y(z^{(k)})}.
\end{align*}
Therefore, using also the $\mathcal{R}^{X} \mapsto \mathcal{R}^{\mathcal{M}}$ transformation given by \eqref{def of reproducing kernel M}, we obtain
\begin{align*}
& \big[ K(2x+\epsilon_{x},2y+j,2x+\epsilon_{x},2y+i) \big]_{i,j=0}^{1} = \frac{1}{(2\pi i)^{2}}\int_{\gamma_{\mathcal{M}}}\int_{\gamma_{\mathcal{M}}} \begin{pmatrix}
\alpha^{2} & \alpha \\
w & 1
\end{pmatrix}^{\epsilon_{x}} 
\begin{pmatrix}
1 \\ \frac{\lambda(w)-\alpha^{2}-w}{1+\alpha}
\end{pmatrix} \\ 
&   \lambda(w)^{2N-x-\epsilon_{x}} \frac{\mathcal{R}^{\mathcal{M}}(w,z)}{w^{2N-y}z^{y+1}} \lambda(z)^{x} \begin{pmatrix}
\frac{(1+\alpha^{3})z}{(1-\alpha)(\lambda(z)-\alpha^{2}-z)} & \frac{1}{1-\alpha}
\end{pmatrix}\begin{pmatrix}
1 & 1 \\
\alpha z & 1
\end{pmatrix}^{\epsilon_{x}}\frac{dzdw}{y(w)y(z)},
\end{align*}
where $\gamma_{\mathcal{M}}$ is a closed contour surrounding once $0^{(1)}$ and $0^{(2)}$ on $\mathcal{M}$ in the positive direction. By performing the change of variables $w = w(\omega)$ and $z = z(\zeta)$ as in \eqref{z in terms of zeta lol}--\eqref{w in terms of omega lol}, using the factorization \eqref{z factorized in terms of zeta} and \eqref{lambda in terms of zeta}, the identity \eqref{dz y in terms of zeta}, and also the $\mathcal{R}^{\mathcal{M}} \mapsto \mathcal{R}^{U}$ transformation given by \eqref{def of Reproducing kernel C}, we get
\begin{align}
& \big[ K(2x+\epsilon_{x},2y+j,2x+\epsilon_{x},2y+i) \big]_{i,j=0}^{1} = \frac{1}{(2\pi i)^{2}}\int_{\gamma_{\mathbb{C}}}\int_{\gamma_{\mathbb{C}}} \begin{pmatrix}
\alpha^{2} & \alpha \\
w & 1
\end{pmatrix}^{\epsilon_{x}} 
\begin{pmatrix}
1 \\ \frac{\lambda(w)-\alpha^{2}-w}{1+\alpha}
\end{pmatrix} \label{lol10} \\ 
& \bigg( \frac{(\omega - \alpha c)(\omega - \alpha c^{-1})}{\omega(\omega-c)(\omega - c^{-1})} \bigg)^{2N} \mathcal{R}^{U}(\omega,\zeta) \frac{\omega^{N}w^{y}\lambda(z)^{x}}{\zeta^{N+1}z^{y+1}\lambda(w)^{x+\epsilon_{x}}} \begin{pmatrix}
\frac{(1+\alpha^{3})z}{(1-\alpha)(\lambda(z)-\alpha^{2}-z)} & \frac{1}{1-\alpha}
\end{pmatrix}\begin{pmatrix}
1 & 1 \\
\alpha z & 1
\end{pmatrix}^{\epsilon_{x}}d\zeta d\omega, \nonumber
\end{align}
where $\gamma_{\mathbb{C}}$ is a closed curve surrounding $c$ and $c^{-1}$ once in the positive direction, such that it does not surround $0$. On the other hand, using again \eqref{z factorized in terms of zeta} and \eqref{lambda in terms of zeta}, we have
\begin{align*}
\frac{w^{y}\lambda(z)^{x}}{z^{y}\lambda(w)^{x}} = \frac{(\omega-c)^{y}(\omega-c^{-1})^{y}}{(\zeta-c)^{y}(\zeta-c^{-1})^{y}} \frac{(\zeta - \alpha c)^{x} (\zeta - \alpha c^{-1})^{x}}{(\omega - \alpha c)^{x} (\omega - \alpha c^{-1})^{x}} \frac{\omega^{x-y}}{\zeta^{x-y}}.
\end{align*}
By using the definition \eqref{def of W in intro} of $W$, we can rewrite \eqref{lol10} as
\begin{align}
& \big[ K(2x+\epsilon_{x},2y+j,2x+\epsilon_{x},2y+i) \big]_{i,j=0}^{1} = \frac{1}{(2\pi i)^{2}}\int_{\gamma_{\mathbb{C}}}\int_{\gamma_{\mathbb{C}}} H_{K}(\omega,\zeta;\epsilon_{x}) \label{new formula for the kernel in the thm 2} \\ 
&  W(\omega) \mathcal{R}^{U}(\omega,\zeta) \frac{\omega^{N+x-y}}{\zeta^{N+x-y}} \frac{(\omega-c)^{y}(\omega-c^{-1})^{y}}{(\zeta-c)^{y}(\zeta-c^{-1})^{y}} \frac{(\zeta - \alpha c)^{x} (\zeta - \alpha c^{-1})^{x}}{(\omega - \alpha c)^{x} (\omega - \alpha c^{-1})^{x}}d\zeta d\omega, \nonumber
\end{align}
where $H_{K}(\omega,\zeta;\epsilon_{x})$ is defined for $\omega, \zeta \in \mathbb{C}$ and $\epsilon_{x} \in \{0,1\}$ by
\begin{align}\label{lol9}
H_{K}(\omega,\zeta;\epsilon_{x}) = & \; \frac{1}{\zeta z \lambda(w)^{\epsilon_{x}}} \begin{pmatrix}
\alpha^{2} & \alpha \\
w & 1
\end{pmatrix}^{\epsilon_{x}} \begin{pmatrix}
1 \\ \frac{1-\alpha + \alpha^{2}}{1-\alpha} \frac{\omega - c}{\omega}
\end{pmatrix}  \begin{pmatrix}
\frac{\alpha}{c(1-\alpha)^{2}}(\zeta - c^{-1}) & \frac{1}{1-\alpha}
\end{pmatrix}\begin{pmatrix}
1 & 1 \\
\alpha z & 1
\end{pmatrix}^{\epsilon_{x}}. 
\end{align}
Using the identities
\begin{align*}
& z = \frac{\alpha (\zeta - c)(\zeta - c^{-1})}{c(1-\alpha)^{2}\zeta}, & & w = \frac{\alpha (\omega - c)(\omega - c^{-1})}{c(1-\alpha)^{2}\omega}, & & \lambda(w) = \frac{\alpha (\omega - \alpha c)(\omega - \alpha c^{-1})}{c(1-\alpha)^{2}\omega},
\end{align*}
it is a simple computation to verify that \eqref{lol9} can be rewritten as \eqref{HK epsx 0}--\eqref{HK epsx 1}.  This finishes the proof.

\section{Lozenge probabilities}\label{section: lozenge probabilities}
This section is about the lozenge probabilities $P_{j}(x,y)$, $j=1,2,3$, defined in \eqref{def of P1 P2 P3}. In Subsection \ref{subsection: double contour formula for Pj}, we use Theorem \ref{thm: correlation kernel final scalar expression} to find double contour formulas for $P_{j}(x,y)$, $j=1,2,3$, in terms of $\mathcal{R}^{U}$. In the rest of this section, we follow \cite[Section 7]{CDKL} and use the symmetries in our model to restrict our attention to the lower left part $\eta \leq \frac{\xi}{2} \leq 0$ of the liquid region for the proof of Theorem \ref{thm:main}.
% These formulas satisfy some symmetries that are studied in Subsection \ref{subsection: symmetries}. In Subsection \ref{subsection: preliminaries to asymp}, we show that these symmetries allow us . The ideas of Subsections \ref{subsection: symmetries} and \ref{subsection: preliminaries to asymp} are inspired by those of \cite[Section 7]{CDKL}.
\subsection{Double contour formulas}\label{subsection: double contour formula for Pj}
Formula \eqref{new formula for the kernel in the thm} for the kernel can be rewritten as
\begin{multline}\label{main formula for the kernel written compactly}
\big[ K(2x+\epsilon_{x},2y+j,2x+\epsilon_{x},2y+i) \big]_{i,j=0}^{1} \\
 = \frac{1}{(2\pi i)^{2}}\int_{\gamma_{\mathbb{C}}}\int_{\gamma_{\mathbb{C}}} H_{K}(\omega,\zeta;\epsilon_{x}) W(\omega) \mathcal{R}^{U}(\omega,\zeta) \frac{\omega^{N}}{\zeta^{N}}q(\omega,\zeta)^{y} \tilde{q}(\omega,\zeta)^{x} d\zeta d\omega,
\end{multline}
where $q$ and $\tilde{q}$ are given by
\begin{align}\label{def of q and qtilde}
&  q(\omega,\zeta) := \frac{\zeta (\omega-c)(\omega-c^{-1})}{\omega(\zeta-c)(\zeta-c^{-1})}, \qquad \tilde{q}(\omega,\zeta) = \frac{\omega(\zeta - \alpha c)(\zeta - \alpha c^{-1})}{\zeta(\omega - \alpha c)(\omega - \alpha c^{-1})}.
\end{align}
The double contour formulas for $P_{j}(x,y)$, $j=1,2,3$, are obtained via a series of lemmas. Let us first recall that the paths $\mathfrak{p}_j:\{0,1,\ldots,4N\} \to \mathbb{Z}+\tfrac{1}{2}$, $j=0,\ldots,2N-1$ are defined in \eqref{def of the paths} via \eqref{non-intersecting paths}. We define the height function $h:\{0,1,\ldots,4N\}\times \mathbb Z \to \mathbb{N}_{\geq 0}$ by
\begin{align}\label{def of h}
h(x,y)= \#\{ j \mid  \mathfrak{p}_j(x)< y \}.
\end{align}
Lemma \ref{lem:height_to_lozenge} below is identical to \cite[Lemma 7.2]{CDKL} and allows to recover the lozenges from the height function.
%The graph of $h$ is a stepped surface and the paths can be thought of as level curves of this random surface.  We can recover the tiling from the height function by using simple identities which relate the positions of the different lozenges to differences of the height function. 
\begin{lemma} \label{lem:height_to_lozenge} 
For $x \in \{0,1,\ldots, 4N\}$ and $y \in \mathbb Z$, the following identities hold:
\begin{align}
h(x,y+1)-h(x+1,y+1) &= 
\begin{cases}
1, & \text{ there is a lozenge 	\tikz[scale=.3,baseline=(current bounding box.center)] { \draw (0,0) \lozr; \filldraw circle(5pt);  \node[below] (i) at (0,-.2) {\tiny{$(x,y)$}};} }\\[-10pt]
0, & \text{ otherwise.}
\end{cases} \label{diff loz 1} \\
h(x+1,y+1)-h(x,y) &= 
\begin{cases}
1, & \text{ there is a lozenge 	\tikz[scale=.3,baseline=(current bounding box.center)] { \draw (0,0) \lozu; \filldraw circle(5pt);  \node[below] (i) at (0,-.2) {\tiny{$(x,y)$}};}  }\\[-5pt]
0, & \text{ otherwise.}
\end{cases} \label{diff loz 2} \\
h(x,y+1)-h(x,y) & = 
\begin{cases}
0, & \text{ there is a lozenge 	\tikz[scale=.3,baseline=(current bounding box.center)] { \draw (0,0) \lozd; \filldraw (1,0) circle(5pt);  \node[below] (i) at (1,-.2) {\tiny{$(x,y)$}};} }\\[-5pt]
1, & \text{ otherwise.}
\end{cases} \label{diff loz 3}
\end{align}
\end{lemma}
\begin{proof}
This is an immediate consequence of \eqref{non-intersecting paths} and \eqref{def of h}. 
\end{proof}
The next lemma establishes a double integral formula for the expectation value of the height function. 

\begin{lemma}\label{lem:doubleintegralheight}
For $x \in \{1,2,\ldots, 2N-1\}$, $y \in \mathbb Z$ and $\epsilon_{x},\epsilon_{y} \in \{0,1\}$, we have
\begin{multline}\label{formula for expectation h in lemma}
\mathbb{E}[h(2x+\epsilon_{x},2y+\epsilon_{y})] = \frac{1}{(2\pi i)^2} \int_{\tilde{\gamma}_{\mathbb{C}}}d\zeta\int_{\gamma_{\mathbb{C}}} \frac{d\omega}{{q(\omega,\zeta)-1}} \mathcal{R}^{U}(\omega,\zeta)  W(\omega) \\   
\times \frac{\omega^{N}}{\zeta^{N}}q(\omega,\zeta)^{y}\tilde{q}(\omega,\zeta)^{x} \big(q(\omega,\zeta)^{\epsilon_{y}}H_{K}(\omega,\zeta;\epsilon_{x})_{11}+H_{K}(\omega,\zeta;\epsilon_{x})_{22} \big). 
\end{multline}
where $\gamma_{\mathbb{C}}$ is a closed curve surrounding both $c$ and $c^{-1}$, but not surrounding $0$, and $\tilde{\gamma}_{\mathbb{C}}$ is a deformation of $\gamma_{\mathbb{C}}$ lying in the bounded region delimited by $\gamma_{\mathbb{C}}$, such that $|q(\zeta,\omega)|>1$ whenever $\zeta \in \tilde{\gamma}_{\mathbb{C}}$ and $\omega \in \gamma_{\mathbb{C}}$.
\end{lemma}
\begin{proof}
Let $\mathcal{X}(\tilde{x},\tilde{y})$ be the random variable that counts the number of paths going through the point $(\tilde{x},\tilde{y})$, $\tilde{x},\tilde{y} \in \{0,1,...,4N\}$. Since $\mathcal{X}(\tilde{x},\tilde{y}) \in \{0,1\}$, we have $\mathbb{P}(\mathcal{X}(\tilde{x},\tilde{y}) = 1) = \mathbb{E}(\mathcal{X}(\tilde{x},\tilde{y}))$. Also, note that the identity \eqref{def of determinantal} with $k=1$ is equivalent to $\mathbb{P}(\mathcal{X}(\tilde{x},\tilde{y}) = 1) = K(\tilde{x},\tilde{y},\tilde{x},\tilde{y})$. Thus, by definition \eqref{def of h} of $h$, we get
\begin{align}
\mathbb{E}[h(2x+\epsilon_{x},2y)] & = \sum_{k<y} \big[ K(2x+\epsilon_{x},2k,2x+\epsilon_{x},2k) + K(2x + \epsilon_{x},2k + 1,2x + \epsilon_{x},2k + 1) \big] \nonumber \\
& = \sum_{k<y} \mbox{Tr} \big[ K(2x+\epsilon_{x},2k+j,2x+\epsilon_{x},2k+i) \big]_{i,j=0}^{1}. \label{lol13}
\end{align}
%for all $x \in \{1,2,\ldots, 2N-1\}$, $y \in \mathbb Z$ and $\epsilon_{x},\epsilon_{y} \in \{0,1\}$. 
%Let us define $f$ by
%\begin{align}\label{def of f in proof lol f}
%\zeta \mapsto f(\zeta) = \left| \frac{\zeta}{(\zeta-c)(\zeta-c^{-1})} \right|,
%\end{align}
Let us define $\tilde{\gamma}_{\mathbb{C}} := C(c,r) \cup C(c^{-1},r)$, where $C(a,r)$ denotes a circle oriented positively centered at $a$ of radius $r$. We see from \eqref{def of q and qtilde} that $|q(\omega,\zeta)| \to + \infty$ as $\zeta$ tends to $c$ or $c^{-1}$. Thus, by choosing $r$ sufficiently small, we can make sure that $\tilde{\gamma}_{\mathbb{C}}$ lies in the interior region of $\gamma_{\mathbb{C}}$, and that 
\begin{align*}
\left| q(\omega,\zeta) \right|>1+\epsilon, \qquad \mbox{for all } \zeta \in \tilde{\gamma}_{\mathbb{C}} \mbox{ and } \omega \in \gamma_{\mathbb{C}},
\end{align*}
for a certain $\epsilon > 0$. Therefore, uniformly for $\zeta \in \tilde{\gamma}_{\mathbb{C}}$ and $\omega \in \gamma_{\mathbb{C}}$, we have
\begin{align}\label{lol14}
& \sum_{k<y} q(\omega,\zeta)^{k} = \frac{q(\omega,\zeta)^{y}}{q(\omega,\zeta)-1}.
\end{align}
The statement \eqref{formula for expectation h in lemma} with $\epsilon_{y} = 0$ follows after combining \eqref{main formula for the kernel written compactly}, \eqref{lol13} and \eqref{lol14}. Then, \eqref{formula for expectation h in lemma} with $\epsilon_{y} = 1$ follows from
\begin{align*}
\mathbb{E}[h(2x+\epsilon_{x},2y+1)] = \mathbb{E}[h(2x+\epsilon_{x},2y)] + K(2x+\epsilon_{x},2y,2x+\epsilon_{x},2y).
\end{align*}
\end{proof}
The double contour formulas for $P_{j}$, $j=1,2,3$ are stated in the following proposition.
\begin{proposition}
For $x \in \{1,2,\ldots, 2N-1\}$ and $y \in \mathbb Z$, we have
\begin{align}
& P_{1}(x,y) = \frac{1}{(2\pi i)^{2}}\int_{\gamma_{\mathbb{C}}}\int_{\gamma_{\mathbb{C}}} H_{1}(\omega,\zeta) W(\omega) \mathcal{R}^{U}(\omega,\zeta) \frac{\omega^{N}}{\zeta^{N}}q(\omega,\zeta)^{y} \tilde{q}(\omega,\zeta)^{x} d\zeta d\omega, \label{P1 double contour} \\
& P_{2}(x,y) = \frac{1}{(2\pi i)^{2}}\int_{\gamma_{\mathbb{C}}}\int_{\gamma_{\mathbb{C}}} H_{2}(\omega,\zeta) W(\omega) \mathcal{R}^{U}(\omega,\zeta) \frac{\omega^{N}}{\zeta^{N}}q(\omega,\zeta)^{y} \tilde{q}(\omega,\zeta)^{x} d\zeta d\omega, \label{P2 double contour} \\
& P_{3}(x,y) = \begin{pmatrix}
1 & 1 \\ 1 & 1
\end{pmatrix} - \frac{1}{(2\pi i)^{2}}\int_{\gamma_{\mathbb{C}}}\int_{\gamma_{\mathbb{C}}} H_{3}(\omega,\zeta) W(\omega) \mathcal{R}^{U}(\omega,\zeta) \frac{\omega^{N}}{\zeta^{N}}q(\omega,\zeta)^{y} \tilde{q}(\omega,\zeta)^{x} d\zeta d\omega, \label{P3 double contour}
\end{align}
where $H_{1}$, $H_{2}$ and $H_{3}$ are given by
\begin{align}
& H_{1}(\omega,\zeta) = \begin{pmatrix} \frac{\alpha c (\omega -c) (\omega-c^{-1})}{(\zeta-c)(\zeta-c^{-1})\omega (\omega - \alpha c)} & \frac{(\zeta-\alpha c)(\omega -c)(\omega-c^{-1})}{(\zeta-c)(\zeta-c^{-1})(\omega-\alpha c)(\omega - \alpha c^{-1})}  \\
\frac{\omega -c}{(\zeta-c)(\omega - \alpha c)} & \frac{\alpha (\zeta - \alpha c)(\omega-c)}{c \zeta (\zeta -c)(\omega-\alpha c)(\omega-\alpha c^{-1})}
\end{pmatrix} \label{def of H1} \\
& H_{2}(\omega,\zeta) = \begin{pmatrix}
\frac{c(1-\alpha)(\omega-c)}{\alpha(\zeta-c)(\zeta-c^{-1})(\omega-\alpha c)} & \frac{(1-\alpha)(\zeta-\alpha c)(\omega-c)}{c(\zeta-c)(\zeta-c^{-1})(\omega-\alpha c)(\omega-\alpha c^{-1})} \\
\frac{(1-\alpha)c}{(\zeta-c)(\omega - \alpha c)} & \frac{\alpha c (1-\alpha) (\zeta- \alpha c)\omega}{\zeta (\zeta-c)(\omega - \alpha c)(\omega - \alpha c^{-1})}
\end{pmatrix}, \label{def of H2} \\
& H_{3}(\omega,\zeta) = \begin{pmatrix}
\frac{\omega -c}{c \omega (\zeta -c)(\zeta -c^{-1})} & \frac{(\zeta-\alpha c)(\omega-c)}{(\zeta -c)(\zeta -c^{-1})(\omega - \alpha c)} \\
\frac{1}{\zeta -c} & \frac{c(\zeta - \alpha c)}{\zeta (\zeta -c) (\omega - \alpha c)}
\end{pmatrix}. \label{def of H3}
\end{align}
\end{proposition}
\begin{proof}
Recall that $\mathcal{P}_{j}$, $j=1,2,3$ are defined by \eqref{def of mathcal P1 P2 P3}. By \eqref{def of determinantal}, for $\epsilon_{x},\epsilon_{y} \in \{0,1\}$, we have
\begin{align*}
& \mathcal{P}_{3}(2x+\epsilon_{x},2y+\epsilon_{y}) = 1 - K(2x+\epsilon_{x},2y+\epsilon_{y},2x+\epsilon_{x},2y+\epsilon_{y}).
\end{align*}
Noting that 
\begin{align*}
H_{3}(\omega,\zeta) = \begin{pmatrix}
H_{K}(\omega,\zeta;0)_{22} & H_{K}(\omega,\zeta;1)_{22} \\
H_{K}(\omega,\zeta;0)_{11} & H_{K}(\omega,\zeta;1)_{11}
\end{pmatrix} = \begin{pmatrix}
\frac{\omega -c}{c \omega (\zeta -c)(\zeta -c^{-1})} & \frac{(\zeta-\alpha c)(\omega-c)}{(\zeta -c)(\zeta -c^{-1})(\omega - \alpha c)} \\
\frac{1}{\zeta -c} & \frac{c(\zeta - \alpha c)}{\zeta (\zeta -c) (\omega - \alpha c)}
\end{pmatrix},
\end{align*}
formula \eqref{P3 double contour} follows by combining \eqref{def of P1 P2 P3} with \eqref{main formula for the kernel written compactly}. The proof of \eqref{P1 double contour} and \eqref{P2 double contour} requires more work and relies on Lemmas \ref{lem:height_to_lozenge} and \ref{lem:doubleintegralheight}. First, we note the following direct consequences of \eqref{formula for expectation h in lemma}:
\begin{align}
& \mathbb{E}[h(2x+\epsilon_{x},2y+1+\epsilon_{y})] = \frac{1}{(2\pi i)^2} \int_{\tilde{\gamma}_{\mathbb{C}}}d\zeta\int_{\gamma_{\mathbb{C}}} \frac{d\omega}{{q(\omega,\zeta)-1}} \mathcal{R}^{U}(\omega,\zeta)  W(\omega) \nonumber \\   
& \hspace{1.5cm} \times \frac{\omega^{N}}{\zeta^{N}}q(\omega,\zeta)^{y}\tilde{q}(\omega,\zeta)^{x} \big(q(\omega,\zeta)H_{K}(\omega,\zeta;\epsilon_{x})_{11}+q(\omega,\zeta)^{\epsilon_{y}}H_{K}(\omega,\zeta;\epsilon_{x})_{22} \big), \label{formula for expectation h number 2} \\
& \mathbb{E}[h(2x+1+\epsilon_{x},2y+1+\epsilon_{y})] = \frac{1}{(2\pi i)^2} \int_{\tilde{\gamma}_{\mathbb{C}}}d\zeta\int_{\gamma_{\mathbb{C}}} \frac{d\omega}{{q(\omega,\zeta)-1}} \mathcal{R}^{U}(\omega,\zeta)  W(\omega) \nonumber \\   
& \hspace{1cm} \times \frac{\omega^{N}}{\zeta^{N}}q(\omega,\zeta)^{y}\tilde{q}(\omega,\zeta)^{x+\epsilon_{x}} \big(q(\omega,\zeta)H_{K}(\omega,\zeta;1-\epsilon_{x})_{11}+q(\omega,\zeta)^{\epsilon_{y}}H_{K}(\omega,\zeta;1-\epsilon_{x})_{22} \big). \label{formula for expectation h number 3}
\end{align}
Using \eqref{diff loz 1}, \eqref{formula for expectation h number 2} and \eqref{formula for expectation h number 3}, we get
\begin{align}
& \mathcal{P}_{1}(2x+\epsilon_{x},2y+\epsilon_{y})  = \mathbb{E}[h(2x+\epsilon_{x},2y+1+\epsilon_{y})]-\mathbb{E}[h(2x+1+\epsilon_{x},2y+1+\epsilon_{y})] \nonumber \\
& = \frac{1}{(2\pi i)^2} \int_{\tilde{\gamma}_{\mathbb{C}}}d\zeta\int_{\gamma_{\mathbb{C}}} \frac{d\omega}{{q(\omega,\zeta)-1}} \mathcal{R}^{U}(\omega,\zeta)  W(\omega) \frac{\omega^{N}}{\zeta^{N}}q(\omega,\zeta)^{y} \tilde{q}(\omega,\zeta)^{x} \nonumber \\
& \times \Big( q(\omega,\zeta) H_{K}(\omega,\zeta;\epsilon_{x})_{11}- q(\omega,\zeta)\tilde{q}(\omega,\zeta)^{\epsilon_{x}}H_{K}(\omega,\zeta;1-\epsilon_{x})_{11} \nonumber \\
& + q(\omega,\zeta)^{\epsilon_{y}}H_{K}(\omega,\zeta;\epsilon_{x})_{22}- q(\omega,\zeta)^{\epsilon_{y}}\tilde{q}(\omega,\zeta)^{\epsilon_{x}}H_{K}(\omega,\zeta;1-\epsilon_{x})_{22}\Big). \label{lol15}
\end{align}
It is a direct computation to verify that the integrand has no pole at $\zeta = \omega$ for any $\epsilon_{x},\epsilon_{y} \in \{0,1\}$, so that $\tilde{\gamma}_{\mathbb{C}}$ can be deformed back to $\gamma_{\mathbb{C}}$. We obtain \eqref{P1 double contour} after writing \eqref{lol15} in the matrix form \eqref{def of P1 P2 P3}. Finally, using \eqref{diff loz 2}, \eqref{formula for expectation h in lemma} and \eqref{formula for expectation h number 3}, we get
\begin{align*}
& \mathcal{P}_{2}(2x+\epsilon_{x},2y+\epsilon_{y})  = \mathbb{E}[h(2x+1+\epsilon_{x},2y+1+\epsilon_{y})]-\mathbb{E}[h(2x+\epsilon_{x},2y+\epsilon_{y})] \\
& = \frac{1}{(2\pi i)^2} \int_{\tilde{\gamma}_{\mathbb{C}}}d\zeta\int_{\gamma_{\mathbb{C}}} \frac{d\omega}{{q(\omega,\zeta)-1}} \mathcal{R}^{U}(\omega,\zeta)  W(\omega) \frac{\omega^{N}}{\zeta^{N}}q(\omega,\zeta)^{y} \tilde{q}(\omega,\zeta)^{x} \\
& \times \Big( q(\omega,\zeta)\tilde{q}(\omega,\zeta)^{\epsilon_{x}} H_{K}(\omega,\zeta;1-\epsilon_{x})_{11}- q(\omega,\zeta)^{\epsilon_{y}}H_{K}(\omega,\zeta;\epsilon_{x})_{11} \\
& + q(\omega,\zeta)^{\epsilon_{y}}\tilde{q}(\omega,\zeta)^{\epsilon_{x}}H_{K}(\omega,\zeta;1-\epsilon_{x})_{22}-  H_{K}(\omega,\zeta;\epsilon_{x})_{22}\Big).
\end{align*}
Another direct computation shows that the integrand has no pole at $\zeta = \omega$ for any $\epsilon_{x},\epsilon_{y} \in \{0,1\}$, so that $\tilde{\gamma}_{\mathbb{C}}$ can be deformed back to $\gamma_{\mathbb{C}}$. The formula \eqref{P2 double contour} is then obtained by rewriting the above in the matrix form \eqref{def of P1 P2 P3}.
\end{proof}

\subsection{Symmetries}\label{subsection: symmetries}
Let $H(\omega,\zeta)$ be a $2 \times 2$ meromorphic function in both $\zeta$ and $\omega$, whose only possible poles in each variable are at $0$, $\alpha c$, $\alpha c^{-1}$, $c$ and $c^{-1}$. Furthermore, we assume that all the poles of $H$ are of order $1$ and that $H(\omega,\zeta)$ is bounded as $\zeta$ and/or $\omega$ tend to $\infty$. For $x \in \{1,2,\ldots, 2N-1\}$ and $y \in \mathbb Z$, we define
\begin{align}\label{mathcalI}
\mathcal{I}(x,y;H) = \frac{1}{(2\pi i)^{2}}\int_{\gamma_{\mathbb{C}}} \int_{\gamma_{\mathbb{C}}} H(\omega, \zeta) W(\omega) \mathcal{R}^{U}(\omega,\zeta) \frac{\omega^{N}}{\zeta^{N}}q(\omega,\zeta)^{y} \tilde{q}(\omega,\zeta)^{x}  d\zeta d\omega.
\end{align}
Since the poles of $H$ are of order at most $1$, recalling \eqref{def of W in intro}, the only poles of the integrand are at $0$, $c$ and $c^{-1}$, in both the $\zeta$ and $\omega$ variables.
The following star-operation will play an important role for a symmetry property of $\mathcal{I}$:
\begin{align}\label{star operation}
\zeta^{\star} = c^{-1} + \frac{R_{1}^{2}}{\zeta -c^{-1}}\qquad \mbox{where} \qquad R_{1} = \frac{1-\alpha}{\sqrt{\alpha}}.
\end{align}
Let $\gamma_{1}$ be the circle centered at $c^{-1}$ of radius $R_{1}$. The star-operation maps $\gamma_{1}$ into itself, but reverses the orientation. Furthermore, it satisfies $(\zeta^{\star})^{\star} = \zeta$ for all $\zeta \in \mathbb{C}\cup \{\infty\}$. We start by proving some symmetries for $\mathcal{R}^{U}$.
\begin{lemma}The reproducing kernel $\mathcal{R}^{U}$ satisfies two symmetries.
\begin{enumerate} 
\item[\rm (a)] We have 
\begin{align}\label{first sym for RU}
\mathcal{R}^{U}(\omega,\zeta) = \mathcal{R}^{U}(\zeta,\omega), \qquad \omega,\zeta \in \mathbb{C}.
\end{align}
\item[\rm (b)] We have
\begin{equation}\label{eq:RNsymmetry} 
\mathcal{R}^{U} \left( \omega^{\star},\zeta^{\star} \right) = \frac{R_{1}^{4N-2}\mathcal{R}^{U}(\omega,\zeta)}{(\omega - c^{-1})^{2N-1} (\zeta - c^{-1})^{2N-1}} , \qquad \omega,\zeta \in \mathbb{C} \setminus \{c^{-1}\}. 
\end{equation}
\end{enumerate} 
\end{lemma}
\begin{proof}
Since $\det U \equiv 1$, it follows from \eqref{reproducing kernel in terms of U intro} that
\begin{align}\label{RU first column}
\mathcal{R}^{U}(\omega,\zeta) = \frac{U_{11}(\omega)U_{21}(\zeta)-U_{11}(\zeta)U_{21}(\omega)}{\zeta-\omega},
\end{align}
from which we deduce \eqref{first sym for RU}. Now we prove (b).
Note that the first column of $U$ only contains polynomials, which are independent of the choice of the contour $\gamma_{\mathbb{C}}$ that appears in the formulation of the RH problem for $U$. Therefore, $\mathcal{R}^{U}$ is independent of the choice of $\gamma_{\mathbb{C}}$ as well by \eqref{RU first column}. Since $\gamma_{1}$ encloses both $c$ and $c^{-1}$, and does not enclose $0$, $\gamma_{1}$ is a valid choice of contour. We use the freedom we have in the choice of $\gamma_{\mathbb{C}}$ by letting $U$ be the solution to the RH problem for $U$ associated to the contour $\gamma_{1}$.  We can verify by direct computations that 
\begin{align}\label{W symmetry}
W(\zeta^{\star})  = \frac{(\omega - c^{-1})^{4N}}{R_{1}^{4N}}W(\zeta),
\end{align}
so that
\begin{equation*}
\widehat{U}(\zeta) := \begin{pmatrix} R_{1}^{2N} & 0 \\ 0 & -R_{1}^{-2N} \end{pmatrix} U(\tfrac{1}{c})^{-1} U\left( \zeta^{\star} \right)
\begin{pmatrix} 
\frac{(\zeta - c^{-1})^{2N}}{R_{1}^{2N}} & 0 \\ 0 & - \frac{ R_{1}^{2N}}{(\zeta - c^{-1})^{2N}} 
\end{pmatrix}
\end{equation*}
also satisfies the conditions of the RH problem for $U$. By uniqueness of the solution of this RH problem, we infer that $U(\zeta) = \widehat{U}(\zeta)$. After replacing $(\omega,\zeta)$ by $(\omega^{\star},\zeta^{\star})$ in \eqref{reproducing kernel in terms of U intro} and using the relations $\widehat{U}(\zeta) = U(\zeta)$ and $\frac{\zeta-\omega}{\zeta^{\star}-\omega^{\star}}=-\frac{(\zeta-c^{-1})(\omega-c^{-1})}{R_{1}^{2}}$, we obtain \eqref{eq:RNsymmetry}.
\end{proof}
\begin{proposition} \label{prop:symmetries} 
The double integral $\mathcal{I}(x,y;H)$ satisfies two symmetries.
\begin{enumerate} 
\item[\rm (a)] The following $(x,y)\mapsto (2N-x,2N-y)$ symmetry holds
\begin{align} \label{eq:Isymmetry1}
\mathcal I(2N-x,2N-y; H) & = \mathcal I(x,y; \widehat{H}), 
\end{align}
with
\begin{equation} \label{eq:hatH} 
\widehat{H}(\omega,\zeta) = H(\zeta,\omega).
\end{equation}
\item[\rm (b)] The following $(x,y)\mapsto (x,N+x-y)$ symmetry holds
\begin{align}	\label{eq:Isymmetry2}
\mathcal I(x,N+x-y; H) & = \mathcal I(x,y; \widetilde{H})
\end{align}
with
\begin{equation} \label{eq:tildeH}
\widetilde{H}(\omega,\zeta) = \frac{R_{1}^{2}H\left(\omega^{\star},\zeta^{\star}\right)}{(\omega-c^{-1})(\zeta-c^{-1})}.
\end{equation}
\end{enumerate}	
\end{proposition}

\begin{proof} (a) From \eqref{def of q and qtilde}, we verify that
\begin{align}\label{lol16}
\frac{\omega^{N}}{\zeta^{N}}q(\omega,\zeta)^{2N-y}\tilde{q}(\omega,\zeta)^{2N-x} = \frac{\zeta^{N}}{\omega^{N}} \frac{W(\zeta)}{W(\omega)}q(\zeta,\omega)^{y}\tilde{q}(\zeta,\omega)^{x}.
\end{align}
Replacing $(x,y)$ in \eqref{mathcalI} by $(2N-x,2N-y)$, and then using \eqref{lol16}, we get
\begin{align}\label{mathcalI sym1}
\mathcal{I}(2N-x,2N-y;H) = \frac{1}{(2\pi i)^{2}}\int_{\gamma_{\mathbb{C}}} \int_{\gamma_{\mathbb{C}}} W(\zeta) \mathcal{R}^{U}(\omega,\zeta) \frac{\zeta^{N}}{\omega^{N}}q(\zeta,\omega)^{y} \tilde{q}(\zeta,\omega)^{x} H(\omega, \zeta) d\zeta d\omega.
\end{align}
Recalling \eqref{first sym for RU}, the identity \eqref{eq:Isymmetry1} follows after interchanging variables in \eqref{mathcalI sym1}.	
	
(b) Note that $\gamma_{1}$ encloses both $c$ and $c^{-1}$, and does not enclose $0$, so we can (and do) deform $\gamma_{\mathbb{C}}$ to $\gamma_{1}$ in \eqref{mathcalI}. We first replace $(x,y)$ by $(x,N+x-y)$ in \eqref{mathcalI}, and then perform the change of variables $\zeta \mapsto \zeta^{\star}$ and $\omega \mapsto \omega^{\star}$. This gives
\begin{align*}
\mathcal{I}(x,N+x-y) = \frac{1}{(2\pi i)^{2}}\int_{\gamma_{1}} \int_{\gamma_{1}} W(\omega^{\star}) \mathcal{R}^{U}(\omega^{\star},\zeta^{\star}) \frac{\omega^{\star \, N}}{\zeta^{\star \, N}} q(\omega^{\star},\zeta^{\star})^{N+x-y}\tilde{q}(\omega^{\star},\zeta)^{x} H(\omega^{\star},\zeta^{\star}) d\zeta^{\star}d\omega^{\star}.
\end{align*}
It is a long but direct computation to verify that
\begin{align*}
%& W(\omega^{\star})  = \frac{(\omega - c^{-1})^{4N}}{R_{1}^{4N}}W(\omega), & & \mathcal{R}^{U}(\omega^{\star},\zeta^{\star}) = \frac{R_{1}^{4N-2}\mathcal{R}^{U}(\omega,\zeta)}{(\omega - c^{-1})^{2N-1} (\zeta - c^{-1})^{2N-1}} \\
& q(\omega^{\star},\zeta^{\star})^{N+x-y} = q(\omega,\zeta)^{y-x-N}, & & \tilde{q}(\omega^{\star},\zeta^{\star})^{x} = q(\omega,\zeta)^{x} \tilde{q}(\omega,\zeta)^{x}, \\
& \frac{\omega^{\star \, N}}{\zeta^{\star \, N}} = \frac{(\omega -c)^{N}(\zeta-c^{-1})^{N}}{(\omega -c^{-1})^{N}(\zeta-c)^{N}}, & & H(\omega^{\star},\zeta^{\star})d\zeta^{\star}d\omega^{\star} = \frac{H(\omega^{\star},\zeta^{\star})R_{1}^{4} d\zeta d \omega}{(\zeta-c^{-1})^{2}(\omega-c^{-1})^{2}}.
\end{align*}
Recalling also \eqref{eq:RNsymmetry} and \eqref{W symmetry}, \eqref{eq:Isymmetry2} follows by deforming back $\gamma_{1}$ to the original contour $\gamma_{\mathbb{C}}$ (in each variable). 
\end{proof}
We recall that $s(\xi,\eta;\alpha)$ is defined for $(\xi,\eta) \in \mathcal L_{\alpha}$ as the unique solution of  \eqref{eq:saddlepointeq} lying in the upper half-plane, and that $\mathcal{Q}$ is defined by \eqref{def of Q in statement of results}. These quantities will appear naturally in the analysis of the next sections. For now, we simply note the following symmetries for $s(\xi,\eta;\alpha)$. 
\begin{proposition} \label{prop:saddlesymmetry}
Let $(\xi,\eta) \in \mathcal L_{\alpha}$. Then also $(-\xi,-\eta) \in \mathcal L_{\alpha}$,
$(\xi,\xi-\eta) \in \mathcal L_{\alpha}$ and
\begin{align} \label{eq:saddlesymmetry1}
s(-\xi,-\eta;\alpha) & = s(\xi,\eta;\alpha) \\
s(\xi,\xi-\eta;\alpha) & =  \label{eq:saddlesymmetry2}
\left( \overline{s(\xi,\eta;\alpha)} \right)^{\star},
\end{align}	
where $\star$ denotes the star-operation defined in \eqref{star operation}.
\end{proposition}
\begin{proof}
The symmetry \eqref{eq:saddlesymmetry1} is part of Proposition \ref{prop:hightemp} and has already been proved in Section \ref{section: easy proofs}. It remains to prove \eqref{eq:saddlesymmetry2}. We define the function $f$ as follows
\begin{align*}
f(\zeta;\xi,\eta) = -\frac{\xi-\eta}{2}\frac{1}{\zeta} + \frac{\xi}{2}\left( \frac{1}{\zeta-\alpha c} + \frac{1}{\zeta - \alpha c^{-1}} \right) - \frac{\eta}{2}\left( \frac{1}{\zeta -c} + \frac{1}{\zeta - c^{-1}} \right),
\end{align*}
so that \eqref{eq:saddlepointeq} can be rewritten as
\begin{align}\label{eq determining the saddles}
f(\zeta;\xi,\eta)^{2} = \mathcal{Q}(\zeta).
\end{align}
Note that both $f$ and $\mathcal{Q}$ depend on $\alpha$, even though this is not indicated in the notation. It is a long but direct computation to verify that
\begin{align}\label{sym for Q and f}
\frac{R_{1}^{4}}{(\zeta-c^{-1})^{4}}\mathcal{Q}( \zeta^{\star}) = \mathcal{Q}(\zeta), \qquad \mbox{ and } \qquad - \frac{R_{1}^{2}}{(\zeta - c^{-1})^2} f ( \zeta^{\star}; \xi,\eta ) = f(\zeta;\xi,\xi-\eta).
\end{align}
By definition of $s(\xi,\eta;\alpha)$, we have $f(s(\xi,\eta;\alpha);\xi,\eta)^{2} = \mathcal{Q}(s(\xi,\eta;\alpha))$, so the symmetry \eqref{sym for Q and f} implies that
\begin{align}\label{lol17}
f(s(\xi,\eta;\alpha)^{\star};\xi,\xi-\eta)^{2} = \mathcal{Q}(s(\xi,\eta;\alpha)^{\star}).
\end{align}
Since the star operation maps the upper half-plane to the lower half-plane, $s(\xi,\eta;\alpha)^{\star}$ lies in the lower half-plane. Therefore, applying the conjugate operation in \eqref{lol17}, and noting that $\overline{f(\zeta)} = f(\overline{\zeta})$ and $\overline{\mathcal{Q}(\zeta)} = \mathcal{Q}(\overline{\zeta})$, we infer that $(\xi,\xi-\eta) \in \mathcal L_{\alpha}$ if and only if $(\xi,\eta) \in \mathcal L_{\alpha}$, and that \eqref{eq:saddlesymmetry2} holds.
\end{proof}

\subsection{Preliminaries to the asymptotic analysis}\label{subsection: preliminaries to asymp}

%Theorem \ref{thm:main} will follow from Theorem \ref{thm:doubleintegrals_for_lozenge_densities} 
%and the following result.

\begin{proposition} \label{prop:doubleintegrallimit}
Let $\{(x_{N},y_{N}\}_{N \geq 1}$ be a sequence satisfying \eqref{good sequence} with $(\xi,\eta) \in \mathcal{L}_{\alpha}$, such that $\eta \leq \frac{\xi}{2}\leq 0$. If $(\xi,\eta)$ lies on the boundary of $\{\eta \leq \frac{\xi}{2}\leq 0\}$, then we assume furthermore that $\frac{y_{N}}{N}-1 \leq \frac{1}{2}(\frac{x_{N}}{N}-1) \leq 0$ for all sufficiently large $N$. Let $(\omega,\zeta) \mapsto H(\omega,\zeta)$ be a $2 \times 2$ meromorphic function in both $\zeta$ and $\omega$, whose only possible poles in each variable are at $0$, $\alpha c$, $\alpha c^{-1}$, $c$ and $c^{-1}$. Furthermore we assume that all the poles of $H$ are of order $1$ and that $H(\omega,\zeta)$ is bounded as $\zeta$ and/or $\omega$ tend to $\infty$. Then $\mathcal I(x_{N},y_{N}; H)$ defined in \eqref{mathcalI} has the limit
\begin{equation} \label{eq:Hintegral} 
\lim_{N \to \infty} \mathcal I(x_{N},y_{N};H) = \frac{1}{2\pi i} \int_{\overline{s}}^s H(\zeta,\zeta) d\zeta 
\end{equation}
where $s = s(\xi,\eta;\alpha)$, and the integration path is from $\overline{s}$ to $s$ and lies in $\mathbb C \setminus (-\infty,c^{-1}]$.
\end{proposition}
The proof of Proposition \ref{prop:doubleintegrallimit} will be given in Section \ref{section: saddle point analysis}, after considerable preparations have been carried out in Sections \ref{section: g-function}-\ref{section: phase functions}.

\vspace{0.2cm}Proposition \ref{prop:doubleintegrallimit} only covers the lower left quadrant $\eta \leq \frac{\xi}{2}\leq 0$ of the liquid region. The next lemma shows that this is sufficient.
\begin{lemma}\label{lemma:from one quadrant to the four}
Assume Proposition \ref{prop:doubleintegrallimit} holds true. Then the statement of Proposition \ref{prop:doubleintegrallimit} still holds without the assumption that $\eta \leq \frac{\xi}{2}\leq 0$, and without the assumption that $\frac{y_{N}}{N}-1 \leq \frac{1}{2}(\frac{x_{N}}{N}-1) \leq 0$ for all sufficiently large $N$. That is, it holds for all $(\xi,\eta) \in \mathcal{L}_{\alpha}$.
\end{lemma}
\begin{proof}
If $\{(x_{N},y_{N})\}_{N \geq 1}$ is a sequence satisfying \eqref{good sequence} with $(\xi,\eta) \in \mathcal{L}_{\alpha} \cap \{\eta > \frac{\xi}{2}> 0\}$, then $\{(2N-x_{N},2N-y_{N})\}_{N \geq 1}$ satisfies \eqref{good sequence} with $(-\xi,-\eta)$ lying in the lower left quandrant of $\mathcal{L}_{\alpha}$. Therefore, Proposition \ref{prop:doubleintegrallimit} applies to the sequence $\{(2N-x_{N},2N-y_{N})\}_{N \geq 1}$, and we rely on the symmetries \eqref{eq:Isymmetry1} and \eqref{eq:saddlesymmetry1} to conclude
\begin{multline}\label{lol20}
\lim_{N \to \infty} \mathcal I(x_{N},y_{N};H) = \lim_{N \to \infty} \mathcal I(2N-x_{N},2N-y_{N};\widehat{H}) \\ = \frac{1}{2\pi i} \int_{\overline{s(-\xi,-\eta;\alpha)}}^{s(-\xi,-\eta;\alpha)} \widehat{H}(\zeta,\zeta) d\zeta = \frac{1}{2\pi i} \int_{\overline{s(\xi,\eta;\alpha)}}^{s(\xi,\eta;\alpha)} H(\zeta,\zeta) d\zeta,
\end{multline}
where we have also used \eqref{eq:hatH} for the last equality. Now, if $\{(x_{N},y_{N})\}_{N \geq 1}$ is a sequence satisfying \eqref{good sequence} with $(\xi,\eta) \in \mathcal{L}_{\alpha} \cap \{\eta > \frac{\xi}{2} < 0\}$, then $\{(x_{N},N+x_{N}-y_{N})\}_{N \geq 1}$ satisfies \eqref{good sequence} with $(\xi,\xi-\eta)$ lying in the lower left quandrant of $\mathcal{L}_{\alpha}$, so that Proposition \ref{prop:doubleintegrallimit} applies. Using the symmetries \eqref{eq:Isymmetry2} and \eqref{eq:saddlesymmetry2}, we arrive at
\begin{multline}\label{lol21}
\lim_{N \to \infty} \mathcal I(x_{N},y_{N};H) = \lim_{N \to \infty} \mathcal I(x_{N},N+x_{N}-y_{N};\widetilde{H}) = \frac{1}{2\pi i} \int_{\overline{s(\xi,\xi-\eta;\alpha)}}^{s(\xi,\xi-\eta;\alpha)} \widetilde{H}(\zeta,\zeta) d\zeta \\ = \frac{1}{2\pi i} \int_{s(\xi,\eta;\alpha)^{*}}^{\overline{s(\xi,\eta;\alpha)}^{*}} \frac{R_{1}^{2}H\left(\zeta^{\star},\zeta^{\star}\right)}{(\zeta-c^{-1})^{2}}d\zeta = \frac{1}{2\pi i} \int_{\overline{s(\xi,\eta;\alpha)}}^{s(\xi,\eta;\alpha)} H(\zeta,\zeta) d\zeta,
\end{multline}
where, for the last equality, we have applied the change of variables $\zeta \to \zeta^{*}$ stated in \eqref{star operation}. The claim for the last quadrant $(\xi,\eta) \in \mathcal{L}_{\alpha} \cap \{\eta < \frac{\xi}{2}> 0\}$ follows by combining \eqref{lol20} with \eqref{lol21}. Finally, if $\{(x_{N},y_{N})\}_{N \geq 1}$ satisfies \eqref{good sequence} with $\eta = \frac{\xi}{2}$ and/or $\xi=0$, then we define a new sequence $\{(\tilde{x}_{N},\tilde{y}_{N})\}_{N \geq 1}$ as follows. For each $N$, $(\tilde{x}_{N},\tilde{y}_{N})$ is equal to
\begin{align}\label{different possibilities}
(x_{N},y_{N}), \quad (2N-x_{N},2N-y_{N}), \quad (x_{N},N+x_{N}-y_{N}) \quad \mbox{ or } \quad (2N-x_{N},N-x_{N}+y_{N}),
\end{align}
in such a way that $\frac{\tilde{y}_{N}}{N}-1 \leq \frac{1}{2}(\frac{\tilde{x}_{N}}{N}-1) \leq 0$. There are four natural subsequences of $\{(\tilde{x}_{N},\tilde{y}_{N})\}_{N\geq 1}$, corresponding to the four sets of indices
\begin{align*}
& A_{1} = \{N:(\tilde{x}_{N},\tilde{y}_{N}) = (x_{N},y_{N})\}, & & A_{2}= \{N:(\tilde{x}_{N},\tilde{y}_{N}) = (2N-x_{N},2N-y_{N})\}, \\
& A_{3} = \{N: (\tilde{x}_{N},\tilde{y}_{N}) = (x_{N},N+x_{N}-y_{N})\}, & & A_{4}= \{N:(\tilde{x}_{N},\tilde{y}_{N}) = (2N-x_{N},N-x_{N}+y_{N})\}. 
\end{align*}
If any of the four subsequences $\{(\tilde{x}_{N},\tilde{y}_{N})\}_{N\geq 1, N \in A_{j}}$, $j=1,2,3,4$ contains infinitely many elements, Proposition \ref{prop:doubleintegrallimit} applies and by \eqref{eq:Hintegral}, \eqref{lol20} and \eqref{lol21}, we have
\begin{align}\label{subsequence limit}
\lim_{\substack{N \to \infty \\ N \in A_{j}}} \mathcal I(x_{N},y_{N};H) = \lim_{\substack{N \to \infty \\ N \in A_{j}}} \mathcal I(\tilde{x}_{N},\tilde{y}_{N};\widetilde{H}_{j}) = \frac{1}{2\pi i} \int_{\overline{s(\xi,\eta;\alpha)}}^{s(\xi,\eta;\alpha)} H(\zeta,\zeta) d\zeta,
\end{align}
where $H_{j}(\omega,\zeta)$, $j=1,2,3,4$ are equal to $H(\omega,\zeta),\widehat{H}(\omega,\zeta),\widetilde{H}(\omega,\zeta)$ and $\widetilde{H}(\zeta,\omega)$, respectively. Since the right-hand-side of \eqref{subsequence limit} is independent of $j$, this shows that 
\begin{align*}
\lim_{N \to \infty} \mathcal I(x_{N},y_{N};H) =\frac{1}{2\pi i} \int_{\overline{s(\xi,\eta;\alpha)}}^{s(\xi,\eta;\alpha)} H(\zeta,\zeta) d\zeta,
\end{align*}
which finishes the proof.
\end{proof}
\begin{proposition}
Proposition \ref{prop:doubleintegrallimit} implies Theorem \ref{thm:main}.
\end{proposition}
\begin{proof}
By \eqref{P1 double contour}--\eqref{P3 double contour} and \eqref{mathcalI}, for $x \in \{1,2,\ldots, 2N-1\}$ and $y \in \mathbb Z$, we can write
\begin{align*}
& \mathcal{P}_{j}(x,y) = \mathcal{I}(x,y;H_{j}), \qquad j=1,2,  \\
& \mathcal{P}_{3}(x,y) = \begin{pmatrix}
1 & 1 \\ 1 & 1
\end{pmatrix} - \mathcal{I}(x,y;H_{3}),
\end{align*}
where the functions $H_{j}$, $j=1,2,3$, are defined in \eqref{def of H1}--\eqref{def of H3}. Let $\{(x_{N},y_{N}\}_{N \geq 1}$ be a sequence satisfying \eqref{good sequence} with $(\xi,\eta) \in \mathcal{L}_{\alpha}$. By Lemma \ref{lemma:from one quadrant to the four}, we do not need to assume $\eta \leq \frac{\xi}{2}\leq 0$ to invoke Proposition \ref{prop:doubleintegrallimit}. Applying Proposition \ref{prop:doubleintegrallimit} with $H = H_{3}$, we obtain
\begin{align}\label{lol19}
\lim_{N \to \infty} P_{3}(x_{N},y_{N}) = \begin{pmatrix}
1 & 1 \\ 1 & 1
\end{pmatrix} - \frac{1}{2\pi i} \int_{\overline{s}}^s H_{3}(\zeta,\zeta) d\zeta.
\end{align}
From \eqref{def of H3}, we see that 
\begin{align}\label{lol18}
H_{3}(\zeta,\zeta) = \begin{pmatrix}
\frac{1}{\zeta - c^{-1}} - \frac{1}{\zeta} & \frac{1}{\zeta - c^{-1}} \\
\frac{1}{\zeta -c} & \frac{1}{\zeta - c} - \frac{1}{\zeta}
\end{pmatrix},
\end{align}
and since the path going from $\overline{s}$ to $s$ does not cross $(-\infty,c^{-1}]$, we get \eqref{P3 limit main result} after substituting \eqref{lol18} in \eqref{lol19} and carrying out the integration. Similarly, using \eqref{def of H1}--\eqref{def of H2}, we have
\begin{align*}
& H_{1}(\zeta,\zeta) = \begin{pmatrix}
\frac{1}{\zeta - \alpha c} - \frac{1}{\zeta} & \frac{1}{\zeta - \alpha c^{-1}} \\
\frac{1}{\zeta - \alpha c} & \frac{1}{\zeta - \alpha c^{-1}} - \frac{1}{\zeta}
\end{pmatrix} \quad \mbox{ and } \quad H_{2}(\zeta,\zeta) = \begin{pmatrix}
\frac{1}{\zeta - c^{-1}}-\frac{1}{\zeta - \alpha c} & \frac{1}{\zeta - c^{-1}} - \frac{1}{\zeta - \alpha c^{-1}} \\
\frac{1}{\zeta -c} - \frac{1}{\zeta -\alpha c} & \frac{1}{\zeta -c} - \frac{1}{\zeta - \alpha c^{-1}}
\end{pmatrix},
\end{align*}
and we obtain \eqref{P1 limit main result}--\eqref{P2 limit main result} after applying Proposition \ref{prop:doubleintegrallimit} with $H = H_{1}$ and $H=H_{2}$, respectively.

\end{proof}

\section{$g$-function}\label{section: g-function}
In Section \ref{section: steepest descent for $U$}, we will perform a Deift/Zhou \cite{DZ} steepest descent analysis on the RH problem for $U$. The first transformation $U \mapsto T$ consists of normalizing the RH problem and requires considerable preparation. This transformation uses a so-called $g$-function \cite{Deift}, which is of the form
\begin{align}\label{def of g}
g(\zeta) = \int_{\supp(\mu)}\log(\zeta - \xi) d\mu(\xi),
\end{align}
where $\mu$ is a probability measure, $d\mu$ is its density, and $\supp \mu$ is its (bounded and oriented) support. For any choice of $\mu$, the $g$-function satisfies
\begin{align*}
& g(\zeta) = \log(\zeta) + \bigO(\zeta^{-1}), & & \mbox{as } \zeta \to \infty,
\end{align*} 
so that $U(\zeta)e^{-2Ng(\zeta)\sigma_{3}}$ is normalized at $\infty$ (with $\sigma_{3} = \diag(1,-1)$), in the sense that $U(\zeta)e^{-2Ng(\zeta)\sigma_{3}} = I_{2} + \bigO(\zeta^{-1})$ as $\zeta \to \infty$. Also, we note that in the definition of $U$, the contour $\gamma_{\mathbb{C}}$ can be chosen arbitrarily, as long as it is a closed curve surrounding $c^{-1}$ and $c$ once in the positive direction, which does not surround $0$. However, in order to successfully perform an asymptotic analysis on the RH problem for $U$, we need to choose $\mu$ and $\gamma_{\mathbb{C}}$ appropriately so that the jumps for $T$ have ``good properties".

\vspace{0.2cm}In this section, we find the key ingredients for the $Y \mapsto T$ transformation of Section \ref{section: steepest descent for $U$}, that is, we find a $g$-function (built in terms of $\mu$) and a relevant contour $\gamma_{\mathbb{C}}$. Let us rewrite $W$ as follows
\begin{equation}
W(\zeta) = \bigg( \frac{(\zeta-\alpha c)(\zeta-\alpha c^{-1})}{\zeta (\zeta-c)(\zeta-c^{-1})} \bigg)^{2N} = e^{-2NV(\zeta)},
\end{equation}
where the potential $V$ is given by
\begin{equation}\label{def potential}
V(\zeta) = \log \zeta + \log (\zeta - c) + \log (\zeta - c^{-1}) - \log (\zeta - \alpha c) - \log (\zeta - \alpha c^{-1})
\end{equation}
and we take the principal branch for the logarithms. We require $g$ and $\gamma_{\mathbb{C}}$ to satisfy the following criteria (we define $\supp (\mu)$ as an oriented open set for convenience):
\begin{enumerate}[label=(\alph*),leftmargin=0.6cm]
\item \label{item a} $\gamma_{\mathbb{C}}$ is a closed curve surrounding $c^{-1}$ and $c$ once in the positive direction, but not surrounding $0$.
\item \label{item b} $e^{g}$ is analytic in $\mathbb{C}\setminus \overline{\supp (\mu)}$, where $\supp (\mu)$ is an open oriented curve satisfying $\supp (\mu) \subset \gamma_{\mathbb{C}}$.
\item \label{item c} For $\zeta \in \supp(\mu)$, we let $g_{+}(\zeta)$ (resp. $g_{-}(\zeta)$) denote the limit of $g(\zeta')$ as $\zeta' \to \zeta$ from the left (resp. right) of $\supp(\mu)$. Here ``left" and ``right" are with respect to the orientation of $\supp(\mu)$. The $g$-function \eqref{def of g} satisfies
\begin{align}
& g_{+}(\zeta) + g_{-}(\zeta) - V(\zeta) + \ell = 0, & & \mbox{for } \zeta \in \supp (\mu), \label{jump relation for g on the support} \\
& \re \big( g_{+}(\zeta) + g_{-}(\zeta) - V(\zeta) + \ell \big) < 0, & & \mbox{for } \zeta \in \gamma_{\mathbb{C}}\setminus \overline{\supp (\mu)}, \label{real part for g negative on gamma setminus S} \\
& \im \big( g_{+}(\zeta) - \tfrac{V(\zeta)}{2} + \tfrac{\ell}{2}\big), & & \mbox{is decreasing along }\supp (\mu), \label{real part for g positive in neighborhood of the support}
\end{align}
for some constant $\ell \in \mathbb{C}$, and where $V$ is given by \eqref{def potential}.
\end{enumerate}
In approximation theory, the equality \eqref{jump relation for g on the support} together with the inequality \eqref{real part for g negative on gamma setminus S} are usually refered to as the Euler-Lagrange variational conditions \cite{ST}, and $\ell$ is the Euler-Lagrange constant. A measure $\mu$ satisfying \eqref{jump relation for g on the support}--\eqref{real part for g negative on gamma setminus S} is called the equilibrium measure \cite{ST} in the external field $V$, because it is the unique minimizer of
\begin{align*}
\tilde{\mu} \mapsto \iint \log \frac{1}{|s-t|}d\tilde{\mu}(s)d\tilde{\mu}(t) + \re \int V(s) d\tilde{\mu}(s)
\end{align*}
among all probability measures $\tilde{\mu}$ supported on $\supp(\mu)$. Here we require in addition that \eqref{real part for g positive in neighborhood of the support} is satisfied. This extra-condition characterizes $\supp(\mu)$ as a so-called $S$-curve \cite{Stahl,GR,Rak,MFR,KS,MFR2}.

\subsection{Definition of $\mathcal{Q}$ and related computations}
By taking the derivative in \eqref{jump relation for g on the support}, we have
\begin{align}\label{lol5}
g_{+}'(\zeta) + g_{-}'(\zeta) - V'(\zeta) = 0, \qquad \zeta \in \supp(\mu),
\end{align}
and by condition (b), $g'$ is analytic in $\mathbb{C}\setminus \overline{\supp(\mu)}$. Therefore, the function
\begin{align}\label{def of Q}
\mathcal{Q}(\zeta) := \left( g^{\prime}(\zeta) - \frac{V^{\prime}(\zeta)}{2} \right)^{2}
\end{align}
is meromorphic on $\mathbb{C}$. By \eqref{def potential}, we get
\begin{align}\label{def of V'}
V'(\zeta) = \frac{1}{\zeta} + \frac{1}{\zeta-c^{-1}} + \frac{1}{\zeta-c} - \frac{1}{\zeta- \alpha c^{-1}} - \frac{1}{\zeta - \alpha c},
\end{align}
from which we conclude that $\mathcal{Q}$ has a double zero at $\infty$, and double poles at $0$, $\alpha c$, $\alpha c^{-1}$, $c$ and $c^{-1}$. Since a meromorphic function on the Riemann sphere (genus $0$) has as many poles as zeros, $\mathcal{Q}$ has eight other zeros. As $\zeta \to \infty$, we have $g'(\zeta) = \zeta^{-1} + \bigO(\zeta^{-2})$, from which we get $\mathcal{Q}(\zeta) = 2^{-2} \zeta^{-2} + \bigO(\zeta^{-3})$. Therefore, $\mathcal{Q}$ can be written in the form
\begin{equation}\label{def xi prime without the cut yet}
\mathcal{Q}(\zeta) =  \frac{\Pi(\zeta)}{4\zeta^{2} (\zeta-\alpha c)^{2} (\zeta-\alpha c^{-1})^{2} (\zeta-c)^{2} (\zeta-c^{-1})^{2}},
\end{equation}
where $\Pi$ is a monic polynomial of degree $8$ which remains to be determined. If we assume that $g'(\zeta)$ remains bounded for $\zeta \in \mathbb{C}$, then we can deduce from \eqref{def of Q} and \eqref{def of V'} the leading order term for $\mathcal{Q}(\zeta)$ as $\zeta \to \zeta_{\star} \in \{0,\alpha c, \alpha c^{-1}, c, c^{-1}\}$:
\begin{align}
& \mathcal{Q}(\zeta) = 2^{-2}\zeta^{-2} + \bigO(\zeta^{-1}), & & \mbox{ as } \zeta \to 0, \label{Q 0} \\
& \mathcal{Q}(\zeta) = 2^{-2}(\zeta - \alpha c)^{-2} + \bigO((\zeta - \alpha c)^{-1}), & & \mbox{ as } \zeta \to \alpha c, \label{Q alpha} \\
& \mathcal{Q}(\zeta) = 2^{-2}(\zeta - \alpha c^{-1})^{-2} + \bigO((\zeta - \alpha c^{-1})^{-1}), & & \mbox{ as } \zeta \to \alpha c^{-1}, \label{Q beta alpha} \\
& \mathcal{Q}(\zeta) = 2^{-2}(\zeta - c)^{-2} + \bigO((\zeta - c)^{-1}), & & \mbox{ as } \zeta \to c, \label{Q 0 2} \\
& \mathcal{Q}(\zeta) = 2^{-2}(\zeta - c^{-1})^{-2} + \bigO((\zeta - c^{-1})^{-1}), & & \mbox{ as } \zeta \to c^{-1}. \label{Q 0 1}
\end{align}
By combining these asymptotics with \eqref{def xi prime without the cut yet}, we get
\begin{align}
& \Pi(0) = \alpha^{4}, & & \Pi(\alpha c) = (1-\alpha)^{8}c^{8}, & & \Pi(\alpha c^{-1}) = \ds (1-\alpha)^{8}\alpha^{4}, \nonumber \\
& \Pi(c) = (1-\alpha)^{8}c^{8}, & & \Pi(c^{-1}) = (1-\alpha)^{8}\alpha^{-4}. \label{sqrt pi8}
\end{align}
This gives $5$ linear equations for the $8$ unknown coefficients of $\Pi$, which is not enough to determine $\Pi$ (and hence, $\mathcal{Q}$). Therefore, one needs to make a further assumption: we assume that we can find $\Pi$ in the form
\begin{equation}\label{assumption on Pi}
\Pi(\zeta) = (\zeta-r_{1})^{2}(\zeta-r_{2})^{2}(\zeta-r_{3})^{2}(\zeta-r_{+})(\zeta-r_{-}).
\end{equation}
As we will see, Assumption \eqref{assumption on Pi} implies that $\supp(\mu)$ consists of a single curve (``one-cut regime"). This assumption is justified if we can: 1) find $r_{1}$, $r_{2}$, $r_{3}$, $r_{+}$, $r_{-}$ so that \eqref{sqrt pi8} holds and 2) construct a $g$-function via \eqref{def of Q} which satisfies the properties \ref{item a}--\ref{item b}--\ref{item c}. 

\vspace{0.2cm}Substituting \eqref{assumption on Pi} in \eqref{sqrt pi8}, we obtain $5$ \textit{non-linear} equations for the $5$ unknowns $r_{1}$, $r_{2}$, $r_{3}$, $r_{+}$, $r_{-}$. This system turns out to have quite a few solutions -- we need to select ``the correct one". Let us define $r_{1}$, $r_{2}$, $r_{3}$, $r_{+}$, $r_{-}$ by \eqref{def of r1 r2 r3}--\eqref{def of r+ r-}. It is a simple computation to verify that indeed \eqref{sqrt pi8} holds in this case. We will show in Subsection \ref{subsection: mu and g function} that this definition of $r_{1}$, $r_{2}$, $r_{3}$, $r_{+}$, $r_{-}$ is ``the correct solution" to \eqref{sqrt pi8}, in the sense that it allows to construct a $g$-function satisfying the properties \ref{item a}--\ref{item b}--\ref{item c}.

\begin{remark}
Let us briefly comment on how to find \eqref{def of r1 r2 r3}--\eqref{def of r+ r-}. Unfortunately, we were not able to solve analytically the non-linear system obtained after substituting \eqref{assumption on Pi} into \eqref{sqrt pi8}. Instead, we have solved numerically (using the Newton--Raphson method) this system for a large number of values of $\alpha \in (0,1)$. As already mentioned, the system \eqref{sqrt pi8} possesses several solutions. In order to ensure numerical convergence to ``the correct solution", we choose starting values of $r_{1}$, $r_{2}$ and $r_{3}$ so that \eqref{ordering of the zeros and the poles} holds. The expressions \eqref{def of r1 r2 r3}--\eqref{def of r+ r-} have then been guessed by an inspection of the plots of $r_{1}(\alpha)$, $r_{2}(\alpha)$, $r_{3}(\alpha)$, $r_{+}(\alpha)$, $r_{-}(\alpha)$. 
\end{remark}
 %Then we obtain a function $g$ satisfying \eqref{jump relation for g on the support} by integration from a convenient starting point. This gives roughly a constructive approach to obtain $g$. However, the message here is that such an approach is not completely satisfactory -- there are many places in the derivation of $g$ where we need to guess some quantities (starting with $\mathcal{Q}$, see Section \ref{subsection: comments}). Therefore, we present this section as follows: we \textit{define} $g$ and all the associated quantities (in particular the function $\mathcal{Q}$) in Section \ref{subsection: Q}. The rest of the section consists of proving that the quantities define in Section \ref{subsection: Q} are such that the criteria (a), (b) and (c) above are indeed fulfilled. Finally, we comment more on some heuristic way of deriving $\mathcal{Q}$ and $\supp(\mu)$ in Section \ref{subsection: comments}.
 
%Now that $\mathcal{Q}$ is determined, one needs to define in a convenient way $\mathcal{Q}(\zeta)^{1/2}$, so that by \eqref{def of Q} we have
%\begin{align*}
%g'(\zeta) = \frac{V'(\zeta)}{2} + \mathcal{Q}(\zeta)^{1/2}.
%\end{align*}
%It turns out that the relevant branch cut for $\mathcal{Q}(\zeta)^{1/2}$ takes place on a critical trajectory of $\mathcal{Q}$.

\subsection{Critical trajectories of $\mathcal{Q}$}\label{subsection: critical trajectories of Q}
In this subsection, we study the critical trajectories of $\mathcal{Q}$, which are relevant to define the $g$-function and study its properties.

\vspace{0.2cm}Let $t \mapsto \zeta(t)$, $t \in [a,b]$ be a smooth parametrization of a curve $\sigma$, satisfying $\zeta'(t) \neq 0$ for all $t \in (a,b)$. $\sigma$ is a {\it trajectory} of the quadratic differential $\mathcal{Q}(\zeta) d\zeta^2$ if $\mathcal{Q}(\zeta(t)) \zeta'(t)^2 < 0$ for every $t \in (a,b)$, and an {\it orthogonal trajectory} if $\mathcal{Q}(\zeta(t)) \zeta'(t)^2 > 0$ for every $t \in (a,b)$. $\sigma$ is {\it critical} if it contains a zero or a pole of $\mathcal{Q}$. Note that these definitions are independent of the choice of the parametrization.

\medskip

Since $r_{+}$ and $r_{-}$ are simple zeros of $\mathcal{Q}$, there are three critical trajectories (and also three orthogonal critical trajectories) emanating from each of the points $r_{\pm}$. Recall the definitions of $\gamma_{0},\gamma_{\alpha},\gamma_{1},\Sigma_{0},\Sigma_{\alpha}$ and $\Sigma_{1}$ given in Subsection \ref{subsection: saddle points and liquid region}. %The next lemma states that these critical trajectories are $\Sigma_{1}$, $\Sigma_{\alpha}$ and $\Sigma_{0}$. 
\begin{figure}
\begin{center}
\begin{tikzpicture}
%\node at (0,0) {\includegraphics[width = 12cm]{../Image/Trajectories/traj_alpha_04.jpg}};
\node at (0,0) {};
\node at (2.1,0) {\color{black} \large $\bullet$};
\node at (-0.38,0) {\color{black} \large $\bullet$};
\node at (-1.05,0) {\color{black} \large $\bullet$};
\node at (-2.035,0) {\color{black} \large $\bullet$};
\node at (-3.13,0) {\color{black} \large $\bullet$};

\node at (0.4,0) {\color{red} \large $\bullet$};
\node at (-1.5,0) {\color{red} \large $\bullet$};
\node at (-5.53,0) {\color{red} \large $\bullet$};
\node at (-1.2,1.42) {\color{red} \large $\bullet$};
\node at (-1.2,-1.42) {\color{red} \large $\bullet$};

% circle gamma1
\draw[red,line width=0.65 mm] ([shift=(-156:3.6cm)]2.1,0) arc (-156:156:3.6cm);
\draw[red,dashed,line width=0.65 mm] ([shift=(155:3.6cm)]2.1,0) arc (155:205:3.6cm);

% circle gamma_0
\draw[red,line width=0.65 mm] ([shift=(-35:2.4cm)]-3.13,0) arc (-35:35:2.4cm);
\draw[red,dashed,line width=0.65 mm] ([shift=(35:2.4cm)]-3.13,0) arc (35:325:2.4cm);

% circle gamma alpha
\draw[red,line width=0.65 mm] ([shift=(95:1.45cm)]-1.05,0) arc (95:265:1.45cm);
\draw[red,dashed,line width=0.65 mm] ([shift=(-95:1.45cm)]-1.05,0) arc (-95:95:1.45cm);

\node at (-0.9,2.535) {$\Sigma_{1}$};
\node at (-0.5,0.535) {$\Sigma_{0}$};
\node at (-2.5,1.035) {$\Sigma_{\alpha}$};

\node at (0.5,1.035) {$+$};
\node at (-3.5,1.035) {$+$};
\node at (-1.8,0.5) {$-$};
\end{tikzpicture}
\end{center}
\caption{\label{fig: crit traj alpha 04}The critical trajectories of $\mathcal{Q}$ (solid red), and the critical orthogonal trajectories (dashed red) for $\alpha = 0.4$. The red dots are the zeros of $\mathcal{Q}$, and the black dots are the poles. The critical trajectories divide $\mathbb{C}$ in three regions. The sign of $\re \phi$ in each of these regions in shown by $+$ or $-$.}
\end{figure}
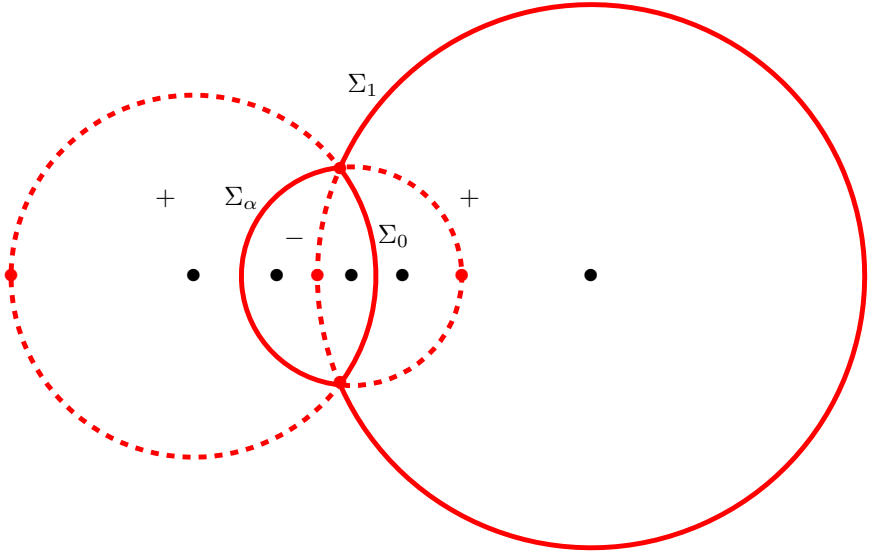
\begin{lemma}\label{lemma:Sigma0 and Sigma alpha}
The arcs $\overline{\Sigma_{0}}$, $\overline{\Sigma_{\alpha}}$ and $\overline{\Sigma_{1}}$ are three critical trajectories of $\mathcal{Q}(\zeta) d\zeta^2$ joining $r_{-}$ with $r_{+}$, and $\gamma_{0} \setminus \Sigma_{0}$, $\gamma_{\alpha} \setminus\Sigma_{\alpha}$ and $\gamma_{1} \setminus\Sigma_{1}$ are each the union of two critical orthogonal trajectories of $\mathcal{Q}(\zeta) d\zeta^2$. An illustration is shown in Figure \ref{fig: crit traj alpha 04}.
\end{lemma}
\begin{proof}
Let $t \mapsto \zeta = \zeta(t) = c^{-1} + R_{1}e^{it}$, $t \in [-\pi,\pi]$, be a parametrization of $\gamma_{1}$. Writing $r_{\pm} = c^{-1} + R_{1}e^{\pm i \theta_{1}}$ with $\theta_{1} \in (\frac{2\pi}{3},\pi)$, and noting that $\zeta' = iRe^{it}$, we have
\begin{align*}
(\zeta-r_{+})(\zeta-r_{-}) = 2 R_{1}^{2} e^{i t}(\cos t - \cos \theta_{1}) \quad \mbox{ and } \quad \frac{(\zeta')^{2}}{(\zeta-c^{-1})^{2}} = -1.
\end{align*}
Therefore, we get
\begin{align}
 \mathcal{Q}(\zeta)(\zeta')^{2} & \; = (\zeta')^{2} \frac{(\zeta-r_{1})^{2}(\zeta-r_{2})^{2}(\zeta-r_{3})^{2}(\zeta-r_{+})(\zeta-r_{-})}{4\zeta^{2} (\zeta-\alpha c)^{2} (\zeta-\alpha c^{-1})^{2} (\zeta-c)^{2} (\zeta-c^{-1})^{2}} \nonumber \\
& \; = - R_{1}^{2}e^{it} \frac{(\zeta-r_{1})^{2}(\zeta-r_{2})^{2}(\zeta-r_{3})^{2}(\cos t - \cos \theta_{1})}{2\zeta^{2} (\zeta-\alpha c)^{2} (\zeta-\alpha c^{-1})^{2} (\zeta-c)^{2}}. \label{Q on circle form 1}
\end{align}
Using \eqref{r2 r+ r- on the circle}, we show that $(\zeta - r_{2}) = 2R_{1}e^{\frac{it}{2}}\cos \frac{t}{2}$, and
\begin{align*}
(\zeta - r_{1})(\zeta - r_{3}) & \; = 2R_{1}^{2}e^{it}\left( \cos t + \frac{\alpha^{2} + (2-\alpha)\sqrt{1-\alpha + \alpha^{2}}}{2(1-\alpha)} \right), \\
\zeta (\zeta-c) & \; = 2R_{1}^{2}e^{it} \left( \cos t + \frac{2-3\alpha + 2\alpha^{2}}{2(1-\alpha)\sqrt{1-\alpha + \alpha^{2}}} \right), \\
(\zeta - \alpha c)(\zeta - \alpha c^{-1}) & \; = 2R_{1}^{2} e^{it} \left( \cos t + \frac{2-\alpha + \alpha^{2}}{2\sqrt{1-\alpha + \alpha^{2}}} \right).
\end{align*}
Substituting the above expressions in \eqref{Q on circle form 1}, we find
\begin{align}\label{Q on circle form 2}
\mathcal{Q}(\zeta)(\zeta')^{2} = \frac{(\cos \theta_{1} - \cos t)\cos^{2} \frac{t}{2} \left( \cos t + \frac{\alpha^{2} + (2-\alpha)\sqrt{1-\alpha + \alpha^{2}}}{2(1-\alpha)} \right)^{2}}{2\left( \cos t + \frac{2-3\alpha + 2\alpha^{2}}{2(1-\alpha)\sqrt{1-\alpha + \alpha^{2}}} \right)^{2}\left( \cos t + \frac{2-\alpha + \alpha^{2}}{2\sqrt{1-\alpha + \alpha^{2}}} \right)^{2}}.
\end{align}
We verify by direct computations that
\begin{align*}
\frac{\alpha^{2} + (2-\alpha)\sqrt{1-\alpha + \alpha^{2}}}{2(1-\alpha)} > \frac{2-3\alpha + 2\alpha^{2}}{2(1-\alpha)\sqrt{1-\alpha + \alpha^{2}}} > \frac{2-\alpha + \alpha^{2}}{2\sqrt{1-\alpha + \alpha^{2}}} > 1,
\end{align*}
and thus the right-hand-side of \eqref{Q on circle form 2} is negative for $t \in (-\theta_{1},\theta_{1})$, positive for $t \in (-\pi,-\theta_{1})\cup(\theta_{1},\pi)$ and zero for $t = -\pi,-\theta_{1},\theta_{1},\pi$. We conclude that $\overline{\Sigma_{1}}$ is a critical trajectory and that $\gamma_{1}\setminus \Sigma_{1}$ is the union of two orthogonal critical trajectories. The statement about $\overline{\Sigma_{\alpha}}$, $\gamma_{\alpha}\setminus \Sigma_{\alpha}$, $\overline{\Sigma_{0}}$, $\gamma_{0}\setminus \Sigma_{0}$ can be proved a similar way, and we provide less details. For $\zeta = \zeta(t) = R_{0}e^{it}$, $t \in [-\pi,\pi]$, after long but straightforward computations, we obtain
\begin{align}\label{Q on circle gamma_0}
\mathcal{Q}(\zeta)(\zeta')^{2} = \frac{(\cos \theta_{0} - \cos t)\cos^{2} \frac{t}{2} \left( \cos t - \frac{(1+\alpha)\sqrt{1-\alpha + \alpha^{2}} - (1-\alpha)^{2}}{2\alpha} \right)^{2}}{2\left( \cos t - \frac{1-\alpha + 2\alpha^{2}}{2\alpha\sqrt{1-\alpha + \alpha^{2}}} \right)^{2}\left( \cos t - \frac{2-\alpha + \alpha^{2}}{2\sqrt{1-\alpha + \alpha^{2}}} \right)^{2}}.
\end{align}
Since
\begin{align*}
\frac{1-\alpha + 2\alpha^{2}}{2\alpha\sqrt{1-\alpha + \alpha^{2}}} > \frac{(1+\alpha)\sqrt{1-\alpha + \alpha^{2}} - (1-\alpha)^{2}}{2\alpha} > \frac{2-\alpha + \alpha^{2}}{2\sqrt{1-\alpha + \alpha^{2}}} > 1,
\end{align*}
we infer that $\overline{\Sigma_{0}}$ is a critical trajectory and that $\gamma_{0}\setminus \Sigma_{0}$ is the union of two orthogonal critical trajectories. For $\zeta = \zeta(t) = \alpha c^{-1}+ R_{\alpha}e^{it}$, $t \in [-\pi,\pi]$, we obtain
\begin{align}\label{Q on circle gamma_a}
\mathcal{Q}(\zeta)(\zeta')^{2} = \frac{(\cos t - \cos \theta_{\alpha})\sin^{2} \frac{t}{2} \left( \cos t + \frac{1-(1-2\alpha)\sqrt{1-\alpha + \alpha^{2}}}{2\alpha (1-\alpha)} \right)^{2}}{2\left( \cos t - \frac{1-\alpha + 2\alpha^{2}}{2\alpha\sqrt{1-\alpha + \alpha^{2}}} \right)^{2}\left( \cos t + \frac{2-3\alpha + 2\alpha^{2}}{2(1-\alpha)\sqrt{1-\alpha + \alpha^{2}}} \right)^{2}}
\end{align}
with
\begin{align*}
\frac{1-(1-2\alpha)\sqrt{1-\alpha + \alpha^{2}}}{2\alpha (1-\alpha)}> \frac{2-3\alpha + 2\alpha^{2}}{2(1-\alpha)\sqrt{1-\alpha + \alpha^{2}}} > 1, \quad \frac{1-\alpha + 2\alpha^{2}}{2\alpha\sqrt{1-\alpha + \alpha^{2}}}>1.
\end{align*}
Therefore, we deduce from an inspection of \eqref{Q on circle gamma_a} that $\overline{\Sigma_{\alpha}}$ is a critical trajectory and that $\gamma_{\alpha}\setminus \Sigma_{\alpha}$ is the union of two orthogonal critical trajectories. This finishes the proof.
\end{proof}
\subsection{Branch cut structure and the zero set of $\re \phi$}
As can be seen in \eqref{def of Q}, $g'$ can be expressed as
\begin{align}\label{lol12}
g'(\zeta) = \frac{V'(\zeta)}{2} + \mathcal{Q}(\zeta)^{1/2},
\end{align}
for a certain branch of $\mathcal{Q}(\zeta)^{1/2}$. To obtain a $g$-function with the desired properties (a)-(b)-(c), it turns out that the branch cut needs to be taken along the critical trajectory $\Sigma_{1}$ (as in Subsection \ref{subsection: saddle points and liquid region}). 
\begin{definition}\label{def:Q^{1/2}}
We define $\mathcal{Q}^{1/2}$ as 
\begin{align}\label{def of Q^{1/2}}
\mathcal{Q}(\zeta)^{1/2} = \frac{(\zeta-r_{1})(\zeta-r_{2})(\zeta-r_{3})\sqrt{(\zeta-r_{+})(\zeta-r_{-})}}{2\zeta (\zeta-\alpha c) (\zeta-\alpha c^{-1}) (\zeta-c) (\zeta-c^{-1})},
\end{align}
where the branch cut for $\sqrt{(\zeta-r_{+})(\zeta-r_{-})}$ is taken on $\Sigma_{1}$ such that
\begin{align*}
\sqrt{(\zeta-r_{+})(\zeta-r_{-})} = \zeta + \bigO(1), \qquad \mbox{as } \zeta \to \infty.
\end{align*}
%We recall that $c$ is given by \eqref{def of W in intro}, and $r_{1},r_{2},r_{3},r_{+}$ and $r_{-}$ are given by \eqref{def of r1 r2 r3}-\eqref{def of r+ r-}.
\end{definition}
It will also be convenient to define a primitive of $\mathcal{Q}^{1/2}$.
\begin{definition}\label{def:phi}
We define $\phi:\mathbb{C}\setminus \big( (-\infty,c^{-1}] \cup \{c^{-1}+R_{1}e^{it}: -\pi \leq t \leq \theta_{1} \} \big) \to \mathbb{C}$ by
\begin{equation}\label{def of phi}
\phi(\zeta) = \int_{r_{+}}^{\zeta} \mathcal{Q}(\xi)^{1/2}d\xi,
\end{equation}
where the path of integration does not intersect $(-\infty,c^{-1}] \cup \{c^{-1}+R_{1}e^{it}: -\pi \leq t \leq \theta_{1} \}$.
\end{definition}
We first state some basic properties of $\phi$. By \eqref{Q 0}--\eqref{Q 0 1}, $\mathcal{Q}^{1/2}$ has simple poles at $0$, $\alpha c$, $\alpha c^{-1}$, $c$ and $c^{-1}$, and the residues are real. Also, since $\Sigma_{1}$ is a critical trajectory of $\mathcal{Q}$, we have $\phi_{\pm}(\zeta) \in i \mathbb{R}$ for $\zeta \in \Sigma_{1}$. Therefore, $\re \phi$ is single-valued and continuous in $\mathbb C \setminus \{0,\alpha c, \alpha c^{-1}, c, c^{-1}\}$, and $\re \phi$ is also harmonic in $\mathbb C \setminus (\Sigma_1 \cup \{0,\alpha c, \alpha c^{-1}, c, c^{-1}\})$. Finally, by combining Definition \ref{def:Q^{1/2}} with \eqref{Q 0}--\eqref{Q 0 1}, we have
\begin{equation}
\label{eq:Nphinearpoles} 
\begin{aligned} 
\phi(\zeta) & = - \frac{1}{2} \log \zeta + \bigO(1)  \text{ as } \zeta \to 0, & 
\lim_{\zeta \to 0} \re \phi(\zeta) = + \infty, \\
\phi(\zeta) & = \frac{1}{2} \log(\zeta-\alpha c) + \bigO(1) \text{ as } \zeta \to \alpha c, &
\lim_{\zeta \to \alpha c} \re \phi(\zeta) = - \infty, \\
\phi(\zeta) & = \frac{1}{2} \log(\zeta -\alpha c^{-1}) + \bigO(1) \text{ as } \zeta \to \alpha c^{-1}, &
\lim_{\zeta \to \alpha c^{-1}} \re \phi(\zeta) = - \infty, \\
\phi(\zeta) & = -\frac{1}{2} \log(\zeta -c) + \bigO(1) \text{ as } \zeta \to c, &
\lim_{\zeta \to c} \re \phi(\zeta) = + \infty, \\ 
\phi(\zeta) & = -\frac{1}{2} \log(\zeta -c^{-1}) + \bigO(1) \text{ as } \zeta \to c^{-1}, &
\lim_{\zeta \to c^{-1}} \re \phi(\zeta) = + \infty, \\ 
\phi(\zeta) & = \frac{1}{2} \log(\zeta) + \bigO(1) \text{ as } \zeta \to \infty, &
\lim_{\zeta \to \infty} \re \phi(\zeta) = + \infty. 
\end{aligned}
\end{equation}

In the rest of this subsection, we determine the zero set $\mathcal N_\phi$ of $\re \phi$. This will be useful in Subsection \ref{subsection: mu and g function} to establish the \ref{item a}-\ref{item b}-\ref{item c} properties of the $g$-function. Let us define
\begin{equation} \label{eq:Nphi}
\mathcal N_\phi = \{ z \in \mathbb{C}  :  \re \phi(z)=0\}.
\end{equation}
\begin{lemma}\label{lem:Nphihigh}
We have
\begin{align}\label{zero set of phi explicitly determined}
\mathcal{N}_{\phi} = \Sigma_{0} \cup \Sigma_{\alpha} \cup \Sigma_{1}.
\end{align}
In particular, $\mathcal{N}_{\phi}$ divides the complex plane in three regions. The sign of $\re \phi$ in these regions is as shown in Figure \ref{fig: crit traj alpha 04}.
\end{lemma}
\begin{proof}
By Lemma \ref{lemma:Sigma0 and Sigma alpha}, it holds that
\begin{align}\label{eq subseteq}
\mathcal{N}_{\phi} \supseteq \Sigma_{0} \cup \Sigma_{\alpha} \cup \Sigma_{1}.
\end{align}
We now prove the inclusion $\subseteq$. We first show that
\begin{align}\label{mathcalN on the real line}
\mathcal{N}_{\phi} \cap \mathbb{R} = (\Sigma_{0} \cup \Sigma_{\alpha} \cup \Sigma_{1})\cap \mathbb{R} = \{\alpha c^{-1} - R_{\alpha}, R_{0}, c^{-1} + R_{1} \}.
\end{align}
By Definitions \ref{def:Q^{1/2}} and \ref{def:phi}, $\phi' = Q_{\alpha}^{1/2}$ changes sign when it crosses each of the nine points $r_{1}$, $0$, $\alpha c$, $r_{2}$, $\alpha c^{-1}$, $c$, $r_{3}$, $c^{-1}$, $c^{-1}+R_{1}$. Since $\phi'(\zeta) = 2^{-1}\zeta^{-1} + \bigO(\zeta^{-2})$ as $\zeta \to \infty$, we have $\phi' > 0$ on the intervals
\begin{align*}
(r_{1},0), \quad (\alpha c,r_{2}), \quad (\alpha c^{-1},c), \quad (r_{3},c^{-1}), \quad (c^{-1}+R_{1},+\infty),
\end{align*}
and $\phi' < 0$ on the intervals
\begin{align*}
(-\infty,r_{1}), \quad (0,\alpha c), \quad (r_{2},\alpha c^{-1}), \quad (c,r_{3}), \quad (c^{-1},c^{-1}+R_{1}).
\end{align*}
By \eqref{eq subseteq}, we have
\begin{align*}
\re \phi(\alpha c^{-1} - R_{\alpha}) = \re \phi (R_{0}) = \re \phi(c^{-1}+R_{1}) = 0,
\end{align*}
so $\re \phi$ admits no other zeros on $(0,\alpha c) \cup (\alpha c^{-1},c) \cup (c^{-1},+\infty)$. On the intervals $(-\infty,0)$ and $(c,c^{-1})$, $\re \phi$ admits a local minimum at $r_{1}$ and $r_{3}$, respectively, and on the interval $(\alpha c, \alpha c^{-1})$, it admits a local maximum at $r_{2}$. Thus \eqref{mathcalN on the real line} holds true if we show that
\begin{align}\label{pos and neg real part at r1 r2 r3}
\re \phi(r_{1}) > 0, \qquad \re \phi(r_{2}) < 0, \quad \mbox{ and } \quad \re \phi(r_{3}) > 0.
\end{align}
By Lemma \ref{lemma:Sigma0 and Sigma alpha}, $\re \phi$ is strictly monotone on each of the curves $(\gamma_{0}\setminus \Sigma_{0}) \cap \mathbb{C}^{+}$, $(\gamma_{\alpha}\setminus \Sigma_{\alpha}) \cap \mathbb{C}^{+}$ and $(\gamma_{1}\setminus \Sigma_{1}) \cap \mathbb{C}^{+}$. The expressions \eqref{Q on circle form 2}, \eqref{Q on circle gamma_0} and \eqref{Q on circle gamma_a}, together with Definition \ref{def:Q^{1/2}}, allow to conclude that $\re \phi$ is strictly increasing on $(\gamma_{0}\setminus \Sigma_{0}) \cap \mathbb{C}^{+}$ oriented from $r_{+}$ to $r_{1}$, strictly increasing on $(\gamma_{\alpha}\setminus \Sigma_{\alpha}) \cap \mathbb{C}^{+}$ oriented from $r_{+}$ to $r_{3}$, and strictly decreasing on $(\gamma_{1}\setminus \Sigma_{1}) \cap \mathbb{C}^{+}$ oriented from $r_{+}$ to $r_{2}$. In particular this proves \eqref{pos and neg real part at r1 r2 r3}, and thus \eqref{mathcalN on the real line}.

\medskip Assume $\mathcal N_{\phi}$ if of the form $\Sigma_{0} \cup \Sigma_{\alpha} \cup \Sigma_{1} \cup \sigma$ for a certain curve $\sigma$ distinct from $\Sigma_{0}$, $\Sigma_{\alpha}$ and $\Sigma_{1}$. Since $\phi_{\pm}'(\zeta) \neq 0$ for $\zeta \in \Sigma_{1}$, we must have $\sigma \cap \Sigma_{1} = \emptyset$. Also, in view of \eqref{mathcalN on the real line}, $\sigma$ cannot intersect the real axis. Then $\sigma$ must be	a closed contour in $\mathbb{C}\setminus \big( \mathbb{R}\cup \Sigma_{1} \big)$, and the max/min principle for harmonic functions would then imply that $\re \phi$ in constant on the whole bounded region delimited by $\sigma$. By \eqref{def of phi}, $\re \phi$ is clearly not constant on such domain, so we arrive at a contradiction, and we conclude that $\mathcal{N}_{\phi} = \Sigma_{0} \cup \Sigma_{\alpha} \cup \Sigma_{1}$.

\medskip Thus $\mathcal{N}_{\phi}$ divides the complex plane in three regions in which $\re \phi$ does not change sign. The signs in each of these regions is then determined immediately by \eqref{eq:Nphinearpoles} (or equivalently, by \eqref{pos and neg real part at r1 r2 r3}).
\end{proof}
\subsection{Definition and properties of $g$}\label{subsection: mu and g function}
\begin{definition}\label{def:mu0g} 
We define the measure $\mu$ by
\begin{align} 	
d\mu(\zeta) & = \frac{1}{\pi i} \mathcal{Q}_{-}(\zeta)^{1/2} d\zeta \nonumber \\
& = \frac{1}{\pi i} \frac{(\zeta-r_{1})(\zeta-r_{2})(\zeta-r_{3})\sqrt{(\zeta-r_{+})(\zeta-r_{-})}_{-}}{2\zeta (\zeta-\alpha c) (\zeta-\alpha c^{-1}) (\zeta-c) (\zeta-c^{-1})} d\zeta, 
\qquad \zeta \in  \Sigma_1, \label{def of mu}
\end{align}
where $\Sigma_1 = \supp(\mu)$ is given by \eqref{def of Sigma 1}, and is oriented from $r_{-}$ to $r_{+}$; so $\mathcal{Q}_{-}(\zeta)^{1/2}$ denotes the limit of $\mathcal{Q}(\xi)^{1/2}$ as $\xi \to \zeta \in \Sigma_1$ with $\xi$ in the exterior of the circle $\gamma_1$. 
\end{definition}
\begin{proposition}
The measure $\mu$ defined in \eqref{def of mu} is a probability measure.
\end{proposition}
\begin{proof}
We compute $\int_{\Sigma_{1}} d\mu$ by residue calculation. Since $\mathcal{Q}_{+} = - \mathcal{Q}_{-}$, we have
\begin{align}
\int_{\Sigma_{1}} d\mu(\zeta) = \frac{1}{2\pi i}\int_{\mathcal{C}} \mathcal{Q}(\zeta)^{1/2}d\zeta,
\end{align}
where $\mathcal{C}$ is a closed curve surrounding $\Sigma_{1}$ once in the positive direction, but not surrounding any of the poles of $\mathcal{Q}$. By deforming $\mathcal{C}$ into another contour $\widetilde{\mathcal{C}}$ surrounding $0$, $\alpha c$, $\alpha c^{-1}$, $c$ and $c^{-1}$, we pick up some residues:
\begin{align}\label{int of mu with some residue}
\int_{\Sigma_{1}} d\mu(\zeta) = -\sum_{\zeta_{\star}\in \mathcal{P}}\mbox{Res} \left( \mathcal{Q}(\zeta)^{1/2}, \zeta = \zeta_{\star} \right) + \frac{1}{2\pi i}\int_{\widetilde{\mathcal{C}}} \mathcal{Q}(\zeta)^{1/2}d\zeta
\end{align}
where $\mathcal{P}=\{0,\alpha c,\alpha c^{-1},c,c^{-1}\}$. By combining Definition \ref{def:Q^{1/2}} with \eqref{Q 0}--\eqref{Q 0 1}, we have
\begin{align}
& \mbox{Res} \left( \mathcal{Q}(\zeta)^{1/2}, \zeta = 0 \right) = - \frac{1}{2}, & & \mbox{Res} \left( \mathcal{Q}(\zeta)^{1/2}, \zeta = \alpha c \right) =  \frac{1}{2}, \nonumber \\
& \mbox{Res} \left( \mathcal{Q}(\zeta)^{1/2}, \zeta = \alpha c^{-1} \right) =  \frac{1}{2}, & & \mbox{Res} \left( \mathcal{Q}(\zeta)^{1/2}, \zeta = c \right) = - \frac{1}{2}, \nonumber \\
& \mbox{Res} \left( \mathcal{Q}(\zeta)^{1/2}, \zeta = c^{-1} \right) = - \frac{1}{2}, \label{residues of Q^1/2}
\end{align}
and since $\mathcal{Q}(\zeta)^{1/2} = \frac{1}{2\zeta} + \bigO(\zeta^{-2})$ as $\zeta \to \infty$, we find 
\begin{align}\label{eq:mu0total}
\int_{\Sigma_{1}}d\mu(\zeta) = 1.
\end{align}
It remains to show that $\mu$ has a positive density on $\Sigma_{1}$. Let $\zeta(t) = c^{-1}+R_{1} e^{i t}$, 
$-\theta_{1} < t < \theta_{1}$, be a parametrization of
$\Sigma_1$. Consider the function
\begin{equation} \label{eq:mu0mass} 
t \mapsto \int_{r_-}^{\zeta(t)} d\mu = \frac{1}{\pi i} \int_{r_-}^{\zeta(t)} \mathcal{Q}_{-}(\xi)^{1/2} d\xi,
\end{equation}
whose derivative is given by
\begin{align}\label{lol11}
\frac{1}{\pi i} \mathcal{Q}_{-}(\zeta(t))^{1/2}  \zeta'(t).
\end{align}
Since $\mathcal{Q}(\zeta(t)) (\zeta'(t))^2 < 0$ for $t \in (-\theta_{1}, \theta_{1})$ by Lemma \ref{lemma:Sigma0 and Sigma alpha}, \eqref{lol11} is real and non-zero. Note also that the function \eqref{eq:mu0mass} vanishes for $t = - \theta_{1}$
and equals $1$ for $t= \theta_{1}$ by \eqref{eq:mu0total}.
Therefore \eqref{lol11} is strictly positive. 
\end{proof}
\begin{definition}\label{def: g}
The $g$-function is defined by
\begin{equation} \label{def of g function} 
g(\zeta) = \int_{\Sigma_1} \log(\zeta-\xi) d\mu(\xi), \qquad \zeta \in \C \setminus 
\left((-\infty, r_{2}] \cup \{c^{-1}+R_{1} e^{i t} : -\pi \leq t \leq \theta_{1} \} \right),
\end{equation}
where for each $\xi \in \Sigma_1$, the function $\zeta \mapsto \log(\zeta-\xi)$ has a branch cut along $(-\infty, r_{2}] \cup 
\{c^{-1}+R_{1} e^{i t} : -\pi \leq t \leq \arg \xi\}$ and behaves like $\log(\zeta-\xi) = \log |\zeta| + \bigO(\zeta^{-1}) $, as  $\zeta \to +\infty$.   

\medskip We define the variational constant $\ell \in \C$  by
\begin{align}\label{def of ell}
\ell = - 2g(r_{+}) + V(r_{+}).
\end{align}
\end{definition}
The next proposition shows, among other things, that Definition \ref{def: g} for $g$ is consistent with \eqref{lol12}, and that $g$ satisfies \eqref{jump relation for g on the support}.
\begin{proposition}\label{prop:phi and g relations}
The functions $g$ and $\phi$ are related by
\begin{align}\label{phi and g relation}
\phi(\zeta) = g(\zeta) - \frac{V(\zeta)}{2} + \frac{\ell}{2}, \qquad \zeta \in \mathbb{C}\setminus \big( (-\infty,r_{2}] \cup \{c^{-1}+R_{1} e^{i t} : -\pi \leq t \leq \theta_{1} \} \big)
\end{align}
and we have
\begin{align}
& g_{+}(\zeta) + g_{-}(\zeta) - V(\zeta) + \ell = 0, & & \mbox{for } \zeta \in \Sigma_{1}, \label{prop g add} \\
& g_{+}(\zeta)-g_{-}(\zeta) = 2\phi_{+}(\zeta) = - 2 \phi_{-}(\zeta), & & \mbox{for } \zeta \in \Sigma_{1}. \label{prop g sous}
\end{align}
Furthermore, the $g$-function satisfies the properties \ref{item a}--\ref{item b}--\ref{item c} listed at the beginning of Section \ref{section: g-function} with $\gamma_{\mathbb{C}} = \gamma_{1}$.
\end{proposition}
\begin{proof}
We first prove \eqref{phi and g relation}. For a fixed $\zeta \in \mathbb{C}\setminus \Sigma_{1}$, we have
\begin{align*}
& g'(\zeta) = \int_{\Sigma_{1}} \frac{d\mu(\xi)}{\zeta-\xi} = \frac{1}{2\pi i} \int_{\mathcal{C}} \frac{Q(\xi)^{1/2}}{\zeta - \xi}d\xi,
\end{align*}
where $\mathcal{C}$ is a closed curve surrounding $\Sigma_{1}$ once in the positive direction, but not surrounding any of the poles of $\mathcal{Q}$, and not surrounding $\zeta$. By deforming $\mathcal{C}$ into another contour $\widetilde{\mathcal{C}}$ surrounding $0$, $\alpha c$, $\alpha c^{-1}$, $c$, $c^{-1}$ and  $\zeta$, we obtain
\begin{align}\label{int of g' with some residue}
\int_{\Sigma_{1}} \frac{d\mu(\xi)}{\zeta - \xi} = -\sum_{\xi_{\star}\in \mathcal{P}}\mbox{Res} \left( \frac{\mathcal{Q}(\xi)^{1/2}}{\zeta- \xi}, \xi = \xi_{\star} \right) + \mathcal{Q}(\zeta)^{1/2} + \frac{1}{2\pi i}\int_{\widetilde{\mathcal{C}}} \frac{\mathcal{Q}(\xi)^{1/2}}{\zeta - \xi}d\xi,
\end{align}
where $\mathcal{P}=\{0,\alpha c,\alpha c^{-1},c,c^{-1}\}$. By deforming $\widetilde{\mathcal{C}}$ to $\infty$, noting that $\mathcal{Q}(\xi)^{1/2} = \bigO(\xi^{-1})$ as $\xi \to \infty$, the integral on the right-hand-side of \eqref{int of g' with some residue} is $0$. The sum can be evaluated using the residues \eqref{residues of Q^1/2}, and we get
\begin{align*}
\int_{\Sigma_{1}} \frac{d\mu(\xi)}{\zeta - \xi} = \frac{1}{2\zeta} - \frac{1}{2(\zeta - \alpha c)} - \frac{1}{2(\zeta - \alpha c^{-1})} + \frac{1}{2(\zeta -c)} + \frac{1}{2(\zeta - c^{-1})} + \mathcal{Q}(\zeta)^{1/2}.
\end{align*}
Using \eqref{def of V'} and $\phi' = \mathcal{Q}^{1/2}$, the above can be rewritten as
\begin{align*}
g'(\zeta) = \frac{V'(\zeta)}{2} + \phi'(\zeta), \qquad \zeta \in \mathbb{C}\setminus \Sigma_{1}.
\end{align*}
Integrating this identity from $r_{+}$ to $\zeta$ along a path that does not intersect $(-\infty,c^{-1}]\cup \{c^{-1}+R_{1}e^{i t}:-\pi \leq t < \theta_{1}\}$, we obtain
\begin{align*}
g(\zeta) - g(r_{+}) = \frac{V(\zeta)}{2} - \frac{V(r_{+})}{2} + \phi(\zeta),
\end{align*}
where we have used $\phi(r_{+}) = 0$. Then \eqref{phi and g relation} follows from the definition of $\ell$ given by \eqref{def of ell}. Since $\mathcal{Q}_{+}^{1/2} = -\mathcal{Q}_{-}^{1/2}$ on $\Sigma_{1}$, by \eqref{def of phi} we have
\begin{align*}
& \phi_{+}(\zeta) + \phi_{-}(\zeta) = 0, \qquad \mbox{for } \zeta \in \Sigma_{1},
\end{align*}
from which \eqref{prop g add} and \eqref{prop g sous} follow. The circle $\gamma_{1}$ encloses both $c$ and $c^{-1}$, and $0$ lies in the exterior of $\gamma_{1}$, so criterion \ref{item a} is fulfilled. For $\zeta \in (-\infty,r_{2}]\cup \{c^{-1}+R_{1}e^{it}:-\pi \leq -\theta_{1} \}$, we have
\begin{align*}
g_{+}(\zeta)-g_{-}(\zeta) = \int_{\Sigma_{1}} \big( \log_{+}(\zeta-\xi)-\log_{-}(\zeta-\xi) \big) d\mu(\xi) = 2 \pi i \int_{\Sigma_{1}} d \mu(\xi) = 2\pi i,
\end{align*}
so $e^{g}$ is analytic in $\mathbb{C}\setminus \overline{\Sigma_{1}}$ and criterion \ref{item b} is also fulfilled. For $\zeta \in \gamma_{1}\setminus \overline{\Sigma_{1}}$, by \eqref{phi and g relation} and Lemma \ref{lem:Nphihigh}, we have
\begin{align*}
\re \big( g_{+}(\zeta) + g_{-}(\zeta) - V(\zeta) + \ell \big) = \re \big( \phi_{+}(\zeta) + \phi_{-}(\zeta) \big) = 2 \, \re \phi(\zeta) < 0,
\end{align*}
as required in \eqref{real part for g negative on gamma setminus S}. Finally, by Definitions \ref{def:phi} and \ref{def:mu0g},  for $\zeta \in \Sigma_{1}$ we have
\begin{align*}
\im \Big( g_{+}(\zeta) - \frac{V(\zeta)}{2} + \frac{\ell}{2}\Big) = \im \phi_{+}(\zeta) = \im \int_{r_{+}}^{\zeta} \mathcal{Q}_{+}^{1/2}(\xi)d\xi = \pi \int_{\zeta}^{r_{+}}  d\mu(\xi)
\end{align*}
which is strictly decreasing as $\zeta$ goes from $r_{-}$ to $r_{+}$. So \eqref{real part for g positive in neighborhood of the support} holds as well, and hence \ref{item c}, which finishes the proof.
\end{proof}

\section{Steepest descent for $U$}\label{section: steepest descent for $U$}
In this section, we will perform an asymptotic analysis of the RH problem for $U$ as $N \to + \infty$, by means of the Deift/Zhou steepest descent method \cite{DZ}. As mentioned in Section \ref{section: g-function}, the relevant contour to consider for the RH problem for $U$ is $\gamma_{\mathbb{C}} = \gamma_{1}$. The analysis is split in a series of transformations $U \mapsto T \mapsto S \mapsto R$. The first transformation $U \mapsto T$ of Section \ref{subsection: g function transformation} uses the $g$-function obtained in Section \ref{section: g-function} to normalize the RH problem at $\infty$. The opening of the lenses $T \mapsto S$ is realised in Section \ref{subsection: opening of the lenses}. The last step $S\mapsto R$ requires some preparations that are done in Section \ref{subsection: parametrices}: it consists of constructing approximations (called ``parametrices") for $S$ in different regions of the complex plane. Finally, the $S \mapsto R$ transformation is carried out in Section \ref{subsection: small norm}.

%a global parametrix $P^{(\infty)}$, and local parametrices $P^{(r_{+})}$ and $P^{(r_{-})}$ around $r_{+}$ and $r_{-}$, respectively. These parametrices are rather standard: local parametrices in the bulk are built out of Airy functions (as in \cite{Deiftetal}). 
 % 

\subsection{First transformation: $U \mapsto T$}\label{subsection: g function transformation}
We normalize the RH problem with the following transformation
\begin{equation}\label{def of T}
T(\zeta) = e^{N\ell\sigma_{3}}U(\zeta)e^{-2Ng(\zeta)\sigma_{3}}e^{- N\ell \sigma_{3}},
\end{equation}
where $g$ and $\ell$ are defined in Definition \ref{def: g}. Using \eqref{phi and g relation}, we can write the jumps for $T$ in terms of the function $\phi$ of Definition \ref{def:phi}. From \eqref{prop g add} and \eqref{prop g sous}, we find that $T$ satisfies the following RH problem.
\subsubsection*{RH problem for $T$}
\begin{itemize}
\item[(a)] $T : \mathbb{C}\setminus \gamma_{1} \to \mathbb{C}^{2\times 2}$ is analytic.
\item[(b)] By using \eqref{def potential}, \eqref{def of g} and \eqref{jump relation for g on the support}, the jumps for $T$ are given by
\begin{align}
& T_{+}(\zeta) = T_{-}(\zeta) \begin{pmatrix}
e^{-4N\phi_{+}(\zeta)} & 1 \\
0 & e^{-4N\phi_{-}(\zeta)}
\end{pmatrix}, & & \mbox{ for } \zeta \in \Sigma_{1} \subset \gamma_{1}, \label{jumps for T inside support} \\
& T_{+}(\zeta) = T_{-}(\zeta) \begin{pmatrix}
1 & e^{2N(\phi_{+}(\zeta)+\phi_{-}(\zeta))} \\
0 & 1
\end{pmatrix}, & & \mbox{ for } \zeta \in \gamma_{1} \setminus \overline{\Sigma_{1}}. \label{jumps for T outside support}
\end{align}
\item[(c)] As $\zeta \to \infty$, we have $T(\zeta) = I + \bigO(\zeta^{-1})$.

As $\zeta$ tends to $r_{+}$ or $r_{-}$, $T(\zeta)$ remains bounded.
\end{itemize}

The following estimates for $T$ will be important for the saddle point analysis of Section \ref{section: saddle point analysis}.

\begin{proposition} \label{prop:TandTinvsmall}
We have $T(\zeta) = \bigO(N^{1/6})$ and $T^{-1}(\zeta) = \bigO(N^{1/6})$ as $N \to \infty$, uniformly for $\zeta \in \mathbb C \setminus \gamma_1$. In addition, for every $\delta > 0$ fixed, we have $T(\zeta) = \bigO(1)$ and $T^{-1}(\zeta) = \bigO(1)$ as $N \to \infty$ uniformly for 
\begin{equation} \label{eq:zawayfrombranchpoints} 
\zeta \in \{ \zeta \in \mathbb C \setminus \gamma_{1} : |\zeta-r_+| \geq \delta, |\zeta-r_-| \geq \delta \}.
\end{equation}
\end{proposition}
The rest of this section is devoted to the proof of Proposition \ref{prop:TandTinvsmall}.

\subsection{Second transformation: $T \mapsto S$}\label{subsection: opening of the lenses}
Note that for $\zeta \in \Sigma_{1}$, the jumps for $T$ can be factorized as follows:
\begin{align}\label{factorization jumps}
\begin{pmatrix}
e^{-4N\phi_{+}(\zeta)} & 1 \\
0 & e^{-4N\phi_{-}(\zeta)}
\end{pmatrix} = \begin{pmatrix}
1 & 0 \\ e^{-4N\phi_{-}(\zeta)} & 1
\end{pmatrix} \begin{pmatrix}
0 & 1 \\ -1 & 0
\end{pmatrix} \begin{pmatrix}
1 & 0 \\ e^{-4N\phi_{+}(\zeta)} & 1
\end{pmatrix},
\end{align}
where we used $\phi_{+}(\zeta) + \phi_{-}(\zeta) = 0$ for $\zeta \in \Sigma_{1}$. We define the lenses $\gamma_{+}$ and $\gamma_{-}$ by
\begin{align*}
& \gamma_{+} := \gamma_{\alpha} \setminus \overline{\Sigma_{\alpha}} \qquad \mbox{ and } \qquad \gamma_{-} := \gamma_{0} \setminus \overline{\Sigma_{0}},
\end{align*}
see also Figure \ref{fig:opening of the lenses}. 
\begin{figure}
\begin{center}
\begin{tikzpicture}
%\node at (0,0) {\includegraphics[width = 12cm]{../Image/Trajectories/traj_alpha_04.jpg}};
\node at (0,0) {};
\node at (2.1,0) {\color{black} \large $\bullet$};
\node at (-0.38,0) {\color{black} \large $\bullet$};
\node at (-1.05,0) {\color{black} \large $\bullet$};
\node at (-2.035,0) {\color{black} \large $\bullet$};
\node at (-3.13,0) {\color{black} \large $\bullet$};

\node at (0.4,0) {\color{red} \large $\bullet$};
\node at (-1.5,0) {\color{red} \large $\bullet$};
\node at (-5.53,0) {\color{red} \large $\bullet$};
\node at (-1.2,1.42) {\color{red} \large $\bullet$};
\node at (-1.2,-1.42) {\color{red} \large $\bullet$};

% circle gamma1
\draw[black,line width=0.65 mm,->-=0.5,->-=0.98] ([shift=(-180:3.6cm)]2.1,0) arc (-180:180:3.6cm);

% contour Sigma_0
\draw[red,line width=0.65 mm] ([shift=(-35:2.4cm)]-3.13,0) arc (-35:35:2.4cm);

% contour Sigma alpha
\draw[red,line width=0.65 mm] ([shift=(95:1.45cm)]-1.05,0) arc (95:265:1.45cm);

% gamma -
\draw[black,line width=0.65 mm,->-=0.6] ([shift=(35:2.4cm)]-3.13,0) arc (35:325:2.4cm);

% gamma +
\draw[black,line width=0.65 mm,->-=0.7] ([shift=(-95:1.45cm)]-1.05,0) arc (-95:95:1.45cm);

\node at (-0.9,2.535) {$\Sigma_{1}$};
\node at (0.6,0.535) {$\gamma_{+}$};
\node at (-2.1,2.435) {$\gamma_{-}$};
\end{tikzpicture}
\end{center}
\caption{\label{fig:opening of the lenses}The jump contour for $T$ (black), and $\Sigma_{\alpha}$ and $\Sigma_{0}$ (in red), for $\alpha = 0.4$. The red dots are the zeros of $\mathcal{Q}$, and the black dots are the poles.}
\end{figure}
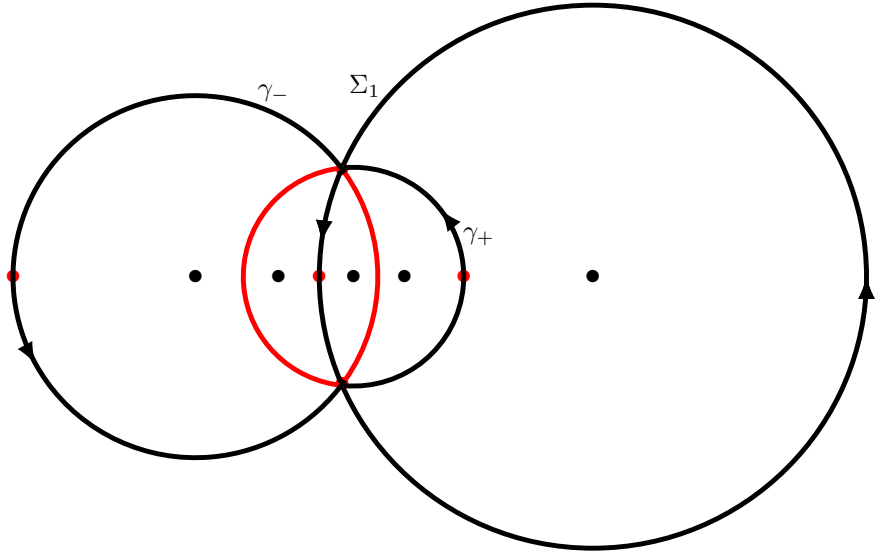
The $T \mapsto S$ transformation is given by $S(\zeta) = T(\zeta) \mathcal{W}(\zeta)$, where
\begin{equation}\label{T to S transformation}
\mathcal{W}(\zeta) = \left\{ \begin{array}{l l}
\begin{pmatrix}
1 & 0 \\ -e^{-4N\phi(\zeta)} & 1 
\end{pmatrix}, & \mbox{for } \zeta \mbox{ in the bounded region delimited by } \overline{\Sigma_{1} \cup \gamma_{+}}, \\
\begin{pmatrix}
1 & 0 \\ e^{-4N\phi(\zeta)} & 1 
\end{pmatrix}, & \mbox{for } \zeta \mbox{ in the unbounded region delimited by } \overline{\Sigma_{1} \cup \gamma_{-}}, \\
I, & \mbox{otherwise}.
\end{array} \right.
\end{equation}
$S$ satisfies the following RH problem.
\subsubsection*{RH problem for $S$}
\begin{itemize}
\item[(a)] $S : \mathbb{C}\setminus (\gamma_{1}\cup \gamma_{+} \cup \gamma_{-}) \to \mathbb{C}^{2\times 2}$ is analytic.
\item[(b)] The jumps for $S$ are given by
\begin{align}
& S_{+}(\zeta) = S_{-}(\zeta) \begin{pmatrix}
0 & 1 \\
-1 & 0
\end{pmatrix}, & & \mbox{ for } \zeta \in \Sigma_{1}, \\
& S_{+}(\zeta) = S_{-}(\zeta) \begin{pmatrix}
1 & 0 \\
e^{-4N\phi(\zeta)} & 1
\end{pmatrix}, & & \mbox{ for } \zeta \in \gamma_{+} \cup \gamma_{-}, \\
& S_{+}(\zeta) = S_{-}(\zeta) \begin{pmatrix}
1 & e^{2N(\phi_{+}(\zeta)+\phi_{-}(\zeta))} \\
0 & 1
\end{pmatrix}, & & \mbox{ for } \zeta \in \gamma_{1} \setminus \overline{\Sigma_{1}}.
\end{align}
\item[(c)] As $\zeta \to \infty$, we have $S(\zeta) = I + \bigO(\zeta^{-1})$.

As $\zeta$ tends to $r_{+}$ or $r_{-}$, $S(\zeta)$ remains bounded.
\end{itemize}
\subsection{Parametrices}\label{subsection: parametrices}
In this subsection, we find good approximations to $S$ in different regions of the complex plane. By Lemma \ref{lem:Nphihigh}, $\re \phi(\zeta) > 0$ for $\zeta \in \gamma_{+}\cup \gamma_{-}$, $\re \phi(\zeta) < 0$ for $\zeta \in \gamma_{1} \setminus \overline{\Sigma_{1}}$, and $\re \phi(\zeta) = 0$ for $\zeta \in \Sigma_{1}$. So the jumps for $S$ on $\gamma_{+}\cup\gamma_{-} \cup (\gamma_{1}\setminus \overline{\Sigma_{1}})$ are exponentially close to the identity matrix matrix as $N \to \infty$, uniformly outside fixed neighborhoods of $r_{-}$ and $r_{+}$. By ignoring these jumps, we are left with the following RH problem, whose solution is denoted $P^{(\infty)}$. We will show in Subsection \ref{subsection: small norm} that $P^{(\infty)}$ is a good approximation to $S$ away from $r_{+}$ and $r_{-}$.

%Therefore, heuristically $S$ can be well approximated outside neighborhoods of $r_{+}$ and $r_{-}$ %As $N \to + \infty$ and simultaneously $\zeta \to r_{+}$ or $\zeta \to r_{-}$, this convergence is slower and depends on the speed at which $N|\re \phi(\zeta)| \to +\infty$. 

%By Let us consider the following RH problem.
\subsubsection*{RH problem for $P^{(\infty)}$}
\begin{itemize}
\item[(a)] $P^{(\infty)} : \mathbb{C}\setminus \overline{\Sigma_{1}} \to \mathbb{C}^{2\times 2}$ is analytic.
\item[(b)] The jumps for $P^{(\infty)}$ are given by
\begin{align}
& P^{(\infty)}_{+}(\zeta) = P^{(\infty)}_{-}(\zeta) \begin{pmatrix}
0 & 1 \\
-1 & 0
\end{pmatrix}, & & \mbox{ for } \zeta \in \Sigma_{1}.
\end{align}
\item[(c)] As $\zeta \to \infty$, we have $P^{(\infty)}(\zeta) = I + \bigO(\zeta^{-1})$.

As $\zeta \to \zeta_{\star} \in \{r_{+},r_{-}\}$, $P^{(\infty)}(\zeta) = \bigO((\zeta-\zeta_{\star})^{-1/4})$.
\end{itemize}
%Conditions (a)-(b) and the behavior at $\infty$ have been obtained from the RH problem for $S$ by ignoring the jump matrices that are close to the identity matrix. 
The condition on the behavior of $P^{(\infty)}(\zeta)$ as $\zeta \to \zeta_{\star}\in\{r_{+},r_{-}\}$ has been added to ensure existence of a solution. This RH problem is independent of $N$, and its unique solution is given by
\begin{equation}
P^{(\infty)}(\zeta) = \begin{pmatrix}
\frac{1}{2}(a(\zeta)+a(\zeta)^{-1}) & \frac{1}{2i}(a(\zeta)-a(\zeta)^{-1}), \\[0.2cm]
\frac{1}{-2i}(a(\zeta)-a(\zeta)^{-1}) & \frac{1}{2}(a(\zeta)+a(\zeta)^{-1})
\end{pmatrix},
\end{equation}
where $a(\zeta) := \left(\frac{\zeta-r_{+}}{\zeta-r_{-}}\right)^{1/4}$ is analytic in $\mathbb{C}\setminus \Sigma_{1}$ and such that $a(\zeta) \sim 1$ as $\zeta \to \infty$.  

\medskip Note that $P^{(\infty)}$ is not a good approximation to $S$ in small neighborhoods of $r_{+},r_{-}$; this can be seen from the behaviors
\begin{align*}
& S(\zeta)=\bigO(1) \quad \mbox{ and } \quad P^{(\infty)}(\zeta) = \bigO((\zeta-\zeta_{\star})^{-1/4}), & & \mbox{as } \zeta \to \zeta_{\star}\in\{r_{+},r_{-}\}.
\end{align*}
Let $\delta > 0$ in Proposition \ref{prop:TandTinvsmall} be fixed, and let $\mathcal{D}_{r_{+}}$ and $\mathcal{D}_{r_{-}}$ be small open disks of radius $\delta/2$ centered at $r_{+}$ and $r_{-}$, respectively. We now construct local approximations $P^{(r_{+})}$ and $P^{(r_{-})}$ (called ``local parametrices") to $S$ in $\mathcal{D}_{r_{+}}$ and $\mathcal{D}_{r_{-}}$, respectively.  We require $P^{(r_{\pm})}$ to satisfy the same jumps as $S$ inside $\mathcal{D}_{r_{\pm}}$, to remain bounded as $\zeta \to r_{\pm}$, and to satisfy the matching condition
\begin{align}\label{matching conditions}
P^{(r_{\pm})}(\zeta) = (I + \bigO(N^{-1}))P^{(\infty)}(\zeta), \qquad \mbox{as } N \to + \infty,
\end{align}
uniformly for $\zeta \in \partial \mathcal{D}_{r_{\pm}}$. The density of $\mu$ vanishes like a square root at the endpoints $r_{+}$ and $r_{-}$, and therefore $P^{(r_{\pm})}$ can be built in terms of Airy functions \cite{DKMVZ1}. These constructions are well-known and standard, so we do not give the details. What is important for us is that
\begin{equation} \label{eq:PandPinvgrow} 
P^{(r_{\pm})}(z) = \bigO(N^{\frac{1}{6}}), \quad P^{(r_{\pm})}(z)^{-1} = \bigO(N^{\frac{1}{6}}) \quad \text{ as } N \to \infty, 
\end{equation}
uniformly for $z \in \mathcal D_{r_{\pm}}$.

\subsection{Small norm RH problem $R$}\label{subsection: small norm}
The final transformation $S \mapsto R$ of the steepest descent is defined by
\begin{equation}\label{S to R transformation}
R(\zeta) = \left\{ \begin{array}{l l}
S(\zeta)P^{(\infty)}(\zeta)^{-1}, & \mbox{ for } \zeta \in \mathbb{C}\setminus (\overline{\mathcal{D}_{r_{+}}\cup \mathcal{D}_{r_{-}}}), \\
S(\zeta)P^{(r_{+})}(\zeta)^{-1}, & \mbox{ for } \zeta \in \mathcal{D}_{r_{+}}, \\
S(\zeta)P^{(r_{-})}(\zeta)^{-1}, & \mbox{ for } \zeta \in \mathcal{D}_{r_{-}}.
\end{array} \right.
\end{equation}
Since $S$ and $P^{(r_{\pm})}$ satisfy the same jumps inside $\mathcal{D}_{r_{\pm}}$, $R$ is analytic inside $(\mathcal{D}_{r_{+}}\setminus \{r_{+}\}) \cup (\mathcal{D}_{r_{-}}\setminus \{r_{-}\})$. Furthermore, $S$ and $P^{(r_{\pm})}$ remain bounded near $r_{\pm}$, so the singularities of $R$ at $r_{\pm}$ are removable. We conclude that $R$ is analytic in 
\begin{align}\label{analyticity of R}
\mathbb C \setminus \Big(\big( (\gamma_1 \cup \gamma_{+} \cup \gamma_{-})\setminus (\mathcal{D}_{r_{+}} \cup \mathcal{D}_{r_{-}})\big) \cup \partial \mathcal{D}_{r_+} \cup \partial \mathcal D_{r_-}\Big).
\end{align}
By \eqref{matching conditions}, the jumps $R_{-}^{-1}R_{+}$ are $\bigO(N^{-1})$ on $\partial \mathcal{D}_{r_{+}} \cup \partial \mathcal{D}_{r_{+}}$, and by Lemma \ref{lem:Nphihigh}, $R_{-}^{-1}R_{+} = \bigO(e^{-c N})$ on $(\gamma_1 \cup \gamma_{+} \cup \gamma_{-})\setminus (\mathcal{D}_{r_{+}} \cup \mathcal{D}_{r_{-}})$ for a certain $c > 0$. It follows by standard theory \cite{DKMVZ1,DKMVZ} that 
\begin{align}\label{lol35}
R(\zeta) = I + \bigO(N^{-1}), \qquad \mbox{as } N \to + \infty,
\end{align}
uniformly for $\zeta$ in the domain \eqref{analyticity of R}.  In particular, $R$ and $R^{-1}$ remain bounded as $N \to \infty$.

Inverting the transformations \eqref{T to S transformation} and \eqref{S to R transformation}, we get
\begin{align*}
T(\zeta) = R(\zeta) \times \left\{ \begin{array}{l l}
P^{(\infty)}(\zeta), & \mbox{ for } \zeta \in \mathbb{C}\setminus (\mathcal{D}_{r_{+}}\cup \mathcal{D}_{r_{-}}) \\
P^{(r_{+})}(\zeta), & \mbox{ for } \zeta \in \mathcal{D}_{r_{+}} \\
P^{(r_{-})}(\zeta), & \mbox{ for } \zeta \in \mathcal{D}_{r_{-}}
\end{array} \right\} \times \mathcal{W}(\zeta)^{-1}.
\end{align*}
By Lemma \ref{lem:Nphihigh}, $\mathcal{W}(\zeta)$ and $\mathcal{W}(\zeta)^{-1}$ are bounded as $N \to + \infty$, uniformly for $\zeta \in \mathbb{C}$. Proposition \ref{prop:TandTinvsmall} follows then straightforwardly by using the estimates \eqref{eq:PandPinvgrow} and \eqref{lol35}.

\section{Phase functions $\Phi$ and $\Psi$}\label{section: phase functions}
%By the symmetries of Subsection \ref{subsection: symmetries}, it is sufficient to prove \eqref{eq:Hintegral} in the lower-left part of the liquid region, that is for 
%\begin{align*}
%(\xi,\eta) \in \mathcal{L}_{\alpha} \cap \{\eta \leq \tfrac{\xi}{2} \leq 0\}. 
%\end{align*}
%We know from the symmetries of Subsection \ref{subsection: symmetries} that it is sufficient to 
In Section \ref{section: saddle point analysis}, we will prove Proposition \ref{prop:doubleintegrallimit} via a saddle point analysis of the double contour integral \eqref{mathcalI}. As it will turn out, the dominant part of the integrand as $N \to + \infty$ will be in the form $e^{2N(\Phi(\zeta;\xi,\eta)-\Phi(\omega;\xi,\eta))}$, for a certain function $\Phi$ which is described below. The analytic continuation of $\Phi$ to the second sheet of $\mathcal{R}_{\alpha}$ is denoted $\Psi$ -- it will also play a role in the saddle point analysis and is presented below.

\medskip The content of this section is a preparation for the saddle point analysis of Section \ref{section: saddle point analysis}. We will study the level set
\begin{align}\label{mathcalNPhi}
\mathcal{N}_{\Phi} = \{ \zeta \in \mathbb{C}: \re \Phi(\zeta) = \re \Phi(s) \},
\end{align}
and also find the relevant contour deformations to consider.
\subsection{Preliminaries}
We start with a definition.
\begin{definition}\label{def: Phi and Psi}
For $(\xi,\eta) \in \mathcal H$ and $\zeta \in \mathbb{C}\setminus \big( (-\infty,c^{-1}] \cup \{c^{-1}+R_{1}e^{it}:-\pi \leq t \leq \theta_{1}\} \big)$, we define $\Phi$ and $\Psi$ by
\begin{align} \nonumber
\Phi(\zeta) & = \Phi(\zeta;\xi,\eta) \\ 
& = g(\zeta) - \frac{1+\xi-\eta}{2} \log \zeta  + \frac{1+\xi}{2}\log \big( (\zeta-\alpha c)(\zeta-\alpha c^{-1}) \big) - \frac{1+\eta}{2}\log \big( (\zeta-c)(\zeta-c^{-1}) \big) + \frac{\ell}{2} \nonumber \\
\label{Phidef} 
& =  \phi(\zeta)  - \frac{\xi-\eta}{2} \log \zeta + \frac{\xi}{2}\log \big( (\zeta-\alpha c)(\zeta-\alpha c^{-1}) \big) - \frac{\eta}{2}\log \big( (\zeta-c)(\zeta-c^{-1}) \big), \\ 
\Psi(\zeta) & =  \Psi(\zeta;\xi,\eta) = -\Phi(\zeta;-\xi,-\eta)
\\ \label{Psidef}
& =  -\phi(\zeta)  - \frac{\xi - \eta}{2} \log \zeta + \frac{\xi}{2}\log \big( (\zeta-\alpha c)(\zeta-\alpha c^{-1}) \big)  - \frac{\eta}{2}\log \big( (\zeta-c)(\zeta-c^{-1}) \big),
\end{align}
where we have used \eqref{def potential} and \eqref{phi and g relation} to write \eqref{Phidef}.
\end{definition}
In the formulas that will be used in Section \ref{section: saddle point analysis}, $\Phi$ and $\Psi$ will always appear in the form 
\begin{align*}
e^{\pm 2N\Phi(\zeta;\xi_{N},\eta_{N})}, \quad e^{\pm 2N\Psi(\zeta;\xi_{N},\eta_{N})}, \qquad \mbox{ with } \qquad \xi_{N} = \frac{x}{N}-1, \quad \eta_{N} = \frac{y}{N}-1,
\end{align*}
for certain integers $x,y \in \{1,\ldots,2N-1\}$. Since $x$ and $y$ are integers, the functions $\zeta \mapsto e^{\pm 2N\Phi(\zeta;\xi_{N},\eta_{N})}$ and $\zeta \mapsto e^{\pm 2N\Psi(\zeta;\xi_{N},\eta_{N})}$ have no jumps along $(-\infty,c^{-1}] \cup \{c^{-1} + R_{1}e^{i t}: -\pi\leq t \leq - \theta_{1}\}$. 
Also, for any $(\xi,\eta) \in \mathcal{H}$, $\re \Phi$ and $\re \Psi$ are harmonic on $\mathbb{C}\setminus (\Sigma_{1} \cup \{0,\alpha c,\alpha c^{-1},c,c^{-1}\})$, and well-defined and continuous on $\mathbb{C}\setminus \{0,\alpha c,\alpha c^{-1},c,c^{-1}\}$. 
For $(\xi,\eta) \in \mathcal{H}^{\mathrm{o}}$, we note the following basic properties of $\Phi$:
\begin{subequations}\label{eq:NPhinearpoles} 
\begin{align}
\Phi(\zeta) & = - \frac{1+\xi-\eta}{2} \log \zeta + \bigO(1)  \text{ as } \zeta \to 0, & 
\lim_{\zeta \to 0} \re \Phi(\zeta) = + \infty, \\
\Phi(\zeta) & = \frac{1+\xi}{2} \log(\zeta-\alpha c) + \bigO(1) \text{ as } \zeta \to \alpha c, &
\lim_{\zeta \to \alpha c} \re \Phi(\zeta) = - \infty, \\
\Phi(\zeta) & = \frac{1+\xi}{2} \log(\zeta -\alpha c^{-1}) + \bigO(1) \text{ as } \zeta \to \alpha c^{-1}, &
\lim_{\zeta \to \alpha c^{-1}} \re \Phi(\zeta) = - \infty, \\
\Phi(\zeta) & = -\frac{1+\eta}{2} \log(\zeta -c) + \bigO(1) \text{ as } \zeta \to c, &
\lim_{\zeta \to c} \re \Phi(\zeta) = + \infty, \\ 
\Phi(\zeta) & = -\frac{1+\eta}{2} \log(\zeta -c^{-1}) + \bigO(1) \text{ as } \zeta \to c^{-1}, &
\lim_{\zeta \to c^{-1}} \re \Phi(\zeta) = + \infty, \\ 
\Phi(\zeta) & = \frac{1-\xi + \eta}{2} \log(\zeta) + \bigO(1) \text{ as } \zeta \to \infty, &
\lim_{\zeta \to \infty} \re \Phi(\zeta) = + \infty,
\end{align}
\end{subequations}
and similarly
\begin{subequations}\label{eq:NPsinearpoles} 
\begin{align}
& \lim_{\zeta \to 0} \re \Psi(\zeta) = \lim_{\zeta \to c} \re \Psi(\zeta) = \lim_{\zeta \to c^{-1}} \re \Psi(\zeta) = \lim_{\zeta \to \infty} \re \Psi(\zeta) = - \infty, \\
& \lim_{\zeta \to \alpha c} \re \Psi(\zeta) = \lim_{\zeta \to \alpha c^{-1}} \re \Psi(\zeta) = + \infty
\end{align}
\end{subequations}
Since the saddle points are the solutions to \eqref{eq:saddlepointeq}, it follows from \eqref{def of phi}, \eqref{Phidef} and \eqref{Psidef} that they are also the zeros of $\Phi'$ and $\Psi'$. For the saddle point analysis, it will be important to know: 1) the sign of $|s-c^{-1}|-R_{1}$ and 2) whether $\Phi'(s)=0$ or $\Psi'(s) = 0$. We summarize the different cases in the next lemma.

\begin{lemma} \label{lem:Ldivision} 
Let $(\xi,\eta) \in \mathcal L_{\alpha}$ and $s = s(\xi,\eta;\alpha)$. 
Then we have 
\begin{enumerate}
\item[\rm (a)] $\Phi'(s) = 0$ and $|s-c^{-1}| < R_{1}$ if and only if $\xi < 0$ and $\eta < \frac{\xi}{2}$, 
\item[\rm (b)] $\Phi'(s)=0$ and $|s-c^{-1}| > R_{1}$ if and only if $\xi < 0$ and $\eta > \frac{\xi}{2}$,
\item[\rm (c)] $\Psi'(s)=0$ and $|s-c^{-1}| < R_{1}$ if and only if $\xi > 0$ and $\eta > \frac{\xi}{2}$,
\item[\rm (d)] $\Psi'(s)=0$ and $|s-c^{-1}| > R_{1}$ if and only if $\xi > 0$ and $\eta < \frac{\xi}{2}$,
\item[\rm (e)] $|s-c^{-1}| = R_{1}$ if and only if $\xi = 0$ or $\eta = \frac{\xi}{2}$.
\end{enumerate}
\end{lemma}
\begin{proof}
This is an immediate consequence of Propositions \ref{prop:hightemp} and \ref{prop: s on the Riemann surface}.
\end{proof}
%\hspace{-0.5cm}We will focus on the case $(\xi,\eta) \in \mathcal{L}_{\alpha} \cap \{\eta \leq \tfrac{\xi}{2} < 0\}$, which corresponds to Lemma \ref{lem:Ldivision} (a) and (e). 

\subsection{The level set $\mathcal{N}_{\Phi}$}\label{subsection: the set NPhi}
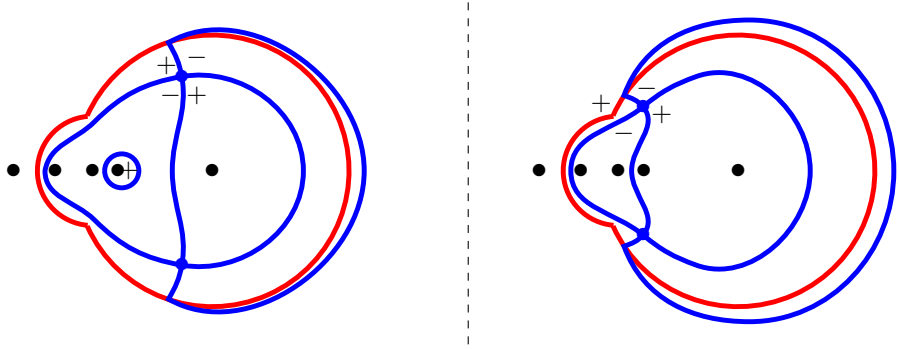
\begin{figure}
\begin{center}
\begin{tikzpicture}[master,scale = 0.5]
%\node at (0,0) {\includegraphics[width = 12cm]{../Image/Trajectories/traj_alpha_04.jpg}};
\node at (0,0) {};
\node at (2.1,0) {\color{black} \large $\bullet$};
\node at (-0.38,0) {\color{black} \large $\bullet$};
\node at (-1.05,0) {\color{black} \large $\bullet$};
\node at (-2.035,0) {\color{black} \large $\bullet$};
\node at (-3.13,0) {\color{black} \large $\bullet$};

% circle gamma1
\draw[red,line width=0.65 mm] ([shift=(-156:3.6cm)]2.1,0) arc (-156:156:3.6cm);

% circle gamma alpha
\draw[red,line width=0.65 mm] ([shift=(95:1.45cm)]-1.05,0) arc (95:265:1.45cm);

\node at (1.3,2.5) {\color{blue} \large $\bullet$};
\node at (1.3,-2.5) {\color{blue} \large $\bullet$};
\draw[blue, line width=0.65 mm] (1.3,2.5) to [out=10, in=90] (4.5,0) to [out=-90, in=-10] (1.3,-2.5);
\draw[blue, line width=0.65 mm] (1.3,2.5) to [out=-80, in=90] (1.05,0) to [out=-90, in=80] (1.3,-2.5);
\draw[blue, line width=0.65 mm] (1.3,2.5) to [out=-170, in=45]
(-1,1.3) to [out=-135, in=80] (-2.3,0) to [out=-90, in=135] (-1,-1.3) to [out=-45, in=170] (1.3,-2.5);
\draw[blue, line width=0.65 mm] (1.3,2.5) to [out=100, in=-60]
(0.95,3.4) to [out=30, in=90] (6.1,0) to [out=-90, in=-30] (0.95,-3.4) to [out=60, in=-100] (1.3,-2.5);
\draw[blue,line width=0.65 mm] ([shift=(-180:0.45cm)]-0.28,0) arc (-180:180:0.45cm);
\node at (-0.1,0) {$+$};
\node at (1.7,2) {$+$};
\node at (1.7,3) {$-$};
\node at (0.9,2.8) {$+$};
\node at (1,2) {$-$};
\end{tikzpicture}\hspace{2cm}\begin{tikzpicture}[slave,scale = 0.5]
%\node at (0,0) {\includegraphics[width = 12cm]{../Image/Trajectories/traj_alpha_04.jpg}};
\node at (0,0) {};
\node at (2.1,0) {\color{black} \large $\bullet$};
\node at (-0.38,0) {\color{black} \large $\bullet$};
\node at (-1.05,0) {\color{black} \large $\bullet$};
\node at (-2.035,0) {\color{black} \large $\bullet$};
\node at (-3.13,0) {\color{black} \large $\bullet$};

% circle gamma1
\draw[red,line width=0.65 mm] ([shift=(-156:3.6cm)]2.1,0) arc (-156:156:3.6cm);

% circle gamma alpha
\draw[red,line width=0.65 mm] ([shift=(95:1.45cm)]-1.05,0) arc (95:265:1.45cm);

\node at (-0.4,1.7) {\color{blue} \large $\bullet$};
\node at (-0.4,-1.7) {\color{blue} \large $\bullet$};
\draw[blue, line width=0.65 mm] (-0.4,1.7) to [out=40, in=-160]
(1,2.5) to [out=20, in=90]
(4,0) to [out=-90, in=-20]
(1,-2.5) to [out=160, in=-40]
(-0.4,-1.7);
\draw[blue, line width=0.65 mm] (-0.4,1.7) to [out=-50, in=90] (-0.7,0) to [out=-90, in=50] (-0.4,-1.7);
\draw[blue, line width=0.65 mm] (-0.4,1.7) to [out=-140, in=90]
(-2.3,0) to [out=-90, in=140] (-0.4,-1.7);
\draw[blue, line width=0.65 mm] (-0.4,1.7) to [out=130, in=-20]
(-0.9,2) to [out=70, in=180]
(2.4,4) to [out=0, in=90]
(6.2,0) to [out=-90, in=0]
(2.4,-4) to [out=-180, in=-70]
(-0.9,-2) to [out=20, in=-130]
(-0.4,-1.7);
\node at (-1.5,1.8) {$+$};
\node at (-0.3,2.2) {$-$};
\node at (0.1,1.5) {$+$};
\node at (-0.9,1) {$-$};

\draw[dashed] (-5,-4.6)--(-5,4.6);
\end{tikzpicture}
\end{center}
\caption{\label{fig: cases 1 and 2 for Phi}
The set $\mathcal{N}_{\Phi}$ is represented in blue, and $\Sigma_{\alpha} \cup \Sigma_{1}$ in red. The parameters are $(\alpha,\xi,\eta) = (0.4,-0.12,-0.86)$ (left) and $(\alpha,\xi,\eta) = (0.4,-0.22,-0.66)$ (right), and they satisfy $\eta < \frac{\xi}{2} < 0$. The sign of $\re (\Phi(\zeta)-\Phi(s))$ in the different regions delimited by $\mathcal{N}_{\Phi}$ is indicated with $\pm$. In each figure, the black dots represent $0$, $\alpha c$, $\alpha c^{-1}$, $c$ and $c^{-1}$ and the blue dots are $s$ and $\overline{s}$.}
\end{figure}
We study the set
\begin{align*}
\mathcal{N}_{\Phi} = \{ z \in \mathbb{C}: \re \Phi(z) = \re \Phi(s) \},
\end{align*}
in case $\eta \leq \frac{\xi}{2}<0$. We have represented $\mathcal{N}_{\Phi}$ for different values of $(\alpha,\xi,\eta)$ in Figures \ref{fig: cases 1 and 2 for Phi}, \ref{fig: cases 3 and 4 for Phi} and \ref{fig: case 5 for Phi}. There are in total eight saddles which are the zeros of $\Phi'$ and $\Psi'$. From \eqref{eq:NPhinearpoles}--\eqref{eq:NPsinearpoles}, both $\Phi'$ and $\Psi'$ vanish at least once on each of the intervals $(-\infty,0)$, $(\alpha c,\alpha c^{-1})$, and $(c,c^{-1})$. This determines the location of $6$ saddles. The remaining two are $s$ and $\overline{s}$, and we already know from Lemma \ref{lem:Ldivision} (a) and (e) that $\Phi'(s) = 0 = \Phi'(\overline{s})$. Therefore, $\Phi' \neq 0$ on $(0,\alpha c) \cup (\alpha c^{-1},c)$. Since $\Phi'(\zeta) \in \mathbb{R}$ for $\zeta \in \mathbb{R}\setminus \{0,\alpha c, \alpha c^{-1},c,c^{-1}\}$, this implies by \eqref{eq:NPhinearpoles} that $\mathcal{N}_{\Phi}$ intersects exactly once each of these two intervals.

\vspace{0.2cm}We show with the next two lemmas that the set $\mathcal{N}_{\Phi} \cap (\overline{\Sigma_{\alpha}\cup \Sigma_{1}})\cap \mathbb{C}^{+}$ is either the empty set or a singleton.

\begin{figure}
\begin{center}
\begin{tikzpicture}[master,scale = 0.5]
%\node at (0,0) {\includegraphics[width = 12cm]{../Image/Trajectories/traj_alpha_04.jpg}};
\node at (0,0) {};
\node at (2.1,0) {\color{black} \large $\bullet$};
\node at (-0.38,0) {\color{black} \large $\bullet$};
\node at (-1.05,0) {\color{black} \large $\bullet$};
\node at (-2.035,0) {\color{black} \large $\bullet$};
\node at (-3.13,0) {\color{black} \large $\bullet$};

% circle gamma1
\draw[red,line width=0.65 mm] ([shift=(-156:3.6cm)]2.1,0) arc (-156:156:3.6cm);

% circle gamma alpha
\draw[red,line width=0.65 mm] ([shift=(95:1.45cm)]-1.05,0) arc (95:265:1.45cm);

\node at (-1.25,0.8) {\color{blue} \large $\bullet$};
\node at (-1.25,-0.8) {\color{blue} \large $\bullet$};
\draw[blue, line width=0.65 mm] (-1.25,0.8) to [out=20, in=90]
(3.3,0) to [out=-90, in=-20]
(-1.25,-0.8);
\draw[blue, line width=0.65 mm] (-1.25,0.8) to [out=-70, in=90] (-0.8,0) to [out=-90, in=70] (-1.25,-0.8);
\draw[blue, line width=0.65 mm] (-1.25,0.8) to [out=-160, in=90]
(-2.3,0) to [out=-90, in=160] (-1.25,-0.8);
\draw[blue, line width=0.65 mm] (-1.25,0.8) to [out=110, in=180]
%(-0.9,2) to [out=70, in=180]
(2.4,4.5) to [out=0, in=90]
(7,0) to [out=-90, in=0]
(2.4,-4.5) to [out=-180, in=-110]
%(-0.9,-2) to [out=20, in=-110]
(-1.25,-0.8);
\node at (-1.5,1) {$+$};
\node at (-0.9,1.2) {$-$};
\node at (-0.7,0.7) {$+$};
\node at (-1.4,0.5) {$-$};
\end{tikzpicture}
\hspace{1cm}
\begin{tikzpicture}[scale = 0.5]
%\node at (0,0) {\includegraphics[width = 12cm]{../Image/Trajectories/traj_alpha_04.jpg}};
\node at (0,0) {};
\node at (2.1,0) {\color{black} \large $\bullet$};
\node at (-0.38,0) {\color{black} \large $\bullet$};
\node at (-1.05,0) {\color{black} \large $\bullet$};
\node at (-2.035,0) {\color{black} \large $\bullet$};
\node at (-3.13,0) {\color{black} \large $\bullet$};

% circle gamma1
\draw[red,line width=0.65 mm] ([shift=(-156:3.6cm)]2.1,0) arc (-156:156:3.6cm);

% circle gamma alpha
\draw[red,line width=0.65 mm] ([shift=(95:1.45cm)]-1.05,0) arc (95:265:1.45cm);

\node at (-1.25,0.6) {\color{blue} \large $\bullet$};
\node at (-1.25,-0.6) {\color{blue} \large $\bullet$};
\draw[blue, line width=0.65 mm] (-1.25,0.6) to [out=40, in=90]
(0.5,0) to [out=-90, in=-40]
(-1.25,-0.6);
\draw[blue, line width=0.65 mm] (-1.25,0.6) to [out=-50, in=90] (-0.9,0) to [out=-90, in=50] (-1.25,-0.6);
\draw[blue, line width=0.65 mm] (-1.25,0.6) to [out=-160, in=90]
(-1.6,0) to [out=-90, in=160] (-1.25,-0.6);
\draw[blue, line width=0.65 mm] (-1.25,0.6) to [out=130, in=180]
%(-0.9,2) to [out=70, in=180]
(2.4,4.5) to [out=0, in=90]
(7,0) to [out=-90, in=0]
(2.4,-4.5) to [out=-180, in=-130]
%(-0.9,-2) to [out=20, in=-130]
(-1.25,-0.6);
\node at (-1.6,0.8) {$+$};
\node at (-0.9,1) {$-$};
\node at (-0.7,0.5) {$+$};
\node at (-1.3,0.3) {$-$};

\node at (1.9,0.5) {$+$};

\draw[blue,line width=0.65 mm] ([shift=(-180:0.25cm)]-2.035,0) arc (-180:180:0.25cm);
\draw[blue,line width=0.65 mm] ([shift=(-180:0.7cm)]1.9,0) arc (-180:180:0.7cm);

\draw[dashed] (-4.8,-4.6)--(-4.8,4.6);
\end{tikzpicture}
\end{center}
\caption{\label{fig: cases 3 and 4 for Phi}
The set $\mathcal{N}_{\Phi}$ is represented in blue, and $\Sigma_{\alpha} \cup \Sigma_{1}$ in red. The parameters are $(\alpha,\xi,\eta) = (0.4,-0.55,-0.414)$ (left) and $(\alpha,\xi,\eta) = (0.4,-0.88,-0.502)$ (right), and they satisfy $\eta < \frac{\xi}{2} < 0$. The sign of $\re (\Phi(\zeta)-\Phi(s))$ in the different regions delimited by $\mathcal{N}_{\Phi}$ is indicated with $\pm$. In each figure, the black dots represent $0$, $\alpha c$, $\alpha c^{-1}$, $c$ and $c^{-1}$ and the blue dots are $s$ and $\overline{s}$.}
\end{figure}
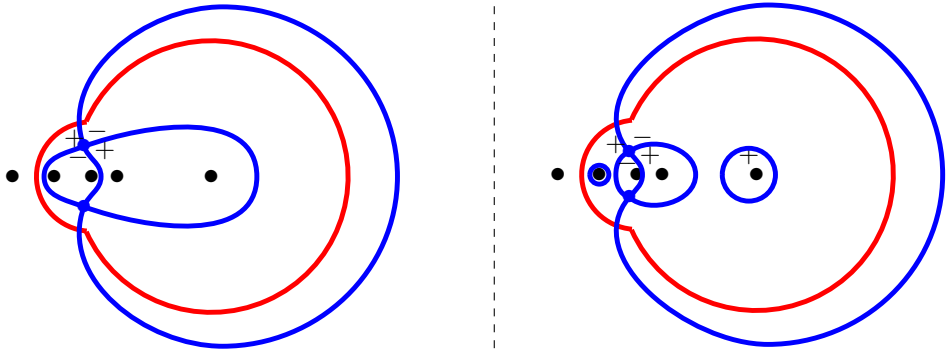

\vspace{0.2cm}For $\zeta \in \mathbb{C}\setminus \{0,\alpha c, \alpha c^{-1},c,c^{-1}\}$, we define the following functions
\begin{align*}
f_{1}(\zeta) = \log \frac{(\zeta-c)(\zeta-c^{-1})}{\zeta}, \quad f_{2}(\zeta) = \log \frac{\zeta}{(\zeta - \alpha c)(\zeta - \alpha c^{-1})},  \quad f_{3}(\zeta) = \log \frac{(\zeta-c)(\zeta-c^{-1})}{(\zeta - \alpha c)(\zeta - \alpha c^{-1})}.
\end{align*}
\begin{lemma}
\label{lem:LogsIncSigmaMinus1} 
If $\zeta$ moves along $(\overline{\Sigma_{\alpha}\cup \Sigma_{1}}) \cap \mathbb C^+$ from left to right, then 
\begin{itemize}
\item[(1)] $\re f_{1}$ is strictly decreasing on $\Sigma_{\alpha}\cap \mathbb{C}^{+}$ and constant on $\Sigma_{1}\cap \mathbb{C}^{+}$,
\item[(2)] $\re f_{2}$ is constant on $\Sigma_{\alpha}\cap \mathbb{C}^{+}$ and strictly decreasing on $\Sigma_{1}\cap \mathbb{C}^{+}$,
\item[(3)] $\re f_{3}$ is strictly decreasing.
\end{itemize}
\end{lemma}
\begin{proof}
A long and tedious computation shows that $\frac{d}{dt} \re f_{1}(\alpha c^{-1} + R_{\alpha} e^{-i t})$ has the same sign as $\sin t$. In particular, $\re f_{1}(\zeta)$ is strictly decreasing along $\Sigma_{\alpha}\cap \mathbb{C}^{+}$ as $\zeta$ moves from left to right. Another (and simpler) computation gives
\begin{align*}
\frac{d}{dt}f_{1}(c^{-1} + R_{1}e^{-it}) = -i \frac{\cos t + \frac{\sqrt{1-\alpha + \alpha^{2}}}{1-\alpha}}{\cos t + \frac{2-3\alpha + 2 \alpha^{2}}{2 (1-\alpha)\sqrt{1-\alpha + \alpha^{2}}}}.
\end{align*}
This expression is purely imaginary, so $\re f_{1}$ is constant on $\Sigma_{1}$. The proofs for $f_{2}$ and $f_{3}$ are similar, so we omit them.
\end{proof}

\begin{corollary} \label{cor:RePhi} For $\eta \leq \frac{\xi}{2} < 0$, the function $\zeta \mapsto \re \Phi(\zeta)$ is strictly decreasing	as $\zeta$ moves along $(\overline{\Sigma_{\alpha} \cup \Sigma_1}) \cap \mathbb C^+$	from left to right.
\end{corollary}
\begin{proof}
We know from Lemma \ref{lem:Nphihigh} that $\re \phi = 0$ on $\Sigma_{\alpha} \cup \Sigma_1$. Therefore, from the expression \eqref{Phidef} for $\Phi$, for $\zeta \in \Sigma_{\alpha} \cup \Sigma_1$ we have 
\begin{align}
\re \Phi (\zeta) & =  - \frac{\xi-\eta}{2} \log |\zeta| + \frac{\xi}{2}\log \big| (\zeta-\alpha c)(\zeta-\alpha c^{-1}) \big| - \frac{\eta}{2}\log \big| (\zeta-c)(\zeta-c^{-1}) \big| \nonumber \\
& = \left( \frac{\xi}{4}-\frac{\eta}{2} \right) \re f_{1}(\zeta) - \frac{\xi}{4} (\re f_{2}(\zeta) + \re f_{3}(\zeta)). \label{eq:RePhionSigma}
\end{align}
The claim follows from Lemma \ref{lem:LogsIncSigmaMinus1}, because $\xi < 0$ and $\frac{\xi}{2} - \eta \geq  0$.
\end{proof}

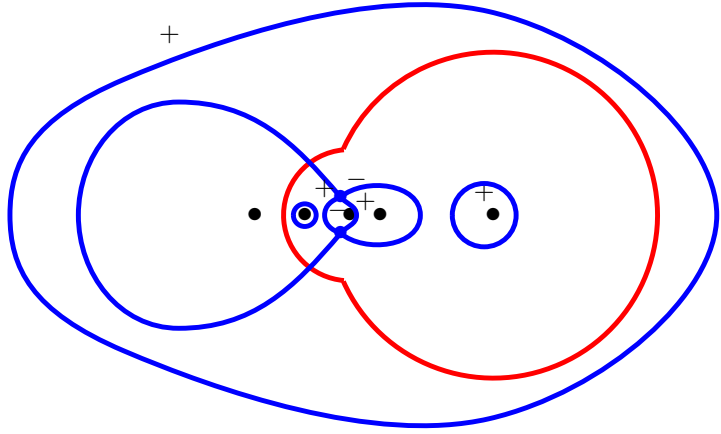
\begin{figure}
\begin{center}
\begin{tikzpicture}[scale = 0.6]
%\node at (0,0) {\includegraphics[width = 12cm]{../Image/Trajectories/traj_alpha_04.jpg}};
\node at (0,0) {};
\node at (2.1,0) {\color{black} \large $\bullet$};
\node at (-0.38,0) {\color{black} \large $\bullet$};
\node at (-1.05,0) {\color{black} \large $\bullet$};
\node at (-2.035,0) {\color{black} \large $\bullet$};
\node at (-3.13,0) {\color{black} \large $\bullet$};

% circle gamma1
\draw[red,line width=0.65 mm] ([shift=(-156:3.6cm)]2.1,0) arc (-156:156:3.6cm);

% circle gamma alpha
\draw[red,line width=0.65 mm] ([shift=(95:1.45cm)]-1.05,0) arc (95:265:1.45cm);

\node at (-1.25,0.4) {\color{blue} \large $\bullet$};
\node at (-1.25,-0.4) {\color{blue} \large $\bullet$};
\draw[blue, line width=0.65 mm] (-1.25,0.4) to [out=40, in=90]
(0.5,0) to [out=-90, in=-40]
(-1.25,-0.4);
\draw[blue, line width=0.65 mm] (-1.25,0.4) to [out=-50, in=90] (-0.9,0) to [out=-90, in=50] (-1.25,-0.4);
\draw[blue, line width=0.65 mm] (-1.25,0.4) to [out=-160, in=90]
(-1.6,0) to [out=-90, in=160] (-1.25,-0.4);
\draw[blue, line width=0.65 mm] (-1.25,0.4) to [out=130, in=0]
%(-0.9,2) to [out=70, in=180]
(-4.8,2.5) to [out=180, in=90]
(-7,0) to [out=-90, in=-180]
(-4.8,-2.5) to [out=0, in=-130]
%(-0.9,-2) to [out=20, in=-130]
(-1.25,-0.4);
\node at (-1.6,0.6) {$+$};
\node at (-0.9,0.8) {$-$};
\node at (-0.7,0.3) {$+$};
\node at (-1.3,0.1) {$-$};

\node at (1.9,0.5) {$+$};

\draw[blue,line width=0.65 mm] ([shift=(-180:0.25cm)]-2.035,0) arc (-180:180:0.25cm);
\draw[blue,line width=0.65 mm] ([shift=(-180:0.7cm)]1.9,0) arc (-180:180:0.7cm);

\draw [blue, line width=0.65 mm] plot [smooth cycle, tension=0.7] coordinates {(-8.5,0) (-6,3) (2,4.5) (7,0) (2,-4.5) (-6,-3)};
\node at (-5,4) {$+$};
\end{tikzpicture}
\end{center}
\caption{\label{fig: case 5 for Phi}
The set $\mathcal{N}_{\Phi}$ is represented in blue, and $\Sigma_{\alpha} \cup \Sigma_{1}$ in red. The parameters are $(\alpha,\xi,\eta) = (0.4,-0.943,-0.538)$, and they satisfy $\eta < \frac{\xi}{2} < 0$. The sign of $\re (\Phi(\zeta)-\Phi(s))$ in the different regions delimited by $\mathcal{N}_{\Phi}$ is indicated with $\pm$. The black dots represent $0$, $\alpha c$, $\alpha c^{-1}$, $c$ and $c^{-1}$ and the blue dots are $s$ and $\overline{s}$.}
\end{figure}
\paragraph{Notation.} For a given closed curve $\sigma$, we denote $\mbox{int}(\sigma)$ for the open and bounded region delimited by $\sigma$.

\vspace{0.2cm}Since $\Phi'(s) = 0$, there are four curves $\{\Gamma_{j}\}_{j=1}^{4}$ emanating from $s$ that belongs to $\mathcal{N}_{\Phi}$. By Corollary \ref{cor:RePhi}, $\mathcal{N}_{\Phi} \cap (\overline{\Sigma_{\alpha}\cup \Sigma_{1}})\cap \mathbb{C}^{+}$ is either the empty set or a singleton, so at least three of the $\Gamma_j$'s, say $\Gamma_1, \Gamma_2, \Gamma_3$, do not intersect $(\overline{\Sigma_{\alpha}\cup \Sigma_{1}}) \cap \mathbb C^+$. The curves $\Gamma_{j}$, $j =1,2,3$ cannot lie entirely in $\mathbb{C}^{+}$; otherwise the max/min principle for harmonic functions would imply that $\re \Phi$ is constant within the region $\mbox{int}(\Gamma_{j})$. Therefore, $\Gamma_{j}$, $j =1,2,3$ have to intersect $\mathbb{R}$. Note that $\overline{\Phi(\zeta)} = \Phi(\overline{\zeta})$ implies that $\mathcal{N}_{\Phi}$ is symmetric with respect to $\mathbb{R}$. In particular, the curves $\Gamma_{j}$, $j=1,2,3$ join $s$ with $\overline{s}$. The next lemma states that $\Gamma_{4}$ is not contained in the region $\mbox{int}(\overline{\Sigma_{\alpha}\cup \Sigma_{1}})$.
% To prove that, we first note the following property of $\mathcal{N}_{\Phi}$:
%\begin{enumerate}[label = (\alph*)]
%\item\label{item a for mathcalPhi} Any potential curve $\sigma$ of $\mathcal{N}_{\phi}\setminus \cup_{j=1}^{4} \Gamma_{j}$ must be surrounding at least one point belonging to $\{0,\alpha c,\alpha c^{-1},c,c^{-1}\}$, in such a way that either $\mbox{int}(\sigma)\cap \mathcal{N}_{\Phi} = \emptyset$ or $\mbox{int}(\sigma)\cap \mathcal{N}_{\Phi} = \mathcal{N}_{\Phi}\setminus \sigma$ holds. This follows again directly by the max/min principle for harmonic functions.
%\end{enumerate}

%We state several properties of $\mathcal{N}_{\Phi}$ that will be important for us:
%\begin{enumerate}[label = (\alph*)]
%\item $\mathcal{N}_{\Phi}$ intersects exactly once each of the intervals $(0,\alpha c)$ and $(\alpha c^{-1},c)$ (see the discussion above Lemma \ref{lem:LogsIncSigmaMinus1}).
%\item Since $\re \Phi$ is continuous in $\mathbb{C}\setminus \{0,\alpha c, \alpha c^{-1},c,c^{-1}\}$, we conclude from \eqref{eq:NPhinearpoles} that $\mathcal{N}_{\Phi}$ intersects each of the intervals $(-\infty,0)$, $(\alpha c^{-1},\alpha c)$, $(c,c^{-1})$ and $(c^{-1},+\infty)$ an even number of times (possibly $0$).
%\item By the max/min principle for harmonic functions, any curve of $\mathcal{N}_{\phi}\setminus \cup_{j=1}^{4} \Gamma_{j}$ must be surrounding a point belonging to $\{0,\alpha c,\alpha c^{-1},c,c^{-1}\}$.
%\end{enumerate}

\begin{lemma}\label{lemma: Gamma4 intersects at one point}
$\mathcal{N}_{\Phi} \cap (\overline{\Sigma_{\alpha}\cup \Sigma_{1}}) \cap \mathbb{C}^{+}$ is a singleton.
\end{lemma}
\begin{proof}
Assume on the contrary that $\Gamma_{4}$ lies entirely in $\mbox{int}(\overline{\Sigma_{\alpha}\cup \Sigma_{1}})$, and denote $p_{j}$ for the intersection point of $\Gamma_{j}$ with $\mathbb{R}$. We assume without loss of generality that $p_{1}<p_{2}<p_{3}<p_{4}$. There is at most one $p_{j}$ inside each of the intervals
\begin{align*}
(\alpha c^{-1} - R_{\alpha},\alpha c), \quad (\alpha c, \alpha c^{-1}), \quad (\alpha c^{-1},c), \quad (c,c^{-1}), \quad (c^{-1},c^{-1}+R_{1}),
\end{align*}
otherwise we again find a contradiction using the max/min principle for harmonic functions. Thus, there are five $5$ possibilities for the location of the $p_{j}$'s, and each of them leads to a contradiction. Let us treat the case
\begin{align}\label{the pj location fake case}
p_{1} \in (\alpha c, \alpha c^{-1}), \quad p_{2} \in (\alpha c^{-1},c), \quad p_{3} \in (c,c^{-1}), \quad p_{4} \in (c^{-1},c^{-1}+R_{1}).
\end{align}
Since $\re (\Phi(\zeta)-\Phi(s))$ changes sign as $\zeta$ crosses $\mathcal{N}_{\Phi}\setminus \{s,\overline{s}\}$, by \eqref{eq:NPhinearpoles} we must have 
\begin{align}\label{two expressions for mathcalN}
\mathcal{N}_{\Phi} = \sigma_{1} \cup \sigma_{2} \cup \bigcup_{j=1}^{4} \Gamma_{j},
\end{align}
where $\sigma_{1}$ is a closed curve surrounding either $\alpha c$ or $\alpha c^{-1}$, such that $\mbox{int}(\sigma_{1}) \cap \mathcal{N}_{\Phi} = \emptyset$, and $\sigma_{2}$ is a closed curve surrounding either $c$ or $c^{-1}$, such that $\mbox{int}(\sigma_{2}) \cap \mathcal{N}_{\Phi} = \emptyset$. Since $\mathcal{N}_{\Phi}$ intersects both $(0,\alpha c)$ and $(\alpha c^{-1},c)$ exactly once, $\sigma_{1}$ surrounds $\alpha c$ and $\sigma_{2}$ surrounds $c^{-1}$. Then, the max/min principle implies that $\re \Phi$ is constant on $\mbox{int}(\overline{\Gamma_{3}\cup \Gamma_{4}})\setminus \mbox{int}(\sigma_{2})$, which is a contradiction. The four other cases than \eqref{the pj location fake case} can be treated similarly, so we omit the proofs.
\end{proof}

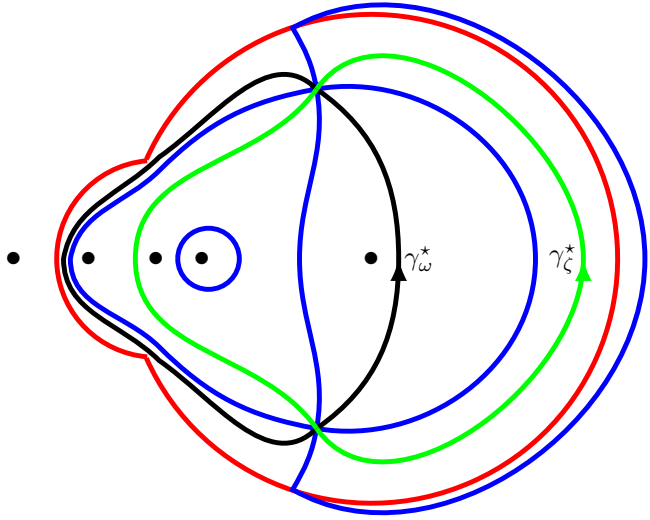
\begin{figure}
\begin{center}
\begin{tikzpicture}[scale = 0.9]
%\node at (0,0) {\includegraphics[width = 12cm]{../Image/Trajectories/traj_alpha_04.jpg}};
\node at (0,0) {};
\node at (2.1,0) {\color{black} \large $\bullet$};
\node at (-0.38,0) {\color{black} \large $\bullet$};
\node at (-1.05,0) {\color{black} \large $\bullet$};
\node at (-2.035,0) {\color{black} \large $\bullet$};
\node at (-3.13,0) {\color{black} \large $\bullet$};

% circle gamma1
\draw[red,line width=0.65 mm] ([shift=(-156:3.6cm)]2.1,0) arc (-156:156:3.6cm);

% circle gamma alpha
\draw[red,line width=0.65 mm] ([shift=(95:1.45cm)]-1.05,0) arc (95:265:1.45cm);

\node at (1.3,2.5) {\color{blue} \large $\bullet$};
\node at (1.3,-2.5) {\color{blue} \large $\bullet$};
\draw[blue, line width=0.65 mm] (1.3,2.5) to [out=10, in=90] (4.5,0) to [out=-90, in=-10] (1.3,-2.5);
\draw[blue, line width=0.65 mm] (1.3,2.5) to [out=-80, in=90] (1.05,0) to [out=-90, in=80] (1.3,-2.5);
\draw[blue, line width=0.65 mm] (1.3,2.5) to [out=-170, in=45]
(-1,1.3) to [out=-135, in=80] (-2.3,0) to [out=-90, in=135] (-1,-1.3) to [out=-45, in=170] (1.3,-2.5);
\draw[blue, line width=0.65 mm] (1.3,2.5) to [out=100, in=-60]
(0.95,3.4) to [out=30, in=90] (6.1,0) to [out=-90, in=-30] (0.95,-3.4) to [out=60, in=-100] (1.3,-2.5);
\draw[blue,line width=0.65 mm] ([shift=(-180:0.45cm)]-0.28,0) arc (-180:180:0.45cm);

\draw[black, line width=0.65 mm] (1.3,2.5) to [out=135, in=35]
(-1,1.5) to [out=-135, in=80] (-2.4,0) to [out=-90, in=135] (-1,-1.5) to [out=-35, in=-135] (1.3,-2.5);
\draw[black, line width=0.65 mm,-<- = 0.5] (1.3,2.5) to [out=-35, in=90] (2.5,0) to [out=-90, in=35] (1.3,-2.5);
\node at (2.8,0) {$\gamma_{\omega}^{\star}$};

\draw[green, line width=0.65 mm] (1.3,2.5) to [out=-125, in=90] (-1.35,0) to [out=-90, in=125] (1.3,-2.5);
\draw[green, line width=0.65 mm, -<-=0.5] (1.3,2.5) to [out=55, in=90] (5.2,0) to [out=-90, in=-55] (1.3,-2.5);
\node at (4.9,0) {$\gamma_{\zeta}^{\star}$};

\end{tikzpicture}
\end{center}
\caption{\label{fig: case 1 contour}The set $\mathcal{N}_{\Phi}$ is represented in blue, and $\Sigma_{\alpha} \cup \Sigma_{1}$ in red. The parameters are $(\alpha,\xi,\eta) = (0.4,-0.12,-0.86)$ as in Figure \ref{fig: cases 1 and 2 for Phi} (left). The contour $\gamma_{\zeta}^{\star}$ is represented in green, and $\gamma_{\omega}^{\star}$ in black. The black dots represent $0$, $\alpha c$, $\alpha c^{-1}$, $c$ and $c^{-1}$ and the blue dots are $s$ and $\overline{s}$.}
\end{figure}

Lemma \ref{lemma: Gamma4 intersects at one point} states that $\Gamma_{4}$ crosses $\Sigma_{\alpha}\cup \Sigma_{1}$ exactly once. We know from \eqref{eq:NPhinearpoles} that $\re \Phi(\zeta) \to + \infty$ as $\zeta \to \infty$, so $\Gamma_{4}$ intersects the real line, and then by symmetry ends at $\overline{s}$. So each of the $\Gamma_{j}$'s intersects $\mathbb{R}$. We denote $p_{j}$ for the intersection point of $\Gamma_{j}$ with $\mathbb{R}$, and choose the ordering such that  $p_1 < p_2 < p_3$. We recall that $\re (\Phi(\zeta)-\Phi(s))$ is harmonic for $\zeta \in \mathbb{C}\setminus (\Sigma_{1} \cup \{0,\alpha c,\alpha c^{-1},c,c^{-1}\})$ and changes sign as $\zeta$ crosses $\mathcal{N}_{\Phi}\setminus \{s,\overline{s}\}$. Therefore, by \eqref{eq:NPhinearpoles}, the region $\mbox{int}(\overline{\Gamma_{1} \cup \Gamma_{2}})$ must contain at least one of the singularities $\alpha c$ and $\alpha c^{-1}$, and $\mbox{int}(\overline{\Gamma_{2} \cup \Gamma_{3}})$ must contain at least one of the singularities $c$ and $c^{-1}$. There are still quite a few cases that can occur. The figures provide a fairly good overview (though not complete) of what can happen:
\begin{enumerate}
\item In Figure \ref{fig: cases 1 and 2 for Phi} (left), $\alpha c, \alpha c^{-1}, c \in \mbox{int}(\overline{\Gamma_{1}\cup \Gamma_{2}})$, $c^{-1} \in \mbox{int}(\overline{\Gamma_{2}\cup \Gamma_{3}})$.
\item In Figures \ref{fig: cases 1 and 2 for Phi} (right) and \ref{fig: cases 3 and 4 for Phi} (left), $\alpha c, \alpha c^{-1} \in \mbox{int}(\overline{\Gamma_{1}\cup \Gamma_{2}})$ and $c,c^{-1} \in \mbox{int}(\overline{\Gamma_{2}\cup \Gamma_{3}})$.
\item In Figures \ref{fig: cases 3 and 4 for Phi} (right) and \ref{fig: case 5 for Phi}, $\alpha c^{-1} \in \mbox{int}(\overline{\Gamma_{1}\cup \Gamma_{2}})$ and $c \in \mbox{int}(\overline{\Gamma_{2}\cup \Gamma_{3}})$.
\end{enumerate}
Furthermore, $\Gamma_{4}$ intersects both $\Sigma_{1}$ and $(c^{-1}+R_{1},+\infty)$ in Figure \ref{fig: cases 1 and 2 for Phi}, intersects both $\Sigma_{\alpha}$ and $(c^{-1}+R_{1},+\infty)$ in Figure \ref{fig: cases 3 and 4 for Phi}, and intersects both $\Sigma_{\alpha}$ and $(-\infty,\alpha c^{-1} - R_{\alpha})$ in Figure \ref{fig: case 5 for Phi}.
There are also some obvious intermediate cases which are not illustrated by a figure. In all cases, we can find contours $\gamma_{\zeta}^{\star}$ and $\gamma_{\omega}^{\star}$ as described in the following proposition. These contours are illustrated for two different situations in Figures \ref{fig: case 1 contour} and \ref{fig: case 2 contour} (left).

\begin{proposition} \label{prop:contoursexist}
Let $(\xi,\eta) \in \mathcal{L}_{\alpha}$ with $\eta < \frac{\xi}{2} < 0$. There exist contours $\gamma_{\zeta}^{\star}$ and $\gamma_{\omega}^{\star}$ such that 
\begin{itemize}
\item $\gamma_{\omega}^{\star} \subset \mbox{int}(\overline{\Sigma_\alpha \cup \Sigma_{1}})$, it surrounds $\alpha c$ and $\alpha c^{-1}$, and it goes through $s$ and $\overline{s}$ in such a way that
\begin{align*}
\re \Phi(\omega) > \re \Phi(s), \qquad \omega \in \gamma_{\omega}^{\star} \setminus \{s, \overline{s}\},
\end{align*}
\item $\gamma_\zeta^{\star} \subset \mbox{int}(\gamma_1)$, surrounds $c$ and $c^{-1}$, and it goes through $s$ and $\overline{s}$ in such a way that 
\begin{align*}
\re \Phi(\zeta) < \re \Phi(s), \qquad \zeta \in \gamma_{\zeta}^{\star} \setminus \{s, \overline{s}\}.
\end{align*}
\end{itemize}
\end{proposition}

If $\eta = \frac{\xi}{2}$, we know from Proposition \ref{prop:hightemp} \ref{item b in prop mapping s} that $s$ lies on $\gamma_{1}\setminus \overline{\Sigma_{1}}$. For the saddle point analysis, we will need $\gamma_{\zeta}^{\star}$ lying inside $\gamma_{1}$ (not necessarily strictly inside). To prove existence of such a contour $\gamma_{\zeta}^{\star}$, we need to know that $\re \Phi(\zeta)-\re \Phi(s)$ is strictly negative for $\zeta \in \gamma_{1}\setminus \overline{\Sigma_{1}}$ (at least in small neighborhoods of $s$ and $\overline{s}$). 
\begin{figure}
\begin{center}
\hspace{-0.2cm}
\begin{tikzpicture}[master,scale = 0.7]
%\node at (0,0) {\includegraphics[width = 12cm]{../Image/Trajectories/traj_alpha_04.jpg}};
\node at (0,0) {};
\node at (2.1,0) {\color{black} \large $\bullet$};
\node at (-0.38,0) {\color{black} \large $\bullet$};
\node at (-1.05,0) {\color{black} \large $\bullet$};
\node at (-2.035,0) {\color{black} \large $\bullet$};
\node at (-3.13,0) {\color{black} \large $\bullet$};

% circle gamma1
\draw[red,line width=0.65 mm] ([shift=(-156:3.6cm)]2.1,0) arc (-156:156:3.6cm);

% circle gamma alpha
\draw[red,line width=0.65 mm] ([shift=(95:1.45cm)]-1.05,0) arc (95:265:1.45cm);

\node at (-1.25,0.6) {\color{blue} \large $\bullet$};
\node at (-1.25,-0.6) {\color{blue} \large $\bullet$};
\draw[blue, line width=0.65 mm] (-1.25,0.6) to [out=40, in=90]
(0.5,0) to [out=-90, in=-40]
(-1.25,-0.6);
\draw[blue, line width=0.65 mm] (-1.25,0.6) to [out=-50, in=90] (-0.9,0) to [out=-90, in=50] (-1.25,-0.6);
\draw[blue, line width=0.65 mm] (-1.25,0.6) to [out=-160, in=90]
(-1.6,0) to [out=-90, in=160] (-1.25,-0.6);
\draw[blue, line width=0.65 mm] (-1.25,0.6) to [out=130, in=180]
%(-0.9,2) to [out=70, in=180]
(2.4,4.5) to [out=0, in=90]
(7,0) to [out=-90, in=0]
(2.4,-4.5) to [out=-180, in=-130]
%(-0.9,-2) to [out=20, in=-130]
(-1.25,-0.6);

\draw[blue,line width=0.65 mm] ([shift=(-180:0.25cm)]-2.035,0) arc (-180:180:0.25cm);
\draw[blue,line width=0.65 mm] ([shift=(-180:0.7cm)]1.9,0) arc (-180:180:0.7cm);

\draw[green, line width=0.65 mm, -<-=0.2] (-1.25,0.6) to [out=85, in=90]
(3,0) to [out=-90, in=-85]
(-1.25,-0.6);
\draw[green, line width=0.65 mm] (-1.25,0.6) to [out=-95, in=90] (-1.25,0) to [out=-90, in=95] (-1.25,-0.6);
\node at (3.3,0) {$\gamma_{\zeta}^{\star}$};

\draw[black, line width=0.65 mm] (-1.25,0.6) to [out=-205, in=90]
(-2.4,0) to [out=-90, in=205] (-1.25,-0.6);
\draw[black,arrows={-Triangle[length=0.24cm,width=0.16cm]}]
($(-2,0.6)$) --  ++(-170:0.001);
\draw[black, line width=0.65 mm] (-1.25,0.6) to [out=-5, in=90]
(-0.6,0) to [out=-90, in=5]
(-1.25,-0.6);
\node at (-0.5,0.35) {$\gamma_{\omega}^{\star}$};
\end{tikzpicture}
\hspace{0.5cm}
\begin{tikzpicture}[slave,scale = 0.7]
%\node at (0,0) {\includegraphics[width = 12cm]{../Image/Trajectories/traj_alpha_04.jpg}};
\node at (0,0) {};
\node at (2.1,0) {\color{black} \large $\bullet$};
\node at (-0.38,0) {\color{black} \large $\bullet$};
\node at (-1.05,0) {\color{black} \large $\bullet$};
\node at (-2.035,0) {\color{black} \large $\bullet$};
\node at (-3.13,0) {\color{black} \large $\bullet$};

% circle gamma1
\draw[red,line width=0.65 mm] ([shift=(-156:3.6cm)]2.1,0) arc (-156:156:3.6cm);

% circle gamma alpha
\draw[red,line width=0.65 mm] ([shift=(95:1.45cm)]-1.05,0) arc (95:265:1.45cm);

\node at ($(2.1,0)+(168:3.6)$) {\color{blue} \large $\bullet$};
\node at ($(2.1,0)+(-168:3.6)$) {\color{blue} \large $\bullet$};
\draw[blue, line width=0.65 mm] ($(2.1,0)+(168:3.6)$) to [out=40, in=90]
(0.5,0) to [out=-90, in=-40]
($(2.1,0)+(-168:3.6)$);
\draw[blue, line width=0.65 mm] ($(2.1,0)+(168:3.6)$) to [out=-50, in=90] (-0.9,0) to [out=-90, in=50] ($(2.1,0)+(-168:3.6)$);
\draw[blue, line width=0.65 mm] ($(2.1,0)+(168:3.6)$) to [out=-160, in=90]
(-2.25,0) to [out=-90, in=160] ($(2.1,0)+(-168:3.6)$);
\draw[blue, line width=0.65 mm] ($(2.1,0)+(168:3.6)$) to [out=130, in=180]
%(-0.9,2) to [out=70, in=180]
(2.4,4.5) to [out=0, in=90]
(7,0) to [out=-90, in=0]
(2.4,-4.5) to [out=-180, in=-130]
%(-0.9,-2) to [out=20, in=-130]
($(2.1,0)+(-168:3.6)$);

% to be changes a bit
%
\draw[blue,line width=0.65 mm] ([shift=(-180:0.7cm)]1.9,0) arc (-180:180:0.7cm);

\draw[green,line width=0.65 mm] ([shift=(161:3.6cm)]2.1,0) arc (161:199:3.6cm);
\draw[green, line width=0.65 mm, -<-=0.2] ($(2.1,0)+(161:3.6)$) to [out=0, in=90]
(3,0) to [out=-90, in=-0]
($(2.1,0)+(-161:3.6)$);
%\draw[green, line width=0.65 mm] (-1.25,0.6) to [out=-95, in=90] (-1.25,0) to [out=-90, in=95] (-1.25,-0.6);
\node at (3.3,0) {$\gamma_{\zeta}^{\star}$};

\draw[black, line width=0.65 mm] ($(2.1,0)+(168:3.6)$) to [out=-205, in=90]
(-2.4,0) to [out=-90, in=205] ($(2.1,0)+(-168:3.6)$);
\draw[black,arrows={-Triangle[length=0.24cm,width=0.16cm]}]
($(-2.2,0.55)$) --  ++(-145:0.001);
\draw[black, line width=0.65 mm] ($(2.1,0)+(168:3.6)$) to [out=-5, in=90]
(-0.6,0) to [out=-90, in=5]
($(2.1,0)+(-168:3.6)$);
\node at (-0.5,0.35) {$\gamma_{\omega}^{\star}$};
\draw[dashed] (-4,-5)--(-4,5);
\end{tikzpicture}
\end{center}
\caption{\label{fig: case 2 contour}The set $\mathcal{N}_{\Phi}$ is represented in blue, and $\Sigma_{\alpha} \cup \Sigma_{1}$ in red. The parameters are $(\alpha,\xi,\eta) = (0.4,-0.88,-0.502)$ (left) and $(\alpha,\xi,\eta) = (0.4,-0.7,-0.35)$ (right). The contour $\gamma_{\zeta}^{\star}$ is represented in green, and $\gamma_{\omega}^{\star}$ in black. The black dots represent $0$, $\alpha c$, $\alpha c^{-1}$, $c$ and $c^{-1}$ and the blue dots are $s$ and $\overline{s}$.}
\end{figure}
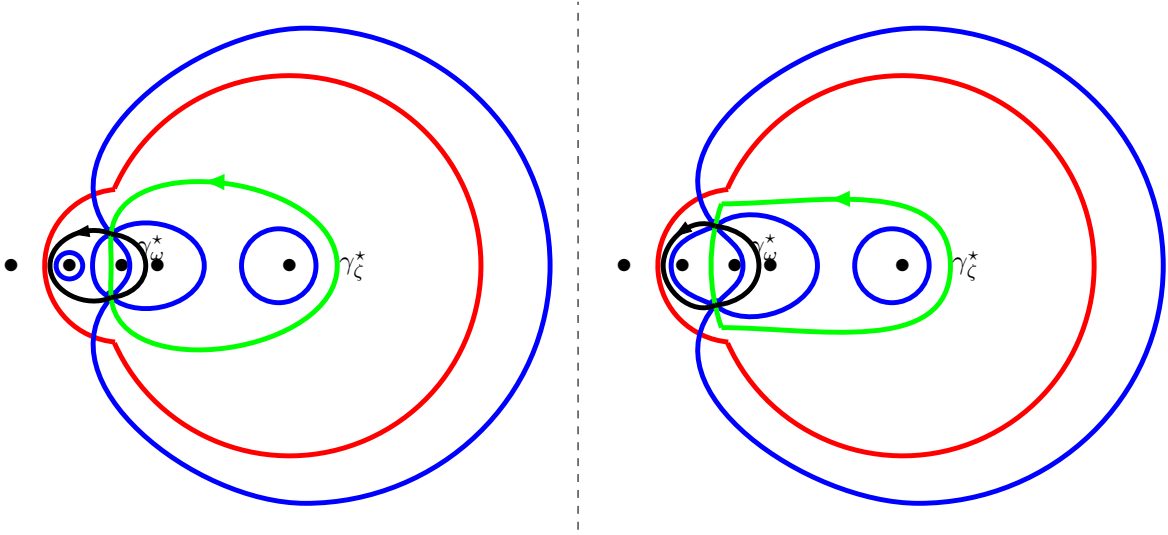
\begin{lemma}
Let $\eta = \frac{\xi}{2}<0$. For $\zeta \in \gamma_{1}\setminus (\overline{\Sigma_{1}} \cup \{s\}) \cap \mathbb{C}^{+}$, we have $\re \Phi(\zeta) < \re \Phi(s)$. 
\end{lemma}
\begin{proof}
Let $\zeta = c^{-1} + R_{1}e^{it}$. For $t \in (\theta_{1},\pi)$, we have
\begin{align}\label{lol25}
\re(\Phi'(\zeta)d\zeta) = \frac{-\cos(\frac{t}{2})\left( \sqrt{\cos \theta_{1}-\cos t}(\cos t + a_{1}) + \frac{\xi}{2}(\cos t + a_{2})\sqrt{1-\cos t} \right)}{\sqrt{2}(\cos t + \frac{2-\alpha + \alpha^{2}}{2\sqrt{1-\alpha + \alpha^{2}}})(\cos t + \frac{2-3\alpha + 2\alpha^{2}}{2(1-\alpha)\sqrt{1-\alpha + \alpha^{2}}})},
\end{align}
where $a_{1},a_{2}$ are given by $a_{1} = \frac{\alpha^{2} + (1-\alpha)\sqrt{1-\alpha + \alpha^{2}}}{2(1-\alpha)}$ and $a_{2} = \frac{2-3\alpha + 2\alpha^{2} + \alpha^{3}}{2(1-\alpha)\sqrt{1-\alpha + \alpha^{2}}}$ and satisfy $a_{1}>a_{2}>1$. The expression \eqref{lol25} vanishes if and only if
\begin{align}\label{lol26}
\frac{\sqrt{\cos \theta_{1} - \cos t}}{\sqrt{1-\cos t}} = - \frac{\xi}{2} \frac{\cos t + a_{2}}{\cos t + a_{1}}.
\end{align}
Since the left-hand-side is strictly decreasing, and the right-hand-side is strictly increasing as $t$ decreases from $\pi$ to $\theta_{1}$, there is a unique $\zeta = c^{-1}+R_{1}e^{it}$, $t \in (\theta_{1},\pi)$, such that $\re(\Phi'(\zeta)d\zeta) = 0$, and this must be $s$. This implies that $\re \Phi(\zeta) - \re \Phi(s)$ is of constant sign on $\gamma_{1}\setminus (\overline{\Sigma_{1}} \cup \{s\}) \cap \mathbb{C}^{+}$. By \eqref{lol25}, $\re(\Phi'(\zeta)d\zeta) > 0$ at $t = \theta_{1}$ (recall that $\xi < 0$), so the claim is proved.
\end{proof}

Therefore, we can find contours $\gamma_{\zeta}^{\star}$ and $\gamma_{\omega}^{\star}$ as described in Proposition \ref{prop:contoursexist eta = xi/2}, see also Figure \ref{fig: case 2 contour} (right).
\begin{proposition} \label{prop:contoursexist eta = xi/2}
Let $(\xi,\eta) \in \mathcal{L}_{\alpha}$ with $\eta = \frac{\xi}{2} < 0$. There exist contours $\gamma_{\zeta}^{\star}$ and $\gamma_{\omega}^{\star}$ such that 
\begin{itemize}
\item $\gamma_{\omega}^{\star} \subset \mbox{int}(\overline{\Sigma_\alpha \cup \Sigma_{1}})$, it surrounds $\alpha c$ and $\alpha c^{-1}$, and it goes through $s$ and $\overline{s}$ in such a way that
\begin{align*}
\re \Phi(\omega) > \re \Phi(s), \qquad \omega \in \gamma_{\omega}^{\star} \setminus \{s, \overline{s}\},
\end{align*}
\item $\gamma_\zeta^{\star} \subset \overline{\mbox{int}(\gamma_1)}$, surrounds $c$ and $c^{-1}$, and it goes through $s$ and $\overline{s}$ in such a way that 
\begin{align*}
\re \Phi(\zeta) < \re \Phi(s), \qquad \zeta \in \gamma_{\zeta}^{\star} 	\setminus \{s, \overline{s}\}.
\end{align*}
\end{itemize}
\end{proposition}

\section{Saddle point analysis}\label{section: saddle point analysis}
In this section, we prove Proposition \ref{prop:doubleintegrallimit} by means of a saddle point analysis that mainly follows the lines of \cite{CDKL}. This analysis relies mostly on Sections \ref{section: steepest descent for $U$}--\ref{section: phase functions} and is only valid for $(\xi,\eta)$ in the lower left part of the liquid region, that is for $(\xi,\eta) \in \mathcal{L}_{\alpha} \cap \{\eta \leq \tfrac{\xi}{2} \leq 0\}$. We divide the proof in three subcases: $\eta \leq \tfrac{\xi}{2} < 0$, $\eta < \tfrac{\xi}{2} = 0$ and $\eta = \xi = 0$. 

\begin{remark}
By adapting the analysis of this section and of Section \ref{section: phase functions}, it is possible to carry out similar saddle point analysis when $(\xi,\eta)$ lies in the other quadrants of the liquid region. Note however that this is not needed, thanks to the symmetries of Subsection \ref{subsection: symmetries} (see also Proposition \ref{prop:doubleintegrallimit}).
\end{remark}

\subsection{The case $\eta \leq \tfrac{\xi}{2} < 0$}
The double integral $\mathcal{I}$ is defined in \eqref{mathcalI}. The associated two contours of integration can be chosen freely, as long as they are closed curves surrounding $c$ and $c^{-1}$ once in the positive direction, and not surrounding $0$. From now, it will be convenient to take different contours in the $\zeta$ and $\omega$ variables, so we indicate this in the notation by rewriting \eqref{mathcalI} as
\begin{align}\label{mathcalI zeta omega}
\mathcal{I}(x,y;H) = \frac{1}{(2\pi i)^{2}}\int_{\gamma_{\zeta}}d\zeta \int_{\gamma_{\omega}}d\omega H(\omega, \zeta) W(\omega) \mathcal{R}^{U}(\omega,\zeta) \frac{\omega^{N}}{\zeta^{N}}q(\omega,\zeta)^{y} \tilde{q}(\omega,\zeta)^{x}.
\end{align}
Only the first column of $U$ appears in \eqref{mathcalI zeta omega}, which is independent of the choice of the contour $\gamma_{\mathbb{C}}$ associated to the RH problem for $U$. However, by using the jumps for $U$, we will find (just below) another formula for $\mathcal{I}$ in terms of the second column of $U$. Therefore, the choice of $\gamma_{\mathbb{C}}$ will matter. To be able to use the steepest descent of Section \ref{section: steepest descent for $U$}, we assume from now that $\gamma_{\mathbb{C}}=\gamma_{1}$. Recall that $T$ is expressed in terms of $U$ via \eqref{def of T}, and define
\begin{align}\label{lol22}
\widetilde{\mathcal{R}}^{T}(\omega,\zeta) = \begin{pmatrix}
1 & 0
\end{pmatrix} T^{-1}(\omega)T(\zeta) \begin{pmatrix}
1 \\ 0
\end{pmatrix}.
\end{align}
By Proposition \ref{prop:TandTinvsmall}, $\widetilde{\mathcal{R}}^{T}(\omega,\zeta)$ is uniformly bounded as $\zeta$ and $\omega$ stay bounded away from $r_{+}$ and $r_{-}$. We will need the analytic continuation in $\omega$ of $\widetilde{\mathcal{R}}^{T}(\omega,\zeta)$ from the interior of $\gamma_{1}$ to the bounded region delimited by $\overline{\Sigma_{1} \cup \Sigma_{\alpha}}$ (see Figure \ref{fig: crit traj alpha 04}). We denote it $\widetilde{\mathcal{R}}^{T,a}(\omega,\zeta)$, and by \eqref{jumps for T outside support} it is given by
\begin{align}\label{def of RTa}
\widetilde{\mathcal{R}}^{T,a}(\omega,\zeta) = \begin{cases}
\begin{pmatrix}
1 & 0
\end{pmatrix} T^{-1}(\omega)T(\zeta) \begin{pmatrix}
1 \\ 0
\end{pmatrix}, & |\omega-c^{-1}| < R_{1}, \, \zeta \in \mathbb{C}\setminus \gamma_{1}, \\[0.4cm]
\begin{pmatrix}
1 & -e^{4N\phi(\omega)}
\end{pmatrix} T^{-1}(\omega)T(\zeta) \begin{pmatrix}
1 \\ 0
\end{pmatrix}, & \omega \in \mbox{int}\big( (\gamma_{1}\setminus \Sigma_{1})\cup \Sigma_{\alpha}\big), \, \zeta \in \mathbb{C}\setminus \gamma_{1}.
\end{cases}
\end{align}
By Lemma \ref{lem:Nphihigh}, $\re \phi(\omega) < 0$ for $\omega \in \mbox{int}\big( (\gamma_{1}\setminus \Sigma_{1})\cup \Sigma_{\alpha}\big)$, so $\widetilde{\mathcal{R}}^{T,a}(\omega,\zeta)$ remains bounded as $N \to +\infty$, uniformly for $\zeta$ and $\omega$ bounded away from $r_{+}$ and $r_{-}$, as long as $\omega \in \mbox{int}(\overline{\Sigma_{1}}\cup \Sigma_{\alpha})$. Our next goal is to prove the following.
\begin{proposition} \label{prop:deformationhigh}
Let $(x,y)$ be coordinates inside the hexagon, such that $\xi:=\frac{x}{N}-1$ and $\eta:=\frac{y}{N}-1$ satisfy $(\xi,\eta) \in \mathcal L_{\alpha}$ with $\eta \leq \frac{\xi}{2} <0$. Take $\gamma_{\zeta}^{\star}$ and $\gamma_{\omega}^{\star}$ as in Proposition \ref{prop:contoursexist} if $\eta < \frac{\xi}{2}$, and as in Proposition \ref{prop:contoursexist eta = xi/2} if $\eta = \frac{\xi}{2}$ (see also Figures \ref{fig: case 1 contour} and \ref{fig: case 2 contour}).  Then the double contour integral \eqref{mathcalI} is equal to
\begin{equation} \label{eq:deformationhigh}
\mathcal I(x,y;H) =
\frac{1}{2\pi i} \int_{\overline{s}}^s H(\zeta,\zeta) d\zeta + \frac{1}{(2\pi i)^2} \int_{\gamma_{\zeta}^{\star}} d\zeta
\int_{\gamma_{\omega}^{\star}}  \frac{d\omega}{\omega-\zeta}H(\omega,\zeta)
\widetilde{\mathcal{R}}^{T,a}(\omega,\zeta) e^{2N(\Phi(\zeta;\xi,\eta)-\Phi(\omega;\xi,\eta))}.
\end{equation}
\end{proposition}
\begin{remark}\label{remark: it is analytic in zeta}
By Proposition \ref{prop:contoursexist eta = xi/2}, $\gamma_{\zeta}^{\star}$ intersects $\gamma_{1}\setminus \overline{\Sigma_{1}}$ whenever $\eta = \frac{\xi}{2}$. We do not indicate whether we take the $+$ or $-$ boundary values in the integrand of \eqref{eq:deformationhigh}. This is without ambiguity, because
\begin{align*}
\zeta \mapsto T(\zeta)\begin{pmatrix}
1 \\ 0
\end{pmatrix}e^{2N\Phi(\zeta;\xi,\eta)}
\end{align*}
has no jumps on $\gamma_{1}$ (this can be verified using \eqref{jumps for T inside support}--\eqref{jumps for T outside support}).
\end{remark}

%We aim to prove Proposition \ref{prop:doubleintegrallimit}, which states that it is sufficient to consider the case $(\xi,\eta) \in \mathcal{L}_{\alpha}\cap\{\eta \leq \frac{\xi}{2}\leq 0\}$.
%\vspace{0.2cm}In Section \ref{}, we will prove \eqref{eq:Hintegral} by performing a saddle point analysis.  
\begin{proof}
Take $\gamma_{\omega} = \gamma_{1}$ and $\gamma_{\zeta}$ lying strictly inside $\gamma_{1}$ in \eqref{mathcalI zeta omega}. From the jumps for $U$ \eqref{jump relations of U}, we have
\begin{align*}
W(\omega)\begin{pmatrix}
0 & 1
\end{pmatrix}U(\omega)^{-1} = \begin{pmatrix}
1 & 0 
\end{pmatrix}U_{-}(\omega)^{-1} - \begin{pmatrix}
1 & 0 
\end{pmatrix}U_{+}(\omega)^{-1}, \qquad \omega \in \gamma_{1}.
\end{align*}
Inserting this in \eqref{mathcalI zeta omega}, and using the $U \mapsto T$ transformation \eqref{def of T}, we get
\begin{align}
\mathcal{I}(x,y;H) & = \frac{1}{(2\pi i)^{2}}\int_{\gamma_{\zeta}}d\zeta \int_{\gamma_{\omega}=\gamma_{1}}\frac{d\omega}{\omega-\zeta} H(\omega, \zeta) \widetilde{\mathcal{R}}_{+}^{T}(\omega,\zeta) e^{2N(g(\zeta)-g_{+}(\omega))} \frac{\omega^{N}}{\zeta^{N}}q(\omega,\zeta)^{y} \tilde{q}(\omega,\zeta)^{x} \nonumber \\
& - \frac{1}{(2\pi i)^{2}}\int_{\gamma_{\zeta}}d\zeta \int_{\gamma_{\omega}=\gamma_{1}}\frac{d\omega}{\omega-\zeta} H(\omega, \zeta) \widetilde{\mathcal{R}}_{-}^{T}(\omega,\zeta) e^{2N(g(\zeta)-g_{-}(\omega))} \frac{\omega^{N}}{\zeta^{N}}q(\omega,\zeta)^{y} \tilde{q}(\omega,\zeta)^{x}, \label{mathcalI splitting}
\end{align}
where $\widetilde{\mathcal{R}}_{+}^{T}(\omega,\zeta)$ and $\widetilde{\mathcal{R}}_{-}^{T}(\omega,\zeta)$ denote the limits of $\widetilde{\mathcal{R}}^{T}(\omega',\zeta)$ as $\omega' \to \omega$ from the interior and exterior of $\gamma_{1}$, respectively. 
\begin{remark}
For $x,y \in \{1,2,\ldots,2N-1\}$, we define
\begin{align}\label{integrand in omega and zeta}
m(\omega,\zeta) = \frac{1}{\omega-\zeta} H(\omega, \zeta) \widetilde{\mathcal{R}}^{T}(\omega,\zeta) e^{2N(g(\zeta)-g(\omega))} \frac{\omega^{N}}{\zeta^{N}}q(\omega,\zeta)^{y} \tilde{q}(\omega,\zeta)^{x}.%, \quad \zeta \in \mathbb{C}, \, \omega \in \mathbb{C}\setminus \gamma_{1}.
\end{align}
The boundary values of $m$ appear in the integrand of \eqref{mathcalI splitting}. We recall that $q$ and $\tilde{q}$ are defined in \eqref{def of q and qtilde}, that $H$ satisfies the conditions stated in Proposition \ref{prop:doubleintegrallimit}, and that $g(\omega)$ is bounded for $\omega$ in compact subsets and satifies $g(\omega) \sim \log(\omega)$ as $\omega \to \infty$. Therefore, the following properties hold:
\begin{enumerate}[label=(\roman*)]
\item\label{item i} The function $\zeta \mapsto m(\omega,\zeta)$ is analytic in $\mathbb{C} \setminus \{\omega, 0, c, c^{-1}\}$,
\item \label{item ii} The function $\omega \mapsto m(\omega,\zeta)$ is analytic in $(\mathbb{C}\cup\{\infty\}) \setminus (\{\zeta, \alpha c, \alpha c^{-1}\} \cup \gamma_{1})$.
\end{enumerate}
The statement that $\omega \mapsto m(\omega,\zeta)$ is analytic at $\infty$ deserves a little computation: since $x,y \in \{1,2,\ldots,2N-1\}$, we have $m(\omega,\zeta) = \bigO(\omega^{-1-2N+N-y+x}) = \bigO(\omega^{-2})$ as $\omega \to \infty$.
\end{remark}
If $\eta < \frac{\xi}{2}$, Proposition \ref{prop:contoursexist} states that $\gamma_{\zeta}$ lies strictly inside $\gamma_{1}$, so in this case we can (and do) take $\gamma_{\zeta} = \gamma_{\zeta}^{\star}$ in \eqref{mathcalI splitting}. If $\eta = \frac{\xi}{2}$, we know from Proposition \ref{prop:contoursexist eta = xi/2} that $\gamma_{\zeta}^{\star}$ intersects $\gamma_{1}\setminus \overline{\Sigma_{1}}$. In this case, we let $\gamma_{\zeta}$ in \eqref{mathcalI splitting} tend to $\gamma_{\zeta}^{\star}$ from the interior of $\gamma_{1}$. In what follows, we will abuse notation and simply write $\gamma_{\zeta}^{\star}$. We will also omit the boundary values in the $\zeta$-variable, see Remark \ref{remark: it is analytic in zeta} (or \ref{item i}).

%, and take the $+$ boundary values in \eqref{mathcalI splitting} for $\widetilde{\mathcal{R}}^{T}(\omega,\zeta)$ and $g(\zeta)$ whenever $\zeta \in \gamma_{1}$

\vspace{0.2cm}Let us deform $\gamma_{\omega}$ from $\gamma_{1}$ to $\Sigma_{1} \cup \overline{\Sigma_{\alpha}}$ in each of the two integrals of \eqref{mathcalI splitting}. For each deformation, we pick up a residue at $\omega = \alpha c$. These residues cancel each other and we get
\begin{align*}
\mathcal{I}(x,y;H) & = \frac{1}{(2\pi i)^{2}}\int_{\gamma_{\zeta}^{\star}}d\zeta \int_{\gamma_{\omega}=\Sigma_{1} \cup \overline{\Sigma_{\alpha}}}\frac{d\omega}{\omega-\zeta} H(\omega, \zeta) \widetilde{\mathcal{R}}_{+}^{T,a}(\omega,\zeta) e^{2N(g(\zeta)-g_{+}(\omega))} \frac{\omega^{N}}{\zeta^{N}}q(\omega,\zeta)^{y} \tilde{q}(\omega,\zeta)^{x} \nonumber \\
& - \frac{1}{(2\pi i)^{2}}\int_{\gamma_{\zeta}^{\star}}d\zeta \int_{\gamma_{\omega}=\Sigma_{1} \cup \overline{\Sigma_{\alpha}}}\frac{d\omega}{\omega-\zeta} H(\omega, \zeta) \widetilde{\mathcal{R}}_{-}^{T}(\omega,\zeta) e^{2N(g(\zeta)-g_{-}(\omega))} \frac{\omega^{N}}{\zeta^{N}}q(\omega,\zeta)^{y} \tilde{q}(\omega,\zeta)^{x}. 
\end{align*}
By \ref{item ii}, the integrand of the second integral has no poles in the exterior region of $\Sigma_{1} \cup \overline{\Sigma_{\alpha}}$, so by deforming $\gamma_{\omega}$ at $\infty$, we find that this integral is $0$. Therefore, we simply get
\begin{align*}
\mathcal{I}(x,y;H) & = \frac{1}{(2\pi i)^{2}}\int_{\gamma_{\zeta}^{\star}}d\zeta \int_{\gamma_{\omega}=\Sigma_{1} \cup \overline{\Sigma_{\alpha}}}\frac{d\omega}{\omega-\zeta} H(\omega, \zeta) \widetilde{\mathcal{R}}_{+}^{T,a}(\omega,\zeta) e^{2N(g(\zeta)-g_{+}(\omega))} \frac{\omega^{N}}{\zeta^{N}}q(\omega,\zeta)^{y} \tilde{q}(\omega,\zeta)^{x}.
\end{align*}
This formula can be written in terms of $\Phi$ (see Definition \ref{def: Phi and Psi}) as follows:
\begin{align}\label{mathcalI in terms of Phi}
\mathcal{I}(x,y;H) & = \frac{1}{(2\pi i)^{2}}\int_{\gamma_{\zeta}^{\star}}d\zeta \int_{\gamma_{\omega}=\Sigma_{1} \cup \overline{\Sigma_{\alpha}}}\frac{d\omega}{\omega-\zeta} H(\omega, \zeta) \widetilde{\mathcal{R}}_{+}^{T,a}(\omega,\zeta) e^{2N(\Phi(\zeta;\xi,\eta)-\Phi_{+}(\omega;\xi,\eta))},
\end{align}
where $\xi := x/N-1$ and $\eta := y/N-1$. Finally, we deform $\gamma_{\omega}$ into $\gamma_{\omega}^{\star}$. This gives the right-most term of \eqref{eq:deformationhigh} plus a residue at $\omega = \zeta$ (by \ref{item ii}). After a small computation, we find that this residue is the first term on the right-hand-side of \eqref{eq:deformationhigh}. This finishes the proof. 
\end{proof}
\begin{proof}[Proof of Proposition \ref{prop:doubleintegrallimit} for $\eta \leq \frac{\xi}{2}< 0$]
Let $\{(x_{N},y_{N}\}_{N \geq 1}$ be a sequence satisfying \eqref{good sequence} with $(\xi,\eta) \in \mathcal{L}_{\alpha}\cap \{\eta \leq \frac{\xi}{2}< 0\}$, and define $\xi_{N} := x_{N}/N-1$ and $\eta_{N}:= y_{N}/N-1$. By \eqref{good sequence}, we have $\xi_{N} \to \xi$ and $\eta_{N} \to \eta$ as $N \to + \infty$. If $\eta = \frac{\xi}{2}$, we assume that $(\xi_{N},\eta_{N}) \in \mathcal{L}_{\alpha}\cap \{\eta \leq \frac{\xi}{2}< 0\}$ for all large enough $N$ (this is without loss of generality, see Lemma \ref{lemma:from one quadrant to the four}). Replacing $(x,y)$ in \eqref{eq:deformationhigh} by $(x_{N},y_{N})$, we get
\begin{equation}\label{lol27}
\mathcal I(x_{N},y_{N};H) -\frac{1}{2\pi i} \int_{\overline{s_{N}}}^{s_{N}} H(\zeta,\zeta) d\zeta =
 \frac{1}{(2\pi i)^2} \int_{\gamma_{\zeta}^{\star}} d\zeta
\int_{\gamma_{\omega}^{\star}}  \frac{d\omega}{\omega-\zeta}H(\omega,\zeta)
\widetilde{\mathcal{R}}^{T,a}(\omega,\zeta) e^{2N(\Phi_{N}(\zeta)-\Phi_{N}(\omega))},
\end{equation}
where $s_{N} = s(\xi_{N},\eta_{N};\alpha)$, $\Phi_{N}(\zeta) := \Phi(\zeta;\xi_{N},\eta_{N})$ and the contours $\gamma_{\zeta}^{\star}$ and $\gamma_{\omega}^{\star}$ also depend on $N$, even though this is not indicated in the notation. Since $\gamma_{\zeta}^{\star}$ and $\gamma_{\omega}^{\star}$ do not pass through $r_{+}$ and $r_{-}$, Proposition \ref{prop:TandTinvsmall} implies that
\begin{align*}
\widetilde{\mathcal{R}}^{T,a}(\omega,\zeta) = \bigO(1), \qquad \mbox{as } N \to + \infty \; \mbox{ uniformly for all } \zeta \in \gamma_{\zeta}^{\star} \mbox{ and } \omega \in \gamma_{\omega}^{\star}.
\end{align*}
We also know from Propositions \ref{prop:contoursexist} and \ref{prop:contoursexist eta = xi/2} that
\begin{align*}
\re \Phi_{N}(\zeta) < \re \Phi_{N}(s_{N}) < \re \Phi_{N}(\omega), \qquad \mbox{for all } \zeta \in \gamma_{\zeta}^{\star}\setminus \{s_{N},\overline{s_{N}}\},\omega \in \gamma_{\omega}^{\star}\setminus \{s_{N},\overline{s_{N}}\},
\end{align*}
which implies that the right-hand-side of \eqref{lol27} is
\begin{align}\label{lol28}
\frac{1}{(2\pi i)^2} \hspace{-0.05cm} \int_{\gamma_{\zeta}^{\star} \cap D_{\epsilon}} \hspace{-0.2cm} d\zeta
\hspace{-0.05cm} \int_{\gamma_{\omega}^{\star}\cap D_{\epsilon}}  \frac{d\omega}{\omega-\zeta}H(\omega,\zeta)
\widetilde{\mathcal{R}}^{T,a}(\omega,\zeta) e^{2N(\Phi_{N}(\zeta)-\Phi_{N}(\omega))} + \bigO(e^{-C_{1}N}), \quad \mbox{as } N \to \infty,
\end{align}
for a certain $C_{1}>0$, and where $D_{\epsilon}$ is the union of two small disks of radii $\epsilon > 0$ surrounding $s$ and $\overline{s}$. Since $s_{N}$ and $\overline{s_{N}}$ are simple zeros of $\Phi_{N}'$, we have the estimates
\begin{align*}
& \re(\Phi_{N}(\zeta)-\Phi_{N}(s_{N})) < -C_{2}|\zeta - s_{N}|^{2}, & & \mbox{for } \zeta \in \gamma_{\zeta}^{\star}\setminus \{s_{N},\overline{s_{N}}\}, \\
& \re(\Phi_{N}(\omega)-\Phi_{N}(s_{N})) \geq C_{2}|\omega - s_{N}|^{2}, & & \mbox{for } \omega \in \gamma_{\omega}^{\star}\setminus \{s_{N},\overline{s_{N}}\},
\end{align*}
for a certain $C_{2}>0$. Therefore, the left-most term in \eqref{lol28} is, in absolute value,
\begin{align}
\leq C_{3}\iint_{|x|^{2}+|y|^{2} \leq \epsilon^{2}}\frac{e^{-4C_{2}N(x^{2}+y^{2})}}{\sqrt{x^{2}+y^{2}}}dxdy = 2\pi C_{3} \int_{0}^{\epsilon}e^{-4C_{2}Nr^{2}}dr \leq C_{4} N^{-\frac{1}{2}} \label{lol32}
\end{align}
for certain $C_{3},C_{4}>0$ and for all large enough $N$. Therefore,
\begin{align*}
\mathcal I(x_{N},y_{N};H) -\frac{1}{2\pi i} \int_{\overline{s_{N}}}^{s_{N}} H(\zeta,\zeta) d\zeta = \bigO(N^{-1/2}), \qquad \mbox{as } N \to + \infty,
\end{align*}
which give \eqref{eq:Hintegral}.
\end{proof}

\subsection{The case $\xi = 0$ and $\eta < 0$}
Let us briefly recall first the situation for $(\xi',\eta') \in \mathcal{L}$, such that $\eta' < \frac{\xi'}{2}<0$. In this case, the set $\mathcal{N}_{\Phi}$ contains four curves emanating from $s$: three of these curves, namely $\Gamma_{1}$, $\Gamma_{2}$ and $\Gamma_{3}$, lie in $\mbox{int}(\overline{\Sigma_{1}\cup \Sigma_{\alpha}})$, the other curve $\Gamma_{4}$ intersects once $\overline{\Sigma_{1}\cup \Sigma_{\alpha}}$. Denote $p_{j}$ for the intersection of $\Gamma_{j}$ with $\mathbb{R}$, and recall that the ordering for $\Gamma_{1}$, $\Gamma_{2}$ and $\Gamma_{3}$ is such that $p_{1}<p_{2}<p_{3}$. 

\vspace{0.2cm}As $(\xi',\eta') \to (0,\eta)$ with $\eta < 0$ (see Figures \ref{fig: cases 1 and 2 for Phi} and \ref{fig: cases xi=0 for Phi and Psi} (left)), we know from Proposition \ref{prop:hightemp} that $s(\xi',\eta';\alpha)$ tends to a point $s=s(0,\eta;\alpha)$ lying on $\Sigma_{1}$. In this limit, both $\Gamma_{3}$ and $\Gamma_{4}$ tend to the arc
\begin{align*}
\Sigma_{s}:= \{c^{-1} + R_{1}e^{it} : -\arg s \leq t \leq \arg s\} \subset \Sigma_{1},
\end{align*}
and a part of $\Gamma_{1}$ tends to $\Sigma_{1}\setminus \Sigma_{s}$. Thus, the case $\xi = 0$ gives less freedom for the contour deformations and the saddle point analysis is more involved. To handle this case, we need information about both $\mathcal{N}_{\Phi}$ and $\mathcal{N}_{\Psi}$, where
\begin{align*}
\mathcal{N}_{\Psi} := \{\zeta \in \mathbb{C} : \re \Psi(\zeta) = \Psi(s)\}.
\end{align*}
The sets $\mathcal{N}_{\Phi}$ and $\mathcal{N}_{\Psi}$ are represented in Figure \ref{fig: cases xi=0 for Phi and Psi} for a particular choice of the parameters. We have the following.
\begin{lemma}
For $\xi = 0$, we have $\Sigma_{1} \subset \mathcal{N}_{\Phi}$ and $\Sigma_{1} \subset \mathcal{N}_{\Psi}$
\end{lemma}
\begin{proof}
Since $\re \phi(\zeta) = 0$ for $\zeta \in \Sigma_{1}$, by Definition \ref{def: Phi and Psi} and \eqref{eq:RePhionSigma} we have
\begin{align*}
\re \Phi(\zeta) = \re \Psi(\zeta) = -\frac{\eta}{2} \re f_{1}(\zeta),
\end{align*}
and by Lemma \ref{lem:LogsIncSigmaMinus1} this expression is constant for $\zeta \in \Sigma_{1}$.
\end{proof}
We choose $\gamma_{\zeta}^{\star}$ and $\gamma_{\omega}^{\star}$ as follows (see also Figure \ref{fig: cases xi=0 for Phi and Psi contour}):
\begin{itemize}
\item $\gamma_{\omega}^{\star} \subset \overline{\mbox{int}(\overline{\Sigma_\alpha \cup \Sigma_{1}})}$, is such that $(\Sigma_{1}\setminus \Sigma_{s}) \subset \gamma_{\omega}^{\star}$,  surrounds $\alpha c$ and $\alpha c^{-1}$, and it satisfies
\begin{align*}
\re \Phi(\omega) > \re \Phi(s), \qquad \omega \in \gamma_{\omega}^{\star} \setminus (\overline{\Sigma_{1}\setminus\Sigma_{s}}),
\end{align*}
\item $\gamma_\zeta^{\star} \subset \overline{\mbox{int}(\gamma_1)}$, is such that $\Sigma_{s} \subset \gamma_\zeta^{\star}$, surrounds $c$ and $c^{-1}$, and it satisfies
\begin{align*}
\re \Phi(\zeta) < \re \Phi(s), \qquad \zeta \in \gamma_{\zeta}^{\star} 	\setminus \Sigma_{s}.
\end{align*}
\end{itemize}

\begin{figure}
\begin{center}
\begin{tikzpicture}[master,scale = 0.5]
\node at (0,0) {};
\node at (2.1,0) {\color{black} \large $\bullet$};
\node at (-0.38,0) {\color{black} \large $\bullet$};
\node at (-1.05,0) {\color{black} \large $\bullet$};
\node at (-2.035,0) {\color{black} \large $\bullet$};
\node at (-3.13,0) {\color{black} \large $\bullet$};

% circle gamma1
\draw[blue,line width=0.65 mm] ([shift=(-156:3.6cm)]2.1,0) arc (-156:156:3.6cm);

% circle gamma alpha
\draw[red,line width=0.65 mm] ([shift=(95:1.45cm)]-1.05,0) arc (95:265:1.45cm);

\node at ($(2.1,0)+(118:3.6)$) {\color{blue} \large $\bullet$};
\node at ($(2.1,0)+(-118:3.6)$) {\color{blue} \large $\bullet$};

\draw[blue, line width=0.65 mm] ($(2.1,0)+(118:3.6)$) to [out=-62, in=90] (1.05,0) to [out=-90, in=62] ($(2.1,0)+(-118:3.6)$);
\draw[blue, line width=0.65 mm] ($(2.1,0)+(156:3.6)$) to [out=-135, in=90] (-2.15,0) to [out=-90, in=135] ($(2.1,0)+(-156:3.6)$);
\draw[blue,line width=0.65 mm] ([shift=(-180:0.45cm)]-0.28,0) arc (-180:180:0.45cm);
\node at (-0.1,0) {$+$};
\node at (1.7,2.5) {$+$};
\node at (0,2.2) {$-$};
\node at (1,4) {$+$};
\end{tikzpicture}\hspace{3cm}\begin{tikzpicture}[slave,scale = 0.5]
\node at (0,0) {};
\node at (2.1,0) {\color{black} \large $\bullet$};
\node at (-0.38,0) {\color{black} \large $\bullet$};
\node at (-1.05,0) {\color{black} \large $\bullet$};
\node at (-2.035,0) {\color{black} \large $\bullet$};
\node at (-3.13,0) {\color{black} \large $\bullet$};

% circle gamma1
\draw[blue,line width=0.65 mm] ([shift=(-156:3.6cm)]2.1,0) arc (-156:156:3.6cm);

% circle gamma alpha
\draw[red,line width=0.65 mm] ([shift=(95:1.45cm)]-1.05,0) arc (95:265:1.45cm);

\node at ($(2.1,0)+(118:3.6)$) {\color{blue} \large $\bullet$};
\node at ($(2.1,0)+(-118:3.6)$) {\color{blue} \large $\bullet$};

\draw[blue, line width=0.65 mm] ($(2.1,0)+(118:3.6)$) to [out=118, in=0] (-3,5) to [out=180,in=90] (-7,0) to [out=-90, in=180] (-3,-5) to [out=0, in=-118] ($(2.1,0)+(-118:3.6)$);
\draw[blue, line width=0.65 mm] ($(2.1,0)+(156:3.6)$) to [out=-80, in=90] (-0.8,0) to [out=-90, in=80] ($(2.1,0)+(-156:3.6)$);
\draw[blue,line width=0.65 mm] ([shift=(-180:0.45cm)]-3.3,0) arc (-180:180:0.45cm);
\node at (-3.4,0) {$-$};
\node at (1.7,2.5) {$-$};
\node at (-0.8,3) {$+$};
\node at (1,4) {$-$};

\draw[dashed] (-8.5,-5)--(-8.5,5);
\end{tikzpicture}
\end{center}
\caption{\label{fig: cases xi=0 for Phi and Psi}
The sets $\mathcal{N}_{\Phi}$ (left) and $\mathcal{N}_{\Psi}$ (right) are represented in blue, and $\Sigma_{\alpha}$ in red. The parameters are $(\alpha,\xi,\eta) = (0.4,0,-0.75)$, and $\Sigma_{1}$ is a subset of both $\mathcal{N}_{\Phi}$ and $\mathcal{N}_{\Psi}$. The signs of $\re (\Phi-\Phi(s))$ and $\re (\Psi-\Psi(s))$ are indicated with $\pm$. In each figure, the black dots represent $0$, $\alpha c$, $\alpha c^{-1}$, $c$ and $c^{-1}$ and the blue dots are $s$ and $\overline{s}$.}
\end{figure}
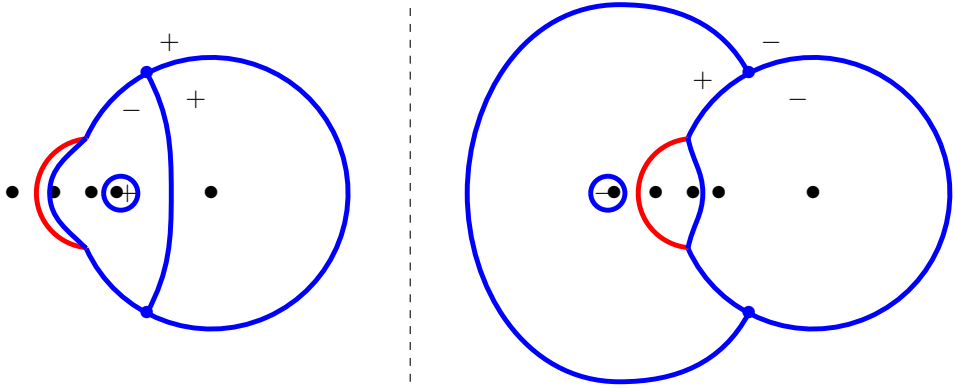
Let $\{(x_{N},y_{N}\}_{N \geq 1}$ be a sequence satisfying \eqref{good sequence} with $(\xi,\eta) \in \mathcal{L}_{\alpha}\cap \{\eta < \frac{\xi}{2}= 0\}$, and define $\xi_{N} := x_{N}/N-1$ and $\eta_{N}:= y_{N}/N-1$. By Lemma \ref{lemma:from one quadrant to the four}, we can (and do) assume without loss of generality that $(\xi_{N},\eta_{N}) \in \mathcal{L}_{\alpha}$ satisfies $\eta_{N}<0$ and $\xi_{N} = 0$ for all $N$. The proof of Proposition \ref{prop:deformationhigh} still goes through with the above choice of $\gamma_{\zeta}^{\star}$ and $\gamma_{\omega}^{\star}$, and as in \eqref{lol27} we obtain
\begin{equation}\label{lol29}
\mathcal I(x_{N},y_{N};H) -\frac{1}{2\pi i} \int_{\overline{s}}^{s} H(\zeta,\zeta) d\zeta =
 \frac{1}{(2\pi i)^2} \int_{\gamma_{\zeta}^{\star}} d\zeta
\int_{\gamma_{\omega}^{\star}}  \frac{d\omega}{\omega-\zeta}H(\omega,\zeta)
\widetilde{\mathcal{R}}^{T,a}(\omega,\zeta) e^{2N(\Phi(\zeta)-\Phi(\omega))},
\end{equation}
where $s = s(\xi_{N},\eta_{N};\alpha)$, $\Phi(\zeta) = \Phi(\zeta;\xi_{N},\eta_{N})$, and the contours $\gamma_{\zeta}^{\star}$ and $\gamma_{\omega}^{\star}$ depend on $N$. We also take the $+$ boundary value in \eqref{lol29} whenever $\omega \in \gamma_{1}$. Since $\re \Phi(\zeta) = \re \Phi(s)$ for all $\zeta \in \Sigma_{s}$ and $\re \Phi(\omega) = \re \Phi(s)$ for all $\omega \in \overline{\Sigma_{1}\setminus \Sigma_{s}}$, we need additional deformation of contours. 

\vspace{0.2cm}We first treat the contour deformations in the $\zeta$-variable. Recall the definition \eqref{def of RTa} of $\widetilde{\mathcal{R}}^{T,a}$. For $\zeta \in \Sigma_{s}$, we use $\Phi_{+}(\zeta) = \Psi_{-}(\zeta)$ and the jumps for $T$ \eqref{jumps for T inside support} to obtain
\begin{align}\label{lol30}
e^{2N\Phi(\zeta)}T(\zeta) \begin{pmatrix}
1 \\ 0
\end{pmatrix} = e^{2N(\Phi_{+}(\zeta)-2\phi_{+}(\zeta))}T_{+}(\zeta) \begin{pmatrix}
0 \\ 1
\end{pmatrix} - e^{2N\Psi_{-}(\zeta)}T_{-}(\zeta)\begin{pmatrix}
0 \\ 1
\end{pmatrix}.
\end{align}
We substitute \eqref{lol30} in \eqref{lol29}, and then split the integral over $\Sigma_{s}\subset \gamma_{\zeta}$ in \eqref{lol29} into two parts. For the second term in \eqref{lol30}, the contour $\Sigma_{s}$ is deformed outwards to $\Sigma_{s,\mathrm{out}}$, see Figure \ref{fig: cases xi=0 for Phi and Psi further deformations}. Because $\Psi_{\pm}(\zeta)=\Phi_{\mp}(\zeta)$ for $\zeta \in \Sigma_{1}$, and since $\Sigma_{1}\subset \mathcal{N}_{\Phi}\cap \mathcal{N}_{\psi}$, the signs of 
\begin{align*}
\re \big(\Phi(\zeta + \epsilon (\zeta-c^{-1}))-\Phi(s)\big) \qquad \mbox{ and } \qquad \re \big(\Psi(\zeta-\epsilon (\zeta-c^{-1})))-\Psi(s)\big)
\end{align*}
are different for all $\zeta \in \Sigma_{1}$, provided $\epsilon = \epsilon(\zeta) \in \mathbb{R}$ is small enough ($\epsilon$ non necessarily positive), see also the signs around $\Sigma_{1}$ in Figure \ref{fig: cases xi=0 for Phi and Psi}. In particular, we have $\re \Psi(\zeta) < \re \Psi(s)$ for $\zeta \in \Sigma_{s,\mathrm{out}}$. For the first term in \eqref{lol30}, the dominant part is $e^{2N(\Phi_{+}(\zeta)-2 \phi_{+}(\zeta))}$, and by Definition \ref{def: Phi and Psi}, we have $\Psi = \Phi-2\phi$. Therefore, we deform $\Sigma_{s}$ inwards to $\Sigma_{s,\mathrm{in}}$, and this contour is chosen such that $\re \Psi(\zeta) < \re \Psi(s)$ for $\zeta \in \Sigma_{s,\mathrm{in}}$, see Figure \ref{fig: cases xi=0 for Phi and Psi further deformations}.

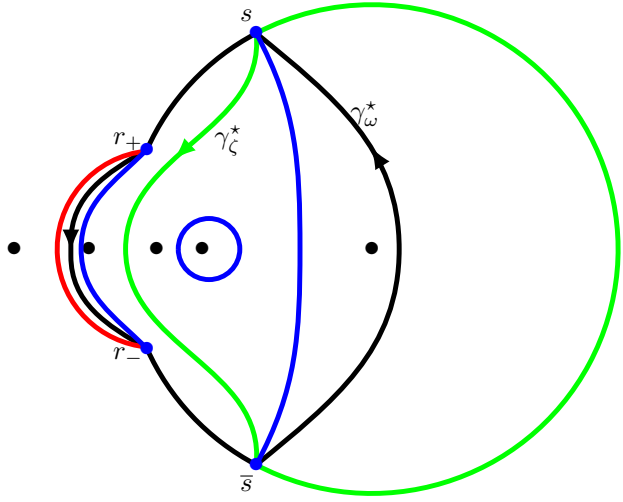
\begin{figure}
\begin{center}
\begin{tikzpicture}[master,scale = 0.9]
\node at (0,0) {};
\node at (2.1,0) {\color{black} \large $\bullet$};
\node at (-0.38,0) {\color{black} \large $\bullet$};
\node at (-1.05,0) {\color{black} \large $\bullet$};
\node at (-2.035,0) {\color{black} \large $\bullet$};
\node at (-3.13,0) {\color{black} \large $\bullet$};

% circle gamma1
\draw[green,line width=0.65 mm] ([shift=(-118:3.6cm)]2.1,0) arc (-118:118:3.6cm);
\draw[green,->-=0.30, line width=0.65 mm] ($(2.1,0)+(118:3.6)$) to [out=-84.5, in=90] (-1.5,0) to [out=-90, in=84.5] ($(2.1,0)+(-118:3.6)$);
\node at (-0,1.6) {$\gamma_{\zeta}^{\star}$};
\draw[black,line width=0.65 mm] ([shift=(-156:3.6cm)]2.1,0) arc (-156:-118:3.6cm);
\draw[black,line width=0.65 mm] ([shift=(118:3.6cm)]2.1,0) arc (118:156:3.6cm);
\draw[black,-<-=0.30, line width=0.65 mm] ($(2.1,0)+(118:3.6)$) to [out=-39.5, in=90] (2.5,0) to [out=-90, in=39.5] ($(2.1,0)+(-118:3.6)$);
\draw[black,->-=0.51, line width=0.65 mm] ($(2.1,0)+(156:3.6)$) to [out=-150, in=90] (-2.3,0.1) to [out=-90, in=90] (-2.3,0) to [out=-90, in=90] (-2.3,-0.1) to [out=-90, in=150] ($(2.1,0)+(-156:3.6)$);
\node at (2,2) {$\gamma_{\omega}^{\star}$};

% circle gamma alpha
\draw[red,line width=0.65 mm] ([shift=(95:1.45cm)]-1.05,0) arc (95:265:1.45cm);

\node at ($(2.1,0)+(118:3.6)$) {\color{blue} \large $\bullet$};
\node at ($(2.1,0)+(-118:3.6)$) {\color{blue} \large $\bullet$};
\node at ($(2.1,0)+(156:3.6)$) {\color{blue} \large $\bullet$};
\node at ($(2.1,0)+(-156:3.6)$) {\color{blue} \large $\bullet$};
\node at ($(2.1,0)+(118:3.9)$) {$s$};
\node at ($(2.1,0)+(-118:3.9)$) {$\overline{s}$};
\node at ($(2.1,0)+(156:3.9)$) {$r_{+}$};
\node at ($(2.1,0)+(-156:3.9)$) {$r_{-}$};

\draw[blue, line width=0.65 mm] ($(2.1,0)+(118:3.6)$) to [out=-62, in=90] (1.05,0) to [out=-90, in=62] ($(2.1,0)+(-118:3.6)$);
\draw[blue, line width=0.65 mm] ($(2.1,0)+(156:3.6)$) to [out=-135, in=90] (-2.15,0) to [out=-90, in=135] ($(2.1,0)+(-156:3.6)$);
\draw[blue,line width=0.65 mm] ([shift=(-180:0.45cm)]-0.28,0) arc (-180:180:0.45cm);
%\node at (-0.1,0) {$+$};
%\node at (1.7,2.5) {$+$};
%\node at (0,2.2) {$-$};
%\node at (1,4) {$+$};
\end{tikzpicture}
\end{center}
\caption{\label{fig: cases xi=0 for Phi and Psi contour}
The contours $\gamma_{\zeta}^{\star}$ (green) and $\gamma_{\omega}^{\star}$ (black), with $(\alpha,\xi,\eta) = (0.4,0,-0.75)$.}
\end{figure}

\vspace{0.2cm}In the $\omega$-variable, we simply analytically continue the integrand and deform $\Sigma_{1}\setminus \Sigma_{s}$ outwards to $(\gamma_{1}\setminus \Sigma_{s})_{\mathrm{ext}}$, see Figure \ref{fig: cases xi=0 for Phi and Psi further deformations}. This contour is chosen such that $\re \Psi(\omega) > \re \Psi(s)$ for $\omega \in (\gamma_{1}\setminus \Sigma_{s})_{\mathrm{ext}}$. Since $\Phi_{+}(\omega) = \Psi_{-}(\omega)$ on $\Sigma_1$, the exponential factor of the integrand is $e^{-2N\Psi(\omega)}$ there. Also, for $\omega \in \Sigma_{1}\setminus \Sigma_{s}$, by \eqref{jumps for T inside support} we have
\begin{align}\label{lol31}
\begin{pmatrix} 1 & 0 \end{pmatrix} T_{+}^{-1}(\omega) = \begin{pmatrix} e^{-4N \phi_-(\omega)} & -1 \end{pmatrix} T_{-1}^{-1}(\omega),
\end{align} 
and we know from Lemma \ref{lem:Nphihigh} that $e^{-4N \phi(\omega)}$ remains bounded for $\omega \in (\gamma_{1}\setminus \Sigma_{s})_{\mathrm{ext}}$. 

\vspace{0.2cm}The result of the above deformations is that the integrand is uniformally exponentially small on the contours, as long as $\zeta$ stays away from $s,\overline{s}$, and that $\omega$ stays away from $s,\overline{s},r_{+},r_{-}$. By a similar analysis as the one done in \eqref{lol32}, we show that the contribution to \eqref{lol29} when $\zeta$ and $\omega$ are close to $s$ or $\overline{s}$ is $\bigO(N^{-\frac{1}{2}})$ as $N \to + \infty$. When $\omega$ is close to $r_{\pm}$, we know by Proposition \ref{prop:TandTinvsmall} that $T^{-1}(\omega) = \bigO(N^{1/6})$. Since $\Phi'(r_{\pm}) \neq 0 \neq \Psi'(r_{\pm})$, the contribution to \eqref{lol29} when $\zeta$ is close to $s$ or $\overline{s}$ and simultaneously $\omega$ close to $r_{+}$ or $r_{-}$ is
\begin{align*}
\leq C_{1}N^{\frac{1}{6}}\iint_{|x|^{2}+|y|^{2} \leq \epsilon^{2}} e^{-C_{2}N(|x|+y^{2})}dxdy \leq C_{3} N^{-\frac{17}{6}}
\end{align*}
for certain constant $C_{1},C_{2},C_{3}>0$ and all large enough $N$. In particular, this proves \eqref{eq:Hintegral}.

\subsection{The case $\xi = 0$ and $\eta = 0$}
At the center of the hexagon, we have $s = s(0,0;\alpha) = r_{+}$, $\overline{s} = r_{-}$, and $\Phi = -\Psi = \phi$ (see also Definition \ref{def: Phi and Psi}). The sets $\mathcal{N}_{\Phi}$ and $\mathcal{N}_{\Psi}$ are then given by Lemma \ref{lem:Nphihigh}:
\begin{align*}
\mathcal{N}_{\Phi} = \mathcal{N}_{\Psi} = \mathcal{N}_{\phi} = \Sigma_{0} \cup \Sigma_{\alpha} \cup \Sigma_{1}.
\end{align*}
Note that for $(\xi',\eta')=(0,\eta') \in \mathcal{L}_{\alpha}$ with $\eta'<0$, part of contour $\gamma_{\omega}^{\star}$ lies in the region $\mbox{int}(\overline{\Sigma_{\alpha} \cup \Gamma_{1}})$, see Figures \ref{fig: cases xi=0 for Phi and Psi} and \ref{fig: cases xi=0 for Phi and Psi further deformations}. As $\eta'\to \eta = 0$, $\Gamma_{1}$ tends to $\Sigma_{\alpha}$, so we need additional contour deformations to handle this case.
Consider the contours $\gamma_{\zeta}^{\star}:= \gamma_{1}$ and $\gamma_{\omega}^{\star} = \gamma_{\alpha}$. By Lemma \ref{lem:Nphihigh}, we have
\begin{align*} 
& \re \Phi(\omega) > 0,  \quad \mbox{for } \omega \in \gamma_\alpha \setminus \overline{\Sigma_{\alpha}} & & \re \Phi(\omega) = 0, \quad \mbox{for }  \omega \in \Sigma_{\alpha}, \\
& \re \Phi(\zeta) < 0, \quad \mbox{for } \zeta \in \gamma_{1} \setminus \overline{\Sigma_{1}} & & \re \Phi(\zeta) = 0, \quad \mbox{for } \zeta \in \Sigma_1.
\end{align*}

\begin{figure}
\begin{center}
\begin{tikzpicture}[master,scale = 0.9]
\node at (0,0) {};
\node at (2.1,0) {\color{black} \large $\bullet$};
\node at (-0.38,0) {\color{black} \large $\bullet$};
\node at (-1.05,0) {\color{black} \large $\bullet$};
\node at (-2.035,0) {\color{black} \large $\bullet$};
\node at (-3.13,0) {\color{black} \large $\bullet$};

% circle gamma1
\draw[blue,line width=0.65 mm] ([shift=(-118:3.6cm)]2.1,0) arc (-118:118:3.6cm);
\draw[green,->-=0.30, line width=0.65 mm] ($(2.1,0)+(118:3.6)$) to [out=-84.5, in=90] (-0.9,0) to [out=-90, in=84.5] ($(2.1,0)+(-118:3.6)$);
\draw[green,-<-=0.5,line width=0.65 mm] ($(2.1,0)+(118:3.6)$) to [out=118-45, in=90] (6.5,0) to [out=-90, in=-118+45] ($(2.1,0)+(-118:3.6)$);
\node at (7,1.6) {$\Sigma_{s,\mathrm{out}}$};
\draw[green,-<-=0.5,line width=0.65 mm] ($(2.1,0)+(118:3.6)$) to [out=118-45-60, in=90] (5,0) to [out=-90, in=-118+45+60] ($(2.1,0)+(-118:3.6)$);
\node at (4,1.6) {$\Sigma_{s,\mathrm{in}}$};
\draw[blue,line width=0.65 mm] ([shift=(-156:3.6cm)]2.1,0) arc (-156:-118:3.6cm);
\draw[blue,line width=0.65 mm] ([shift=(118:3.6cm)]2.1,0) arc (118:156:3.6cm);
\draw[black,->-=0.50, line width=0.65 mm] ($(2.1,0)+(118:3.6)$) to [out=118+45, in=-156-75] ($(2.1,0)+(156:3.6)$);
%\draw[black,->-=0.50, line width=0.65 mm] ($(2.1,0)+(118:3.6)$) to [out=118+45+80, in=-156-75-120] ($(2.1,0)+(156:3.6)$);
\draw[black,-<-=0.50, line width=0.65 mm] ($(2.1,0)+(-118:3.6)$) to [out=-118-45, in=156+75] ($(2.1,0)+(-156:3.6)$);
%\draw[black,-<-=0.40, line width=0.65 mm] ($(2.1,0)+(-118:3.6)$) to [out=-118-45-80, in=156+75+120] ($(2.1,0)+(-156:3.6)$);
\draw[black,-<-=0.30, line width=0.65 mm] ($(2.1,0)+(118:3.6)$) to [out=-39.5, in=90] (2.5,0) to [out=-90, in=39.5] ($(2.1,0)+(-118:3.6)$);
\draw[black,->-=0.51, line width=0.65 mm] ($(2.1,0)+(156:3.6)$) to [out=-150, in=90] (-2.3,0.1) to [out=-90, in=90] (-2.3,0) to [out=-90, in=90] (-2.3,-0.1) to [out=-90, in=150] ($(2.1,0)+(-156:3.6)$);
\node at (-2.2,2.6) {$(\gamma_{1}\setminus \Sigma_{s})_{\mathrm{out}}$};
%\node at (-3.3,1.8) {$(\gamma_{1}\setminus \Sigma_{s})_{\mathrm{in}}$};
%\draw[->-=1] (-2.5,1.9)--(-0.45,1.9);

\node at (-2.2,-2.6) {$(\gamma_{1}\setminus \Sigma_{s})_{\mathrm{out}}$};
%\node at (-3.3,-1.8) {$(\gamma_{1}\setminus \Sigma_{s})_{\mathrm{in}}$};
%\draw[->-=1] (-2.4,-1.9)--(-0.5,-1.9);

% circle gamma alpha
\draw[red,line width=0.65 mm] ([shift=(95:1.45cm)]-1.05,0) arc (95:265:1.45cm);

\node at ($(2.1,0)+(118:3.6)$) {\color{blue} \large $\bullet$};
\node at ($(2.1,0)+(-118:3.6)$) {\color{blue} \large $\bullet$};
\node at ($(2.1,0)+(156:3.6)$) {\color{blue} \large $\bullet$};
\node at ($(2.1,0)+(-156:3.6)$) {\color{blue} \large $\bullet$};
\node at ($(2.1,0)+(118:3.9)$) {$s$};
\node at ($(2.1,0)+(-118:3.9)$) {$\overline{s}$};
\node at ($(2.1,0)+(156:3.9)$) {$r_{+}$};
\node at ($(2.1,0)+(-156:3.9)$) {$r_{-}$};

\draw[blue, line width=0.65 mm] ($(2.1,0)+(118:3.6)$) to [out=-62, in=90] (1.05,0) to [out=-90, in=62] ($(2.1,0)+(-118:3.6)$);
\draw[blue, line width=0.65 mm] ($(2.1,0)+(156:3.6)$) to [out=-135, in=90] (-2.15,0) to [out=-90, in=135] ($(2.1,0)+(-156:3.6)$);
\draw[blue,line width=0.65 mm] ([shift=(-180:0.45cm)]-0.28,0) arc (-180:180:0.45cm);
\end{tikzpicture}
\end{center}
\vspace{-2cm}
\caption{\label{fig: cases xi=0 for Phi and Psi further deformations}
The further contour deformations that we need to consider to handle the case $\xi = 0$ and $\eta < 0$.}
\end{figure}
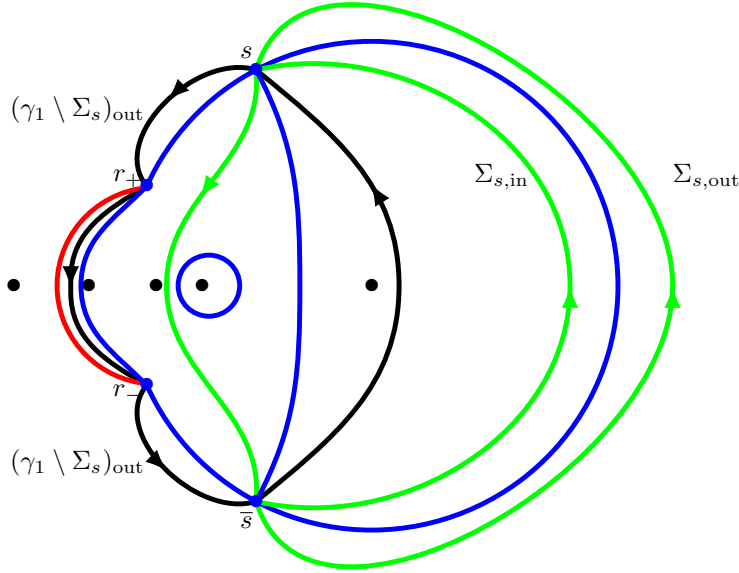
For simplicity, we consider the sequence $\{(x_{N},y_{N})=(N,N)\}_{N \geq 1}$, so that $\xi_{N} := x_{N}/N-1 =0$ and $\eta_{N}:= y_{N}/N-1 =0$ for all $N$. In the same way as done in Proposition \ref{prop:deformationhigh}, we find
\begin{equation}\label{lol33}
\mathcal I(x_{N},y_{N};H) -\frac{1}{2\pi i} \int_{r_{-}}^{r_{+}} H(\zeta,\zeta) d\zeta =
 \frac{1}{(2\pi i)^2} \int_{\gamma_{\zeta}^{\star}} d\zeta
\int_{\gamma_{\omega}^{\star}}  \frac{d\omega}{\omega-\zeta}H(\omega,\zeta)
\widetilde{\mathcal{R}}^{T,a}(\omega,\zeta) e^{2N(\Phi(\zeta)-\Phi(\omega))},
\end{equation}
and we take the $+$ boundary value in \eqref{lol33} whenever $\omega \in \Sigma_{\alpha}$.

\vspace{0.2cm}For $\zeta \in \Sigma_{1}$, we use \eqref{lol30} to split the integrand into two parts, and again we deform the integral associated to the first term slightly inwards, and the other one slightly outwards. As a result, both deformed integrals have exponentially decaying integrands.

\vspace{0.2cm}For $\omega \in \Sigma_{\alpha}$, $\widetilde{\mathcal{R}}^{T,a}(\omega,\zeta)$ is given by the second line of \eqref{def of RTa}, and thus the dominant $\omega$-part in the integrand is
\begin{align*}
e^{-2N \Phi(\omega)} \begin{pmatrix} 1 & - e^{4N \phi(\omega)} \end{pmatrix} T^{-1}(\omega) = e^{-2N \phi(\omega)} \begin{pmatrix} 1 & 0 \end{pmatrix} T^{-1}(\omega) - e^{2N \phi(\omega)} \begin{pmatrix} 0 & 1 \end{pmatrix} T^{-1}(\omega).
\end{align*}
For the first term, we deform $\Sigma_{\alpha}$ outwards so that $\re \phi(\omega) > 0$, and for the first term, we deform $\Sigma_{\alpha}$ inwards so that $\re \phi(\omega) < 0$. 

\vspace{0.2cm}On the deformed contours, the integrand is uniformly exponentially small, as long as $\zeta$ and $\omega$ are bounded away from $r_{+}$ and $r_{-}$. For $\zeta$ and $\omega$ close to $r_{\pm}$, by Proposition \ref{prop:TandTinvsmall} we have $T(\zeta) = \bigO(N^{1/6})$ and $T^{-1}(\omega) = \bigO(N^{1/6})$. The contribution to \eqref{lol33} when $\zeta$ and $\omega$ are close to $r_{+}$ and $r_{-}$ is thus bounded by
\begin{align*}
\leq C_{1}N^{\frac{1}{3}}\iint_{|x|^{2}+|y|^{2} \leq \epsilon^{2}}\frac{e^{-C_{2}N(x^{2}+y^{2})}}{\sqrt{x^{2}+y^{2}}}dxdy \leq C_{3} N^{-\frac{1}{6}}
\end{align*}
for certain $C_{1},C_{2},C_{3}>0$ and for all large enough $N$. This finishes the proof of Proposition \ref{prop:doubleintegrallimit}.

\appendix


\small
\begin{thebibliography}{99}


%\bibitem{AFvM} M. Adler, P.L. Ferrari, and P. van Moerbeke, Non-intersecting random walks in the neighborhood of a symmetric tacnode, \emph{Ann. of Prob}. {\bf 41} (2013), 2599--2647.
	
%\bibitem{AJvM} M. Adler, K. Johansson, and P. van Moerbeke, Double Aztec diamonds and the tacnode process, \emph{Adv. Math.} {\bf 252} (2014), 518--571. 
	
\bibitem{AvMJ} M. Adler, K. Johansson and P. van Moerbeke, Tilings of non-convex polygons, skew-Young tableaux and determinantal processes,  \emph{Comm.  Math. Phys.}  {\bf 364} (2018),  287--342.

\bibitem{Aggarwal} A. Aggarwal, Universality for Lozenge Tiling Local Statistics, arXiv:1907.09991.
 	
%\bibitem{AOvM} M. Adler, N. Orantin, and P. van Moerbeke, Universality for the Pearcey process, \emph{Physica  D} {\bf 239} (2010), 924--941.
 	
%\bibitem{Ahn} A. Ahn, Global universality of Macdonald plane partitions, preprint arXiv:1809.02698.
	
\bibitem{BKMM} J. Baik, T. Kriecherbauer, K. T.-R. McLaughlin, and P. D. Miller, Discrete Orthogonal Polynomials: Asymptotics and Applications, \emph{Annals of Math. Studies} {\bf 164}, Princeton University Press, Princeton, NJ, 2007.
	
\bibitem{BCJ} V. Beffara, S. Chhita, and K. Johansson, Airy point process at the liquid-gas boundary, \emph{Ann. Probab.} {\bf 46} (2018), 2973--3013. 
	
\bibitem{Berggren} T. Berggren, Domino tilings of the Aztec diamond with doubly periodic weightings, arXiv:1911.01250.
	
\bibitem{BeD} T. Berggren and M. Duits, Correlation functions for determinantal processes defined by infinite block Toeplitz minors, \textit{Adv. Math.} \textbf{356} (2019), 106766.
	
\bibitem{BK} P.M. Bleher and A.B.J. Kuijlaars, Large $n$ limit of Gaussian random matrices with external source III, double scaling limit, \emph{Comm. Math. Phys.} {\bf 270} (2007), 481--517.
	
\bibitem{Borodin} A. Borodin, \textit{Determinantal point processes}, The Oxford handbook of random matrix theory, 231--249, Oxford Univ. Press, Oxford, 2011.
	
%\bibitem{BD} A. Borodin and M. Duits, Limits of determinantal processes near a tacnode, \emph{Ann. Inst. Henri Poincar\'e} (B) {\bf 47} (2011), 243--258.
	
\bibitem{BF} A. Borodin and P. L. Ferrari, Random tilings and Markov chains for interlacing particles, \emph{Markov Process. Related Fields} {\bf 24} (2018), 419--451.
	
\bibitem{BG} A. Borodin and V. Gorin, Lectures on integrable probability, In: Probability and Statistical Physics in St. Petersburg (V. Sidoravicius and S. Smirnov, eds.), \emph{Proc. Sympos. Pure Math.} {\bf 91}, Amer. Math. Soc., Providence, RI, 2016, pp. 155--214.
	
%\bibitem{BGG} A. Borodin, V. Gorin, and A. Guionnet, Gaussian asymptotics of discrete $\beta$-ensembles, \emph{Publ. Math. Inst. Hautes \'Etudes Sci.} {\bf 125} (2017), 1--78.
	
\bibitem{BGR} A. Borodin, V. Gorin, and E. M. Rains, $q$-distributions on boxed plane partitions, \emph{Selecta Math. (N.S.)} {\bf 16} (2010), 731--789. 
		
%\bibitem{BOO} A. Borodin, A. Okounkov, and G. Olshanski, On asymptotics of the Plancherel measures for symmetric groups,  \emph{J. Amer. Math. Soc.}  {\bf 13} (2000), 481--515.
	
\bibitem{BO} A. Borodin and G. Olshanski, Asymptotics of Plancherel-type random partitions, \emph{J. Algebra} {\bf 313}, (2007), 40--60. 
	
%\bibitem{BMRT} C. Boutillier, S. Mkrtchyan, N. Reshetikhin, and P. Tingley, Random skew plane partitions with a piecewise periodic back wall, \emph{Ann. Henri Poincar\'{e}} {\bf 13} (2012), 271--296. 
	
%\bibitem{BreuerD} J. Breuer and M. Duits, Central limit theorems for biorthogonal ensembles and asymptotics of recurrence coefficients, \emph{J. Amer. Math. Soc.}  {\bf 30} (2017), 27--66. 
	
%\bibitem{BuGo1}	A. Bufetov and V. Gorin, Fluctuations of particle systems determined by Schur generating functions, \emph{Adv. Math} {\bf 338} (2018), 702--781.
	
\bibitem{BuGo2} A. Bufetov and V. Gorin, Fourier transform on high-dimensional unitary groups with applications to random tilings, to appear in \textit{Duke Math. J}.
	
%\bibitem{BuKn} A. Bufetov and A. Knizel, Asymptotics of random domino tilings of rectangular Aztec diamonds, \emph{Ann. Inst. Henri Poincar\'e Probab. Statist.} {\bf 54} (2018), 1250--1290.
	
\bibitem{CassaManas2012} G. A. Cassatella-Contra and M. Ma\~{n}as, Riemann-Hilbert problems, matrix orthogonal polynomials and discrete matrix equations with singularity confinement, \textit{Stud. Appl. Math.}, \textbf{128} (2012), 252--274.
	
\bibitem{CDKL} C. Charlier, M. Duits, A.B.J. Kuijlaars and J. Lenells, A periodic hexagon tiling model and non-Hermitian orthogonal polynomials, to appear in \textit{Comm. Math. Phys.}.
	
\bibitem{CY2014} S. Chhita and B. Young, Coupling functions for domino tilings of Aztec diamonds, \textit{Adv. Math.} \textbf{259} (2014), 173--251.
	
\bibitem{CJ} S. Chhita and K. Johansson, Domino statistics of the two-periodic Aztec diamond, \emph{Adv. Math.} {\bf 294} (2016), 37--149. 
	
%\bibitem{CK} T. Claeys and A.B.J. Kuijlaars, Universality of the double scaling limit in random matrix models, \emph{Comm. Pure Appl. Math.} {\bf 59} (2006), 1573--1603.
	
\bibitem{CEP} H. Cohn, N. Elkies, and J. Propp, Local statistics for random domino tilings of the Aztec diamond, \emph{Duke Math. J.} {\bf 85} (1996), 117--166. 
	
\bibitem{CKP} H. Cohn, R. Kenyon, and J. Propp, A variational principle for domino tilings, \emph{J. Amer. Math. Soc.}  {\bf 13} (2000), 481--515.
	
\bibitem{CLP} H. Cohn, M. Larsen, and J. Propp, The shape of a typical boxed plane partition, \emph{New York J. Math.} {\bf 4} (1998), 137--165.
	
\bibitem{Deift} P. Deift, Orthogonal Polynomials and Random Matrices: A Riemann–Hilbert Approach, \emph{Courant Lecture Notes} {\bf 3}, New York University, New York, 1999.
	
\bibitem{DKMVZ1}
P. Deift, T. Kriecherbauer, K.T-R McLaughlin, S. Venakides and X. Zhou, Strong asymptotics of orthogonal polynomials with respect to exponential weights, {\em Comm. Pure Appl. Math.} {\bf 52} (1999), 1491--1552.
	
\bibitem{DKMVZ} P. Deift, T. Kriecherbauer, K.T-R. McLaughlin, S. Venakides, and X. Zhou, Uniform asymptotics for polynomials orthogonal with respect to varying exponential weights and applications to universality questions in random matrix theory, \emph{Comm. Pure Appl. Math.} {\bf 52} (1999), 1335--1425. 
	
\bibitem{DZ} P. Deift and X.Zhou, A steepest descent method for oscillatory Riemman-Hilbert problems; asymptotics for the MKdV Equation, \emph{Ann. Math.} {\bf 137} (1993), 295--368.  
	
\bibitem{Delvaux} S. Delvaux, Average characteristic polynomials for multiple orthogonal polynomial ensembles, \textit{J. Approx. Theory} \textbf{162} (2010), 1033--1067. 
	
%\bibitem{DKZ} S. Delvaux, A.B.J. Kuijlaars, and L. Zhang,  Critical behavior of nonintersecting Brownian motions at a tacnode, \emph{Comm. Pure Appl. Math.} {\bf 64} (2011), 1305--1383.
	
%\bibitem{DD} K. Driver and P. Duren, Zeros of the hypergeometric polynomials $F(-n,b; 2b; z)$ for $b < - \tfrac{1}{2}$, \emph{Indag. Math.} {\bf 11}  (2000), 43--51.
	
\bibitem{Duits1} M. Duits, Gaussian free field in an interlacing particle system with two jump rates, \emph{Comm. Pure Appl. Math.} {\bf 66} (2013), 600--643. 
	
%\bibitem{Duits2} M. Duits, On global fluctuations for non-colliding processes, \emph{Ann. Probab.} {\bf 46} (2018), 1279--1350.
	
\bibitem{DK} M. Duits and A.B.J. Kuijlaars, The two periodic Aztec diamond and matrix orthogonal polynomials, to appear in \emph{J. Eur. Math. Soc.}, preprint arXiv:1712.05636.
	
%\bibitem{DM1} E. Duse and A. Metcalfe, Asymptotic geometry of discrete interlaced patterns: Part I, \emph{Internat. J. Math.} {\bf 26} (2015), No. 11, 1550093, 66 pp.
	
%\bibitem{DM2} E. Duse and A. Metcalfe, Asymptotic geometry of discrete interlaced patterns: Part II, to appear in \emph{Ann. Inst. Fourier}, preprint arXiv:1507.00467.
	
%\bibitem{DM3} E. Duse and A. Metcalfe, Universal edge fluctuations of discrete interlaced particle systems, \emph{Ann. Math. Blaise Pascal} {\bf 25} (2018), No. 1, 75--197.
	
\bibitem{EM} B. Eynard and M.L. Mehta, Matrices coupled in a chain I. Eigenvalue correlations, \emph{J. Phys A.} {\bf 31} (1998), 4449--4456.
	
%\bibitem{FV} P. Ferrari and B. Vet\H{o}, Non-colliding Brownian bridges and the asymmetric tacnode process, \emph{Electron. J. Probab.} {\bf 17} (2012), 1--17.
	
\bibitem{FIK} A.S. Fokas, A.R. Its, and A.V. Kitaev, The isomonodromy approach to matrix models in 2D quantum gravity, \emph{Comm. Math. Phys.} {\bf 147} (1992), 395--430. 
	
\bibitem{GV} I. Gessel and G. Viennot, Binomial determinants, paths, and hook length formulae, \emph{Adv. Math.} {\bf 58} (1985), 300--321. 
	
\bibitem{GR} A. Gonchar and E.A. Rakhmanov, Equilibrium distributions and degree of rational approximation of analytic functions, \emph{Math. USSR Sbornik} {\bf 62} (1987), 305--348.
	
\bibitem{Gorin} V. E. Gorin, Nonintersecting paths and the Hahn orthogonal polynomial ensemble, \emph{Funct. Anal. Appl. }{\bf 42} (2008), 180--197.
	
\bibitem{GrunIglesia} F.A. Gr\"{u}nbaum, M.D. de la Iglesia and A. Mart\'{i}nez-Finkelshtein, Properties of matrix orthogonal polynomials via their Riemann-Hilbert characterization, \textit{SIGMA} \textbf{7} (2011), 31 pp.
	
\bibitem{JPS} W. Jockusch, J. Propp and P. Shor, Random domino tilings and the arctic circle theorem (1995), unpublished manuscript available at arXiv:math/9801068.
	
\bibitem{Jptrf} K. Johansson, Non-intersecting paths, random tilings and random matrices, \emph{Probab. Theory and Related Fields} {\bf 123} (2002), 225--280.

\bibitem{J05} K. Johansson, The arctic circle boundary and the Airy process, \emph{Ann. Probab.} {\bf 33} (2005), 1--30.
	
%\bibitem{Jdet} K. Johansson, Random matrices and determinantal processes, in: Mathematical Statistical Physics (A. Bovier et al., eds.),	Elsevier B.V., Amsterdam, 2006, pp. 1--55.
	
%\bibitem{Jtac} K. Johansson, Non-colliding Brownian motions and the extended tacnode process, \emph{Comm. Math. Phys.} {\bf 319} (2013), 231--267. 
	
\bibitem{J17} K. Johansson, Edge fluctuations of limit shapes,  in: Current developments in mathematics 2016 (D. Jerison et al., eds.) Int. Press, Somerville, MA, 2018, pp. 47--110.
	
%\bibitem{JN} K. Johansson and E. Nordenstam, Eigenvalues of GUE minors, \emph{Electron. J. Probab.} {\bf 11} (2006), 1342--1371.
	
\bibitem{KeaSri} D. Keating and A. Sridhar, Random tilings with the GPU, \textit{Journal of Mathematical Physics} \textbf{59} (2018), 094120.
	 
\bibitem{K} R. Kenyon, Lectures on dimers, in: Statistical Mechanics (S. Sheffield and T. Spencer, eds.), Amer. Math. Soc., Providence, RI, 2009, pp. 191–-230.
	
\bibitem{KO} R. Kenyon and A. Okounkov, Limit shapes and the complex Burgers equation, \emph{Acta Math.} {\bf 199} (2007), 263--302. 
	
\bibitem{KOS} R. Kenyon, A. Okounkov and S. Sheffield, Dimers and amoebae, \emph{Ann. of Math. (2)} {\bf 163} (2006), 1019--1056.
	
%\bibitem{KMF} A.B.J. Kuijlaars and A. Mart\'inez-Finkelshtein, Strong asymptotics for Jacobi polynomials with varying nonstandard parameters, \emph{J. Anal. Math.} {\bf 94} (2004), 195--234.
	
\bibitem{KS} A.B.J. Kuijlaars and G.L.F. Silva, $S$-curves in polynomial external fields, \emph{J. Approx. Theory } {\bf 191} (2015), 1--37. 
	
\bibitem{L} B. Lindstr\"om, On the vector representations of induced matroids, \emph{Bull. London Math. Soc.} {\bf 5} (1973), 85--90.
	
%\bibitem{MFMGT} A. Mart\'inez-Finkelshtein, P. Mart\'inez-Gonz\'alez, and F. Thabet, Trajectories of quadratic differentials for Jacobi polynomials with complex parameters, \emph{Comput. Methods Funct. Theory} {\bf 16} (2016), 347--364.
	
%\bibitem{MFO} A. Mart\'inez-Finkelshtein and R. Orive, Riemann-Hilbert analysis of Jacobi polynomials orthogonal on a single contour, \emph{J. Approx. Theory}  134 (2005), 137--170.
	
\bibitem{MFR} A. Mart\'inez-Finkelshtein and E.A. Rakhmanov,  Critical measures, quadratric differentials, and weak limits of zeros of Stieltjes polynomials, \emph{Comm. Math. Phys.} {\bf 302} (2011), 53--111.
	
\bibitem{MFR2} A. Mart\'inez-Finkelshtein and E.A. Rakhmanov, Do orthogonal polynomials dream of symmetric curves? \emph{Found. Comput. Math.} {\bf 16} (2016), 1697--1736.
	
%\bibitem{M1} S. Mkrtchyan, Plane partitions with two-periodic weights, \emph{Lett. Math. Phys.} {\bf 104} (2014), 1053--1078. 
	
%\bibitem{M2} S. Mkrtchyan, Scaling limits of random skew plane partitions with arbitrary sloped back walls, \emph{Comm. Math. Phys.} {\bf 305} (2011), 711--739.
	
%\bibitem{Ok1} A. Okounkov, Infinite wedge and random partitions,  \emph{Selecta Math. (N.S.)} {\bf 7} (2001), 57--81.
	
\bibitem{Ok2} A. Okounkov, Symmetric functions and random partitions, in: Symmetric Functions 2001: Surveys of Developments and Perspectives (S. Fomin ed.), Kluwer Academic Publishers, Dordrecht, 2002, pp. 223--252.
	
\bibitem{OR1} A. Okounkov and N. Reshetikhin, Correlation function of Schur process with application to local geometry of a random $3$-dimensional Young diagram, \emph{J. Amer. Math. Soc.} {\bf 16} (2003), 581--603. 
	
%\bibitem{OR2} A. Okounkov and N. Reshetikhin, Random skew plane partitions and the Pearcey process, \emph{Comm. Math. Phys.} {\bf 269} (2007), 571--609.
	
\bibitem{Petrov1} L. Petrov, Asymptotics of random lozenge tilings via Gelfand-Tsetlin schemes, \emph{Probab. Theory Related Fields} {\bf 160} (2014), 429--487.
	
\bibitem{Petrov2} L. Petrov, Asymptotics of uniformly random lozenge tilings of polygons, Gaussian free field, \emph{Ann. Probab.} {\bf 43} (2015), 1--43.
	
%\bibitem{Pommerenke} C. Pommerenke, \textit{Univalent Functions}, Vandenhoeck \& Ruprecht, G\"{o}ttingen, 1975.
	
\bibitem{ProppShuffling} J. Propp, Generalized domino-shuffling, \textit{Theoret. Comput. Sci.} \textbf{303} (2003), 267--301.

\bibitem{Rak} E.A. Rakhmanov, Orthogonal polynomials and S-curves, in: Recent Advances in Orthogonal Polynomials, Special Functions, and their Applications (J. Arves\'u and G. L\'opez Lagomasino, eds.) Contemp. Math. 578, Amer. Math. Soc., Providence, RI, 2012, pp. 195--239.
	
\bibitem{Soshnikov} A. Soshnikov, Determinantal random point fields, \textit{ Russian Math. Surveys} \textbf{55} (2000), 923--975.
	
\bibitem{Stahl} H. Stahl, Orthogonal polynomials with complex-valued 	weight function. I, II. \emph{Constr. Approx.} {\bf 2} (1986), 225--240, 241--251.
	
\bibitem{ST} E.B. Saff and V. Totik, Logarithmic Potentials with External Fields, Springer Verlag, Berlin, 1997.
	
%\bibitem{TW} C. Tracy and H. Widom, The Pearcey process,  \emph{Comm. Math. Phys.} {\bf 263} (2006), 381--400.

\end{thebibliography}
\end{document}